\pdfoutput=1
\documentclass[a4paper,onecolumn,oneside,12pt]{mwrep}
\usepackage[titletoc]{appendix}

\usepackage{times}
\usepackage[utf8x]{inputenc}
\usepackage[T1]{fontenc}
\usepackage{setspace}

\usepackage[english]{babel}

\usepackage{mathrsfs}
\usepackage{amssymb}
\usepackage{amsmath}
\usepackage{amsthm}
\usepackage{mathabx}
\usepackage[ruled]{algorithm2e}
\usepackage{longtable}
\usepackage{natbib}
\usepackage{tikz,pgfplots}
\usetikzlibrary{positioning, arrows, shapes, calc}

\usepackage{hyperref}
\hypersetup{colorlinks=true,allcolors=black}

\usepackage{supertabular}
\usepackage{multirow}
\usepackage{hhline}
\usepackage{enumitem}
\usepackage{caption}

\usepackage{xcolor}
\usepackage{multirow}
\usepackage{hhline}

\usepackage{graphics}

\usepackage{enumitem}
\usepackage{booktabs}

\usepackage{xcolor}

\newcommand{\highlight}[1]{%
  \colorbox{black!20}{$\displaystyle#1$}}

\newtheoremstyle{phdStyle}
  {\topsep}   % ABOVESPACE
  {\topsep}   % BELOWSPACE
  {\itshape}  % BODYFONT
  {0pt}       % INDENT (empty value is the same as 0pt)
  {\bfseries} % HEADFONT
  {.}         % HEADPUNCT
  {5pt plus 1pt minus 1pt} % HEADSPACE
  {\thmname{#1}\thmnumber{ #2}\thmnote{ (#3)}}          % CUSTOM-HEAD-SPEC

\theoremstyle{phdStyle}

\newtheorem{theorem}{Theorem}
\newtheorem{example}{Example}
\newtheorem{proposition}{Proposition}
\newtheorem{definition}{Definition}  %\begin{definition}{(\textbf{Representation function})}
\newtheorem{corollary}{Corollary}
\newtheorem{property}{Property}
\newtheorem{problem}{Problem}

\newenvironment{dedication}
  {\clearpage           % we want a new page
   \thispagestyle{empty}% no header and footer
   \vspace*{\stretch{4}}% some space at the top 
   \itshape             % the text is in italics
   \raggedleft          % flush to the right margin
  }
  {\par % end the paragraph
   \vspace{\stretch{3}} % space at bottom is three times that at the top
   \clearpage           % finish off the page
  }

\def\specialsubseteq{\widearrow{\subseteq}}
\def\specialsubset{\widearrow{\subset}}
\def\SV{\mathit{SV}}
\def\NR{\mathit{NR}}
\def\SB{\mathit{SB}}
\def\SEMI{\mathit{SEMI}}
\def\OV{\mathit{OV}}
\def\CSEMI{\mathit{CSEMI}}
\def\R{\mathbb {R}}

\def\T{\mathcal{T}}

\def\mC{\mathcal{C}}

\def\cutoff{\mathit{cutoff}}

\newcommand{\Comment}[1]{{\scriptsize{\tcp*[h]{#1}}}}
\newcommand{\set}[1]{\{#1\}}
\newcommand{\midd}{\mathrel{:}}

\newcommand{\MC}[2]{\text{{\upshape MC}}({#1},{#2})}
\newcommand{\MCNR}[2]{\text{{\upshape MC}}^{\mathit{NR}}({#1},{#2})}
\newcommand{\MCSB}[2]{\text{{\upshape MC}}^{\mathit{SB}}({#1},{#2})}

 \DeclareMathOperator*{\argmax}{arg\,max}

\begin{document}

\begin{titlepage}
	\begin{center}
	\thispagestyle{empty}
	\fontsize{28pt}{34pt}\selectfont
	WARSAW UNIVERSITY OF TECHNOLOGY \\
	\vspace*{.6\baselineskip}
	\fontseries{b}\fontsize{24pt}{18pt}\selectfont
	Faculty of Electronics \\and Information Technology \\
	\vspace*{5\baselineskip}
	\fontseries{m}\fontsize{26pt}{1pt}\selectfont
	DOCTORAL DISSERTATION \\	
	\vspace*{5\baselineskip}
	\vspace*{1.15\baselineskip}
	\fontsize{20pt}{15pt}\selectfont
	mgr Piotr Szczepa\'{n}ski\\
	\vspace*{1.15\baselineskip}
	\fontseries{b}\fontsize{15pt}{18pt}\selectfont
	Fast Algorithms for Game-Theoretic Centrality Measures \\
	\end{center}

	\vspace*{6.5\baselineskip}
	\begin{flushright}
	\fontseries{m}\fontsize{13pt}{10pt}\selectfont
	Supervisor\\
	Prof. Mieczys{\l}aw Muraszkiewicz\\
	\vspace*{1\baselineskip}
	Secondary Supervisor\\
	PhD. Tomasz Michalak\\
	\end{flushright}

	\vspace*{5\baselineskip}
	\begin{center}
	Warsaw 2015
	\end{center}
\end{titlepage}

\begin{dedication}
\setcounter{page}{2}
\pagebreak
\vspace*{4cm}
   \thispagestyle{empty}% no header and footer
For my wonderful wife, Paula.
\pagebreak
   \thispagestyle{empty}% no header and footer
\end{dedication}

\null
   \thispagestyle{empty}% no header and footer
\newpage
\begin{center}
\textbf{Abstract}
\end{center}
In this dissertation, we analyze the computational properties of game-theoretic centrality measures. The key idea behind game-theoretic approach to network analysis is to treat nodes as players in a cooperative game, where the value of each coalition of nodes is determined by certain graph properties. Next, the centrality of any individual node is determined by a chosen game-theoretic solution concept (notably, the Shapley value) in the same way as the payoff of a player in a cooperative game. On one hand, the advantage of  game-theoretic centrality measures is that nodes are ranked not only according to their individual roles but also according to how they contribute to the role played by all possible subsets of nodes. On the other hand, the disadvantage is that the game-theoretic solution concepts are typically computationally challenging. The main contribution of this dissertation is that we show that a wide variety of game-theoretic solution concepts on networks can be computed in polynomial time.

Our focus is on centralities based on the Shapley value and its various extensions, such as the Semivalues and Coalitional Semivalues. Furthermore, we prove \#P-hardness of computing the Shapley value in connectivity games and propose an algorithm to compute it. Finally, we analyse computational properties of generalized version of cooperative games in which order of player matters. We propose a new representation for such games, called generalized marginal contribution networks, that allows for polynomial computation in the size of the representation of two dedicated extensions of the Shapley value to this class of games.

\vspace*{\baselineskip}

\noindent\textbf{Keywords:} \textit{the Shapley value, the Owen value, Semivalues, social networks, centrality measures}
\newpage

\begin{center}
\textbf{Streszczenie}
\end{center}
	\thispagestyle{plain}
W niniejszej rozprawie autor porusza problem z{\l}o\.{z}ono\'{s}ci obliczeniowej teoriogrowych centralno\'{s}ci. W teoriogrowym podej\'{s}ciu do analizy sieci traktujemy wierzcho{\l}ki jako graczy w koalicyjnej grze, w kt\'{o}rej warto\'{s}\'{c} koalicji wierzcho{\l}k\'{o}w wynika ze struktury sieci. W takiej grze wa\.{z}no\'{s}\'{c} ka\.{z}dego wierzcho{\l}ka mo\.{z}e by\'{c} okre\'{s}lona przez dowolne rozwi\k{a}zanie gry koalicyjnej (w szczeg\'{o}lno\'{s}ci przez warto\'{s}\'{c} Shapleya). Z jednej strony zalet\k{a} takiego podej\'{s}cia do centralno\'{s}ci jest fakt, \.{z}e ranking wierzcho{\l}k\'{o}w wynika nie tylko z indywidualnej roli ka\.{z}dego wierzcho{\l}ka, lecz także z tego, ile dany wierzcho{\l}ek wnosi do roli ka\.{z}dego podzbioru wierzcho{\l}k\'{o}w w grafie. Jednak\.{z}e z drugiej strony liczenie rozwi\k{a}za\'{n} gier koalicyjnych jest bardzo trudne. G{\l}\'{o}wn\k{a} kontrybucj\k{a} tej rozprawy jest pokazanie, jak liczy\'{c} w czasie wielomianowym teoriogrowe centralno\'{s}ci dla wielu r\'{o}\.{z}nych rozwi\k{a}za\'{n} gier koalicyjnych.

W tej pracy autor skupia si\k{e} na szybkim liczeniu teoriogrowych centralno\'{s}ci opartych na warto\'{s}ci Shapleya i jej r\'{o}\.{z}nych rozszerze\'{n}. W szczeg\'{o}lno\'{s}ci analizuje on obliczeniowe w{\l}asno\'{s}ci P\'{o}{\l}warto\'{s}ci, b\k{e}d\k{a}ce uog\'{o}lnieniem warto\'{s}ci Shapleya, i koalicyje P\'{o}{\l}warto\'{s}ci, b\k{e}d\k{a}ce uog\'{o}lnieniem warto\'{s}ci Owena. Dodaktowo, autor dowodzi \#P-trudno\'{s}ci liczenia warto\'{s}ci Shapleya w grach sp\'{o}jno\'{s}ciowych i proponuje szybki algorytm radz\k{a}cy sobie z tym problemem. Na zako\'{n}czenie, w pracy analizowane są uog\'{o}lnione gry koalicyjne, w kt\'{o}rych koleno\'{s}\'{c} graczy wyznacza warto\'{s}\'{c} koalicji. Autor proponuje now\k{a}, zwi\k{e}z\l\k{a} reprezentacj\k{e} tych gier, kt\'{o}ra pozwala na obliczanie w czasie wielomianowym ze wzgl\k{e}du na wielko\'{s}\'{c} reprezentacji dw\'{o}ch dedytkowanych tym grom rozszerze\'{n} warto\'{s}ci Shapley.

\vspace*{\baselineskip}

\noindent\textbf{S{\l}owa kluczowe:} \textit{warto\'{s}\'{c} Shapleya, warto\'{s}\'{c} Owena, P\'{o}{\l}warto\'{s}ci, sieci spo{\l}eczno\'{s}ciowe, miary centralno\'{s}ci}

\vspace*{2\baselineskip}

\tableofcontents
	\thispagestyle{plain}

\clearpage
\null	\thispagestyle{empty}
\newpage

\chapter{Introduction}

\noindent \textit{The Shapley} value is arguably the most well-known normative payoff division scheme  (solution concept) in cooperative game theory \citep{Shapley:1953}. One of its interesting applications is to measure importance (or centrality) of nodes in networks. Unfortunately, although such game-theoretic approach to centrality is often more advantageous than the classical measures, the Shapley value as well as other, related solution concepts from cooperative game-theory are computationally challenging.  In particular, they typically require to consider all marginal contributions that $n$ players in the cooperative game could make to all $2^n$ coalitions (subsets of players). Clearly, applying such a direct approach to compute game-theoretic solution concepts on networks would be prohibitive even for very small systems. 

The main thesis of this dissertation is that it is possible to compute various measures based on game-theoretic solution concepts in polynomial time. This includes, among others, game-theoretic extensions of three most well-known classical centrality measures: degree, closeness and betweenness. Polynomial running time in our algorithms is achieved by probabilistic analysis of marginal contributions of players to coalitions and taking advantage of the network topology.

Networks underpin the number of real-life domains of our everyday live: from the quite obvious examples of computer networks, communications networks, or road networks to the more hidden networks like protein networks, internet network, or social networks. The studies of the structure of real-life networks were undertaken by many scientists resulting with a number of astonishing breakthroughs. Discovered and explained by \cite{Albert:et:al:2000} the scale-free nature of many real-life networks were undoubtedly one of those. The other was done by \cite{Watts:Strogatz:1998} who proposed the model inspired by the network phenomenon of the \emph{Six Degrees of Separation}. Another significant breakthrough was done much earlier by \cite{Freeman:1979} who provided the conceptual classification of centrality measures---functions evaluating the importance of nodes in the network. More specifically, Freeman distinguished three basic concepts upon which various centrality measures are built: the distance between nodes (\emph{closeness centralities}), the number of shortest paths between nodes (\emph{betweenness centralities}), and direct connections to other nodes (\emph{degree centralities}). %These concepts can be used to evaluate individual nodes, as well as the group of nodes. 

Due to a wide variety of applications of centrality measures, analyzing and developing them is one of the major research topics in the network science. For instance, identifying key nodes in network can be used to study network vulnerability \citep{Holme:et:al:2002}, to build recommendation systems \citep{Liu:et:al:2013}, to identify the focal hubs in a road network \citep{Schultes:Sanders:2007}, to point the most critical functional entities in a protein network \citep{Jeong:et:al:2001}, or to find the most influential people in a social network \citep{Kempe:et:al:2003}.

The common feature of the standard centrality measures introduced by \cite{Freeman:1979} is that they assess the importance of a node by focusing only on the role that a node plays by itself. However, in many applications such an approach is inadequate due to existence \emph{synergies} that may occur if the functioning of nodes is considered in groups.\footnote{\footnotesize Intuitively, synergy can be fought of as a value added from group performance. Note that synergy can be also negative and, in such a case, it is called antergy. For an overview of various concepts of synergy see the work by \cite{Rahwan:et:al:2014}.} In order to capture such synergies \cite{Everett:Borgatti:1999} introduced the concept of group centrality. Its idea is broadly the same as the one of standard centrality, but now the focus is on the functioning of a given group of nodes, rather than individual nodes. In particular, \cite{Everett:Borgatti:1999} extended three Freeman's measures to \emph{group degree centrality}, \emph{group closeness centrality}, and \emph{group betweenness centrality}.

However, while the concept of group centrality addresses the issue of quantifying synergy among nodes in a particular group, there is still another fundamental problem to be solved. In particular, it is unclear how to rank individual nodes given an exponential number of potential groups of nodes they may belong to. In other words, we need to answer the question: \emph{how to rank individual nodes based on their group centralities?} Here, the coalitional game theory comes into play.

In particular, the coalitional (or cooperative) game theory is a part of game theory in which individual players are allowed to form coalitions with an aim to increase their profits. Now, assuming that players have all decided to cooperate (i.e. to form the \textit{grand coalition}), one of the fundamental problems in coalitional game theory is \emph{how to divide the payoff achieved by cooperation?} The most popular answer to this question was offered by L.S. Shapley \citep{Shapley:1953}---the 2012 Noble Prize Laureate---who proposed to consider \textit{marginal contributions of players to all coalitions they could potential belong to}. He proved that there exists the unique payoff division scheme, now called the \textit{Shapley value}, that satisfies the following four fairness axioms: \emph{Efficiency}---the whole available payoff is distributed among players; the \emph{Null Player}---the player that cannot contribute anything should receive zero; \emph{Symmetry}---two players whose contributions to any coalition are always the same should be given the same payoff; and \emph{Additivity}---the payoff division scheme should be additive.
            
The alternative axiomatizations of the Shapley value have been studied by numerous authors in the literature. Unfortunately, while this value has many interesting properties, computing the Sahpley value is often \#P-complete \citep{Deng:Papdimitriou:1994}. This obstacle will be overcome in this thesis in the context of game-theoretic centrality.

Now we have all necessary tools to introduce game-theoretic centrality measures. In a nutshell, the idea behind them is to treat nodes as players in a coalitional game, where the value of each coalition of nodes is determined by a group centrality. In such settings, the value of each individual node can be determined by cooperative game solution such as the Shapley value. In other words, in order to \emph{provide the ranking of individual nodes that accounts for group centrality measure}, we use the Shapley value that evaluates (fairly in a certain sense) marginal contributions of each node to all groups it could potentially belong to.
The key advantage of such an approach is that nodes are ranked not only according to their individual roles in the network but also according to how they contribute to the role played by all possible subsets of nodes. 

Unfortunately, as already mentioned, potential downside of game-theoretic solution concepts is that they are, in general, computationally challenging. However, in this dissertation, we show that this is not always true in the network context, i.e., we are able to compute the Shapley value and the related solution concepts in polynomial time for various centrality-related games on networks.  

Our contributions in this dissertation can be summarized as follows. Firstly, we compute on networks in polynomial time the Semivalues \citep{Dubey:et:al:1981} which are the generalization of the Shapley value and offer more flexibility to define any particular game-theoretic network centrality. Secondly, we propose polynomial-time algorithm on networks for the \emph{Owen value} \citep{Owen:1977} that is the most important solution concept to games with \emph{coalitional structure}, and the \emph{Coalitional Semivalue} our extension of the Owen value. Finally, we propose the new representation of \emph{generalized coalitional games}, the games in which the permutations of players are considered. Our representation allows to compute Nowak-Radzik value \citep{Nowak:Radzik:1994} and S\'{a}nchez-Berganti\~{n}os value \citep{Sanchez:Bergantinos:1999}. The polynomial time algorithms for the last two values are still unknown, but we propose algorithms that significantly reduce computational complexity of the problem. We also show how to use our representation on networks.

\section{Publications and Author's Contribution}

\noindent The results covered by this thesis were published in four international journals and presented on the four top conferences from Artificial Intelligence. Additionally, they gathered attention of scientific community what resulted in the number of citations.

\begin{itemize}
\item Chapter~\ref{chap:gtcm} is mostly based on the work \emph{Efficient Computation of the Shapley Value for Game-theoretic Network Centrality} published in \textbf{Journal of Artificial Intelligence Research (JAIR)} \citep*{Michalak:et:al:2013}. The main two contributions presented in this chapter are: developing efficient algorithms for computing Shapley value-based degree and closeness centralities (together with their various extensions), and providing theoretical, as well as experimental analysis of the introduced algorithms.  The contribution of the author of this dissertation mostly covers the experimental part of this article. More specifically,  Algorithms~\ref{algo:gtc:game1}-\ref{algo:gtc:game5} were developed by Karthik Aaditha, and partianlly formalized by the author of this thesis. More specifically, Propositions~\ref{eq:proposition1} and \ref{prop3} and proof of correctness for Algorithms~\ref{algo:gtc:game1} and \ref{algo:gtc:game3} are the contribution of the author. Additionally, Algorithms~\ref{algo:gtc:mc}-\ref{mcg5} (for the Monte Carlo sampling of the game theoretic-centrality measures) were developed by the author of this dissertation.

\item Chapter~\ref{chap:bet} is based on two publications: \emph{A New Approach to Betweenness Centrality Based on the Shapley Value} \citep*{Szczepanski:et:al:2012} published in \textbf{Proceedings of the $\mathbf{11^{th}}$ International Conference on Autonomous Agents and Multiagent Systems (AAMAS 2012)} and \emph{Efficient Algorithms for Game-Theoretic Betweenness Centrality} published in \textbf{Artificial Intelligence Journal (AI)} \citep*{Szczepanski:et:al:2016}. All the technical content, i.e., algorithms, experiments and theorems are exclusive contributions of the author of this thesis.

\item Chapter~\ref{chap:OV} is based on two publications: \emph{A Centrality Measure for Networks With Community Structure Based on a Generalization of the Owen Value} \citep*{Szczepanski:et:al:2014} published in \textbf{Proceedings of the $\mathbf{21^{st}}$ European Conference on Artificial Intelligence (ECAI 2014)} and \emph{A New Approach to Measure Social Capital using Game-Theoretic Techniques} \citep*{Michalak:et:al:2015} published in \textbf{ACM SIGecom Exchanges}. All algorithms, experiments and theorems were developed by the author of this thesis.

\item Chapter~\ref{chap:connectivity} is based on the publication \emph{Computational Analysis of Connectivity Games with Application to Investigation of Terrorist Networks} \citep*{Michalak:et:al:2013b} published in \textbf{Proceedings of the $\mathbf{23^{rd}}$ International Joint Conference on Artificial Intelligence (IJCAI 2013)}. The contribution presented in this chapter is threefold, it formally proves that computing the Shapley value in connectivity games in NP-hard, it develops faster exact algorithm for connectivity games, and it develops one approximation algorithm for these family of games. The complexity proof consists of two theorems: Theorem~\ref{first:theorem} developed by the author with help of Colin McQuillan and Theorem~\ref{second:theorem}, which was developed exclusively by the author of this dissertation. Algorithm~\ref{ter:algo:faster} is a joint work done by the author, Tomasz Michalak and Talal Rahwan, and finally Algorithm~\ref{ter:algo:approx} was developed by the author together with Oskar Skibski.

\item The final Chapter~\ref{chap:mcnets} is based on the second part of the publication \emph{Implementation and Computation of a Value for Generalized Characteristic Function Games} \citep*{Michalak:et:al:2014} published in journal \textbf{ACM Transactions on Economics and Computation (ACM TEAC)}. The technical content in the whole chapter is an exclusive contribution of the author.

\end{itemize}

Some parts of the above chapters and additionally Chapter~\ref{chap:pre} (including definitions and formalization) are also based on the article \emph{Efficient Computation of Semivalues for Game-Theoretic Network Centrality} \citep*{Szczepanski:et:al:2015b} published in \textbf{Proceedings of the $\mathbf{29^{th}}$ AAAI Conference on Artificial Intelligence (AAAI 2015)}.

\chapter{Preliminaries}\label{chap:pre}

\noindent In this chapter we present the basic notation and definitions from both graph theory and game theory that will be used throughout the thesis. 

Firstly, we introduce \emph{coalitional games} (also known as \emph{cooperative games}), where we focus on the key solution concepts and their computational aspects. In particular, we define the \emph{Shapley value}---the most popular normative solution concept to cooperative games. We also introduce a parametrized generalization to the Shapley value called the \emph{Semivalue}. Since the standard model (or representation) of cooperative games, i.e. the \emph{characteristic function}, is computationally challenging, we discuss an alternative representation, called the \emph{Marginal Contribution Networks}. For certain games, this representation allows for computing the Shapley value in polynomial time in the number of agents.

Secondly, we define the basic notation of a graph and introduce elementary concepts related to social network analysis such as: \emph{centralities} and \emph{communities}.

\section{Coalitional Games}\label{sec:cg}

\noindent Game Theory consists of two broad areas: \emph{non-cooperative} (or \emph{strategic}) games and \emph{cooperative} (or \emph{coalitional}) games. The first class of games is one in which players make decisions independently. In the second class, groups of players, called \emph{coalitions}, are allowed to form profitable coalitions \citep{Osborne:Rubinstein:1994}.

\begin{definition}[Coalitional game]\label{def:coalitional:game}
A coalitional game consists of a set of \emph{players} (or \emph{agents}) $A = \{a_1, a_2, \ldots , a_{|A|}\}$, and the \emph{characteristic function} $\nu: 2^A \rightarrow \mathbb{R}$, which assigns to each \emph{coalition} of players~$C\subseteq A$ a real value (or payoff) indicating its performance, where $\nu(\emptyset) = 0$. Thus, the coalitional game $g$ in the characteristic function form is a pair $g=(A,\nu)$.
\end{definition}

In such defined game the number of all coalitions equals $2^{|A|}$---the number of all subsets of~$A$. The set of all coalitional games will be denoted by $\mathscr{V}$, so we can write $(A,\nu) \in \mathscr{V}$, or sometimes we simplify notation and write $\nu \in \mathscr{V}$. The coalition of all players $A$ is called \emph{grand coalition}.

\begin{example}\label{ex:coalitional:game}
Let us consider the game where players collect cards with famous rock stars. Each card has its own value, which could be doubled if only player possesses all cards from the set. The game consists of two sets of cards $ K=\{k_1,k_2\} ~ M=\{m_1,m_2,m_3\}$, and four
players $ A=\{a_1,a_2,a_3,a_4\} $  who have the following cards:
\begin{center}
	\begin{tabular}{ll}
		$a_1$ & $ \longleftrightarrow k_1, m_3$ \\
		$a_2$ & $ \longleftrightarrow k_2$ \\
		$a_3$ & $ \longleftrightarrow m_1$ \\
		$a_4$ & $ \longleftrightarrow m_2, m_3$ \\
	\end{tabular}
\end{center}
Assuming that each card has individually value $1$ the characteristic function describing this game $ g_1=(A,\nu_1) $ is:
\begin{center}
	\begin{tabular}{rcl}
		$\nu_1(\emptyset)$ & $ = $ & $ 0 $ \\
		$\nu_1(\{a_1\}) = \nu_1(\{a_2\}) = \nu_1(\{a_3\}) = \nu_1(\{a_4\})$ & $ = $ & $1 $ \\
		$\nu_1(\{a_2,a_3\}) $ & $ = $ & $ 2 $ \\
		$\nu_1(\{a_1,a_3\}) = \nu_1(\{a_2,a_4\}) $ & $ = $ & $ 3 $ \\
		$\nu_1(\{a_1,a_4\}) $ & $ = $ & $ 4$ \\
		$\nu_1(\{a_1,a_2\}) $ & $ = $ & $ 5$ \\
		$\nu_1(\{a_3,a_4\}) = \nu_1(\{a_1,a_2,a_3\}) $ & $ = $ & $6$ \\
		$\nu_1(\{a_2,a_3,a_4\}) = \nu_1(\{a_1,a_2,a_4\}) $ & $ = $ & $7$ \\
		$\nu_1(\{a_1,a_3,a_4\}) $ & $ = $ & $8$ \\
		$\nu_1(A) $ & $ = $ & $ 11 $ \\
		
	\end{tabular}
\end{center}

Player $a_3$ has one card $m_1$, so individually its value is $\nu_1(\set{a_3})=1$. The player $a_4$ has two cards from the same set $m_2,m_3$ and playing alone he can receive $\nu_1(\set{a_4})=2$. Now, if these two players cooperate, they will have all cards from the same set, what makes them more valuable. Playing together they will receive $\nu_1(\set{a_3, a_4})=6$, so they have a big incentive to cooperate.
\end{example}

Form the above example we see that the assumption $\nu_1(\emptyset) = 0$ is natural, because the empty coalition without players cannot obtain any profits. Additionally, the characteristic function in this example incites players to collaborate, because joining new members to coalition can only improve its performance. This property of characteristic function is called \emph{superadditivity} and may cause the forming of the grand coalition.

\begin{definition}[Superadditivity]
The characteristic function in the game $(A,v)$  is superadditive if only the following expression holds:
\[ \displaystyle\mathop{\mathlarger{\mathlarger{\forall}}}_{\substack{C_1,C_2 \subseteq A \\ C_1 \cap C_2 = \emptyset }}
v(C_1 \cup C_2) \geq v(C_1) + v(C_2) \]
\end{definition}

The property strongly connected to superadditivity is called \emph{monotonicity}.

\begin{definition}[Monotonicity]
The characteristic function in the game $(A,v)$ jest monotonic if only the following expression holds:
\[ \displaystyle\mathop{\mathlarger{\mathlarger{\forall}}}_{C_1 \subseteq C_2 \subseteq A }
v(C_2) \geq v(C_1) \]
\end{definition}

\citet{Caulier:2009} showed that for the games with non-negative characteristic function the superadditivity implies monotonicity. However, in general the implication in neither direction holds.

When the game is superadditive, or monotonic the  \textit{grand coalition}, i.e., the coalition of all the players in the game, has the highest value and, therefore, is formed.  One of the fundamental questions in cooperative game theory is then how to divide the payoff of the grand coalition ($\nu(A)$) among the players and the most popular normative solution to that problem is the Shapley value.

\section{The Shapley Value}\label{sec:SV}

\noindent  In principle, there may be an infinite number of divisions of the value of grand coalition among players, however, we are interested in those that meet certain desirable criteria. In order to present these criteria we introduce the notion of \emph{marginal contribution}.

\begin{definition}[Marginal contribution]\label{def:mc}
The marginal contribution of player $a_i \in A$ to the coalition $C \subseteq A \setminus \set{a_i}$ is:
\begin{equation}
	\MC{C}{a_i} = \nu(C \cup \set{a_i}) - \nu(C)
\end{equation}
\end{definition}

\noindent Now, let us consider four \emph{fairness} properties that a given division scheme should meet. We denote by $\phi_i(A,\nu)$ the payoff received by a player $a_i$ in game $(A,\nu)$.
\begin{center}
\begin{tabular}{ll}
\textbf{Symmetry} &  payoffs do not depend on the players’ names. \\ & That is, $\phi(A,\pi(\nu))=\pi(\phi)(A,v)$ for every game $\nu$ and permutation $\pi \in \Pi(A)$. \\[2ex]
\textbf{Null Player}& players that make no contribution should receive nothing. \\ &  In other words, $\big(\forall C\subseteq A, \MC{C}{a_i}=0 \big) \ \Rightarrow\ \big(\phi_i(A,\nu) = 0\big)$. \\[2ex]
\textbf{Efficiency}& the entire payoff of each coalition should be distributed among its members. \\ &  That is, $\sum_{a_i\in A} \phi_i(A) = \nu(A)$. \\[2ex]
\textbf{Additivity}& given three games, $g_1=(A,\nu_1)$, $g_2=(A,\nu_2)$ and $g_3=(A,\nu_3)$, \\ &  where $\nu_1(C)=\nu_2(C)+\nu_3(C)$ for all $C \subseteq A$, the payoff of a player in $g_1$ \\ &  should be the sum of his payoff in $g_2$ and in $g_3$.
\end{tabular}
\end{center}

\noindent The \emph{Symmetry} property implies that two players that contribute the same value to all coalitions should be given the same payoff. More formally:
\begin{equation}
\big(\forall_{C \subseteq A \setminus \set{a_i,a_j}}, \MC{C}{a_i} = \MC{C}{a_j}\big)\ \Rightarrow\ \big(\phi_i(A,\nu) = \phi_j(A,\nu)\big)
\end{equation}

\noindent The \emph{Efficiency} property assures that the whole available payoff is distributed among players. The \emph{Null Player} property indicates that player that can not contribute anything should receive value $0$. The \emph{Additivity} property says that if players anticipate in two separate games then their importance in the game consisting of these two games, is the sum of the importance in these two subgames. Interestingly, there exist only one division scheme that satisfies the above four criteria---Shapley value.

\begin{example}\label{ex:shapley:apples}
Let us consider the coalitional game $ g_2=(A,v) $, in which three players $ A=\{a_1,a_2,a_3\} $ collect apples from trees. Each player posses a stand allowing him to collect apples from the top of the trees. Additionally, on two stands of the same height it is possible to put third stand. Now, let us introduce the function $ s: A \rightarrow \mathbb{N} $ evaluating the high of each stand that a given player have:
\[
	s(a_i) = \left\{ \begin{array}{l l}
					20 & \text{if $i = 1$} \\
					40 & \text{if $i = 2$} \\
					20 & \text{if $i = 3$} \\
					 \end{array} \right.
\]
The characteristic function $\nu_2$ in this game is defined as follows:
\[
	\nu_2(C) = \left\{ \begin{array}{l l}
					0 & \text{if $C = \emptyset$}  \\
					60*10 & \text{if $C = A$} \\
					\displaystyle\max_{a_i \in C}\{s(a_i)\}*10 & \text{otherwise.} \\
					 \end{array} \right.
\]
The players $a_1$ and $a_3$ contributes the same value to each coalition, so based on the \textbf{Symmetry} property they should have the same payoff. We also know that the \textbf{Efficiency} property implies that the whole gain $600$ should be distributed among all players. Based on these two information we obtain two equations:
\[
	\begin{array}{l l}
		\phi_1 = \phi_2  \\
		\phi_1 + \phi_2 = \phi_3 = 600 \\
	\end{array}
\]
At this point we haven't enough information to solve the above equation. In order to do it, we need to use \textbf{Additivity} and \textbf{Null Player} properties.
\end{example}

The solution of the above game should satisfy the four \emph{fairness} properties. To this end, \citet{Shapley:1953} proposed to evaluate the role of each player in the game proportionally to a weighted average marginal contribution of this player to all possible coalitions. Formally:

\begin{definition}[Shapley value]\label{def:sv}
In a coalitional game $(A,\nu)$ the Shapley value for a player $a_i$ is given by:
\begin{equation}\label{eq:sv:original}
\phi^{\SV}_i(A,\nu) = \sum_{C \subseteq A \setminus \{a_i\}} \frac{ |C|!(|A| - |C| - 1)!}{A!} (v(C \cup \{a_i\}) - v(C))
\end{equation}
\end{definition}

\noindent Importantly, the Noble prize winner L. S. Shapley \citep{Shapley:1953} showed that the vector $\phi^{\SV}$ is the only solution concept satisfying properties $ (1) - (4) $.

Furthermore, if the characteristic function is superadditive than Shapley value satisfies one additional property:

\begin{center}
\begin{tabular}{ll}
\textbf{Individual Rationality:} & $ \displaystyle\sum_{a_i \in A} \phi^{\SV}_i(\nu) \geq \nu(\{a_i\})  \  ~~~~~~~~~~~~~~~~~~~~~~~~~~~~~~~~~~~~~~~~~~~~~~~~~~~~~~~~~~~~~~~~~~~~~~~~~~~~~~~~~~~~~~~~~~ $
\end{tabular}
\end{center}

\emph{Individual Rationality} indicates that each player has no incentive to play alone. If he do so, he will receive the payoff no greater than the payoff obtained during his collaboration with others.

The set of all players $A$ will be sometimes considered as a ordered list ($|A|$-tuple) $ A = (a_1, a_2, \ldots, a_{|A|}) $ in order to use the notion of the permutation of all players $ \pi \in \Pi(A) $.

The important fact about the Shapley value is that for a given player $ a_i \in A $ in the game $(A,\nu)$ it is the expected marginal contribution of this player to the set of players $P_{\pi}(a_i)$ that precede $a_i$ in a random permutation $ \pi \in \Pi(A)$ of all players in the game.

\begin{definition}[Shapley value as an expected value]\label{def:sv:expected}
In a coalitional game $(A,\nu)$ the Shapley value for a player $a_i$ is given by:
\begin{equation}\label{eq:sv:expected}
\phi^{\SV}_i(A,\nu) = \mathbb{E}[\MC{P_{\pi}(a_i)}{a_i}] =\frac{1}{|A|!}\sum_{\pi \in \Pi(A)}{\MC{P_{\pi}(a_i)}{a_i}}
\end{equation}
\end{definition}

The equivalence between Definition~\ref{def:sv} and Definition~\ref{def:sv:expected} can be easily derived by using a simple combinatorial fact: the number of permutations of $A$ in which the set of players occurring before $i$-th player equals $P_{\pi}(a_i)$ is exactly  $ |P_{\pi}(a_i)|!(|A| - 1 - |P_{\pi}(a_i)|)! $.
\[
\begin{array}{rcl}

 \phi^{\SV}_i(v) & \hspace{-0.1cm}=\hspace{-0.1cm} & \frac{1}{|A|!} \Big( \MC{P_{\pi_1}(a_i)}{a_i} + \MC{P_{\pi_2}(a_i)}{a_i} + \ldots + \MC{P_{\pi_{n!}}(a_i)}{a_i} \Big) \\
		   & \hspace{-0.1cm}=\hspace{-0.1cm} & \frac{1}{|A|!} \Big( \overbrace{(\MC{\emptyset}{a_i} + \ldots + \MC{\emptyset}{a_i} )}^{(|\emptyset|!)(|A|-1)!}  + \ldots +
		   							\overbrace{(\MC{C}{a_i} + \ldots + \MC{C}{a_i} ) }^{ |C|!(|A|-1-|C|)!} + \ldots  \\
		   &   &	 \multicolumn{1}{r}{ \ldots +\overbrace{(\MC{A\setminus \set{a_i}}{a_i} +  \ldots +\MC{A\setminus \set{a_i}}{a_i}  )}^{(|A|-1)!(|A|-1 - (|A| -1))!} \Big) }\\
		   & \hspace{-0.1cm}=\hspace{-0.1cm} & \frac{1}{|A|!} \Big( (|A|-1)!(\MC{\emptyset}{a_i})  + \ldots +
		   							|C|!(|A|-1-|C|)!(\MC{C}{a_i}) + \ldots  \\
		   &   &	 \multicolumn{1}{r}{ \ldots +(|A|-1)!(\MC{A\setminus \set{a_i}}{a_i} )\Big) }\\
		   &\hspace{-0.1cm} =\hspace{-0.1cm} & \frac{1}{|A|!}  \displaystyle\sum_{C \subseteq A \setminus \{a_i\}} |C|!(|A| - |C| - 1 )! ~ \MC{C}{a_i} \\
\end{array}
\]

Using the Definition~\ref{def:sv:expected} we can compute the Shapley value for the game in Example~\ref{ex:shapley:apples}.

\begin{example}\label{ex:shapley:compute}
We will consider all permutations of the set $A$ from the game $g_2$ defined in Example~\ref{ex:shapley:apples}. For each permutation we will compute the marginal contribution of the agent $a_2$ to the set $P_{\pi}(a_2)$.
\[
\begin{array}{l| l}
	a_1,a_2,a_3 & \MC{\set{a_1}}{a_2} = \nu_2(\{a_1,a_2\}) - \nu_2(\{a_1\}) = 200\\
	a_1,a_3,a_2 & \MC{\set{a_1,a_3}}{a_2} = 400\\
	a_2,a_1,a_3 & \MC{\emptyset}{a_2}  = 400\\
	a_2,a_3,a_1 & \MC{\emptyset}{a_2} = 400\\
	a_3,a_1,a_2 & \MC{\set{a_1,a_3}}{a_2} = 400\\
	a_3,a_2,a_1 & \MC{\set{a_3}}{a_2} = 200\\
\end{array}
\]
Thus, for the game $g_2$ we have:
\[
\begin{array}{r l}
\phi^{\SV}_2 = &\frac{200+400+400+400+400+200}{6!} = 333 \frac{1}{3} \\
\phi^{\SV}_1 = &133 \frac{1}{3} \\
\phi^{\SV}_3 = &133 \frac{1}{3} \\
\end{array}
\]
\end{example}

Generally, using equation \eqref{eq:sv:original} computing Shapley value requires iterating through $ 2^{|A|}$ subsets of $A$, and using equation \eqref{eq:sv:expected} requires iterating through all $|A|!$ permutations of $A$. 
Furthermore, computing the Shapley value has been shown to be NP-Hard (or even worse, \#P-Complete) for many specific classes of games \citep{Deng:Papdimitriou:1994,Nagamochi:et:al:1997}.

 Therefore, many new succinctly representations were proposed in the literature. The most popular we introduce in Section~\ref{sec:mcnets} and Section~\ref{sec:read:once:mcnets}.

% \textbf{cos o reprezentacji -- nie ladnie bo taki input.sa ladniejze inputy (section ~~)  more sucictly , UNFORTUNATELY PRZENIESC DO SEKCJI Z REPREZENTACJAMI} Unfortunately, even for some succinctly representable games, computing the Shapley value has been shown to be NP-Hard (or even worse, \#P-Complete) for many domains, including weighted voting games \citep{Deng:Papdimitriou:1994}, threshold network flow games \citep{Bachrach:Rosenschein:2009} and minimum spanning tree games \citep{Nagamochi:et:al:1997}.

\section{Semivalues}\label{sec:semivalue}

\noindent The Shapley value is designed to divide payoff, which is equivalent to evaluating the power of agents in cooperative game. However, if we focus only on the power index of players, we can drop the Efficiency axiom. \textit{Semivalues} introduced by \cite{Dubey:et:al:1981} represent an important class of power indexes, among which only Shapley value satisfies Efficiency. To define Semivalues, we will use the notion of marginal contribution of the player~$i$ to the coalition~$C$ introduced in Definition~\ref{def:mc}. Let $\beta: \{0,1,\ldots,|A|-1\} \to [0,1]$ be a function such that $\sum_{k=0}^{|N|-1}\beta(k) = 1$. Intuitively,~when we calculate the expected marginal contribution of a node, $\beta(k)$ will be the probability that a coalition of size~$k$ is chosen for this node to join. This is why $\beta(k)$ is defined on values ranging from $0$ to $|N|-1$. Note that the function $\beta$ is a discrete probability distribution. Since a player $a_i$ cannot join a coalition that it is already in, we only need to look at coalitions not containing $a_i$.

\begin{definition}[Semivalue]\label{def:semivalue}
Given~$\beta$ and a coalitional game $(A,\nu)$ the Semivalue for a player $a_i$ is given by:
\begin{equation}\label{eq:semi}
\phi^{\SEMI}_i(A,\nu)  = \hspace{-0.2cm} \sum_{0 \leq k < |V|} \beta(k) \mathbb{E}_{C_k}[\text{MC}(C_k,i)],
\end{equation}
where $C_k$ is the random variable of all possible coalitions of size $k$ drawn with uniform probability form the set $A \setminus \{a_i\}$, and $\mathbb{E}_{C_k}[\cdot]$ is the expected value operator for the random variable~$C_k$.
\end{definition}

The \emph{Shapley value} (Definition~\ref{def:sv}) and the \emph{Banzhaf index} of power~\citep{Banzhaf:1965} are two prominent and well-known examples of Semivalues. They are defined by $\beta$-functions~$\beta^{\mathit{Shapley}}$ and $\beta^{\mathit{Banzhaf}}$, respectively:

\begin{equation}\nonumber
 \beta^{\mathit{Shapley}}(k)=
 \frac{1}{|A|} ~~\mbox{ and }~~
 \beta^{\mathit{Banzhaf}}(k)=\frac{\binom{|A|-1}{k}}{2^{|A|-1}}.
\end{equation}

All Semivalues satisfy \emph{Symmetry}, \emph{Null player} and \emph{Additivity} axioms and except for the Shapley value, all these solutions violate the \emph{Efficiency} axiom, which makes them of limited use as fair cost sharing values. However, they are widely used as indicators of power in cooperative games \citep{Monderer:Samet:2002}.

%\textbf{TODO axioms of Banzhaf value}

\section{Marginal Contribution Networks}\label{sec:mcnets}

\noindent The size of the representation of coalitional games in the characteristic function form is exponential in the number of agents and consequently computing Shapley value for such games is generally hard. In order to overcome this limitation the researchers look for for more concise representations. One of the most interesting and very intuitive techniques for representing coalitional games was introduced by \citet{Ieong:Shoham:2005}. The authors proposed \emph{marginal contribution nets} (\emph{MC-Nets}). This representation has a number of desirable properties: it is fully expressive, concise for many characteristic function games of interest, and facilitates a very efficient technique for computing the Shapley value. More specifically, in such representation we can compute the Shapley value in the linear time in respect to the size of a representation.

Let us introduce the syntax of Boolean formulas:
\[
\begin{array}{rll}
	A := & a_i ~|~ \neg a_i & i \in \{1\ldots n\},\\
	\mathcal{F} := & (\mathcal{F}) ~|~ \mathcal{F} \wedge \mathcal{F} ~|~  \mathcal{F} \vee \mathcal{F} ~|~  \mathcal{F} \oplus \mathcal{F} ~|~ A& ,
\end{array}
\]

\noindent where each literal $a_i$ represents the agent in the game. In other words, this formula is a binary rooted tree with positive or negative literals on leafs and internal nodes labeled with one of the following logical connectives: conjunction ($\wedge$), disjunction ($\vee$) and exclusive disjunction ($\oplus$).

\begin{definition}[General Marginal Contribution Network]\label{def:general:mcnets}
The coalitional game in the form of General Marginal Contribution Network is a pair $ (A,\mathcal{R})$, where $ A $ is a set of players, and $\mathcal{R}$ is the set of rules. Rules are of the form of $\mathcal{F} \to V$, where $\mathcal{F}$ is a Boolean formula.	
\end{definition}

\citet{Ieong:Shoham:2005} obtained positive computational results for Marginal Contribution Network, where literals are connected only by conjunction ($\wedge$).

\begin{definition}[Marginal Contribution Network]\label{def:mcnets}
The coalitional game in the form of Marginal Contribution Network is a pair $ (A,\mathcal{R}^*)$, where $ A $ is a set of players, and $\mathcal{R}^*$ is the set of rules. Rules are of the form of $\mathcal{F}^* \to V$, where $\mathcal{F}^*$ is a conjunction of positive or negative literals.		
\end{definition}

A coalition $C$ is said to \textit{meet} a given formula $\mathcal{F}^*$ if and only if $\mathcal{F}^*$ evaluates to \texttt{true} when the values of all Boolean variables that correspond to players in $C$ are set to \texttt{true}, and the values of all Boolean variables that correspond to players not in $C$ are set to \texttt{false}. We write $C \models \mathcal{F}^*$ to mean that $C$ meets $\mathcal{F}^*$. More formally:
\[
	C \models \mathcal{F}^* \iff \forall_{a_i \in \mathcal{F}^*} a_i \in C ~~and~~  \forall_{\neg a_i \in \mathcal{F}^*} a_i \notin C
\]
In MC-Nets, if coalition $C$ does not meet any rule then its value is $0$. Otherwise, the value of $C$ is the sum of every $V$ from the rules of which the $\mathcal{F}^*$ are met by $C$. More formally:
\[
v(C) = \sum_{\mathcal{R}^* \ni (\mathcal{F}^* \to V) \rightarrow \mathit{V}:~C \models \mathcal{F}^*} \mathit{V}.
\]
Ieong and Shoham showed that, given an MC-Net representation in which rules are made only of conjunctions of positive and/or negative literals, the Shapley value can be computed in time linear in the number of rules. Taking advantage of the additivity property of the Shapley value, every basic rule can be considered separately as a game on its own. Thus, we need to compute the Shapley value for a single rule, and based on this we will be able to compute the Shapley value for the entire game.

\begin{theorem}\label{th:mcnets} \textbf{\citep{Ieong:Shoham:2005}}
	In coalitional game $ (A,\mathcal{R}^*)$ in which $\mathcal{R}^*$ consists of a single rule $ \mathcal{F}^* \to V $, where $p$ is the number of positive literals and $n$ is the number of negative ones, the Shapley value for the player $ a_i \in \mathcal{F}^* $ is:
	\[
		\phi^{\SV}_i = \frac{\mathit{V}}{p\binom{p+n}{n}},
	\]
\noindent and for a player $ \neg a_j \in \mathcal{N} $ the Shapley value is:
	\[
		\phi^{\SV}_j = \frac{-\mathit{V}}{n\binom{p+n}{p}}.
	\]
\end{theorem}

\begin{proof}
From the Symmetry axiom we know that all players $ a_i \in \mathcal{F}^* $ have the same Shapley value, because they are indistinguishable. This observation also  applies to all players $ \neg a_j \in \mathcal{F}^* $. In this proof we will use Definition~\ref{def:sv:expected}. We want to count the number of permutations such that for the given player $ a_i \in \mathcal{F}^* $ the $\MC{P_{\pi}(a_i)}{a_i}) \neq 0$. To satisfy this condition all players corresponding to positive literals from $ \mathcal{F}^* $ must appear before $a_i$ and all players corresponding to negative literals from $ \mathcal{F}^* $ after him. Now, we will construct such permutations:
	\begin{itemize}
	\item Let us choose $ p + n $ positions in the sequence of all elements from $A$. We can do this in $ \binom{|A|}{p+n} $ ways.
	\item Than, in the first $p$ chosen positions place all players corresponding to positive literals and in the last $n$ place all players corresponding to negative literals . The number of such line-ups is $ (p-1)!n! $.
	\item The remaining players can be arranged in $ (|A| - (p+n))! $ ways.
	\end{itemize}
	
	Thus, we have:
	\[
		\begin{array}{rcl}

		 \phi^{\SV}_i & = & \frac{1}{|A|!}\sum_{\pi \in \Pi(A)}{\MC{P_{\pi}(a_i)}{a_i}} \\
				   & = &  \frac{1}{|A|!}\binom{|A|}{p+n}(p-1)!n!(|A| - (p+n))! \mathit{V}  \\

				   & = & \frac{(p-1)!n!}{(p+n)!}\mathit{V} \\
				   & = &  \frac{\mathit{V}}{p\binom{p+n}{n}} \\
		\end{array}		
	\]
	The proof for $ a_j \in \mathcal{N} $ is analogous.
	
\end{proof}

Now, let us consider the game $g_2$ from Example~\ref{ex:shapley:apples} and transform its representation from characteristic function form to Marginal Contribution Nets. Next, we will compute the Shapley value based on the above theorem.

\begin{example}
Firstly, we define the set of rules $\mathcal{R}^*$ that corresponds to the characteristic function $\nu_2$ form the game $g_2$.
\[
\begin{array}{rrcl}
	r_1: & a_2 & \rightarrow & 400 \\
	r_2: & a_1 \wedge \neg a_2 & \rightarrow & 200 \\
	r_3: & a_3 \wedge \neg a_1 \wedge \neg a_2 & \rightarrow & 200 \\
	r_4: & a_1 \wedge a_2 \wedge a_3 & \rightarrow & 200
\end{array}
\]

Secondly, in order to compute $ \nu_2(\{a_1,a_2,a_3\}) $ we need to iterate through all rules that this coalition meets. We see that $ \{a_1,a_2,a_3\} \models a_2 $ and $ \{a_1,a_2,a_3\} \models a_1 \wedge a_2 \wedge a_3 $, so  $ \nu_2(\{a_1,a_2,a_3\}) = 400 + 200 = 600 $. In the same manner we can compute the value of any other coalition. Now, let us consider the Shapley value vector for each game defined by a single rule.

\[
\begin{array}{rcl}
	r_1: & (0,\frac{400}{1\binom{1}{1}},0) & = (0,400,0)\\
	r_2: & (\frac{200}{1\binom{2}{1}},\frac{-200}{1\binom{2}{1}},0) & = (100,-100,0) \\
	r_3: & (\frac{-200}{2\binom{3}{1}},\frac{-200}{2\binom{3}{1}},\frac{200}{1\binom{3}{2}}) & = (-33\frac{1}{3},-33\frac{1}{3},66\frac{2}{3} ) \\
	r_4: & (\frac{200}{3\binom{3}{0}},\frac{200}{3\binom{3}{0}},\frac{200}{3\binom{3}{0}}) & = (66\frac{2}{3},66\frac{2}{3},66\frac{2}{3} )
\end{array}
\]
\noindent We sum up all vectors, and we obtain the final solution of this game:
\[
\begin{array}{r l}
\phi^{\SV}_1 = & 0 + 100 - 33\frac{1}{3} + 66\frac{2}{3} = 133 \frac{1}{3} \\
\phi^{\SV}_2 = & 400 - 100 -33\frac{1}{3} + 66\frac{2}{3} = 333 \frac{1}{3} \\
\phi^{\SV}_3 = & 0 + 0 + 66\frac{2}{3} + 66\frac{2}{3} = 133 \frac{1}{3} \\
\end{array}
\]
\end{example}

To conclude the above example, we have obtain the same results as in Example~\ref{ex:shapley:apples} but in order to compute the Shapley value for each player we have performed less operation than when iterating through all permutations. However, creating rules only with conjunctions is not an effective way for many games. Thus, in the next section we will introduce more effective representation of the coalitional games.

\section{Read-once Marginal Contribution Networks}\label{sec:read:once:mcnets}
\noindent Many coalitional games can be represented more succinct than with Marginal Contribution Networks (Definition~\ref{def:general:mcnets}).  Elkind et al.~\citeyear{Elkind:et:al:2009} defined a class of rules, called \emph{read-once MC-Net rules}, which are considerably more succinct than basic rules, while enjoying similarly attractive computational properties. The syntax of read-once Boolean formulas is the same as in General Marginal Contribution Networks, but each literal $a_i$ can only appear once. In other words, this formula is also a binary rooted tree with positive or negative literals on leafs and internal nodes labeled with one of the following logical connectives: conjunction ($\wedge$), disjunction ($\vee$) and exclusive disjunction ($\oplus$).

\begin{definition}[Read-once Marginal Contribution Network]\label{def:read:once:mcnets}
The coalitional game in the form of Read-once Marginal Contribution Network is a pair $ (A,\mathcal{R})$, where $ A $ is a set of players, and $\mathcal{R}$ is the set of rules. Rules are of the form of $\mathcal{F} \to V$, where $\mathcal{F}$ is read-once Boolean formulas.	
\end{definition}

\citet{Elkind:et:al:2009} proved that, in General Marginal Contribution Networks, where each literal can appear in a formula many times, the problem of computing the Shapley value is intractable (i.e., NP-hard). However, for the read-once Boolean formulas the authors presented a polynomial-time algorithm to compute the Shapley value.

In Chapter~\ref{chap:mcnets} we extend the positive results obtained by \citet{Elkind:et:al:2009}  to generalized coalitional games.

\section{Graph Theory}\label{sec:graphs}
\noindent The notion of a \textit{graph} (or a \textit{network}) will play a the key role in this dissertation. These mathematical entities models many real-life problems from different disciplines, such as biology, chemistry, or sociology \citep{Barabasi:2003}.

\begin{definition}[Graph]\label{def:graph}
The simple graph $ G $ is a pair $ (V,E) $, where $ V = \{v_1,v_2,\ldots\,v_{|V|}\} $ is the set of \emph{nodes} (or \emph{vertices}), and $ E = \{e_1,e_2,\ldots\,e_{|E|}\} $ is the set of \emph{edges} (or \emph{links}). Edges $e=(v_i,v_j)$ are a two-elements sets of nodes from $V$.
\end{definition}

In the same manner, but with different definition of edge, we define the directed graph.

\begin{definition}[Directed graph]\label{def:graph:directed}
The directed graph $ G $ is a pair $ (V,E) $, where $ V = \{v_1,v_2,\ldots\,v_{|V|}\} $ is the set of \emph{node}, and $ E = \{e_1,e_2,\ldots\,e_{|E|}\} $ is the set of \emph{edges}. Edges $e=(v_i,v_j)$ are a two-elements ordered pairs of nodes from $V$.
\end{definition}

\noindent In such graphs we say that an edge $(v_i,v_j)$ is directed from the node $v_i$ to the node $v_j$.

\begin{definition}[Weighted graph]\label{def:graph:weighted}
The weighted graph $ G =(V,E,\lambda) $ is a simple/directed graph, with a function  $ \lambda(v,u) : E \to \mathbb{R} $ that evaluates each edge. The value $ \lambda(e) $ is called the wight of an edge $e$.
\end{definition}

A \textit{path}~$p_{st}$ from node~$s$ to node~$t$ in a graph~$G$ is an ordered set $(v_0, v_1, \ldots, v_n)$ such that~$v_0 = s$ and~$v_n = t$ and $(v_i, v_{i+1}) \in E$ for all~$i$ with $1\le i<n$. We define the set of \emph{neighbours} of a node~$v$ and a subset~$C$ of nodes by $N(v_i)=\set{v_j\midd (v_i,v_j)\in E}$ and  $N(C)=\bigcup_{v_i \in C}N(v_i)\setminus C$, respectively. We refer to the \emph{degree} of a node~$v$ by $\mathit{deg}(v) = |N(v)|$, and the \emph{degree} of a set~$C$ as $\mathit{deg}(C) = |N(C)|$. The \emph{distance} from a node~$s$ to a node~$t$ is denoted by~$\mathit{dist}(s, t)$ and is defined as the size of the shortest path between~$s$ and~$t$. The distance between a node~$v$ and a subset of nodes~$C \subseteq V$ is denoted by $\mathit{dist}(C, v) = \min_{u \in C} \textit{dist}(u, v)$.

For a further comprehensive introduction to graph theory the reader is kindly referred to the work by Cormen \citeyear{Cormen:2001}.

\section{Centrality Measures}\label{sec:centralities}

\noindent One of the important aspects of research on networks is to determine which nodes and edges are more critical (or \emph{central}) to the functioning of the entire network than others. For example, one may want to identify the focal hubs in a road network \citep{Schultes:Sanders:2007}, the most critical functional entities in a protein network \citep{Jeong:et:al:2001}, or the most~influential people~ in~ a~ social~ network \citep{Kempe:et:al:2003}. To this end, various centrality metrics, such as \emph{degree}, \emph{closeness}, \emph{eigenvalue} or \emph{betweenness}, have been extensively studied in the literature.\footnote{\footnotesize  \citet{Koschutzki:et:al:2005} presented an overview of the most important centrality metrics.}

\begin{definition}[Centrality]\label{def:centrality}
For a given graph $ G =(V,E) $ the \emph{classic} (also standard) centrality is a function $c: V \rightarrow \mathbb{R} $ that for each node $v \in V$ it determines the importance of this node measured in real value.
\end{definition}

Generally speaking, centrality analysis aims to create a consistent ranking of nodes within a network. To this end, \emph{centrality measures} assign a score to each node that in some way corresponds to the importance of that node given a particular application. Since ``importance'' depends on the context of the problem at hand, many different centrality measures have been developed. The three well-known and widely applied measures introduced by \citet{Freeman:1979} are: \emph{degree centrality}, \emph{closeness centrality} and \emph{betweenness centrality}.

\begin{definition}[Degree centrality]\label{def:cen:degree}
For a given graph $ G =(V,E) $ degree centrality is a function $c_d: V \rightarrow \mathbb{R} $:
\[
	c_d(v) = deg(v)
\]
\noindent where $deg(v)$ is a degree of the node $v$.
\end{definition}

The intuition standing behind this measure is that it reflects the activity of a given node. This activity can be expressed as a number of friends in social portal, or number of articles published together with other peers in science collaboration network.

\begin{definition}[Betweenness centrality]\label{def:cen:bet}
For a given graph $ G =(V,E) $ betweenness centrality is a function $c_b: V \rightarrow \mathbb{R} $:
\[
	c_b(v) =  \sum_{s \neq v \neq t} \frac{\sigma_{st}(v)}{\sigma_{st}}
\]
\noindent where $ \sigma_{st} $ is the number of shortest paths between $ s $ and $ t $ (if $s=t$ then $ \sigma_{st} = 1 $), and $ \sigma_{st}(v) $ is the number of shortest paths between $ s $ and $ t $ passing through a node $ v $ (if $v \in \{s,t\}$ then $ \sigma_{st}(v) = 0 $).\footnote{\footnotesize To deal with unconnected graphs it is assumed that $\frac{0}{0} = 0 $.}
\end{definition}

This measure and its different variants is generally used to indicate the most important nodes in the process of controlling information flow. For instance it can be used to determine the most important routers in computer network.

\begin{definition}[Closeness centrality]\label{def:cen:close}
For a given graph $ G =(V,E) $ closeness centrality is a function $c_c: V \rightarrow \mathbb{R} $:
\[
	c_c(v) = \sum_{s \neq v} d(s,v)
\]
\noindent where $ dist(s,v) $ is a distance between nodes $ s $ and $ v $.
\end{definition}

This centrality has two aspects: one that evaluates the most important nodes in the process of discrimination of information, and one that evaluates the most independent nodes. For instance, if we want to construct an efficient network of communication in our company, the CEO should be in the middle of this structure that his orders would immediately hit the targets. On the other hand, the person with the lowest closeness centrality is the most independent person, because the information coming to him from others requires the smallest number of brokers. For unconnected networks this measure is sometimes presented as a sum of fractions: $\sum_{s \neq v} \frac{1}{d(s,v)}$.

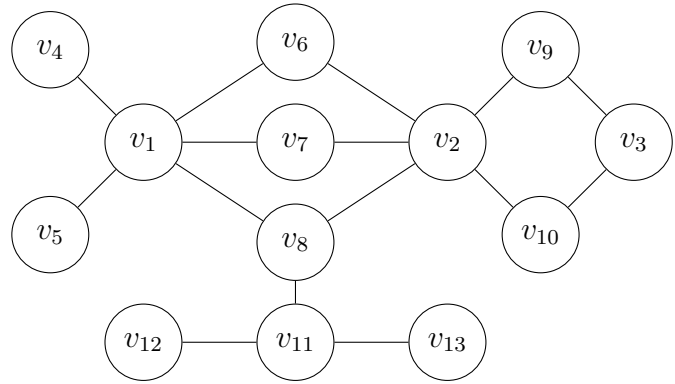
\begin{figure}[t]\centering
\begin{tikzpicture}[node distance = 2cm]
\tikzstyle{every node}=[draw,shape=circle,minimum size=2.4em];

\node (1) {$v_1$};
\node (7) [right of=1]  {$v_7$};
\node (2) [right of=7] {$v_2$};
\node (4) [above left= 0.5 cm and 0.5cm of 1] {$v_4$};
\node (5) [below left= 0.5 cm and 0.5cm of 1] {$v_5$};
\node (6) [above=0.3cm of 7] {$v_6$};
\node (8) [below=0.3cm of 7] {$v_8$};
\node (9) [above right= 0.5 cm and 0.5cm of 2] {$v_9$};
\node (10) [below right= 0.5 cm and 0.5cm of 2] {$v_{10}$};
\node (3) [below right= 0.5 cm and 0.5cm of 9] {$v_3$};
\node (11) [below=0.3cm of 8] {$v_{11}$};
\node (12) [left of=11] {$v_{12}$};
\node (13) [right of=11] {$v_{13}$};

%\node (4) [below left of1] {$v_4$};

\draw (1) -- (7)
(7) -- (2)
(4) -- (1)
(5) -- (1)
(1) -- (6) -- (2) -- (9) -- (3)
(1) -- (8) -- (2) -- (10) -- (3)
(8) -- (11) -- (12)
(11) -- (13);
\end{tikzpicture}
\caption{Sample network of 13 nodes}
\label{fig:cen:example}
\end{figure}

In the sample network in Figure \ref{fig:cen:example}, nodes $v_1$ and $v_2$ have degree 5 and, if judged by degree centrality, these are the most important nodes within the entire network ($c_d(v_1)=c_d(v_2) = 5$). Conversely, closeness centrality focuses on distances among nodes and gives high value to the nodes that are close to all other nodes. With this measure, node $v_8$ in Figure \ref{fig:cen:example} is ranked top with lowest value $c_c(v_8)=22$. The last measure---betweenness centrality---considers shortest paths (i.e., paths that use the minimal number of links) between any two nodes in the network. The more shortest paths the node belongs to, the more important it is. With this measure, $c_b(v_2)=32$ and the node $v_2$ in Figure \ref{fig:cen:example} is more important than all the other nodes (including $v_1$ and $v_8$, which are chosen by other measures as the most important node). Clearly, all these measures expose different characteristics of a node. Consider, for instance, an epidemiology application, where the aim is to identify those people (i.e., nodes) in the social network who have the biggest influence on the spread of the disease and should become a focal point of any prevention or emergency measures. Here, degree centrality ranks top nodes with the biggest immediate \emph{sphere of influence}---their infection would lead to the highest number of adjacent nodes being exposed to the disease. On the other hand, closeness centrality identifies those nodes whose infection would lead to the \emph{fastest} spread of the disease throughout the society. Finally, betweenness centrality reveals the nodes that play a crucial role in \emph{passing} the disease from one person in a network to another.

\section{Community Structures}\label{sec:communities}
\noindent Apart from centrality measures, the other important aspect of networks extensively studied in literature is their community structure. Intuitively, a \textit{community} (also called a cluster) is a group of nodes in a network such that there are comparatively many edges among the nodes within the group and fewer with nodes outside the group. The communities can be non-overlapping, which mean that node can belong to only one community, and overlapping. In this dissertation we focus on non-overlapping communities. The graph divided into communities forms a \emph{community structure}.

\begin{definition}[Community structure]\label{def:CS}
For a given graph $ G =(V,E) $ the community structure $\mathit{CS}=\set{C_1, C_2, \ldots, C_m}$ is the partition of the set $V$, which implies that $\emptyset \notin \mathit{CS}$, $\sum_{C \in \mathit{CS}} C = V$ and  $\forall_{C_1,C_2 \in \mathit{CS}}{C_1 \neq C_2 \implies C_1 \cap C_2 = \emptyset}$.
\end{definition}

Community detection is a key research topic in network analysis \citep{Fortunato:2010}. It is relevant to many fields, where graph representations are common, including sociology \citep{Girvan:Newman:2002}, biology \citep{vanLaarhoven:Marchiori:2012} and computer science \citep{Misra:et:al:2012}. Community detection raises various conceptual question. One such question is how to evaluate a given division of a network into communities. One of the prominent method of doing this was introduced by \citet{Newman:2006} and is called \emph{modularity}.

\begin{definition}[Modularity]\label{def:modularity}
For a given graph $ G =(V,E) $ and community structure $\mathit{CS}=\set{C_1, C_2, \ldots, C_m}$ the modularity is:
\begin{equation}
Q(G,\mathit{CS}) = \sum_{C \in \mathit{CS}}\bigg(\frac{|E(C)|}{|E|}-\Big(\frac{\#deg(C)}{2|E|}\Big)^2\bigg),
\end{equation}
\noindent where $E(C) = \set{(u,v) \in E : u,v \in C}$, and $\#deg(C) = \sum_{v \in C} deg(v)$.
\end{definition}

The first term of the above equation is the fraction of edges inside community $C$ and the second term, in contrast, is the expected fraction of edges in that community under the assumption that edges were located at random in the network, in which the degrees of nodes remains unchanged in respect to the original graph.

\chapter{Related Work}\label{chap:related}

\noindent In this dissertation we are concerned with computational aspects of game-theoretic centrality measures. Hence, there are three main bodies of related literature: on centrality measures (Section~\ref{sec:rw:cen}), on game-theoretic centrality measures (Section~\ref{sec:rw:gtc}), and on computational aspects of coalitional games (Section~\ref{sec:rw:representation}). The related work on the more specific subjects will be presented in the chapters corresponding to these topics.

\section{Centrality Measures and Their Classification}\label{sec:rw:cen}
\noindent Centrality measures constitute one of the main research directions in the literature on network analysis. In his classic work, \cite{Freeman:1979} proposed and formalized three different concepts of centrality: (i) \emph{degree centrality} is based on being directly connected to other nodes; (ii) \emph{closeness centrality} is based on being close to all other nodes; and (iii) \emph{betweenness centrality} is based on lying between other nodes. Importantly, these three centralities are build upon three conceptually different properties extracted from networks: connections, distances and geodesics, respectively, and all three can be easily used as a basic ground for various extensions. To name only the few, \cite{Newman:2004} proposed weighted degree centrality, a appropriate measure of degree for weighted networks, \cite{Brandes:2008} introduced distance-scaled betweenness centrality, and \cite{Boldi:Vigna:2013} harmonic centrality, the extension of closeness one.

The fourth classical measure is \emph{eigenvector centrality}, which is based on the idea that connections to more important nodes should be valued more than otherwise equal connections to less important nodes~\citep{Bonacich:1972}. One of the most known extension of this measure was proposed by \cite{Brin:Page:1998}, who proposed famous \emph{PageRank} metric

The topics closely related to centralities are that of the \emph{percolation} on complex networks and \emph{social capital}. Regarding percolation, \cite{Callaway:et:al:2000} studied different models of nodes failures, where each element of the graph have some probability of being disabled. Such studies gives answer to the problem how a given network is vulnerable to a random node or edge failure. Here the centralities can be used to determine the most important nodes to protect. Such studies were done by \cite{Holme:et:al:2002}, who compared degree and betweenness centralities. He postulated that betweenness centrality is a good basis for strategies of protecting nodes against random nodes failures.

 The ''social capital'' term is much wider than ''centrality'' and is  is used to describe both  ``\textit{the value of an individual's social relationships}'' \citep{Burt:1992} and ``\textit{a quality of groups}'' \citep{Everett:Borgatti:1999} within the society. \cite{Borgatti:et:al:1998} propose to classify social capital along two dimensions. The first is whether it concerns individual actor or a group, or whether we focus on internal or external links (see Table~\ref{table:borgatti:classification}).

 \begin{table}[h]
 \begin{center}
	\begin{tabular}{|c|c|c|}
	\cline{2-3}
	\multicolumn{1}{c}{} & \multicolumn{2}{|c|}{Type of Focus}\\
	\hline
	Type of Node & internal & external \\
	\hline
	\multirow{2}{*}{Individual} & \textbf{(A)} &   \textbf{(B)}  \\
	 & ---- & standard centrlities~\small{\citep{Freeman:1979}} \\
	\hline
	\multirow{2}{*}{Group} & \textbf{(C)} & \textbf{(D)}  \\
	& centralisation~\small{\citep{Freeman:1979}} & group centralities~\small{\citep{Everett:Borgatti:1999}}  \\
	\hline
	\end{tabular}
	\caption{ Classification by \citet{Borgatti:et:al:1998} of different conceptions (forms) of social capital (including different \emph{centralities}), with examples of such measures.}
 \label{table:borgatti:classification}
 \end{center}
 \end{table}

This classification distinguishes three fields of social capital studies (or centrality): \textbf{(A)} in this category are examined the internal properties of autonomous actors  and this is not part of the studies on social capital nor social networks analysis;  \textbf{(B)} in this field the individual actor is examined in the context of its relations to external world and most standard centralities belong to this category; \textbf{(C)} in this category the group of actors are studied in the context of internal relations among them; \textbf{(D)} the last type of social capital is engaged in studying group of actors and their relationships with external world. For all types of social capital many different measures have been proposed (see Table~\ref{table:borgatti:classification}) and most of them the literature refers to as \emph{centralities}. Although the three categories of social capital are considered separately, there are cases where they influence each other. For instance, in situations where actors are divided into groups, the role of each individual strictly depends on the role of the entire group and the role of the group simply depends on its members.

Social capital is related to social networks describing people and their relationships. Centrality measures can be used to analyze any kind of networks.

\section{Game-theoretic Apporach to Centrality Measures}\label{sec:rw:gtc}

\noindent Game-theoretic centrality is an extension of standard centrality based on cooperative game theory. The basic idea behind game-theoretic centrality is to consider the nodes of the network to be a \emph{players} in a cooperative game, where the payoffs of coalitions are derived from the network topology. Such an analogy between graph and game theory allows to use various game-theoretic solution concepts as a centrality measure of individual nodes.\footnote{ \footnotesize See the next chapter for more details.}

The first author to introduce game-theoretic centrality were \cite{Grofman:Owen:1982}, who used \textit{the Banzhaf index of power}~\citep{Banzhaf:1965}. This is a widely-used solution concept that quantifies the power of a player by counting the number of coalitions in which this player is a \textit{swing} player---a player whose addition to the coalition changes it from being unsuccessful (e.g., in achieving some goal) to being successful. \citeauthor{Grofman:Owen:1982} mapped a \emph{simple coalitional} game, a coalitional game where every coalition has a value of $1$ or $0$, onto a directed graph by considering paths in the graph to be coalitions. Here, a swing player is a node that is crucial to maintain communication among the nodes in the path. Then, Grofman and Owen's Banzhaf index-based centrality says that the more times a node is a swing node within all possible paths of nodes, the more important this node is in the entire network.

The game-theoretic centrality were reinvented by \cite{Gomez:et:al:2003}, who proposed a slightly different approach. Instead of restricting themselves to simple coalitional games, the authors study arbitrary games where a coalition's value can be any real number. In this context, the authors propose a centrality measure based on the two prominent solution concepts: \emph{the Shapley value} \citep{Shapley:1953} and \emph{the Myerson value} \citep{Myerson:1977}.

The first value is arguably the best normative solution concept known to date, and was discussed in Section~\ref{sec:SV}. \emph{The Myerson value} is based on \textit{graph-restricted games}, where the underlying assumption is that only connected coalitions (i.e., the group of nodes that between any pair exists a path passing through only members of this group) can communicate and create any value added. \citeauthor{Myerson:1977} assumed that the value of disconnected coalition would simply be the sum of the values of its connected components. A celebrated result of Myerson for this setting is a dedicated solution concept---closely related to the Shapley value---called the Myerson value. In fact, it is the Shapley value, but evaluated for the characteristic function sensible to coalitions connectivity, as was described above.

Taking Shapley value and Myerson value,  \cite{Gomez:et:al:2003} defined the centrality of a node $v_i$ to be the difference between the Shapley value of $v_i$ (which is independent of the network connectedness) and the Myerson value of $v_i$ (which takes the network connectedness into consideration). Thus, by interpreting the Shapley value as a power index in a coalitional game, G{\'o}mez et al.'s Shapley/Myerson-based centrality measure represents the increase (or decrease) in $v_i$'s power due to its position in the network.

This direction of research took also \cite{delPozo:et:al:2011}, who proposed a measure for directed graphs, based on \textit{generalised coalitional games}~\citep{Nowak:Radzik:1994}, where the value of a coalition is influenced by the sequence in which the players have joined it. More specifically, the authors adapted these games to networks in the spirit of Myerson and then based their measure on a parametric family of solution concepts that include two extensions of the Shapley value to generalised coalitional games: one by \cite{Nowak:Radzik:1994} and the other by \cite{Sanchez:Bergantinos:1999}. In a similar spirit, \cite{Amer:et:al:2007} defined a game-theoretic centrality measure called accessibility in oriented networks.

A computational study of game-theoretic centrality measures based on connectivity games on networks (inspired by Myerson) was done in Chapter~\ref{chap:connectivity}. We showed that the efficient computation of the Shapley value is impossible even in the case of a very simple definition of the connectivity game on graph. We also proposed algorithm that significantly improve the computation of such game-theoretic centrality measures. We need to mention that our algorithm was further improved by \cite{Skibski:et:al:2014} and modified and applied to analyze social networks by \cite{Narayanam:et:al:2014}.

All the above works are inspired in one way or the other by the graph-restricted games in the spirit of \citet{Myerson:1977}, where the value of a group of nodes is mainly determined by connectedness of this group.

This dissertation contributes to a parallel stream of literature in which coalitions of nodes are not graph-restricted, i.e., their values do not depend directly on whether they are connected or not. The origins of this approach can be traced back to the concept of \textit{group centrality} introduced by \cite{Everett:Borgatti:1999} which takes into account synergies among nodes. By computing group centrality for every possible group of nodes (not necessarily connected) we obtain a well-defined coalitional game. Next, one can apply a solution concept, such as the Shapley value or the Banzhaf index, to evaluate each individual nodes. Such evaluation attribute synergies to individual nodes, that were captured by group centrality measure.

The measure that captured the ability of nodes to dominate other nodes in directed networks was proposed by \cite{Brink:Borm:2002}. The authors do not refer to their measure as \emph{centrality}, but from the social network analysis perspective it is. \cite{Brink:Gilles:2000} continue this line of research by analysing these measures more deeply and, more importantly, by providing an axiomatic approach to them. These are the first works to use general Shapley value approach to game-theoretic centralities.

This second line of research was also used by \cite{Suri:Narahari:2010} in the interesting application of influence propagation in networks. Specifically, the authors constructed a coalitional game by assigning to coalitions payoffs equal to their group degree. Then, the authors used the Shapley value to approximate the solution of the top $k$-node problem, i.e., the problem of identifying the $k$ most influential nodes in the network. This approach was followed up by \cite{Adamczewski:et:al:2014}, who studied different versions of group degree.

%\section{The Coalitional Games in Network Analysis}\label{sec:rw:other}
% TODO
%We also mention that there exist game-theoretic contributions that are closely related to, but do not explicitly study, the problem of centrality in networks. In fact, any allocation rule on the network can be interpreted as a centrality measure. As such, any solution concept of any coalitional game defined over a graph is a centrality measure. \cite{Slikker:Nouweland:2001} survey works on allocation rules over the network. Some more recent contributions in this line of research include \cite{Brink:et:al:2008,Hendrickx:et:al:2009,Kim:Tackseung:2008} and \cite{Brink:et:al:2011}. Also interesting research on wiretapping a communication network was done by \cite{Aziz:et:al:2009}.

\section{Computational Aspects of Coalitional Games -- Representations}\label{sec:rw:representation}

\noindent Let us now turn our attention to the computational issues. In more detail, the computational challenges related to the game-theoretic solution concepts that are tackled in this dissertation have been extensively studied in the literature on algorithmic aspects of coalitional games. Since the path-breaking work by \citet{Deng:Papdimitriou:1994}, a considerable attention has been paid to developing more efficient \emph{representations} of games. Such representation can be divided into two main categories \citep{Wooldrdige'':Dunne'06}.

The first category, to which also this dissertation indirectly contributes, models the characteristic function using specific combinatorial structures such as graphs. This line of research include works by \citet{Deng:Papdimitriou:1994}, \citet{Greco:et:al:2009}, \citet{Wooldrdige'':Dunne'06}, and the representations discussed by \citet{Aziz:Keijzer:2011}. Such representations are guaranteed to be succinct, however they can express only certain games. In our case, the characteristic function is the function of the group betweenness centrality.

In the second category of representations the emphasis is placed on full expressiveness, often at the expense of succinctness. These representations include MC-Nets \citep{Ieong:Shoham:2005}, its read-once extension~\citep{Elkind_et_al_2009}, Synergy Coalition Groups \citep{Conitzer:Sandholm:2006}, the Decision-Diagrams-based representations \citep{Aadithya:et:al:2011,Sakurai:et:al:2011},  the vector-based representation~\citep{Tran-Than:et:al:2013}, and the representation for graph-restricted weighted voting games \citep{Skibski:et:al:2015} .% While all the above models concern games with no externalities, representations for games with externalities include Embedded MC-Nets~\citep{Michalak10amas,Ichimura:et:al:2011} and Weighted MC-Nets~\citep{Michalak10ecai}.

We note that even for some succinctly representable games, computing the Shapley value and other game-theoretic solution concepts can be challenging (NP-Hard or even \#P-Complete). This is the case with weighted voting games \citep{Deng:Papdimitriou:1994}, threshold network flow games \citep{Bachrach:Rosenschein:2009} and minimum spanning tree games \citep{Nagamochi:et:al:1997}. Also \cite{Aziz:et:al:2009} proved negative results for the Shapley-Shubik power index for the spanning connectivity games on multigraphs, and \cite{Bachrach:et:al:2008} for the Banzhaf index for certain class of connectivity games.

Perhaps the most widely known positive results are those aforementioned due to \cite{Deng:Papdimitriou:1994} and \cite{Ieong:Shoham:2005}. In particular, Deng and Papadimitriou proposed a representation based on graphs, weighted edges of which represent the marginal contribution of both adjacent nodes (agents) to any coalition they belong to. While this representation is clearly not fully expressive, the Shapley value can be computed in time linear in the size of the graph. The representation of \cite{Ieong:Shoham:2005} consists of a finite set of logical rules of the following form: \emph{Boolean Expression $\rightarrow$ Real Number}, where agents are represented by atomic boolean variables (see Definition~\ref{def:mcnets}). The value of a coalition is equal to the sum of the Real Numbers in those rules whose that are satisfied by the coalition. Unlike the representation by Deng and Papadimitriou, the representation by Ieong and Shoham is fully expressive and, for certain types of Boolean Expressions, it allows for the Shapley value to be computed in time linear in the size of the representation.

The number of authors proposed alternative representations of coalitional games: coalitional skill games \citep{Bachrach:Rosenschein:2008} or synergy coalitional groups \citep{Aadithya:et:al:2011}, which allow for representing input compactly for certain games and admit algorithms (e.g. for computing the Shapley value) that run in polynomial time in the size of such a compact input. Finally, one can find in the literature the representation of the extension of coalitional games to games with externalities \citep{Michalak:et:al:2009,Michalak:et:al:2010}.

\chapter{Polynomial Algorithms for Game-theoretic Centrality Measures}\label{chap:gtcm}

\noindent The Shapley value has recently been advocated as a useful measure of centrality in networks. However, although this approach has a variety of real-world applications (including social and organisational networks, biological networks and communication networks), its computational properties have not been widely studied. To date, the only practicable approach to compute Shapley value-based centrality has been via Monte Carlo simulations which are computationally expensive and not guaranteed to give an exact answer. Against this background, in this chapter we present the study of the computational aspects of the Shapley value for network centralities. Specifically, we develop exact analytical formulae for Shapley value-based centrality in both weighted and unweighted networks and develop efficient (polynomial time) and exact algorithms based on them. We empirically evaluate these algorithms on two real-life examples (an infrastructure network representing the topology of the Western States Power Grid and a collaboration network from the field of astrophysics) and demonstrate that they deliver significant speedups over the Monte Carlo approach. For instance, in the case of unweighted networks our algorithms are able to return the exact solution about 1600 times faster than the Monte Carlo approximation, even if we allow for a generous $10\%$ error margin for the latter method.

Against this background, in this chapter we present the study of the computational aspects of the Shapley value for network centralities. Specifically, we develop exact analytical formulae for Shapley value-based centrality in both weighted and unweighted networks and develop efficient (polynomial time) and exact algorithms based on them. We empirically evaluate these algorithms on two real-life examples (an infrastructure network representing the topology of the Western States Power Grid and a collaboration network from the field of astrophysics) and demonstrate that they deliver significant speedups over the Monte Carlo approach. For instance, in the case of unweighted networks our algorithms are able to return the exact solution about 1600 times faster than the Monte Carlo approximation, even if we allow for a generous $10\%$ error margin for the latter method.

The contribution of the author of this dissertation mostly covers the experimental part of this article. More specifically,  Algorithms~\ref{algo:gtc:game1}-\ref{algo:gtc:game5} were developed by Karthik Aaditha, and formalized together with the author of this thesis. The Propositions \ref{prop1} and \ref{prop3} were proved by the author of this thesis. Also Algorithms~\ref{algo:gtc:mc}-\ref{mcg5} were developed by the author of this thesis.

%The remainder of the chapter is organized as follows. In Section \ref{chap:gtc:intro}, we provide the motivation for game-theoretic centralities. Section~\ref{chap:gtc:related:approxSV} describes the related work. Key notations and definitions are introduced in Section~\ref{sec:gtc:def}. Finally, Section \ref{sec:gtc:algo} describes five algorithms for Shapley value-based degree and closeness centralities and Section~\ref{sec:gtc:simulations} provides extensive empirical evaluation on two real-life networks against the Monte Carlo algorithm.

\section{Why Use Game-Theoretic Centrality Measures?}\label{chap:gtc:intro}

\noindent In many network applications, it is important to chose which nodes and edges are the most important. To this end, the concept of \emph{centrality} (see Section~\ref{sec:centralities}), which aims to quantify the importance of individual nodes/edges, has ~been ~extensively studied in  the literature \citep{Koschutzki:et:al:2005,Brandes:Erlebach:2005}.

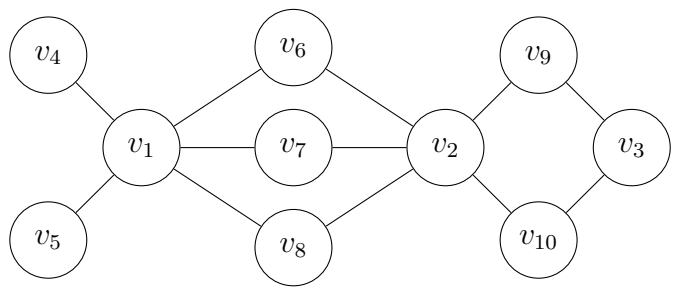
\begin{figure}[h]\centering
\begin{tikzpicture}[node distance = 2cm]
\tikzstyle{every node}=[draw,shape=circle,minimum size=2.4em];

\node (1) {$v_1$};
\node (7) [right of=1]  {$v_7$};
\node (2) [right of=7] {$v_2$};
\node (4) [above left= 0.5 cm and 0.5cm of 1] {$v_4$};
\node (5) [below left= 0.5 cm and 0.5cm of 1] {$v_5$};
\node (6) [above=0.3cm of 7] {$v_6$};
\node (8) [below=0.3cm of 7] {$v_8$};
\node (9) [above right= 0.5 cm and 0.5cm of 2] {$v_9$};
\node (10) [below right= 0.5 cm and 0.5cm of 2] {$v_{10}$};
\node (3) [below right= 0.5 cm and 0.5cm of 9] {$v_3$};

\draw (1) -- (7)
(7) -- (2)
(4) -- (1)
(5) -- (1)
(1) -- (6) -- (2) -- (9) -- (3)
(1) -- (8) -- (2) -- (10) -- (3);
\end{tikzpicture}
\caption{Sample network of 10 nodes}
\label{fig:example:gtcen}
\end{figure}

Generally speaking, centrality analysis aims to create a \emph{consistent ranking} of nodes within a network. To this end, \emph{centrality measures} assign a score to each node that in some way corresponds to the importance of that node given a particular application. Since ``importance'' depends on the context of the problem at hand, many different centrality measures have been developed. Three of the most well-known that were introduced in Section~\ref{sec:centralities}: \emph{degree centrality}, \emph{closeness centrality} and \emph{betweenness centrality}. Recall that we refer to these measures as \emph{classic/standard} centralities (Definition~\ref{def:centrality}).

The common feature of all standard measures is that they assess the importance of a node by focusing only on the role that a node plays by itself. However, in many applications such an approach is inadequate because of \emph{synergies} that may occur if the functioning of nodes is considered in groups. Referring to Figure \ref{fig:example:gtcen} and epidemiology example, a vaccination of individual node $v_6$ (or $v_7$ or $v_8$) would not prevent the spread of the disease from the left to the right part of the network (or vice versa). However, the simultaneous vaccination of $v_6$, $v_7$ and $v_8$ would achieve this goal. Thus, in this particular context, nodes $v_6$, $v_7$ and $v_8$ do not play any significant role individually, but together they do. To quantify the importance of such groups of nodes, the notion of \emph{group centrality} was introduced by \cite{Everett:Borgatti:1999}. Intuitively, group centrality works broadly the same way as standard centrality, but now the focus is on the functioning of a given group of nodes, rather than individual nodes. For instance, in Figure \ref{fig:example:gtcen}, the group degree centrality of $\{v_1,v_2\}$ is 7 as they both have 7 distinct adjacent nodes.

Although the concept of group centrality addresses the issue of synergy between the functions is played by various nodes, it suffers from a fundamental deficiency. It focuses on particular, a priori determined, groups of nodes and it is not clear how to construct a consistent ranking of \emph{individual} nodes using such group results. Specifically, should the nodes from the most valuable group be ranked top? Or should the most important nodes be those which belong to the group with the highest average value per node? Or should we focus on the nodes which contribute most to every coalition they join? In fact, there are very many possibilities to choose from.

A framework that does address this issue is the \emph{game theoretic network centrality measure}. In more detail, it allows the consistent ranking of individual nodes to be computed in a way that accounts for various possible synergies occurring within possible groups of nodes \citep{Grofman:Owen:1982,Gomez:et:al:2003}. Specifically, the concept builds upon cooperative game theory---a part of game theory in which agents (or players) are allowed to form coalitions in order to increase their payoffs in the game. Now, one of the fundamental questions in cooperative game theory is how to distribute the surplus achieved by cooperation among the agents. To this end, Shapley proposed to remunerate agents with payoffs that correspond to their individual marginal contributions to the game~(Definition~\ref{def:sv}). In more detail, for a given agent, such an individual marginal contribution is measured as the weighted average marginal increase in the payoff of any coalition that this agent could potentially join. Shapley famously proved that his concept---known since then as the Shapley value---is the only division scheme that meets certain desirable normative properties. Given this, the key idea of the game theoretic network centrality is \emph{to define a cooperative game over a network} in which agents are the nodes, coalitions are the groups of nodes, and payoffs of coalitions are defined so as to meet requirements of a given application. This means that the Shapley value of each agent in such a game can then be interpreted as a \emph{centrality measure} because it represents \emph{the average marginal contribution made by each node to every coalition of the other nodes}.\footnote{\footnotesize We note that other division schemes or power indices from cooperative game theory can also serve as a good solution.}. In other  words, the Shapley value answers the question of how to construct a consistent ranking of individual nodes once groups of nodes have been evaluated.

In more detail, the Shapley value-based approach to centrality is, on one hand, much more \emph{sophisticated} than the conventional measures, as it accounts for any group of nodes from which the Shapley value derives a consistent ranking of individual nodes. On the other hand, it confers a high degree of \emph{flexibility} as the cooperative game over a network can be defined in a variety of ways. This means that \emph{many different versions of Shapley value-based centrality} can be developed depending on the particular application under consideration, as well as on the features of the network to be analyzed. As a prominent example, in which a specific Shapley value-based centrality measure is developed that is crafted to a particular application, consider  the work of \cite{Suri:Narahari:2010} who study the problem of selecting the top-$k$ nodes in a social network. This problem is relevant in all those applications where the key issue is to choose a group of nodes that together have the biggest \emph{influence} on the entire network. These include, for example, the analysis of co-authorship networks, the diffusion of information, and viral marketing. As a new approach to this problem, \citeauthor{Suri:Narahari:2010} define a cooperative game in which the value of any group of nodes is equal to the number of nodes within, and adjacent to, the group. In other words, it is assumed that the agents' \emph{sphere of influence} reaches the immediate neighbors of the group. Whereas the definition of the game is a natural extension of the (group) degree centrality discussed above, the \emph{Shapley value of nodes in this game} constitutes a new centrality metric that is, arguably, qualitatively better than standard degree centrality as far as the node's influence is concerned. The intuition behind it is visible even in our small network in Figure \ref{fig:example:gtcen}. In terms of influence, node $v_1$ is more important than $v_2$, because it is the only node that is connected to $v_4$ and $v_5$. Without $v_1$ it is impossible to influence $v_3$ and $v_4$, while each neighbor of $v_2$ is accessible from some other node. Thus, unlike standard degree centrality, which evaluates $v_1$ and $v_2$ equally, the centrality based on the Shapley value of the game defined by\citeauthor{Suri:Narahari:2010} recognizes this difference in influence and assigns a higher value to $v_1$ than to $v_2$.

Unfortunately, despite the advantages of Shapley value-based centrality over conventional approaches, efficient algorithms to compute it have not yet been developed. Indeed, given a network $G(V,E)$, where $V$ is the set of nodes and $E$ the set of edges, using the original Shapley value formula involves computing the marginal contribution of every node to every coalition which is $O(2^{|V|})$. Such an exponential computation is clearly prohibitive for bigger networks (of, e.g, 100 or 1000 nodes). For such networks, the only feasible approach currently outlined in the literature is Monte-Carlo sampling \citep{Suri:Narahari:2010,Castro:et:al:2009}. However, this method is not only inexact, but can be also very time-consuming (see Section~\ref{sec:gtc:simulations}).

The computational challenges are hard to overcome, but fortunately in this dissertation we develop the number of algorithms allowing to compute game-theoretic centralities in polynomial time.

\section{General Definition of Game-Theoretic Network Centrality}\label{sec:gtc:def}

\noindent In this section we provide the formal definitions of group centralities and game-theoretic network centralities. 

In order to formalize the game-theoretic network centrality firstly we need to define the group centralities. These measured introduced by \citet{Everett:Borgatti:1999} are the natural extension of standard centralities (Section~\ref{sec:centralities}) to evaluate group of nodes. \citeauthor{Everett:Borgatti:1999} introduced three requirements that group centrality should meet. Firstly, these metrics should be build upon existing concepts of centralities. Secondly, the group centrality applied to the group containing a single individual should yield the same result as standard centrality of this individual. Finally, the group centrality must stem from the relationships between individuals, not between groups. In other words, the value of the group $C \subseteq V$ of nodes is an intrinsic property of this group and its relationships, and the other group of nodes has no impact on it as long as the relationships of $C$ are not changed.

\begin{definition}[Group centrality]\label{def:group:centrality}
For a given graph $ G =(V,E) $ the \emph{group} centrality is a function $c_g: 2^{V} \rightarrow \mathbb{R} $ that for each subset of nodes $C \subseteq V$ determines the importance of this group measured in real value.
\end{definition}

Now, the extension of degree, centrality~(Definition~\ref{def:cen:degree}) is:

\begin{definition}[Group degree centrality]\label{def:group:degree}
For a given graph $ G =(V,E) $ and a group of nodes $C \subseteq V$ group degree centrality is a function $c_{gd}: 2^{V} \rightarrow \mathbb{R} $:
\[
	c_{gd}(C) = deg(C)
\]
\noindent where $deg(C)$ is a degree of the set $C$ (see Section~\ref{sec:graphs}).
\end{definition}

The extension of betweenness centrality~(Definition~\ref{def:cen:bet}) is:

\begin{definition}[Group betweenness centrality]\label{def:group:bet}
The group betweenness centrality of the set of nodes $C \subseteq V$ is defined as a function $ c_{gb}: 2^V \rightarrow \mathbb{R}$ such that:
\[
c_{gb}(C) =  \sum_{ \substack{s \notin C \\ t \notin C}} \frac{\sigma_{st}(C)}{\sigma_{st}},
\]
\noindent where $ \sigma_{st}(C) $ is the number of shortest paths from $ s $ to $ t $ passing through at least one vertex in $C$ (if $s \in C$ or $t \in C$ then $ \sigma_{st}(S) = 0 $), and $\sigma_{st}$ is the number of all shortest paths between $s$ and $t$.
\end{definition}

Eventually, the extension of closeness centrality~(Definition~\ref{def:cen:close}) is defined as follows:

\begin{definition}[Group closeness centrality]\label{def:group:close}
For a given graph $ G =(V,E) $ and a group of nodes $C \subseteq V$ closeness centrality is a function $c_{gc}: 2^V \rightarrow \mathbb{R} $:
\[
	c_{gc}(C) = \sum_{s \notin C} d(C,s)
\]
\noindent where $ dist(C,s) $ is a distance between group $ C $ and node $ s $  (see Section~\ref{sec:graphs}).
\end{definition}

In this dissertation, we will be looking at cooperative games on graphs, with the set of players $V$ for some $(V,E)=G \in \mathscr{G}$, where $\mathscr{G}$ is an infinite set of all graphs.\footnote{\footnotesize The graphs can be simple, weighted, or directed, it will depend on the context.} A coalition of players $C$ will simply be any subset of $C \subseteq V$. The characteristic function, $\nu_{G} \in \mathscr{V}_{G}$, will be any function $\nu_{G}: 2^{V(G)} \rightarrow~\mathbb{R}$, but in this work we focus on the group centralities, so we have $\nu_{G} = c_{g}$. For the purposes of defining a cooperative game on any network, we will systematically associate characteristic functions with graphs through \emph{representation functions}.

\begin{definition}[Representation function]\label{def:rep:fun}
A \emph{representation function} is a function~$\psi : \mathscr{G} \to \mathscr{V}_{G} $ that maps every graph~$G=(V,E)$ onto a cooperative game~$(N,\nu_G) \in \mathscr{V}_{G}$ with~$N=V$.
\end{definition}

Now, solution concepts like Shapley value can be used in the network setting by applying them to the characteristic function that a network represents.

\begin{definition}[Game-theoretic centrality measure]\label{def:gt:cen}
Formally, we define a \textit{game-theoretic centrality measure} as a pair~$(\psi, \phi)$ consisting of a representation function $\psi$ and a solution concept $\phi$.
\end{definition}

\begin{example}
Let us consider a \textit{game-theoretic centrality measure}~$(\psi_B, \phi^{\SV})$. We say that $\psi_B$ is a representation function since it associates a coalitional game with any graph $G=(V,E)$, i.e. every graph \textit{represents} a cooperative game. We have $\psi_B(G) = (V,\nu_G)$, where $V$ is the set of nodes and $\nu_G: 2^V \to \mathbb{R}$ is the characteristic function. Let $\nu_G$ be the ranking of groups of nodes in $G$ based on group betweenness centrality. In other words, $\psi_B$ is simply the group betweenness $c_{gb}$ centrality for any graph. For a specific graph $G$, the importance of each node $u \in V$ according to $(\psi_B,\phi^{\SV})$ is evaluated by the Shapley value of the game $\psi_B(G)$, i.e. $\phi^{\SV}(\nu_G)$. Since we started off with a well-known centrality measure (betweenness centrality) and applied a game-theoretic solution concept to it (the Shapley value), we call the resulting centrality measure a game-theoretic extension of betweenness centrality, or the Shapley-value based betweenness centrality.
\end{example}

The algorithms presented in the next section are designed for the generalization of the following two game-theoretic centralities.

\begin{definition}[Shapley-value based degree centrality]\label{def:gt:deg:sv}
\textit{Shapley-value based degree centrality} is a pair~$(\psi_D, \phi^{\SV})$ consisting of a representation function $\psi_D(G) = c_{gd}$ that maps each graph to the game with group degree centrality and a Shapleu value solution concept $\phi^{\SV}$.
\end{definition}

\begin{definition}[Shapley-value based closeness centrality]\label{def:gt:clo:sv}
\textit{Shapley-value based closeness centrality} is a pair~$(\psi_C, \phi^{\SV})$ consisting of a representation function $\psi_C(G) = c_{gc}$ that maps each graph to the game with group closeness centrality and a Shapleu value solution concept~$\phi^{\SV}$.
\end{definition}

\section{The Algorithms}\label{sec:gtc:algo}

\noindent In this section, we present five characteristic function formulations, each designed to convey a specific centrality notion. These games are a generalizations of Shapley-value based degree centrality and Shapley-value based closeness centrality.

We need to clarify one important aspect of group centralities. From the Definitions~\ref{def:group:degree} and \ref{def:group:close} we see that for each graph $G=(V,E)$ we have $c_{gd}(V)=c_{gc}(V)=0$. However, in the four games $g_1$,$g_2,g_3$ and $g_5$ being under consideration in this chapter we slightly modify Definition~\ref{def:group:degree} and assume that for each coalition $C \subseteq V$ group degree centrality is: $c_{gd}(C) + |C|$. This assumptions were made in order to catch the intuition that the value of coalition $C$ is in fact the size of the 'sphere of influence' made by this coalition on the network. The coalition $C$ is expected to make influence on all its members. So, the five games analyzed in this section are:

\begin{itemize}
\item[$g_1$] In this game the value of coalition $C$ is a function of its own size and of the number of nodes that are immediately reachable from $C$. It is simply the Shapley-value based degree centrality.
\item[$g_2$] In this game the value of coalition $C$ is a function of its own size and of the number of nodes that are immediately reachable in at least $k$ different ways from $C$. This game is inspired by \citet{Bikhchandani:et:al:1992} and is an instance of the general threshold model introduced by \citet{Kempe:et:al:2005}. It has a natural interpretation: an agent ``becomes influenced'' (with ideas, information, marketing message, etc.) only if at least $k$ of his neighbors have already become influenced.
 \item[$g_3$] This game concerns weighted graphs (unlike $g_1$ and $g_2$). Here, the value of coalition $C$ depends on its size and on the set of all nodes within a cutoff distance of $C$, as measured by the shortest path lengths on the weighted graph.
\item[$g_4$]  This game generalizes $g_3$ by allowing the value of $C$ to be an arbitrary non-increasing function $f(.)$ of the distance between $C$ and the other nodes in the network. The intuition here is that the coalition has more influence on closer nodes than on those further away---a property that cannot be expressed with the standard closeness centrality. Thus, $g_4$ is the Shapley-value based closeness centrality.
\item[$g_5$] The last game is an extension of $g_2$ to the case of weighted networks. Here, the value of $C$ depends on the adjacent nodes that are connected to the coalition with weighted edges whose sum exceeds a given threshold $w_{cutoff}$ (recall that in $g_2$ this threshold is defined simply by the integer $k$). Whereas in $g_3$ and $g_4$ weights on edges are interpreted as distance, in $g_5$ they should be interpreted as a \textit{power of influence}.
\end{itemize}

The relationships among all five games are graphically presented in Figure \ref{figure:relationships}.

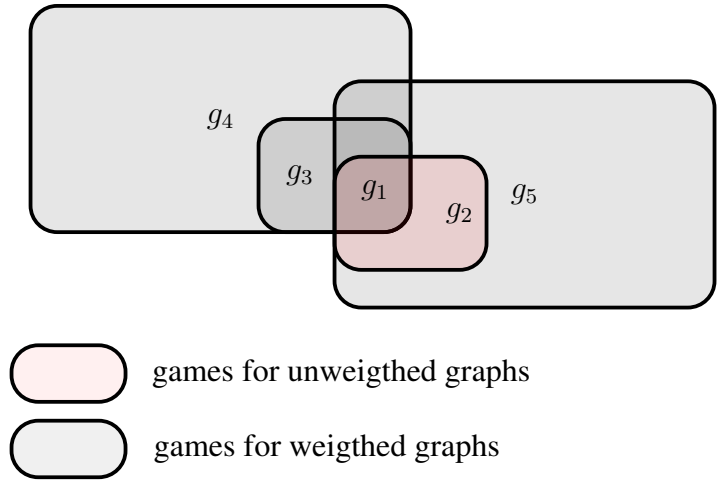
\begin{figure}[t]
\begin{center}
\begin{tikzpicture}

\tikzstyle{abstract}=[rectangle, draw=black, line width=0.05cm, rounded corners=10pt, fill=red, fill opacity=0.1, text centered, anchor=north, text=black, text opacity=1]

   		\draw (2,3)  node[abstract,minimum width=5cm,minimum height=3cm,fill=black]{$g_4$};

   		\draw (3.5,1.5)  node[abstract,minimum width=2cm,minimum height=1.5cm,fill=black, label={[label distance=-0.9cm]180:$g_3$}]{};

   		\draw (6,2)  node[abstract,minimum width=5cm,minimum height=3cm,fill=black]{$g_5$};
   		
   		\draw (4.5,1)  node[abstract,minimum width=2cm,minimum height=1.5cm,fill=red,label={[label distance=-2cm]180:$g_2$},label={[label distance=-0.9cm]170:$g_1$}]{};
   		
\draw (0,-1.5)  node[abstract,minimum width=1.5cm,minimum height=0.75cm,fill=red!60,label={[label distance=-7cm]180:games for unweigthed graphs}]{};
\draw (0,-2.5)  node[abstract,minimum width=1.5cm,minimum height=0.75cm,fill=black!60,label={[label distance=-6.6cm]180:games for weigthed graphs}]{};

\end{tikzpicture}

\caption{Euler diagram showing the relationships among all five games considered in this chapter. Specifically, game $g_2$ generalizes $g_1$; $g_3$ generalizes $g_1$ and is further generalized by $g_4$; $g_5$ generalizes $g_2$. Finally, we note that there are certain instances of games that can be represented as $g_3$, $g_4$ and $g_5$.}
\label{figure:relationships}
\end{center}
\end{figure}

The computation of the Shapley value for each of the above five games (see Table \ref{front_table} for an overview) is the main focus of the paper by \citet{Michalak:et:al:2013}. In this dissertation we skip the technical part of constructing algorithms, and instead of this, for each game we present the extensive evaluation.

\setlength{\tabcolsep}{0.2em}
\begin{center}
\begin{table}[t]
\begin{center}
		\begin{tabular}{cclcc}
		 	\hline
		 	\hline
		Game & Graph & \multicolumn{1}{c}{Value of a coalition $C$, i.e., $\nu(C)$}  & Complexity & Accuracy \tabularnewline
		\hline
		\hline
		$g_{1}$ & $UW$ & $\nu(C)$ is the number of nodes in $C$ and & $O(|V|+|E|)$ & exact \tabularnewline
		        &      &  those immediately reachable from $C$ &  &\tabularnewline
		\hline
		$g_{2}$ & $UW$ & $\nu(C)$ is the number of  nodes in $C$ and          &  $O(|V|+|E|)$  &exact \tabularnewline
		        &      & those immediately reachable from $C$,           &   &\tabularnewline
		        &      & but via at least $k$ different edges &  &\tabularnewline
		\hline
		$g_{3}$ & $W$ & $\nu(C)$ is the number of nodes in $C$ and  &  $O(|V||E|+|V|^{2}\textnormal{log}|V|)$  &exact \tabularnewline
		        &      &  those not further than $d_{{\cutoff}} \textnormal{ away}$     &   &\tabularnewline

		\hline
		$g_{4}$ & $W$ & $\nu(C)$ is the sum of $f(.)$'s --- the non-  &  $O(|V||E|+|V|^{2}\textnormal{log}|V|)$ &exact \tabularnewline
		        &     &  -increasing functions of the distance     &   &\tabularnewline
		        &     &  between $C$ and other nodes     &   &\tabularnewline
		\hline
		$g_{5}$ & $W$ & $\nu(C)$ is the number of nodes in $C$ and	& $O(|V||E|)$ &approx. \tabularnewline
				&	  & those directly connected to $C$ via edges  &    & $\sim 5\mbox{-}10\%$ \tabularnewline
				&	  & which sum of weights  exceeds  $W_{{\cutoff}}$ &   &\tabularnewline
		\hline
		\hline
		\end{tabular}
\end{center}
\caption{Games considered in this dissertation and our results ($UW$ denotes unweighted graphs and $W$ weighted).}
\label{front_table}
\end{table}
\end{center}

\subsection{Shapley Value-Based Degree Centrality}

\noindent Let $G(V,E)$ be an unweighted, undirected network. The characteristic function is defined as:

\begin{equation*}
\psi_1(G)(C) = \nu_{1}(C) = \nu_{gd}(C) + |C|
\end{equation*}.

The above game was applied by \citet{Suri:Narahari:2010} to find out influential nodes in social networks and it was shown to deliver very promising results concerning the target set selection problem (see \citet{Kempe:et:al:2003}). It is therefore desired to compute the Shapley values of all nodes for this game. We shall now present an exact algorithm for this computation rather than obtaining results through Monte Carlo simulation as was done by Suri and Narahari.

In more detail, to evaluate the Shapley value of node $v$, consider all possible permutations of the nodes in which $v$ would make a positive marginal contribution to the coalition of nodes occurring before itself. Let the set of nodes occurring before node $v$ in a random permutation of nodes be denoted $P_{\pi}(v)$. Let the neighbours of node $v$ in the graph $G(V,E)$ be denoted $N(v)$ and the degree of node $v$ be denoted $deg(v)$.

The key question to ask is: what is the necessary and sufficient condition for node $v$ to marginally contribute node $u\in N(v)\cup \{v\}$ to $\textnormal{fringe}(P_{\pi}(v))$? Clearly, this happens if and only if neither $u$ nor any of \emph{its} neighbours are present in $P_{\pi}(v)$. Formally, $(N(u)\cup \{v\})\cap P_{\pi}(v)=\emptyset$.

Now we are going to show that the above condition holds with probability $\frac{1}{1+deg(u)}$.

\begin{proposition}
\label{prop1}
The probability that in a random permutation none of the vertices from $N(u)\cup \{u\}$ occurs before $v$, where $v$ and $u$ are neighbours, is  $\frac{1}{1+deg(u)}$.
\end{proposition}

\begin{proof}
We need to count the number of permutations that satisfy:
\begin{equation}
\label{eq:proposition1}
\forall_{n \in (N(v)\cup \{u\})}{\pi(v) < \pi(n) }.
\end{equation}

To this end:

	\begin{itemize}
	\item Let us choose $ |(N(u)\cup \{u\}|$ positions in the sequence of all elements from $V$. We can do this in $ \binom{|V|}{1+deg(u)}$ ways.
	\item Then, in the last $deg(u)$ chosen positions, place all elements from $(N(u)\cup \{u\})\setminus\{v\}$. Directly before these, place the element $v$. The number of such line-ups is $(deg(u))!$.
	\item The remaining elements can be arrange in $ (|V| - (1+deg(u))! $ different ways.
	\end{itemize}

All in all, the number of permutations satisfying condition (\ref{eq:proposition1}) is:
 \begin{equation}
 \textstyle\binom{|V|}{1+deg(u)}(deg(u))!(|V| - (1+deg(u))! =  \frac{|V|!}{1+deg(u)}; \nonumber
\end{equation}
\noindent thus, the probability that one of such permutations is randomly chosen is $\frac{1}{1+deg(u)}$.
\end{proof}

\begin{algorithm}[t]
\caption{Computing the Shapley value for Game 1}
\label{algo:gtc:game1}
\SetAlgoVlined %The old version of this is: \SetVline
\LinesNumbered %The old version of this is: \linesnumbered
\SetAlgoLined
\BlankLine
\KwIn{Unweighted graph $G(V,E)$}
\KwOut{Shapley values of all nodes in $V$ for game $g_{1}$}
\BlankLine
\ForEach{$v\in V$}{
$\phi^{\SV}_v = \frac{1}{1+deg(v)}$\;
\ForEach{$u\in N(v)$}{
$\phi^{\SV}_v += \frac{1}{1+deg(u)}$\;
}
}
return $\phi^{\SV}$\;
\end{algorithm}

It is possible to derive some intuition from the above algorithm. If a node has a high degree, the number of terms in its Shapley value summation above is also high. But the terms themselves will be inversely related to the degree of neighboring nodes. This gives the intuition that a node will have high centrality not only when its degree is high, but also whenever its degree tends to be higher in comparison to the degree of its neighboring nodes. In other words, \emph{power comes from being connected to those who are powerless}, a fact that is well-recognized by the centrality literature \citep{Bonacich:1987}. Following the same reasoning, we can also easily predict how dynamic changes to the network, such as adding or removing an edge, would influence the Shapley value.\footnote{\footnotesize Many real-life networks are in fact dynamic and the challenge of developing fast streaming algorithms has recently attracted considerable attention in the literature \citep{Lee:et:al:2012}.} Adding an edge between a powerful and a powerless node will add even more power to the former 	and will decrease the power of the latter. Naturally, removing an edge would have the reverse effect.

Algorithm \ref{algo:gtc:game1} cycles through all nodes and through their neighbours, so its running time is $O(|V|+|E|)$.

Finally, we note that Algorithm \ref{algo:gtc:game1} can be adopted to directed graphs with a couple of simple modifications. Specifically, in order to capture how many nodes we can access a given node from, the degree of a node should be replaced with indegree. Furthermore, a set of neighbours of a given node $v$ should consist of those nodes to which an edge is directed from $v$.

\subsection{Shapley Value-Based $k$-influence Degree Centrality}\label{sec:game2}

\noindent We now consider a more general game formulation for an unweighted graph $G(V,E)$, where the value of a coalition includes the number of agents that are either in the coalition or are adjacent to at least $k$ agents who are in the coalition. Formally, we consider game $g_{2}$ characterised by $\nu_{2}:2^{V(G)}\to\mathbb{R}$, where

\begin{equation*}
\psi_2(G)(C) = \nu_{2}(C) =
\begin{cases}
0 & \textnormal{ if } C = \emptyset \\
|\{v : v\in C \textnormal{ (or) } |N(v)\cap C|\geq k\}| & \textnormal{ otherwise.}
\end{cases}
\end{equation*}

The second game is an instance of the General Threshold Model that has been widely studied in the literature \citep{Kempe:et:al:2005,Granovetter:1978}. Intuitively, in this model each node can become active if a monotone activation function reaches some threshold. The instance of this problem has been proposed by  \citet{Goyal:et:al:2010}, where the authors introduced a method of learning influence probabilities in social networks (from users' action logs). However, in many realistic situations much less information is available about a network so it is not possible to assess specific probabilities with which individual nodes become active. Consequently, much simpler models are studied. \citet{Bikhchandani:et:al:1992}, for instance, ``consider a teenager deciding whether or not to try drugs. A strong motivation for trying out drugs is the fact that friends are doing so. Conversely, seeing friends reject drugs could help persuade the teenager to stay clean''. This situation is modelled by the second game; the threshold for each node is $k$ and the activation function is $f(S) = |S|$. Another example is viral marketing or innovation diffusion analysis. Again, in this application, it is often assumed that an agent will ``be influenced'' only if at least $k$ of his neighbors have already been convinced by \cite{Valente:1996}. Note that this game reduces to game $g_{1}$ for $k=1$.

Although this game is a generalization of game $g_{1}$, it can still be solved to obtain the Shapley values of all nodes in $O(|V|+|E|)$ time, as formalised by Algorithm \ref{algo:gtc:game2}.

\begin{algorithm}[t]
\caption{ Computing the Shapley value for Game 2}
\label{algo:gtc:game2}
\SetAlgoVlined %The old version of this is: \SetVline
\LinesNumbered %The old version of this is: \linesnumbered
\SetAlgoLined
\BlankLine
\KwIn{Unweighted graph $G(V,E)$, positive integer $k$}
\KwOut{Shapley value of all nodes in $V$ for game $g_{2}$}
\BlankLine
\ForEach{$v\in V$}{
$\phi^{\SV}_v =\textnormal{min}(1,\frac{k}{1+deg(v)})$\;
\ForEach{$u\in N(v)$}{
$\phi^{\SV}_v += \textnormal{max}(0,\frac{deg(u)-k+1}{deg(u)(1+deg(u))})$\;
}
}
return $\phi^{\SV}$\;
\end{algorithm}

An even more general formulation of the game is possible by allowing $k$ to be a function of the agent, i.e.,  each node $v\in V$ is assigned its own unique attribute $k(v)$. This translates to an application of the form: agent $i$ is convinced if and only if at least $k_{v}$ of his neighbors are convinced, which is a frequently used model in the literature \citep{Valente:1996}.

The above argument does not use the fact that $k$ is constant across all nodes. So this generalized formulation can be solved by a simple modification to the original Shapley value expression:

$$SV(v)=\frac{k(v)}{1+deg(v)}+\sum_{u\in N(v)}\frac{1+deg(u)-k(u)}{deg(u)(1+deg(u))}.$$

The above equation (which is also implementable in $O(|V|+|E|)$ time) assumes that $k(v)\leq1+deg(v)$ for all nodes $v \in V$. This condition can be assumed without loss of generality because all cases can still be modelled (we set $k(v)=1+deg(v)$ for the extreme case where node $v$ is never convinced no matter how many of its neighbors are already convinced).

Finally, we note that Algorithm \ref{algo:gtc:game2} can be adapted to a case of directed graphs along the same lines as Algorithm \ref{algo:gtc:game1}.

\subsection{Shapley Value-Based $d_{\cutoff}$-distance Degree Centrality}\label{sec:game3}

\noindent Hitherto, our games have been confined to unweighted networks. But in many applications, it is necessary to model real-world networks as weighted graphs. For example, in the co-authorship network mentioned in the introduction, each edge is often assigned a weight proportional to the number of joint publications the corresponding authors have produced \citep{Newman:2001}.

This subsection extends game $g_{1}$ to the case of weighted networks. Whereas game $g_{1}$ equates $\nu(C)$ to the number of nodes located within one hop of some node in $C$, our formulation in this subsection equates $\nu(C)$ to the number of nodes located within a distance $d_{{\cutoff}}$ of some node in $C$. Here, distance between two nodes is measured as the length of the shortest path between the nodes in the given weighted graph $G =(V,E,\lambda)$ (see Definition~\ref{def:graph:weighted})

Formally, we define game $g_{3}$, where for each coalition $C\subseteq V(G)$,

\begin{equation*}
\psi_3(G)(C) =\nu_{3}(C) =
\begin{cases}
0 & \textnormal{ if } C = \emptyset \\
|\{v : \mathit{dist}(C,u)\leq d_{{\cutoff}}\}| & \textnormal{ otherwise.}
\end{cases}
\end{equation*}

 Clearly, $g_3$ can be used in all the settings where $g_1$ is applicable; for instance, in the diffusion of information in social networks or to analyse research collaboration networks \citep{Suri:Narahari:2010}. Moreover, as a more general game, $g_3$ provides additional modelling opportunities. For instance, \citeauthor{Suri:Narahari:2010} suggest that a ``\textit{more intelligent}'' way for sieving nodes in the neighbourhood would improve their algorithm for solving the target selection problem (\emph{top-k} problem). Now, $g_3$ allows us to define a different $\cutoff$ distance for each node in Suri and Narahari's setting. Furthermore, $g_3$ is a specific case of the more general model $g_4$ which will be discussed in next subsection.

Before we propose an algorithm, we need to introduce some extra notation. Define the \emph{extended neighborhood} $N(v,d_{{\cutoff}}) = \{u \neq v : \mathit{dist}(v,u) \leq d_{{\cutoff}}\}$, i.e.,  the set of all nodes whose distance from $v$ is at most $d_{{\cutoff}}$. Denote the size of $N_{G}(v,d_{{\cutoff}})$ by $deg(v,d_{{\cutoff}})$.

		\begin{algorithm}[t]
		\caption{Computing the Shapley value for Game 3}
		\label{algo:gtc:game3}
		\SetAlgoVlined %The old version of this is: \SetVline
\LinesNumbered %The old version of this is: \linesnumbered
		\SetAlgoLined
		\BlankLine
		\KwIn{Weighted graph $G =(V,E,\lambda)$, $d_{{\cutoff}} > 0$}
		\KwOut{Shapley value of all nodes in $G$ for game $g_{3}$}
		\BlankLine

		\ForEach {$v\in V(G)$}
		{
			DistanceVector D = Dijkstra($v$,$G$)\;

			extNeighbors($v$) = $\emptyset$; extDegree($v$) = 0\;

			\ForEach {$u\in V(G)$ such that $u\neq v$}
			{
				\If{$D(u) \leq d_{{\cutoff}}$}
				{
					extNeighbors($v$).push($u$)\;
					extDegree($v$)++\;
				}
			}
		}

		\ForEach{$v\in V(G)$}
		{
			$\phi^{\SV}_v = \frac{1}{1+extDegree(v)}$\;
			
			\ForEach{$u \in extNeighbors(v)$}
			{
				$\phi^{\SV}_v += \frac{1}{1+extDegree(u)}$\;
			}
		}

		return $\phi^{\SV}$\;

		\end{algorithm}

	Algorithm \ref{algo:gtc:game3} works as follows: for each node $v$ in the network $G(V,E)$, the extended neighborhood $N(v,d_{{\cutoff}})$ and its size $deg(v,d_{{\cutoff}})$ are first computed using Dijkstra's algorithm in $O(|E|+|V|\textnormal{log}|V|)$ time \citep{Cormen:2001}. The results are then used to compute Shapley value, which takes maximum time $O(|V|^{2})$. In practice this step runs much faster because the worst case situation only occurs when every node is reachable from every other node within $d_{{\cutoff}}$. Overall the complexity is $O(|V||E|+|V|^{2}\textnormal{log}|V|)$.
	
Furthermore, to deal with directed graphs we need to redefine the notion of $extDegree$ and $extNeighbors$ for a given node $u$ in Algorithm \ref{algo:gtc:game3}. The former will be the number of vertices from which the distance to $u$ is smaller than, or equal to, $d_{\cutoff}$. The latter will be the set of nodes whose distance from $u$ is at most $d_{\cutoff}$.

Finally, we make the observation that the the above algorithm does not depend on $d_{{\cutoff}}$ being constant across all nodes. Indeed, each node $v\in V(G)$ may be assigned its own unique value $d_{{\cutoff}}(v)$, where $\nu(C)$ would be the number of agents $v$ who are within a distance $d_{{\cutoff}}(v)$ from $C$. For this case, the above proof gives:

$$\phi^{\SV}_v \ \ = \hspace{-1em} \sum_{  \stackrel{\scriptstyle u : \mathit{dist}(v,u)} {\leq d_{{\cutoff}}(u)}  } \frac{1}{1+deg(u,d_{{\cutoff}}(u))}.$$

\subsection{Shapley Value-Based Closeness Centrality}\label{Ssec:game4}

\noindent This subsection further generalizes game $g_{3}$, again taking motivation from real-life network problems. In game $g_{3}$, all agents at distances $d_{\textnormal{agent}}\leq d_{{\cutoff}}$ contributed equally to the value of a coalition. However, this assumption may not always hold true because in some applications we intuitively expect agents closer to a coalition to contribute more to its value. For instance, we expect a Facebook user to exert more influence over his immediate circle of friends than over ``friends of friends'', even though both may satisfy the $d_{{\cutoff}}$ criterion. Similarly, we expect a virus-affected computer to infect a neighboring computer more quickly than a computer two hops away.

In general, we expect that an agent at distance $d$ from a coalition would contribute $f(d)$ to its value, where $f(.)$ is a positive valued decreasing function of its argument. More formally, we define game $g_{4}$, where the value of a coalition $C$ is given by:

\begin{equation*}
\psi_4(G)(C) =\nu_{4}(C) =
\begin{cases}
0 & \textnormal{ if } C = \emptyset \\
\sum_{v \in V}f(\mathit{dist}(C,v)) & \textnormal{ otherwise.}
\end{cases}
\end{equation*}

In the case of Shapley value-based closeness centrality, the key question to ask is: what is the expected value of the marginal contribution of $v$ through node $u \neq v$ to the value of coalition $P_{\pi}(v)$? Let this marginal contribution be denoted $MC(v,u)$. Clearly:

%The key question to ask is: what is the expected value of the marginal contribution of $v$ through node $u\neq v$ to the value of coalition $P_{\pi}(v)$? Let this marginal contribution be denoted $MC(v,u)$. Clearly:

\begin{equation*}
MC(v,u) =
\begin{cases}
0 & \textnormal{ if } dist(v,u)\geq dist(u,P_{\pi}(v)) \\
f(dist(v,u)) - f(dist(u,P_{\pi}(v))) & \textnormal{ otherwise.}
\end{cases}
\end{equation*}

Let $D_{u}=\{d_{1},d_{2}...d_{|V|-1}\}$ be the distances of node $u$ from all other nodes in the network, sorted in increasing order. Let the nodes corresponding to these distances be $\{w_{1},w_{2}...w_{|V|-1}\}$, respectively. Let $k_{vu}+1$ be the number of nodes (out of these $|V|-1$) whose distances to $u$ are $\leq \textnormal{distance}(v,u)$. Let $w_{k_{vu}+1}=v$ (i.e.,  among all nodes that have the same distance from $u$ as $v$, $v$ is placed last in the increasing order).

	We use \emph{literal} $w_{i}$ to mean $w_{i}\in P_{\pi}(v)$ and the literal $\overline{w_{i}}$ to mean $w_{i}\notin P_{\pi}(v)$. Define a sequence of boolean variables $p_{k}=\overline{u}\wedge\overline{w_{1}}\wedge\overline{w_{2}}\wedge...\wedge\overline{w_{k}}$ for each $0\leq k \leq |V|-1$. Finally denote expressions of the form $MC(v,u | F)$ to mean the marginal contribution of $v$ to $P_{\pi}(v)$ through $u$ given that the coalition $P_{\pi}(v)$ satisfies the boolean expression $F$.

		\begin{gather*}
		MC(v,u|p_{k_{vu}+1}\wedge w_{k_{vu}+2}) = f(d_{k_{vu}+1})-f(d_{k_{vu}+2}),\\
		MC(v,u|p_{k_{vu}+2}\wedge w_{k_{vu}+3}) = f(d_{k_{vu}+1})-f(d_{k_{vu}+3}),\\
		\hspace{5.5em} \vdots \hspace{3em} \vdots \hspace{3em} \vdots \hspace{3em} \\
		MC(v,u|p_{|V|-2}\wedge w_{|V|-1}) = f(d_{k_{vu}+1})-f(d_{|V|-1}),\\
		MC(v,u|p_{|V|-1}) = f(d_{k_{vu}+1}).
		\end{gather*}

With this notation, we obtain expressions for $MC(v,u)$ by splitting over the above \emph{mutually exclusive} and \emph{exhaustive} (i.e.,  covering all possible non-zero marginal contributions) cases.

Now, we need to determine the probability of $\textnormal{Pr}(p_{k}\wedge w_{k+1})$.

\begin{proposition}
\label{prop3}
The probability that in a random permutation none of the nodes from $\{v_j,w_1,\ldots,w_k\}$ occur before $v_i$ and the node $w_{k+1}$ occurs before $v_i$ is  $\frac{1}{(k+1)(k+2)}$.
\end{proposition}

\begin{proof}
Let us count the number of permutations that satisfy:
\begin{equation}
\label{eq:proposition3}
\forall_{v \in \{v_j,w_1,\ldots,w_k\}}{\pi(v) < \pi(v)} ~\wedge~ \pi(v < \pi(w_{k+1}).
\end{equation}

To this end:

	\begin{itemize}
	\item Let us choose $ |\{v_j,w_1,\ldots,w_k\} \cup \{u\} \cup \{w_{k+1}\}|$ positions in the sequence of all elements from $V$. We can do this in $ \binom{|V|}{k+3}$ ways.
	\item Then, in the last $k+1$ chosen positions, we place all elements from $\{v_j,w_1,\ldots,w_k\}$. Directly before these, we place the element $v$, and then vertex $w_{k+1}$ . The number of such line-ups is $(k+1)!$.
	\item The remaining elements can be arrange in $(|V| - (k+3)!$ different ways.
	\end{itemize}

Thus, the number of permutations satisfying (\ref{eq:proposition3}) is:
 \begin{equation}
 \textstyle\binom{|V|}{k+3}(k+1)!(|V| - (k+3))! =  \frac{|V|!}{(k+1)(k+2)}, \nonumber
\end{equation}
\noindent and the probability that one of such permutations is randomly chosen is $\frac{1}{(k+1)(k+2)}$.
\end{proof}

With the above proposition we find that:

$$\textnormal{Pr}(p_{k}\wedge w_{k+1}) \frac{1}{(k+1)(k+2)} \ \forall \ 1+k_{vu} \leq k \leq |V|-2.$$

Using the $MC(v,u)$ equations and the probabilities $\textnormal{Pr}(p_{k}\wedge w_{k+1})$:

\begin{align*}
E[MC(v,u)]  & = \left[ \sum_{k=1+k_{vu}}^{|V|-2}\!\! \frac{f(\textnormal{distance}(v,u))-f(d_{k+1})}{(k+1)(k+2)}\right] + \frac{f(\textnormal{distance}(v,u))}{|V|}\\
& = \frac{f(\textnormal{distance}(v,u))}{k_{vu}+2} \ - \sum_{k=k_{vu}+1}^{|V|-2}\frac{f(d_{k+1})}{(k+1)(k+2)}.\\
\end{align*}

		\begin{algorithm}[t]
		\caption{Computing the Shapley value for Game 4}
		\label{algo:gtc:game4}
		\SetAlgoVlined %The old version of this is: \SetVline
\LinesNumbered %The old version of this is: \linesnumbered
		\SetAlgoLined
		\BlankLine
		\KwIn{Weighted graph $G =(V,E,\lambda)$, function $f:\mathbb{R}^{+}\to\mathbb{R}^{+}$}
		\KwOut{Shapley value of all nodes in $G$ for game $g_{4}$}
		\BlankLine
		
		\textbf{Initialise: } $\forall v\in V(G)$ \textbf{set} $\phi^{\SV}_v = 0$\;
		\ForEach{$v \in V(G)$}
		{
			[Distances D, Nodes w] = Dijkstra($v$,$G$)\;			
			sum = 0; index = |V|-1; prevDistance = -1, prevSV = -1\;
			\While {index > 0}
			{
				\eIf{D(index) == prevDistance}
				{
					currSV = prevSV\;
				}
				{
					currSV = $\frac{f(D(index))}{1+index}-sum$\;
				}
				node = w(index)\;
				$\phi^{\SV}_{node}  += currSV$\;
				sum += $\frac{f(D(index))}{index(1+index)}$\;
				prevDistance = D(index), prevSV = currSV\;
				$\textnormal{index-}\textnormal{-}$\;
			}
			$\phi^{\SV}_v += f(0) - sum$\;
		}
		return $\phi^{\SV}$\;
		\end{algorithm}

In Algorithm \ref{algo:gtc:game4} we use $D_{v}=\{d_{1},d_{2}...d_{|V|-1}\}$ as a vector of distances from node $v$ to all other nodes in the network, sorted in increasing order. Also, the vector $w=\{w_{1},w_{2}...w_{|V|-1}\}$ represents nodes corresponding to these distances, respectively. For each vertex $v$, a vector of distances to every other vertex is first computed using Dijkstra's algorithm \citep{Cormen:2001}. This yields a vector $D_{v}$ that is already sorted in increasing order. This vector is then traversed in reverse to compute the Shapley value of the appropriate node $w$. After the traversal, the Shapley value of $v$ itself is updated.  This process is repeated for all nodes $v$ so that at the end of the algorithm, the Shapley value is computed exactly in $O(|V||E|+|V|^{2}\textnormal{log}|V|)$ time.	

Our final observation is that Algorithm \ref{algo:gtc:game4} works also for directed graphs as long as we use the appropriate version of Dijkstra's algorithm (see, e.g., \citet{Cormen:2001}).

\subsection{Shapley Value-Based $W_{{\cutoff}}$-influence Degree Centrality}\label{sec:game5}

\noindent In this subsection, we generalize game $g_{2}$ for the case of weighted networks. Given a positive weighted network $G =(V,E,\lambda)$ and a value $W_{{\cutoff}}(v)$ for every node $v\in V$, we first define $\lambda(v,C)=\sum_{u\in C}\lambda(v,u)$ for every coalition $C$, where $\lambda(v,u)$ is the weight of the edge between nodes $v$ and $u$ (or $0$ if there is no such edge). With this notation, we define game $g_{5}$ by the characteristic function:

\begin{equation*}
\nu_{5}(C) =
\begin{cases}
0 & \textnormal{ if } C = \emptyset \\
\textnormal{size}(\{v:v\in C\ (or)\ \lambda(v,C)\geq W_{{\cutoff}}(v)\}) & \textnormal{ otherwise.}
\end{cases}
\end{equation*}

The formulation above has applications in, for instance, the analysis of information diffusion, adoption of innovations, and viral marketing. Indeed, many cascade models of such phenomena on weighted graphs have been proposed \citep{Granovetter:1978,Kempe:et:al:2003,Young:2005} which work by assuming that an agent will change state from ``inactive'' to ``active'' if and only if the sum of the weights to all active neighbors is at least equal to an agent-specific cutoff.

Although we have not been able to come up with an exact formula for the Shapley value in this game\footnote{\footnotesize
Computing the Shapley value for this game involves determining whether the sum of weights on specific edges, adjacent to a random coalition, exceeds the threshold. This problem seems to be at least as hard as computing the Shapley value in weighted voting games, which is \#P-Complete \citep{Elkind_et_al_2009}.}, our analysis yields an approximate formula which was found to be accurate in practice.

We will need the following additional notation:
 \begin{itemize}
 \item let the weights of edges between $v$ and each of the nodes in $N(v)$ be $\lambda_{v}=\{\lambda(v,u),\lambda_{1},\lambda_{2}...\lambda_{deg(v)-1}\}$ in that order;
 \item let $\alpha_{v}$ be the sum of all the weights in $\lambda_{v}$ and $\beta_{v}$ be the sum of the squares of all the weights in $\lambda_{v}$; and finally
 \item let $X_{t}^{vu}$ be the sum of a $t$-subset of $\lambda_{v} \setminus \{\lambda(v,u)\}$ drawn uniformly at random from the set of all such possible $t$-subsets.
 \end{itemize}

	\begin{algorithm}[t]
	\caption{Computing the Shapley value for Game 5}
	\label{algo:gtc:game5}
	\SetAlgoVlined %The old version of this is: \SetVline
\LinesNumbered %The old version of this is: \linesnumbered
	\SetAlgoLined
	\BlankLine
	\KwIn{Weighted network $G =(V,E,\lambda)$, cutoffs $W_{{\cutoff}}(v)$ for each $v\in V$}
	\KwOut{Shapley value of all nodes in $G$ for game $g_{5}$}
	\BlankLine

	\ForEach{$v \in V$}
	{
		compute and store $\alpha_{i}$ and $\beta_{i}$\;
	}

	\ForEach{$v \in V$}
	{
		$\phi^{\SV}_v = 0$\;
		\ForEach {$m$ in $0$ to $deg(v)$}
		{
			compute $\mu = \mu(X_{m}^{vv})$, $\sigma = \sigma(X_{m}^{vv})$, $ p = \textnormal{Pr} \{ {\cal {N}} (\mu,\sigma^{2}) < W_{{\cutoff}}(v)\}$\;
			$\phi^{\SV}_v += \frac{p}{1+deg(v)}$\;
		}

		\ForEach {$u \in N(v)$}
		{
			p = 0\;		
			\ForEach {$m$ in $0$ to $deg(u)-1$}
			{
				compute $\mu = \mu(X_{m}^{vu})$, $\sigma = \sigma(X_{m}^{vu})$ and $z$ = $Z_{m}^{vu}$\;
				$p$ += $z\frac{deg(u)-m}{deg(u) (deg(u)+1)}$\;
			}
			$\phi^{\SV}_v += p$\;
		}
	}
	return $\phi^{\SV}$\;
	\end{algorithm}

While in each graph it holds that $ \sum_{v \in V} deg(v) \leq 2|E|$, Algorithm \ref{algo:gtc:game5} implements an $O(|V|+\sum_{v \in V} \sum_{u \in N(v)} deg(u)) \leq O(|V| + |V||E|) = O(|V||E|)$ solution to compute the Shapley value for all agents in game $g_{5}$ using the above approximation. Furthermore, we make the following observation: for small $deg(u)$, one might as well use the brute force computation to determine Shapley value in $O(2^{deg(u)-1})$ time.

As far as directed graphs are concerned, in all calculations in Algorithm~\ref{algo:gtc:game5} we have to consider the indegree of a node instead of degree. Furthermore, the set of neighbours of a node $u$ should be defined as the set of nodes $v$ connected with directed edge $(u,v)$.

\section{Simulations}\label{sec:gtc:simulations}

\noindent In this section we evaluate the time performance of our exact algorithms for games $g_1$ to $g_4$ and our approximation algorithm for game $g_5$. In more detail, we compare our exact algorithms to the method of approximating the Shapley value via Monte Carlo sampling which has been the only feasible approach to compute game-theoretic network centrality available to date in the literature. 		\noindent In this section we evaluate the time performance of our exact algorithms for games $g_1$ to $g_4$ and our approximation algorithm for game $g_5$. In more detail, we compare our exact algorithms to the method of approximating the Shapley value via Monte Carlo sampling which has been the only feasible approach to compute game-theoretic network centrality available to date in the literature. First, we provide a detailed description of the simulation setup; then, we present data sets and the simulation results.
In the remainder of this section we first we comment on Monte Carlo sampling. Next, we provide a detailed description of the simulation setup. Finally, we present data sets and the simulation results.		
\subsection{Approximation Methods for the Shapley Value}\label{chap:gtc:related:approxSV}		
\noindent In this section, we discuss different approximation techniques to compute the Shapley value from the literature. Generally speaking, they can be divided into three groups---each referring to a specific subclass of coalitional games under consideration:		
\begin{enumerate}		
\item First, let us consider the method proposed by \citet{Fatima:et:al:2007} and elaborated further by \citet{Fatima:et:al:2008}. This approach concerns \emph{weighted voting games}. In these games, each player has a certain number of votes (or in other words, a weight). A coalition is ``winning'' if the number of votes in this coalition exceeds some specific threshold, or ``losing'' otherwise. \citeauthor{Fatima:et:al:2007} propose the following method to approximate the Shapley value in weighted voting games. Instead of finding marginal contributions of players to all $2^n$ coalitions, the authors consider only $n$ randomly-selected coalitions, one of each size (i.e., from $1$ to $n$). Only for these $n$ coalitions are the player's expected marginal contributions calculated and the average of these contributions yields an approximation of the Shapley value. Whereas \citeauthor{Fatima:et:al:2007} method is certainly attractive, it is only applicable to games in which the value of a coalition depends on the sum of associated weights being in some bounds. This is not the case for our games $g_1$ to $g_4$.\footnote{\footnotesize Recall that our approximation algorithm for $g_5$ builds upon \citeauthor{Fatima:et:al:2007} method. This is because in this game the marginal contribution of each node depends on the weights assigned to its incident edges.}		
\item Another method was proposed by \citet{Bachrach:et:al:2008} in the context of \emph{simple coalitional games} \footnote{\footnotesize Note that weighted voting games are simple coalitional games.} in which the characteristic function is binary---i.e., each coalition has a value of either zero or one. For these games, \citeauthor{Bachrach:et:al:2008} extend the approach suggested by Mann and Shapley \citeyear{Mann:Shapley:1960} and provide more rigorous statistical analysis. In particular, Mann and Shapley described the Monte Carlo simulations to estimate the Shapley value from a random sample of coalitions. Bachrach at al. use this technique to compute the Banzhaf power index and then they suggested using a random sample of permutations of all players in order to compute the Shapley-Shubik index for simple coalitional games.\footnote{\footnotesize The Shapley-Shubik index is a very well-known application of the Shapley value that evaluates the power of individuals in voting \citep{Shapley:Shubik:1954}.} The computation of the confidence interval, which is crucial in such an approach, hinges upon the binary form of the characteristic function for simple coalitional games. This method is more general than the one proposed by \citet{Fatima:et:al:2007}---as weighted voting games are a subset of simple coalitional games---but still it cannot be effectively used for our games $g_1$ to $g_4$, where the characteristic functions are not binary.		
\item Unlike the first two methods, the last method described by \citet{Castro:et:al:2009} can be efficiently applied to all coalitional games in characteristic function game form, assuming that the worth of every coalition can be computed in polynomial time. Here, approximating the Shapley value involves generating permutations of all players and computing the marginal contribution of each player to the set of players occurring before it. The solution precision increases (statistically) with every new permutation analysed. Furthermore, the authors show how to estimate the appropriate size of a permutation sample in order to guarantee a low error. Given its broad applicability, this method is used in our simulations as a comparison benchmark.		
\end{enumerate}

\subsection{Simulation Setup}

\SetKwBlock{Block}{\emph{Marginal Contribution} block}{}

	\begin{algorithm}[t]
	%\linesnumbered
	\caption{Monte Carlo method to approximate the Shapley value}
	\label{algo:gtc:mc}
	\SetAlgoVlined %The old version of this is: \SetVline
\LinesNumbered %The old version of this is: \linesnumbered
	\SetAlgoLined
	\BlankLine
	\KwIn{
	\begin{itemize}
	\item[$\diamond$]{Characteristic function $v$, maximum iteration $maxIter$}
	\end{itemize}
	}
	\KwOut{Aproximation of Shapley value for game $v$}
	\BlankLine

	\For{$v \in V$}
	{
		$\phi^{\SV}_v = 0$ \;
	}

	\For{$i = 1$ to $maxIter$} {
		shuffle($V$)\;
		\Block{
			P = $\emptyset$ \;					 \label{constart}
			\For{$v \in V$} {
				$\phi^{\SV}_v += \nu(P \cup \set{v}) - \nu(P)$ \label{contri} \;
				$P = P \cup \set{v}$ \;  \label{conend}
			}											
		}
	}

	\For{$v \in V$}
	{
		$\phi^{\SV}_v = \frac{\phi^{\SV}_v}{maxIter}$ \;
	}

	return $\phi^{\SV}_v$ \;

	\end{algorithm}

\noindent In more detail, in a preliminary step, we test what is the maximum number of Monte Carlo iterations that can be performed in a reasonable time for any given game. This maximum number of iterations, denoted $maxIter$, becomes an input to Algorithm \ref{algo:gtc:mc} for Monte Carlo sampling. In this algorithm, in each one of the $maxIter$ iterations, a random permutation of all nodes is generated. Then, using a characteristic function from the set $\nu \in \{\nu_1,\nu_2,\nu_3,\nu_4,\nu_5\}$, it calculates the marginal contribution of each node to the set $P$ of nodes occurring before a given node in a random permutation.\footnote{\footnotesize Recall that the characteristic functions $v_1,v_2,\ldots,v_5$ correspond to games $g_1,g_2,\ldots,g_5$, respectively.} Finally, the algorithm divides the aggregated sum of all contributions for each node by the number of iterations performed. The time complexity of this algorithm is  $O(maxIter*con)$, where $con$ denotes the number of operations necessary for computing the \textit{{Marginal Contribution}} block. This block is specifically tailored to the particular form of the characteristic function of each of the games $g_1$ to $g_5$.  In particular, for game $g_1$ (see Algorithm \ref{algo:gtc:mc}), it is constructed as follows. Recall that, in this game, node $v_i$ makes a positive contribution to coalition $P$ through itself and through some adjacent node $u$ under two conditions. Firstly, neither $v_i$ nor $u$ are in $P$. Secondly, there is no edge from $P$ to $v_i$ or $u$. To check for these conditions in Algorithm \ref{algo:gtc:mc} we store those nodes that have already contributed to the value of coalition $P$ in an array called: $Counted$. For each node $v_i$, the algorithm iterates through the set of its neighbours and for each adjacent node it checks whether this adjacent node is counted in the array $Counted$. If not, the marginal contribution of the node $v_i$ is increased by one.

  \begin{algorithm}[h]
		\RestyleAlgo{box}
		\caption{\textit{\textbf{Marginal Contribution}} block of Algorithm \ref{algo:gtc:mc} for $g_2$}
		\label{mcg2}
		\SetAlgoVlined %The old version of this is: \SetVline
\LinesNumbered %The old version of this is: \linesnumbered
		\SetAlgoLined
		Counted $\gets$ false \;
		Edges $\gets$ 0 \;
		\ForEach{$v \in V(G)$} {
			\ForEach{$u \in N_{G}(v) \cup \{v\}$}{
				Edges[u]++ \;
				\If{!\text{Counted[}$u$\text{]} and ( Edges[$u$] $\ge$ k[$u$] or $u = v$ ) } {
					$\phi^{\SV}_v \gets \phi^{\SV}_v + 1$\;
					Counted[$u$] = true \;
				}
			}
		}
	\end{algorithm}
\vspace{-0.5cm}
	\begin{algorithm}[h]
		\RestyleAlgo{box}
		\caption{\textit{\textbf{Marginal Contribution}} block of Algorithm \ref{algo:gtc:mc} for $g_3$}
		\label{mcg3}
		\SetAlgoVlined %The old version of this is: \SetVline
\LinesNumbered %The old version of this is: \linesnumbered
		\SetAlgoLined
		Counted $\gets$ false \;
		\ForEach{$v \in V(G)$} {
			\ForEach{$u\in extNeighbors(v) \cup \{v\}$}{
				\If{!\text{Counted[}$u$\text{]}} {
					$\phi^{\SV}_v \gets \phi^{\SV}_v + 1$\;
					Counted[$u$] = true \;
				}
			}
		}
	\end{algorithm}

 For the rest of the games considered in this chapter the \textit{\textbf{Marginal Contribution}} block of Algorithm \ref{algo:gtc:mc} takes a slightly different form. We have explained the functioning of this block for $g_1$. Now, we discuss this block for the remaining four games. In particular:
\begin{itemize}
\item[$g_2$:] Here, node $v$ makes a positive contribution to a coalition $P$ both through itself and through some adjacent node $u$ also under two conditions. Firstly, neither $v$ nor $u$ are in $P$. Secondly, there is less than $k$ edges from $P$ to $v$ and there is exactly $k-1$ edges from $P$ to $u$. In order to check the first condition in Algorithm \ref{mcg2} we use the array $Counted$, and to check the second one, we use the array $Edges$. For each node $v_i$, the algorithm iterates through the set of its neighbours and for each adjacent node it checks whether this adjacent node meets these two conditions. If so, then the marginal contribution of the node $v$ is increased by one.
\end{itemize}

	\begin{algorithm}[t]
		\RestyleAlgo{box}
		\caption{\textit{\textbf{Marginal Contribution}} block of Algorithm \ref{algo:gtc:mc} for $g_4$}
		\label{mcg4}
		\SetAlgoVlined %The old version of this is: \SetVline
\LinesNumbered %The old version of this is: \linesnumbered
		\SetAlgoLined
		dist $\gets$ infinity \;
		\ForEach{$v \in V(G)$} {
			\ForEach{$u \in V(G)$} {
				\If{D[$u$] < dist[$u$]} {	
					$\phi^{\SV}_v \gets \phi^{\SV}_v + f(D[$u$]) - f(dist[$u$])$\;
					dist[$u$] = D[$u$] \;
				}
			}
			$\phi^{\SV}_v \gets \phi^{\SV}_v + f(dist[$u$]) - f(0)$\;
			dist[$v$] = 0 \;
		}
	\end{algorithm}

    \begin{algorithm}[t]
		\RestyleAlgo{box}
		\caption{\textit{\textbf{Marginal Contribution}} block of Algorithm \ref{algo:gtc:mc} for $g_5$}
		\label{mcg5}
		\SetAlgoVlined %The old version of this is: \SetVline
\LinesNumbered %The old version of this is: \linesnumbered
		\SetAlgoLined
		Counted $\gets$ false \;
		Weights $\gets$ 0 \;
		\ForEach{$v \in V(G)$} {
			\ForEach{$u\in N_{G}(v) \cup \{v\}$}{
				weights[u]+= $\lambda(v,u)$\;
				\If{!\text{Counted[}$u$\text{]} and ( \text{weights[}$u$\text{]} $\ge$ $W_{\cutoff}(u)$ or $u = v$ ) } {
    				$\phi^{\SV}_v \gets \phi^{\SV}_v + 1$\;
					Counted[$u$] = true \;
				}
			}
		}
	\end{algorithm}

\begin{itemize}
\item[$g_{3/4}$:] In \textit{{Marginal Contribution}} blocks for games $g_3$ and $g_4$ (Algorithms \ref{mcg3} and \ref{mcg4}), all the values that are dependent on the distance ($extNeighbours$ and $D$) are calculated using Dijkstra's algorithm and stored in memory. These pre-computations allow us to significantly speed up Monte Carlo methods. Now, in $g_3$ node $v$ makes a positive contribution to coalition $P$ through itself and through some adjacent node $u$ under two conditions. Firstly, neither $v$ nor $u$ are in $P$. Secondly, there is no edge length of $d_{\cutoff}$ from $P$ to $v$ or $u$. To check for these conditions in Algorithm \ref{mcg3} we again use the array $Counted$. For each node $v$, the algorithm iterates through the set of its extended neighbours and for each of them it checks whether this neighbour meets the conditions. If so, the marginal contribution of the node $v$ is increased by one. In game $g_4$, node $v$ makes a positive contribution to coalition $P$ through each node (including itself) that is closer to $v_i$ than to $P$. In Algorithm \ref{mcg4} we use array $Dist$ to store distances from coalition $P$ to all nodes in the graph and array $D$ to store all distances from $v$ to all other nodes. For each node $v_i$, the algorithm iterates through all nodes in the graph, and for each node $u$, if the distance from $v$ to $u$ is smaller than from $P$ to $u$, the algorithm computes the marginal contribution as $f(D[u])-f(Dist[u])$. The value $Dist[u]$ is then updated to $D[u]$---this is a new distance from $P$ to $u$.
\item[$g_5$:] In game $g_5$, which is an extension of $g_2$ to weighted graphs, node $v$ makes a positive contribution to coalition $P$ (both through itself and through some adjacent node $u$) under two conditions. Firstly, neither $v$ nor $u$ are in $P$. Secondly, the sum of weights on edges from $P$ to $v$ is less than $W_{\cutoff}(v_i)$ and the sum of weights on edges from $P$ to $u$ is greater than, or equal to, $W_{\cutoff}(u)-\lambda(v,u)$ and smaller than $W_{\cutoff}(v_i)+\lambda(v,u)$. In order to check the first condition in Algorithm \ref{mcg5} we use the array $Counted$, and to check the second one, we use the array $Weights$. For each node $v_i$, the algorithm iterates through the set of its neighbours and for each adjacent node it checks whether this adjacent node meets these two conditions. If so, then the marginal contribution of the node $v$ is increased by one.
\end{itemize}

Some details of how Algorithm \ref{algo:gtc:mc} is applied to generate the Shapley value approximations for games $g_1$ to $g_4$, for which we propose exact polynomial solutions, differ from $g_5$, for which we developed an approximate solution. Specifically, for games $g_{1}$ to $g_{4}$:
\begin{enumerate}
\item We use the exact algorithm proposed in this chapter to compute the Shapley value.
\item Then, we run Monte Carlo simulations $30$ times.\footnote{\footnotesize For the purpose of comparison to our method, it suffices to use 30 iterations, as the standard errors converge significantly to indicate the magnitude of the cost of using the Monte Carlo method.} In every run:
\begin{itemize}
\item We perform $maxIter$ Monte Carlo iterations.
\item After every five iterations, we compare the approximation of the Shapley value obtained via Monte Carlo simulation with the exact Shapley value obtained with our algorithm.
\item We record the algorithm's runtime and the error, where the error is defined as the maximum discrepancy between the actual Shapley value and the Monte Carlo-based approximation of the Shapley value.
\end{itemize}
\item Finally, we compute the confidence interval using all iterations (0.95\% confidence level).\footnote{\footnotesize Since for $g_{4}$ each Monte Carlo iteration is relatively time consuming, we run it only once; thus, no confidence interval is generated, i.e., the third step is omitted.}
\end{enumerate}

In the case of game $g_{5}$ we cannot determine the exact Shapley value for larger networks. Therefore, we performed two levels of simulation: one level on small networks and one level on large networks. Specifically:

\begin{enumerate}
\item For small networks, we generate 30 random instances of weighted complete graphs with 6 nodes (denoted $K_6$) and the same number of graphs with 12 nodes (denoted $K_{12}$) with weights drawn from a uniform distribution $U(0,1)$. Then, for each graph and each of the two parameters $W_{{\cutoff}}(v)= \frac{1}{4}\alpha(v)$ and $W_{{\cutoff}}(v)= \frac{3}{4}\alpha(v)$:\footnote{\footnotesize Recall that $\alpha_{j}$ is the sum of all the weights in $W_{j}$ as defined in Section \ref{sec:game5}.}
	\begin{itemize}
		\item We compute the exact Shapley value using formula (\ref{def:sv:expected}).
		\item Then, we run our approximation algorithm and determine the error in our approximation.
		\item Finally, we run $2000$ and $6000$ Monte Carlo iterations for $K_6$ and  $K_{12}$, respectively.
	\end{itemize}
\item For large networks, we again generate 30 random instances of weighted complete graphs, but now with 1000 nodes (we denote them $K_{1000}$). Then, for each graph and each of the three parameters $W_{{\cutoff}}(v) = \frac{1}{4}\alpha(v)$, $W_{{\cutoff}}(v) = \frac{2}{4}\alpha(v)$, and $W_{{\cutoff}}(v) = \frac{3}{4}\alpha(v)$):
	\begin{itemize}
		\item We run our approximation algorithm for the Shapley value.
		\item Then, we run the fixed number ($200000$) of Monte Carlo iterations.
		\item Finally, we compute how the Monte Carlo solution converges to the results of our approximation algorithm.
	\end{itemize}
\end{enumerate}

Having described the simulation setup, we will now discuss the data sets and, finally, the simulation results.

\subsection{Data Used in Simulations}

\noindent We consider two networks that have already been well-studied in the literature. Specifically, for games $g_{1} - g_{3}$ we present simulations on an undirected, unweighted network representing the topology of the Western States Power Grid (WSPG).\footnote{\footnotesize Note that with the distance threshold $d_{{\cutoff}}$ replaced with a hop threshold $k_{{\cutoff}}$, game $g_{3}$ can be played on an unweighted network.} This network (which has 4940 nodes and 6594 edges) has been studied in many contexts before (see, for instance, \citet{Watts:Strogatz:1998}) and is freely available online (see, e.g., \texttt{http://networkdata.ics.uci.edu/} \texttt{data.php?id=107}). For games $g_{3} - g_{5}$  (played on weighted networks), we used the network of astrophysics collaborations (abbreviated henceforth APhC) between Jan 1, 1995 and December 31, 1999. This network (which has 16705 nodes and 121251 edges) is also freely available online (see, e.g., \texttt{http://networkdata.ics.uci.edu/} \texttt{data.php?id=13}) and has been used in previous studies like \citet{Newman:2001}.

\begin{figure}[t]
\begin{minipage}[b]{0.5\linewidth}\centering
\includegraphics[natheight=3.3cm, natwidth=8cm, height=5cm, width=1\textwidth]{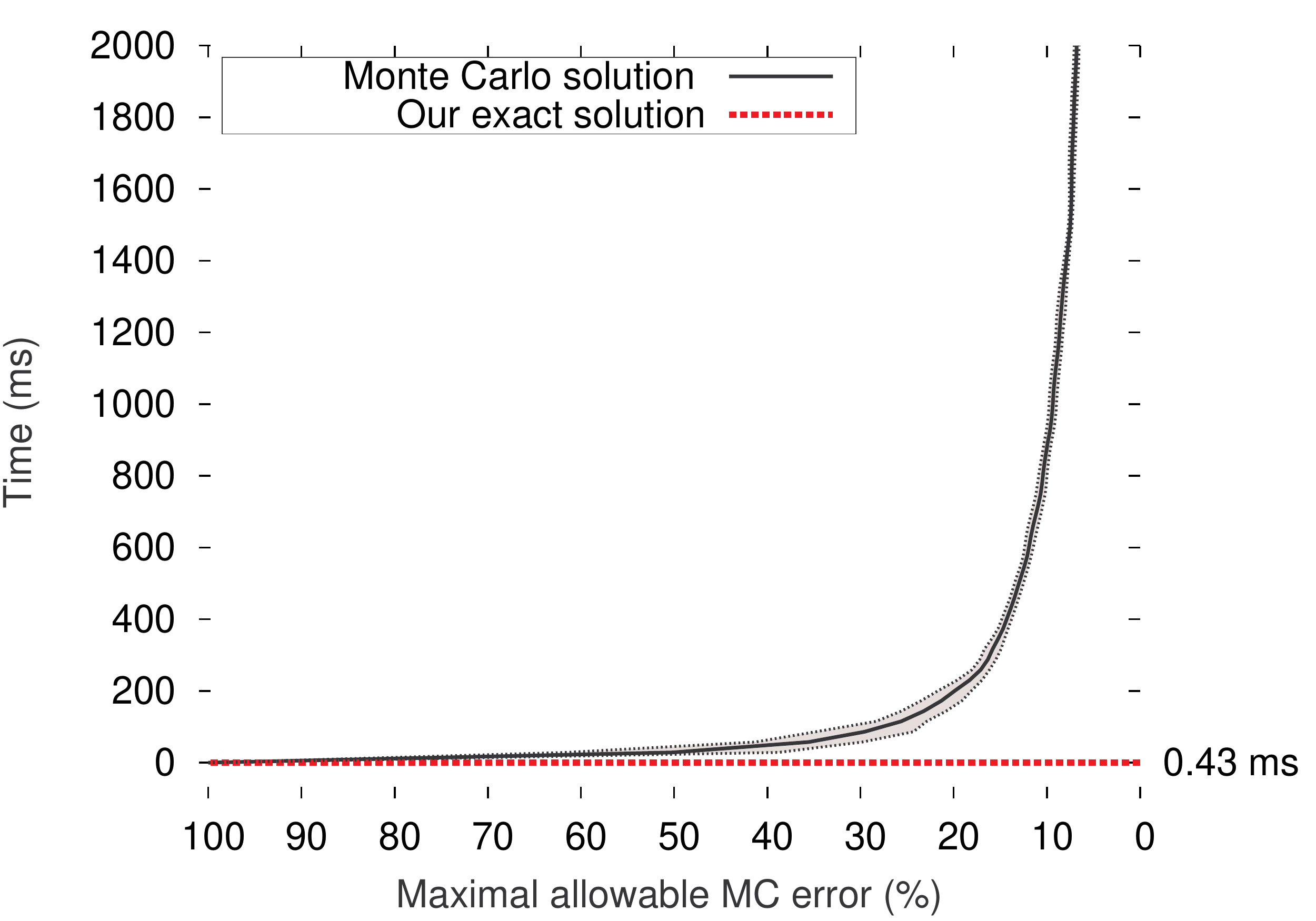}
\caption{$g_1$,  WSPG (UW)}
\label{fg1}
\end{minipage}
\begin{minipage}[b]{0.5\linewidth}\centering
\includegraphics[natheight=3.3cm, natwidth=8cm, height=5cm, width=1\textwidth]{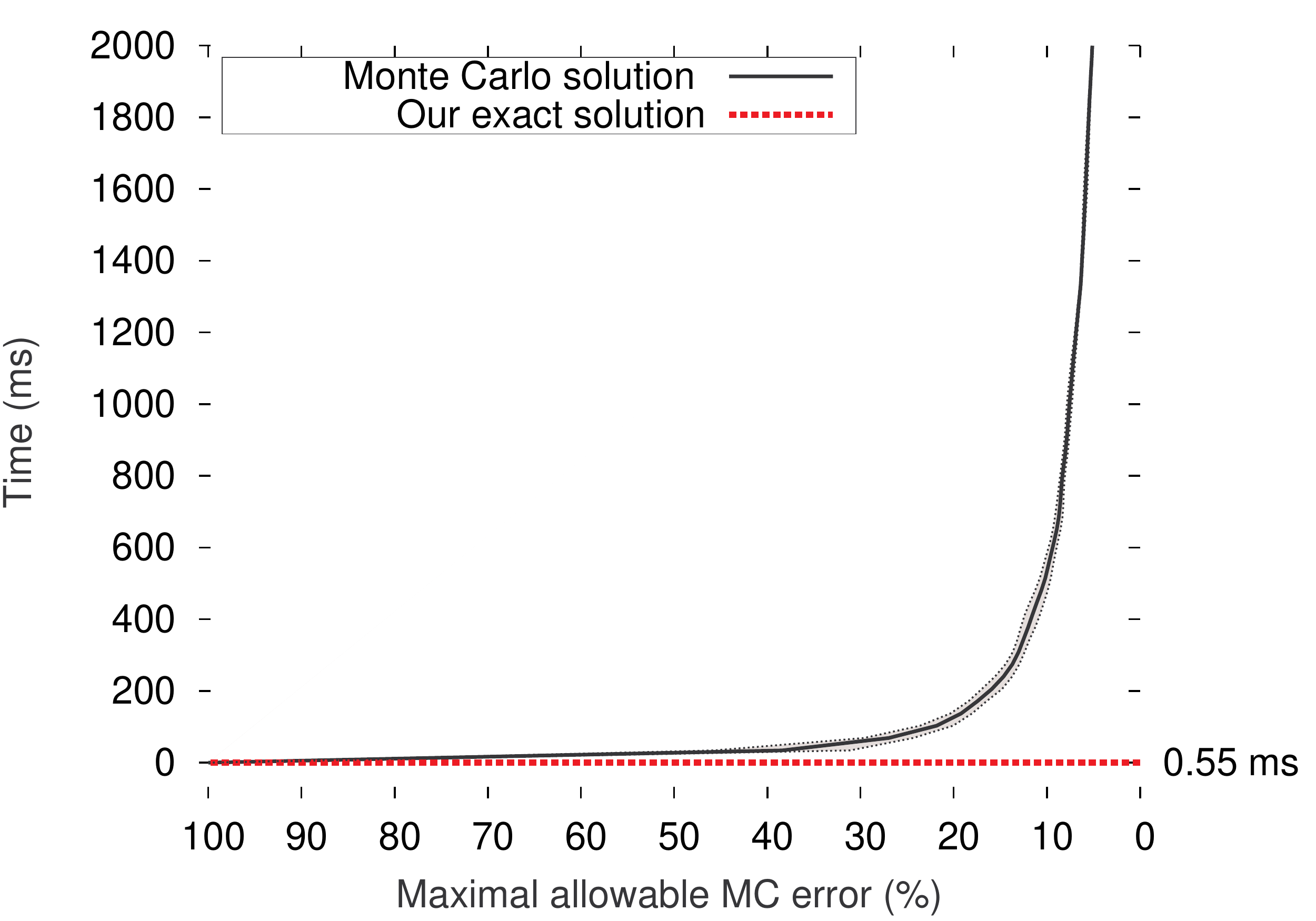}
\caption{$g_2$, $k=2$,  WSPG (UW)}
\label{fg21}
\end{minipage} \\
\\
\begin{minipage}[b]{0.5\linewidth}\centering
\includegraphics[natheight=3.3cm, natwidth=8cm, height=5cm, width=1\textwidth]{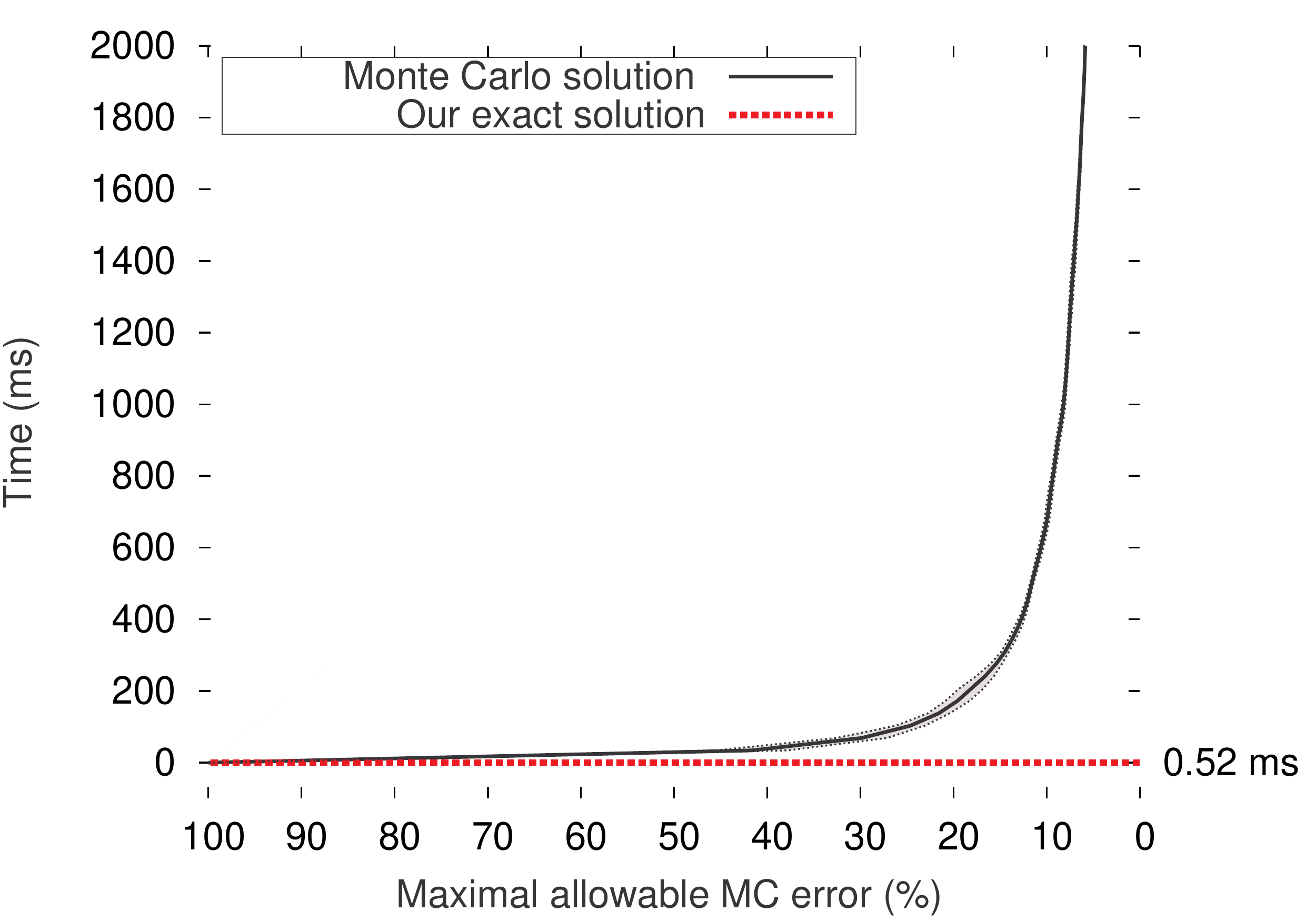}
\caption{$g_2$, $k_{i}=\frac{deg_{i}}{2}$, WSPG (UW)}
\label{fg22}
\end{minipage}
\begin{minipage}[b]{0.5\linewidth}\centering
\includegraphics[natheight=3.3cm, natwidth=8cm, height=5cm, width=1\textwidth]{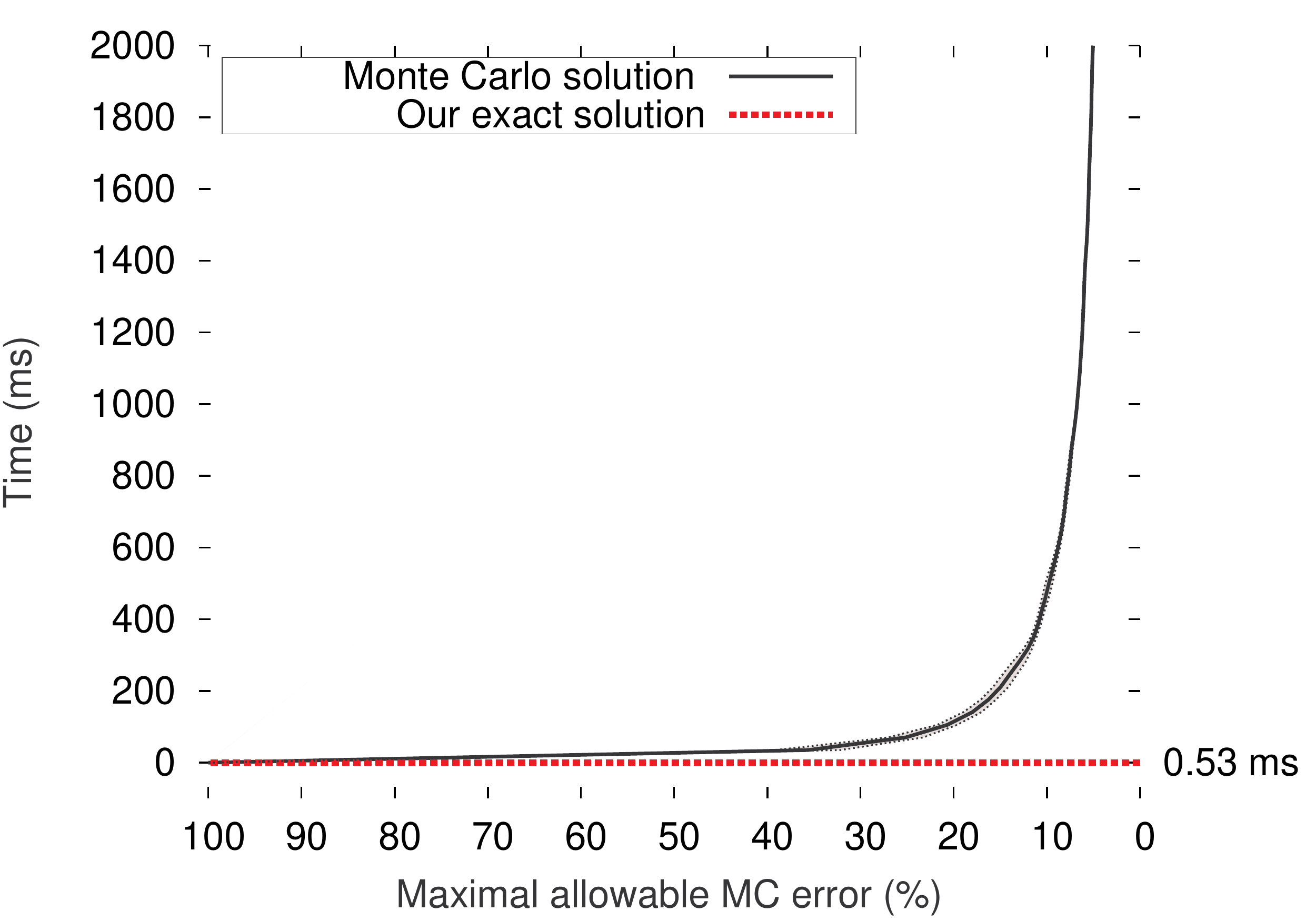}
\caption{$g_2$, $k_{i}=\frac{3}{4}\ deg_{i}$,  WSPG (UW)}
\label{fg23}
\end{minipage}
\end{figure}

\subsection{Simulation Results}%\footnote{\footnotesize The test platform used is Windows XP running on Intel 1.9GHz Dual Core PC with 4 Gigabytes of RAM (DDR2).}}

\noindent The results presented in this section show that our exact algorithms are, in general, much faster then the Monte Carlo sampling, and this is the case even if we allow for generous error tolerance. Furthermore, requiring smaller Monte Carlo errors makes the Monte Carlo runtime exponentially slower than our exact solution.

In more detail, the simulation results for game $g_1$ are shown in Figure \ref{fg1}. The dotted line shows the performance of our exact algorithm which needs $0.43ms$ to compute the Shapley value. In contrast, generating any reasonable Monte Carlo result takes a substantially longer time (the solid line shows the average and the shaded area depicts the confidence interval for Monte Carlo simulations). In particular, it takes on average more than $200ms$ to achieve a $20\%$ error and more than $2000ms$ are required to guarantee a $5\%$ error (which is more than 4600 times slower than our exact algorithm).

Figures \ref{fg21} - \ref{fg23} concern game $g_2$ for different values of $k$ ($k=2$, $k_{i}=\frac{deg_{i}}{2}$, and $k_{i}=\frac{3}{4}\ deg_{i}$, respectively, where $deg_i$ is the degree of node $v_i$).\footnote{\footnotesize Recall that in $g_2$ the meaning of parameter $k$ is as follows: the value of coalition $C$ depends on the number of nodes in the network with at least $k$ neighbours in $C$.}  The advantage of our exact algorithm over Monte Carlo simulation is again exponential.

Replacing the distance threshold $d_{{\cutoff}}$ with a hop threshold $k_{{\cutoff}}$ enables game $g_3$ to be played on an unweighted network. Thus, similarly to games $g_1$ and $g_2$, we test it on the Western States Power Grid. The results are shown on Figures \ref{fg31} and \ref{fg32} for $k_{{\cutoff}}$ being equal to $2$ and $3$, respectively. The third game is clearly more computationally challenging than $g_1$ and $g_2$ (note that the vertical axis is in seconds instead of milliseconds). Now, our exact algorithm takes about $13s$ to complete. The much lower speedups of the exact methods with respect to Monte Carlo approach stem from the fact that both algorithms have to start with Dijkstra's algorithm. Although this algorithm has to be run only once in both cases it takes more than $12.5s$ for the considered network. This means that the exact solution is slower by orders of magnitude (compared to games $g_1$ and $g_2$). The Monte Carlo approach is also slower, but this slowdown is much less significant in relative terms.

Figures \ref{fg33} and \ref{fg34} show the performance of the algorithms for game $g_3$ on the astrophysics collaboration network that, unlike the Western States Power Grid, is a weighted network. We observe that increasing the value of $d_{{\cutoff}}$ (here from $d_{{\cutoff}}=\frac{d_{avg}}{8}$ to $d_{{\cutoff}}=\frac{d_{avg}}{4}$) significantly worsens the performance of the Monte Carlo-based algorithm. This is because the increasing number of nodes that have to be taken into account while computing marginal contributions (see the inner loop in Algorithm \ref{mcg3}) is not only more time consuming, but also increases the Monte Carlo error.

For game $g_4$ the performance of algorithms is shown in Figures \ref{fg41} - \ref{fg43} (for $f(d)=\frac{1}{1+d}$, $f(d)=\frac{1}{1+d^2}$ and $f(d)=e^{-d}$, respectively). Whereas the Monte Carlo methods for the first three games are able to achieve a reasonable error bound in seconds or minutes, for the fourth game it takes more than $40$ hours to approach $50\%$ error. This is because the inner loop of the \textit{{Marginal Contribution}} block (see Algorithm \ref{mcg4}) iterates over all nodes in the network. Due to the time consuming performance we run the simulations only once. Interestingly, we observe that the error of the Monte Carlo method sometimes increases slightly when more iterations are performed. This confirms that the error of the Monte Carlo method to approximate the Shapley value proposed in \citet{Castro:et:al:2009} is only statistically decreasing in time. Certain new randomly chosen permutations can actually increase the error.

\begin{figure}[t]
\begin{minipage}[b]{0.5\linewidth}\centering
\includegraphics[natheight=3.3cm, natwidth=8cm, height=5cm, width=1\textwidth]{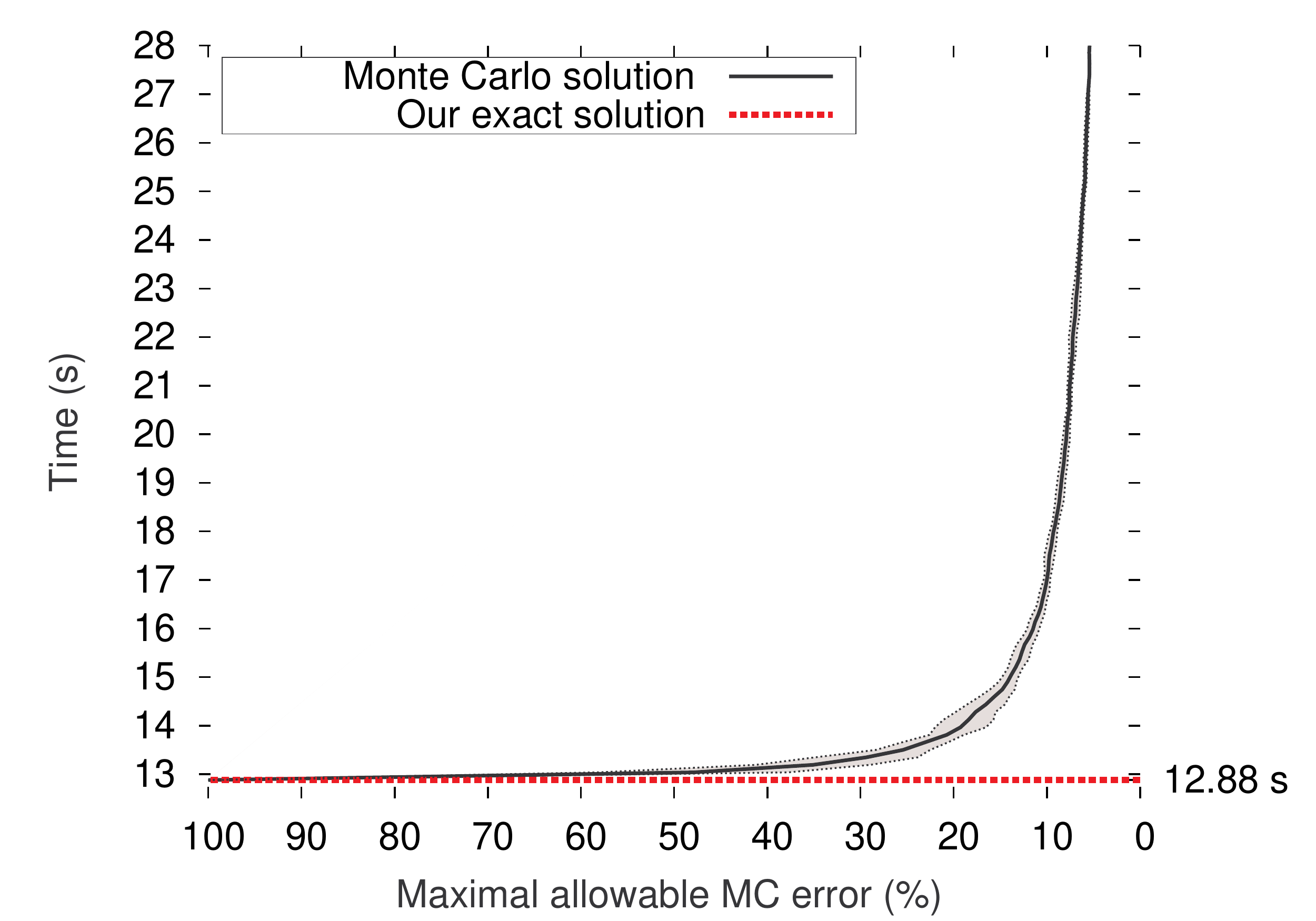}
\caption{$g_3$, $k_{{\cutoff}}=2$,  WSPG (UW)}
\label{fg31}
\end{minipage}
\begin{minipage}[b]{0.5\linewidth}\centering
\includegraphics[natheight=3.3cm, natwidth=8cm, height=5cm, width=1\textwidth]{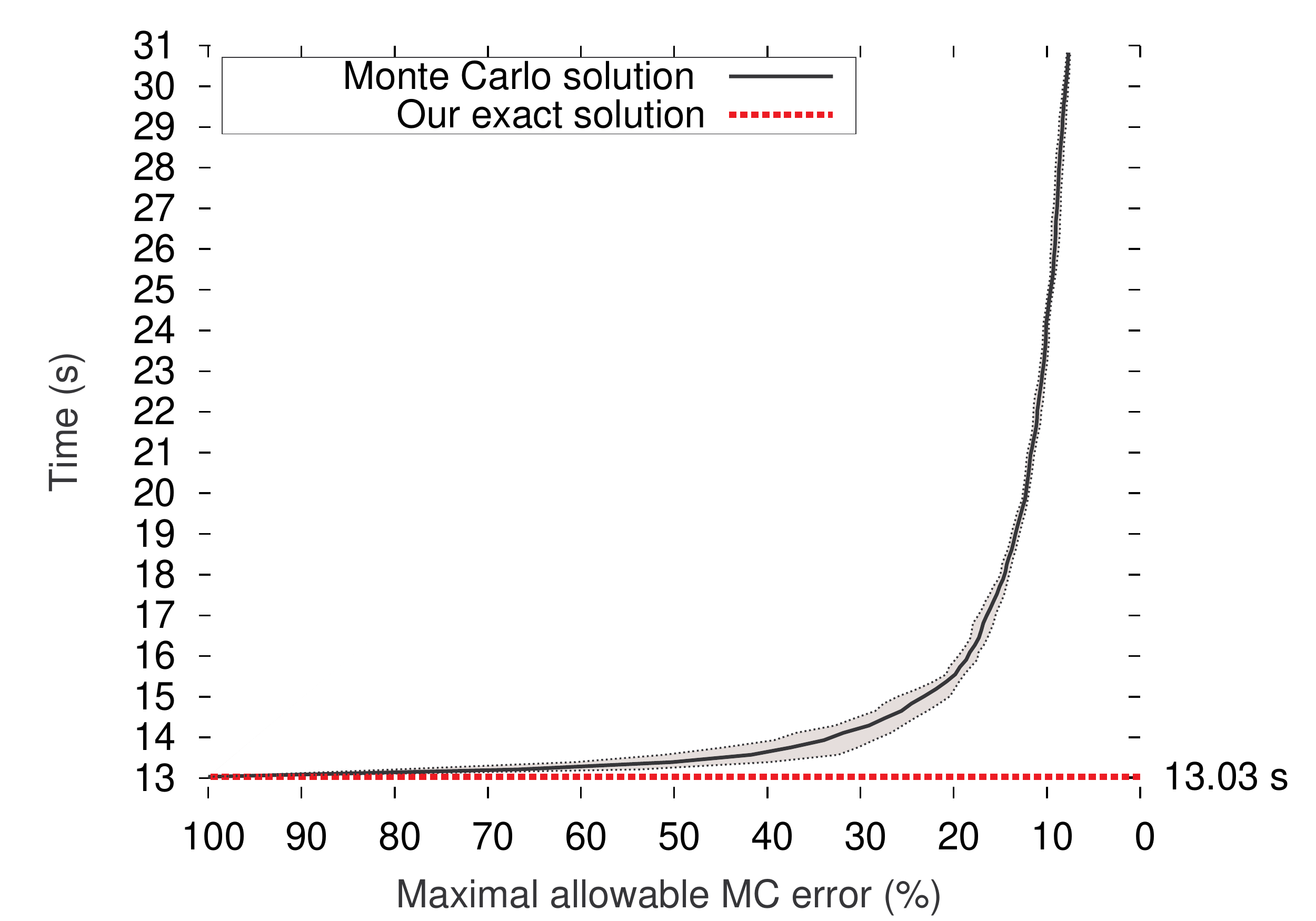}
\caption{$g_3$, $k_{{\cutoff}}=3$,  WSPG (UW)}
\label{fg32}
\end{minipage} \\
\end{figure}

\begin{figure}[t]
\begin{minipage}[b]{0.5\linewidth}\centering
\includegraphics[natheight=3.3cm, natwidth=8cm, height=5cm, width=1\textwidth]{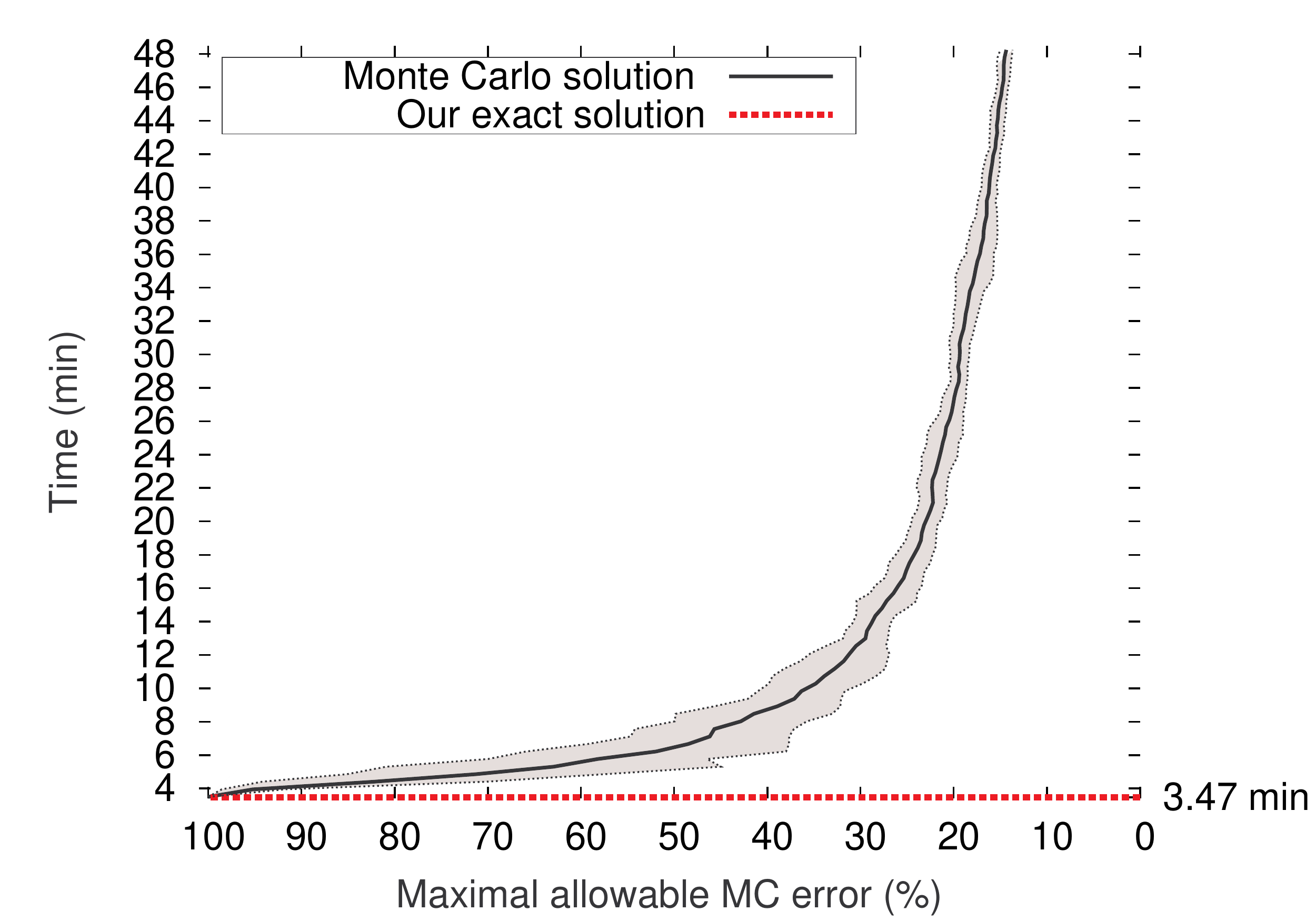}
\caption{$g_3$, $d_{{\cutoff}}=\frac{d_{avg}}{8}$, APhC (W)}
\label{fg33}
\end{minipage}
\begin{minipage}[b]{0.5\linewidth}\centering
\includegraphics[natheight=3.3cm, natwidth=8cm, height=5cm, width=1\textwidth]{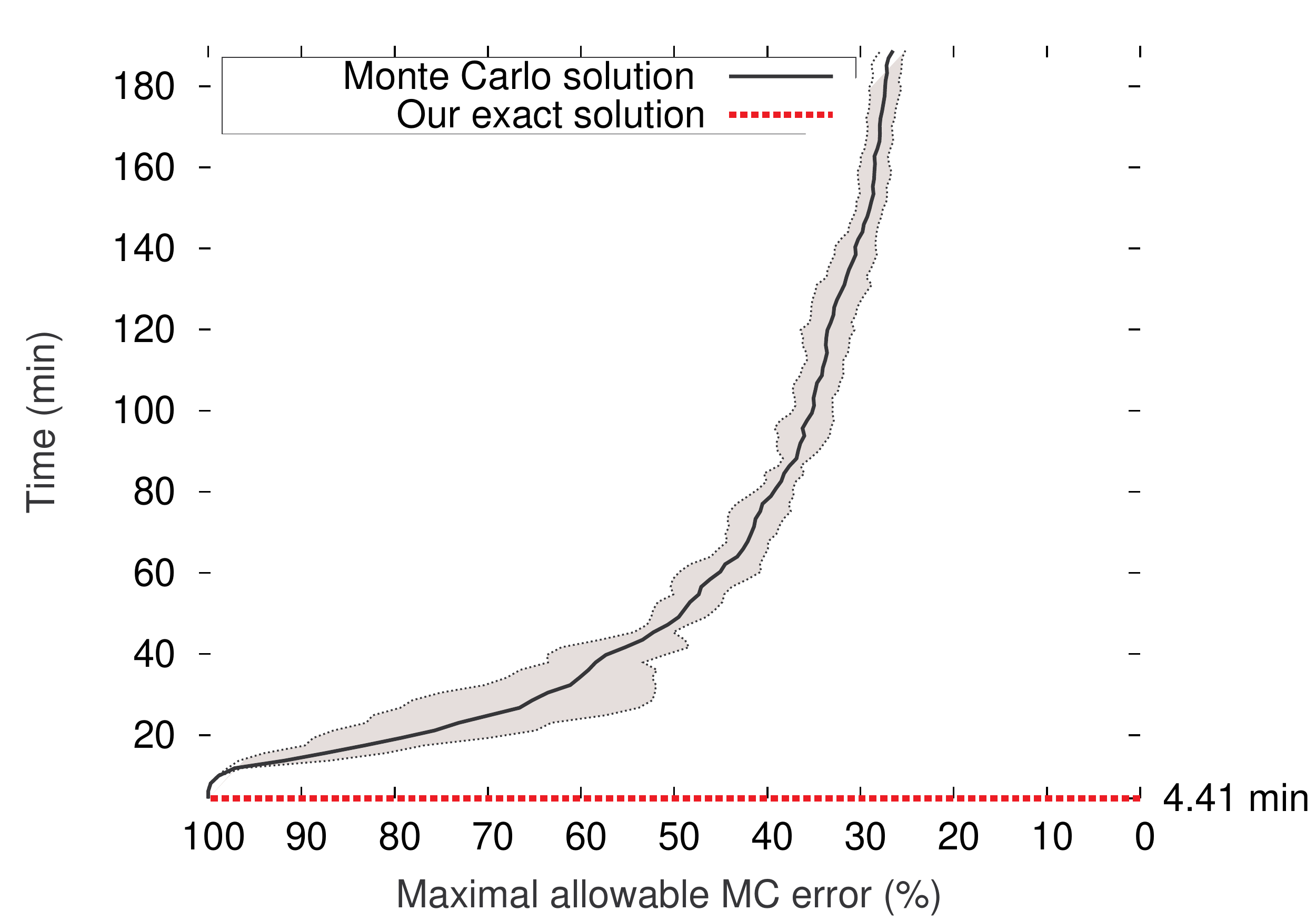}
\caption{$g_3$, $d_{{\cutoff}}=\frac{d_{avg}}{4}$, APhC (W)}
\label{fg34}
\end{minipage} \\
\end{figure}

\begin{figure}[t]
\begin{minipage}[b]{0.5\linewidth}\centering
\includegraphics[natheight=3.3cm, natwidth=8cm, height=5cm, width=1\textwidth]{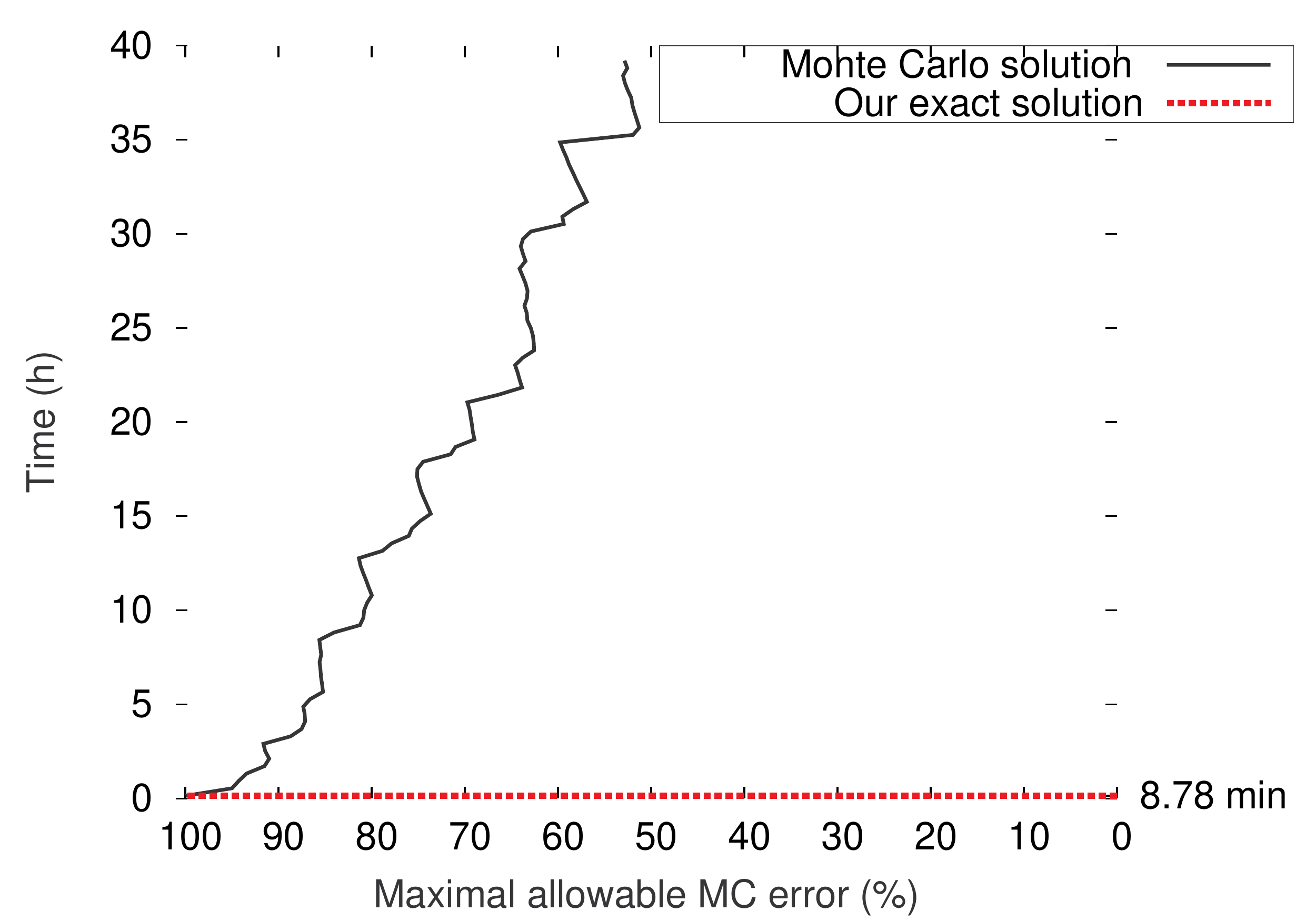}
\caption{$g_4$, $f(d)=\frac{1}{1+d}$, APhC (W)}
\label{fg41}
\end{minipage}
\begin{minipage}[b]{0.5\linewidth}\centering
\includegraphics[natheight=3.3cm, natwidth=8cm, height=5cm, width=1\textwidth]{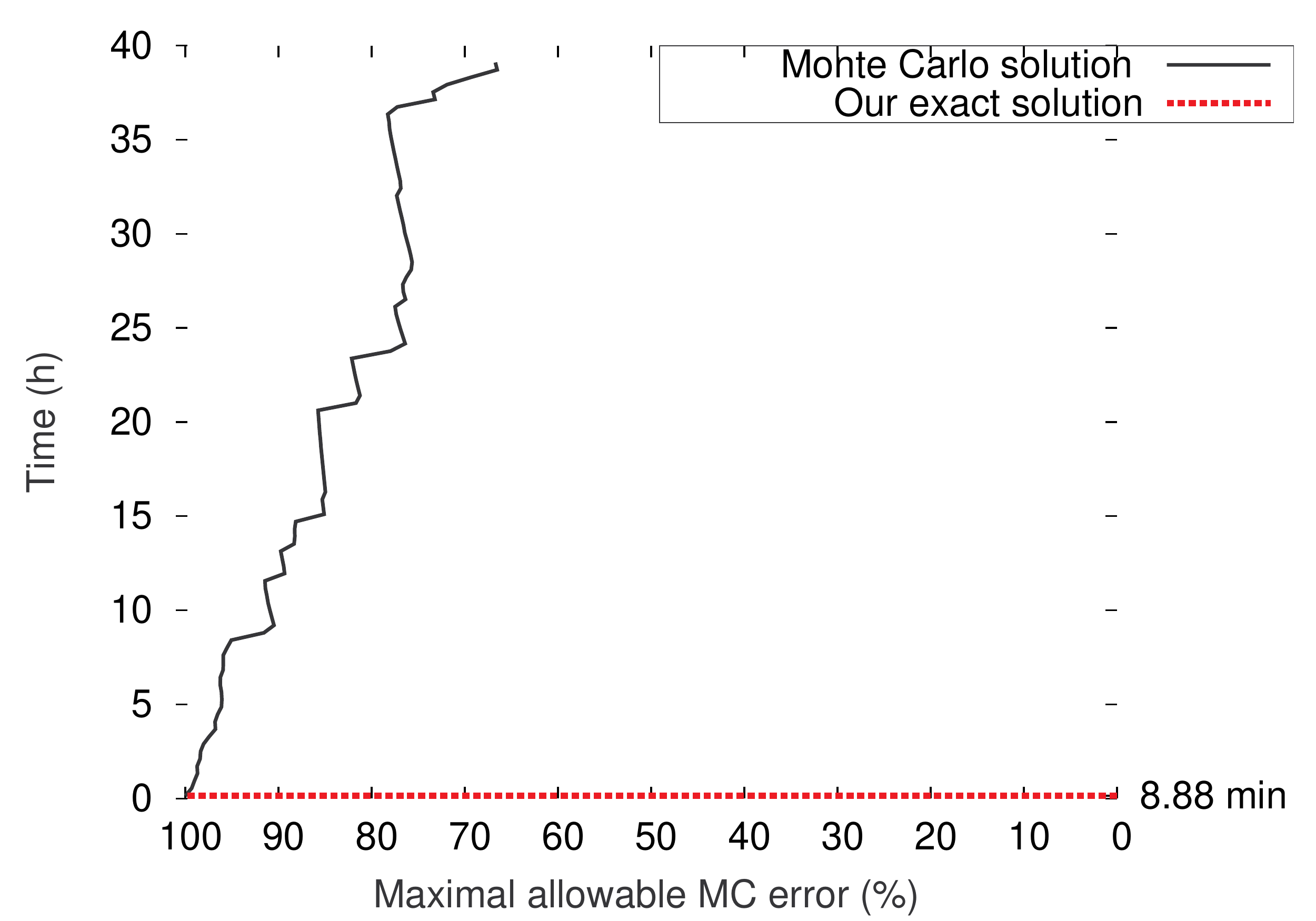}
\caption{$g_4$, $f(d)=\frac{1}{1+d^{2}}$, APhC (W)}
\label{fg42}
\end{minipage} \\
\begin{minipage}[b]{0.25\linewidth}\centering
~
\end{minipage}
\begin{minipage}[b]{0.5\linewidth}\centering
\includegraphics[natheight=3.3cm, natwidth=8cm, height=5cm, width=1\textwidth]{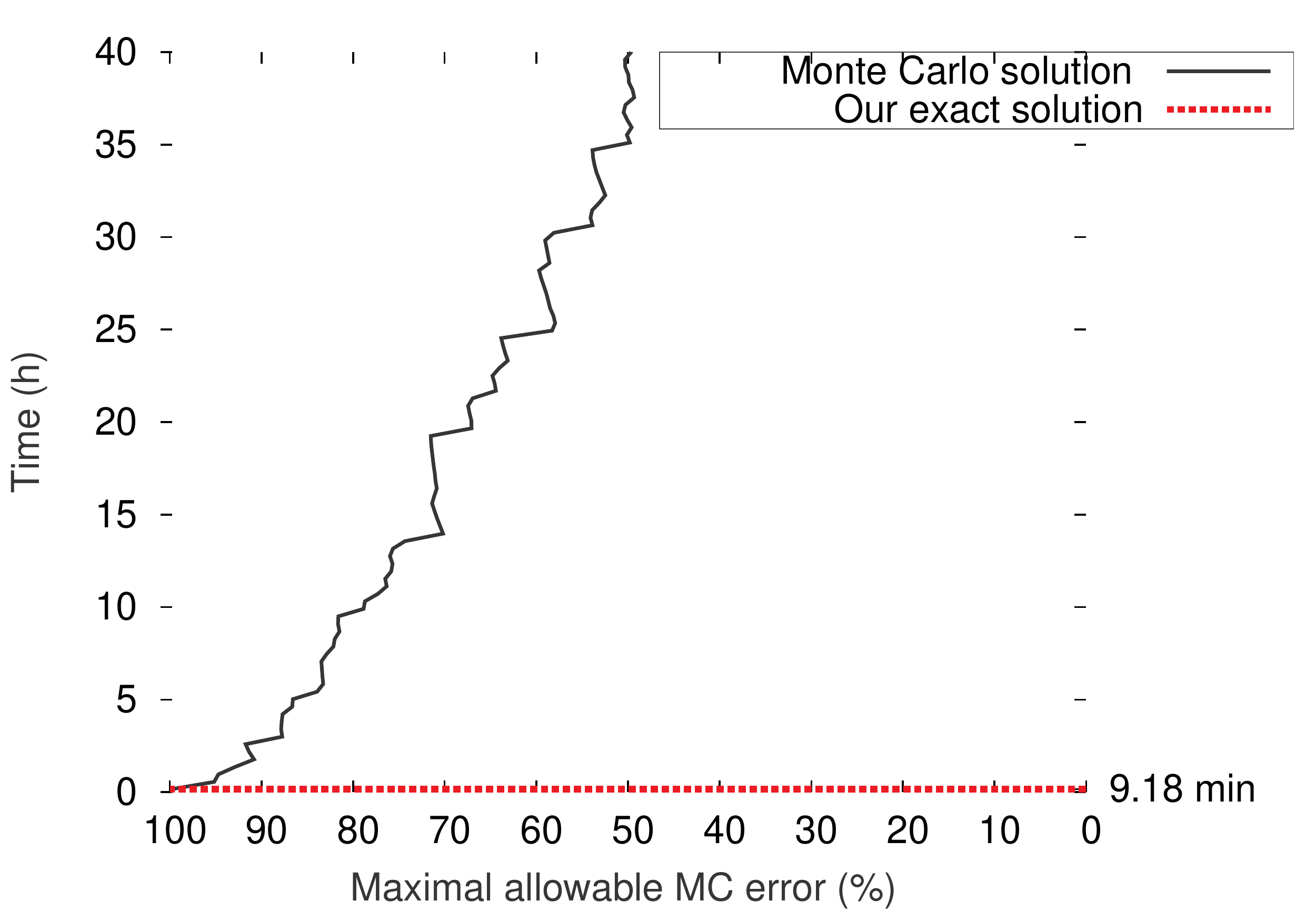}
\caption{$g_4$, $f(d)=e^{-d}$, APhC (W)}
\label{fg43}
\end{minipage}
\begin{minipage}[b]{0.25\linewidth}\centering
\end{minipage}
\end{figure}

\begin{figure}[t]
\begin{minipage}[b]{0.5\linewidth}\centering
\includegraphics[natheight=3.3cm, natwidth=8cm, height=5cm, width=1\textwidth]{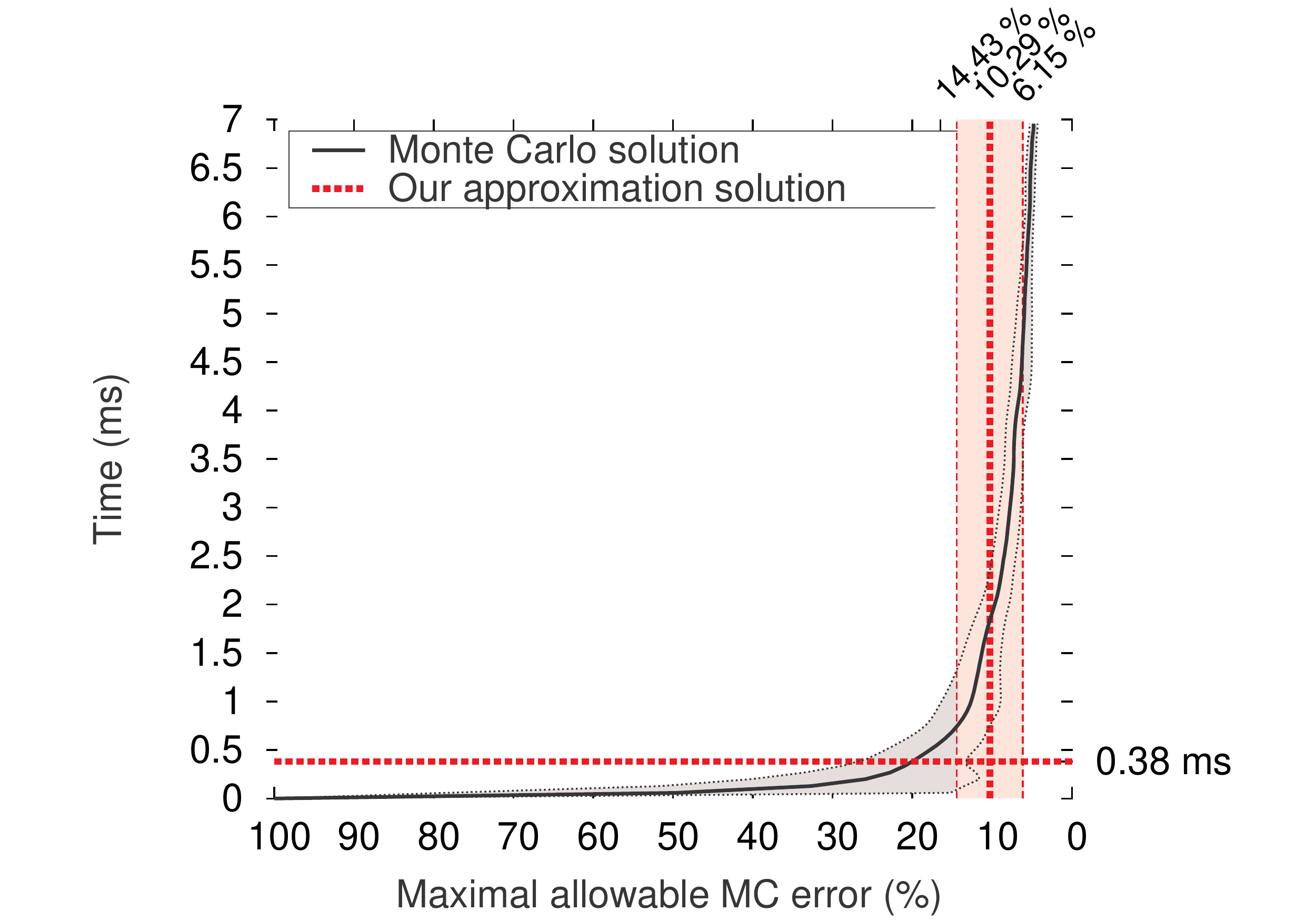}
\caption{$g_5$, $W_{{\cutoff}}=\frac{1}{4}\alpha_i$, $K_{6}$ (W)}
\label{fg51}
\end{minipage}
\begin{minipage}[b]{0.5\linewidth}\centering
\includegraphics[natheight=3.3cm, natwidth=8cm, height=5cm, width=1\textwidth]{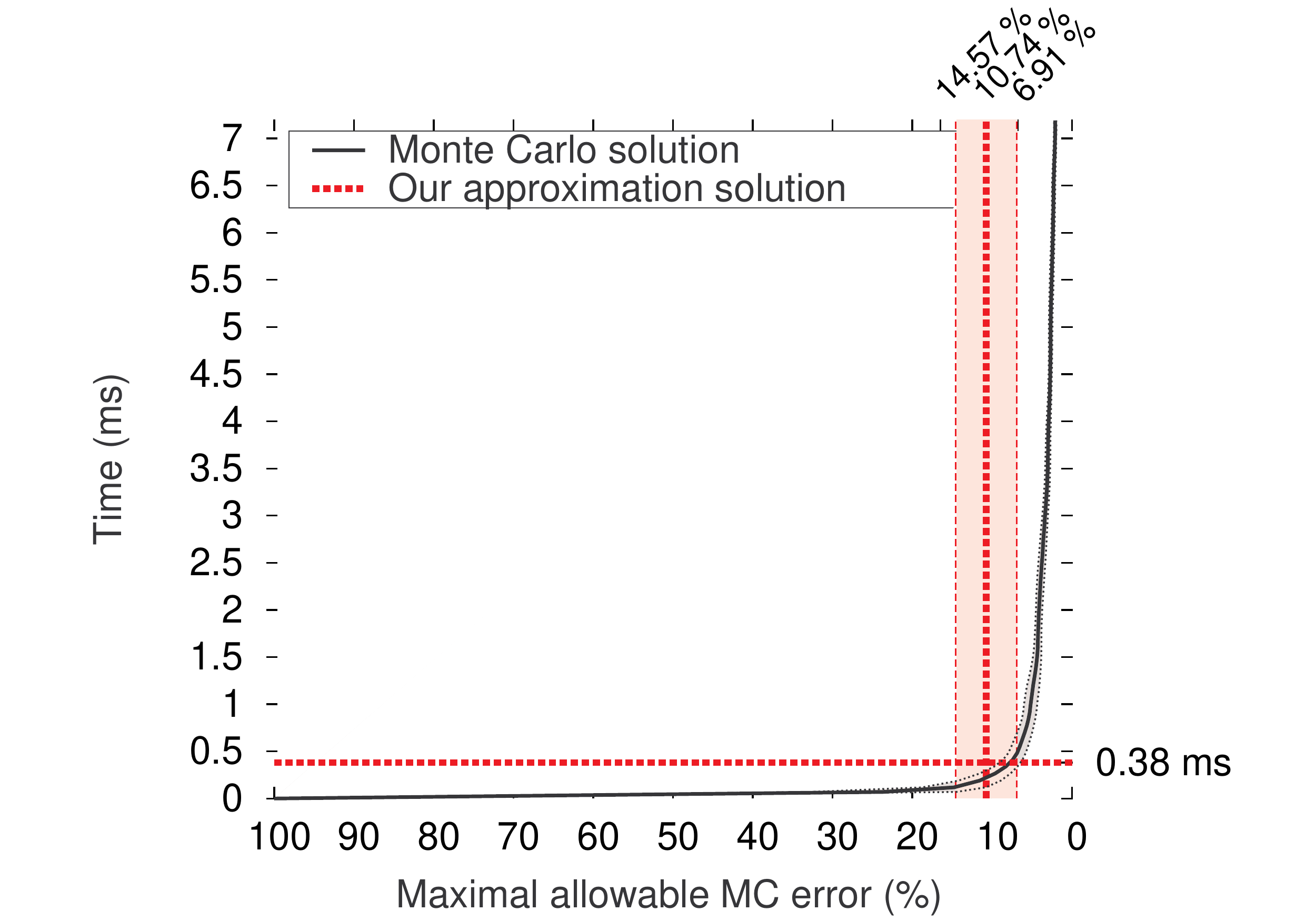}
\caption{$g_5$, $W_{{\cutoff}}=\frac{3}{4}\alpha_i$, $K_{6}$ (W),}
\label{fg52}
\end{minipage} \\
\end{figure}

\begin{figure}[t]
\begin{minipage}[b]{0.5\linewidth}\centering
\includegraphics[natheight=3.3cm, natwidth=8cm, height=5cm, width=1\textwidth]{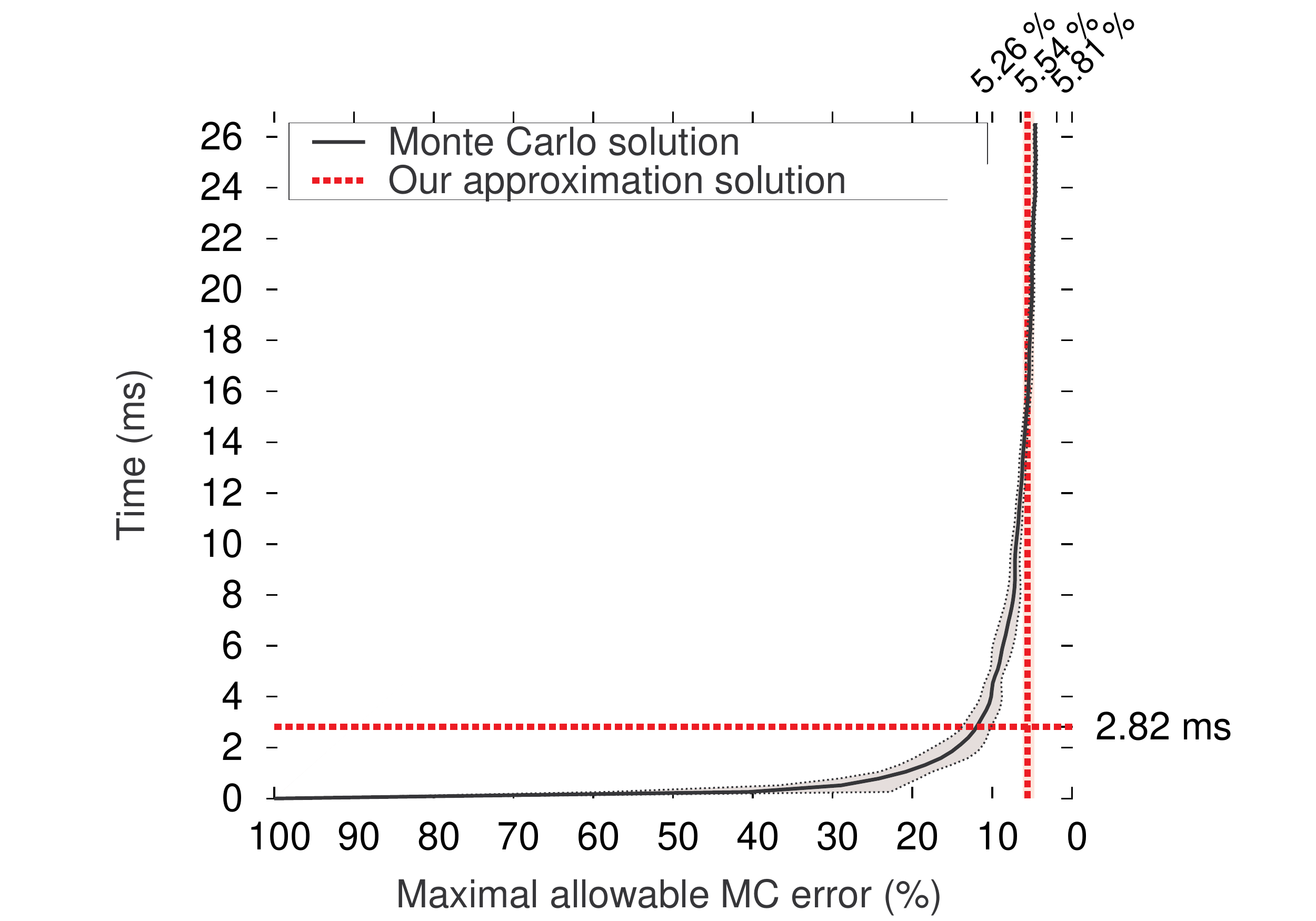}
\caption{$g_5$, $W_{{\cutoff}}=\frac{1}{4}\alpha_i$, $K_{12}$ (W)}
\label{fg53}
\end{minipage}
\begin{minipage}[b]{0.5\linewidth}\centering
\includegraphics[natheight=3.3cm, natwidth=8cm, height=5cm, width=1\textwidth]{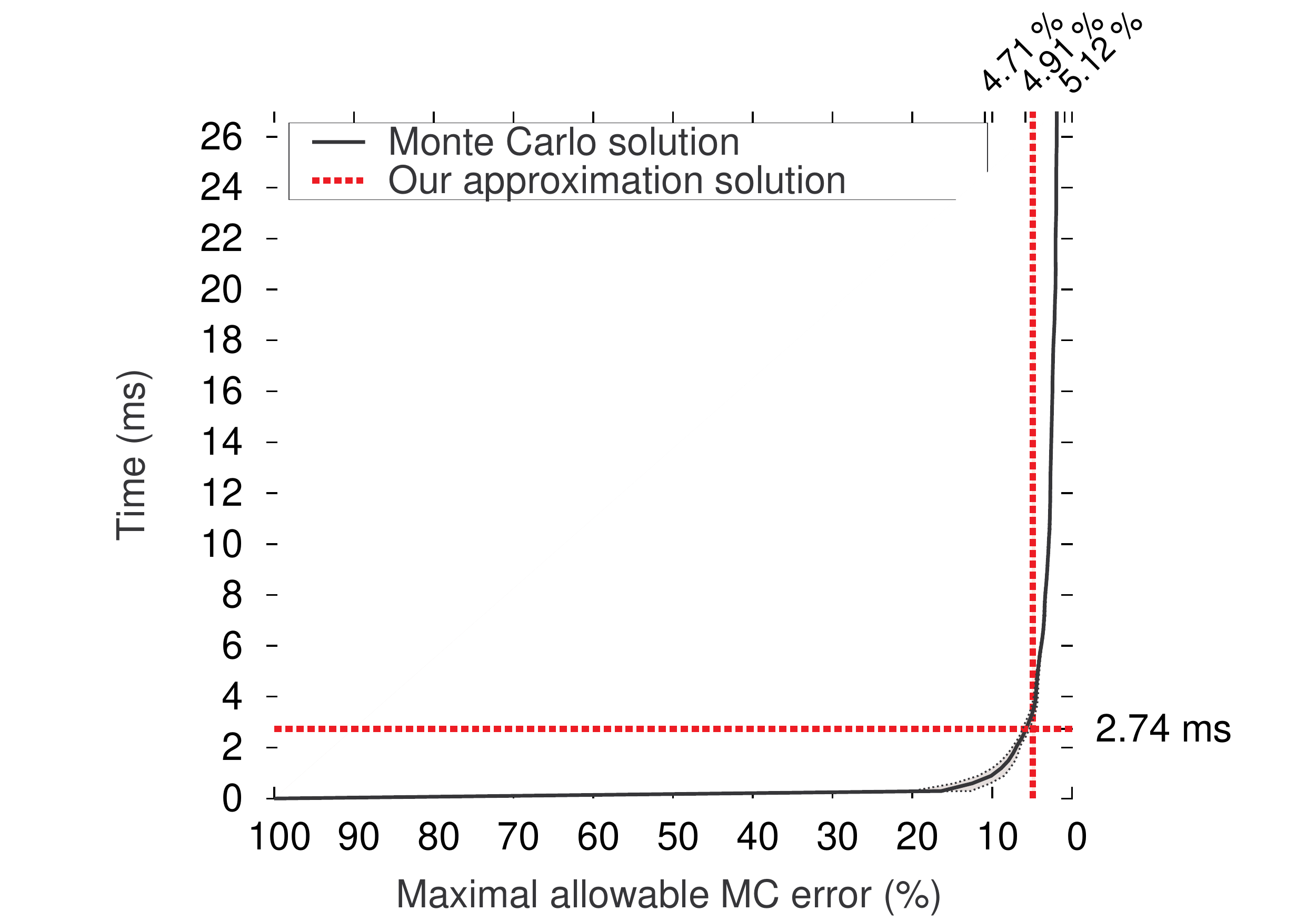}
\caption{$g_5$, $W_{{\cutoff}}=\frac{3}{4}\alpha_i$, $K_{12}$ (W)}
\label{fg54}
\end{minipage} \\
\end{figure}

\begin{figure}[t]
\centering
\includegraphics[natheight=3.3cm, natwidth=8cm, height=5cm, width=0.5\textwidth]{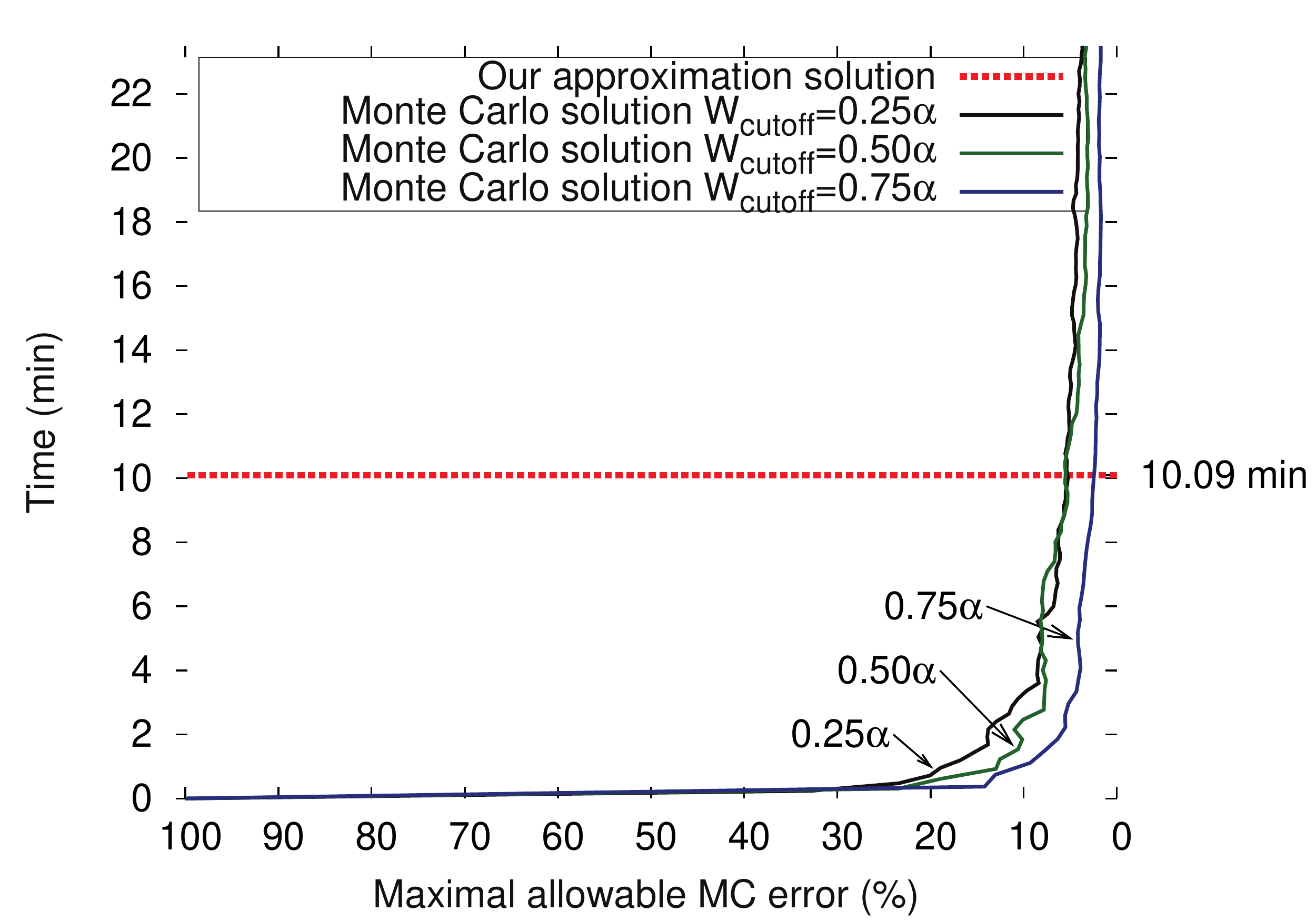}
\caption{$g_5$, $W_{{\cutoff}}=\frac{1}{4}\alpha_i$, $W_{{\cutoff}}=\frac{2}{4}\alpha_i$, and $W_{{\cutoff}}=\frac{3}{4}\alpha_i$, $K_{1000}$ (W)}
\label{fg5b1}
\end{figure}

Figures \ref{fg51}, \ref{fg52}, \ref{fg53} and \ref{fg54} present comparisons of our approximation algorithm for game $g_5$ against Monte Carlo sampling for small networks (for which the exact Shapley value can be computed from the Definition~\ref{def:sv})). In these figures, the horizontal dotted line shows the running time of our solution, while the vertical dotted line shows its average approximation error with the shaded area being the confidence interval. As previously, the solid line shows the average, and the shaded area depicts the confidence interval for the Monte Carlo simulations. We see in Figures \ref{fg51}, \ref{fg52} and \ref{fg53} that the approximation error in our proposed algorithm is well-contained for small networks. Specifically, for $K_{6}$ it is about $10\%$; whereas for the bigger network $K_{12}$ it is about $5\%$. However, we notice that, for higher values of $W_{{\cutoff}}$, the Monte Carlo method may slightly outperform our solution. See in Figure \ref{fg52} how the average approximation error of the Monte Carlo sampling achieved in $0.38ms$ is lower than the average error achieved by our method. Already for $K_{12}$ this effect does not occur (see Figure \ref{fg54}).

 For large networks, where the exact Shapley value cannot be obtained, we are naturally unable to compute exact approximation error. We believe that this error may be higher than the values obtained for $K_6$ and $K_{12}$. However, the mixed strategy, that we discussed in Section \ref{sec:game5} and that uses our approximation only for large degree vertices, should work towards containing the error within practical tolerance bounds. As far as we believe that Monte Carlo gives good results, from Figure \ref{fg5b1}, we can infer that our approximation solution for large networks gives good results (within $5\%$) and is at least two times faster than the Monte Carlo algorithm.

To summarise, our exact solutions outperform Monte Carlo simulations even if relatively wide error margins are allowed. However, this is not always the case for our approximation algorithm for game $g_5$.
Furthermore, it should be underlined that if the centrality metrics under consideration cannot be described with any of the games $g_1$ to $g_4$ for which exact algorithms are now available, then Monte Carlo simulations are still a viable option.

\section{Conclusions}

\noindent The key finding of this chapter is that the Shapley value for many centrality-related cooperative games of interest played on networks can be solved analytically. In particular, we proposed the polynomial time algorithms for computing Shapley value-based degree and closeness cnetralities. The resulting algorithms are not only error-free, but also run in polynomial time (see Table~\ref{front_table}) and, in practice, are much faster than Monte Carlo methods. Approximate closed-form expressions and algorithms can also be constructed for some classes of games played on weighted networks. Simulation results show that these algorithms are acceptable for a range of situations involving big networks.

\chapter{Efficient Algorithms for a Game-Theoretic Measure of Betweenness Centrality}\label{chap:bet}

\noindent Betweenness centrality is a measure of the potential of individual nodes to control information flow in a network. In this chapter, we extend the standard definition of betweenness centrality and propose a family of game-theoretic measures based on \emph{Semivalues}---a large family of solution concepts from cooperative game theory that include, among others, the Shapley value and the Banzhaf power index. The key advantage of the game-theoretic betweenness centrality over the standard measure is the ability to rank nodes while taking into account any synergies that may exist in different groups of nodes. We develop polynomial time algorithms to compute the Semivalue-based betweenness centrality in general, and the Shapley value-based betweenness centrality in particular, both in weighted and unweighted graphs. Interestingly, for the unweighted case, our algorithm for computing the Shapley value-based centrality has the same complexity as the best known algorithm due to \cite{Brandes:2001} for computing the standard betweenness centrality. The new measures are evaluated in simulated scenario where simultaneous node failures occur. Compared to the standard measure, the ranking obtained by our measures reflects more accurately the influence that different nodes have on the functionality of the network. Finally, we empirically evaluate our algorithms on weighted and unweighted random scale-free graphs, and on two real network: an email communication network, and the neural system of C. elegans. This evaluation shows that our algorithms are significantly faster than the fastest alternative from the literature.

All algorithms, experiments and theorems presented in this chapter are an exclusive contribution of the author of this thesis.

%The remainder of the chapter is organized as follows. In Section \ref{chap:bet:intro}, we provide the motivation why betweenness centrality and its extension to game-theoretic framework are important. Section~\ref{chap:bet:related} describes the related work. The motivation scenarios and definition of new centrality is provided in Section~\ref{chap:bet:our_centrality}. Section \ref{chap:bet:algorithms} describes algorithms for game-theoretic betweenness centrality on weighted and unweigthed graphs and Section~\ref{chap:bet:simulations} provides extensive empirical evaluation.

\section{Introduction}\label{chap:bet:intro}

\noindent Betweenness centrality is a measure of control introduced independently by \cite{Anthonisse:1971} and \cite{Freeman:1979}, which aims to determine the importance of a given node in the context of network flow. More specifically, in its standard form, betweenness centrality considers all shortest paths that pass through the node. The more shortest paths go through the node, the more important is the role of this node in the network.

Betweenness centrality (Definition~\ref{def:cen:bet}) is one of the classical measures of centrality in network analysis, with the other measures being degree centrality, closeness centrality, and eigenvector centrality (see Section~\ref{sec:centralities}). Several applications of betweenness centrality have been considered in the literature \citep{Puzis:et:al:2007,Dolev:et:al:2010}. It has also been implemented in numerous software tools for network analysis \citep{Brandes:2008}. Furthermore, a variety of extensions of betweenness centrality have been proposed in order to suit the needs of particular applications. For instance, a well-known community detection algorithm by \cite{Girvan:Newman:2002} is based on edge betweenness centrality.

One of the characteristics of betweenness centrality, as well as the other classical centrality measures, is that it focuses on the role played by each node individually. However, such an approach can be inadequate if there exist non-negligible positive or negative synergies between nodes when their role is considered in groups. Simply put, classical centrality measures are not able to account for such synergies.

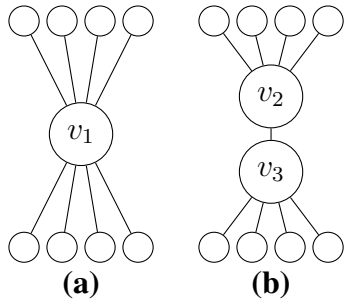
\begin{figure}[t!]\centering
\begin{tikzpicture}[node distance = 1.5cm]
\tikzstyle{every node}=[draw,shape=circle];

\path (0,0cm) node (v0) {};
\path (0.5,0cm) node (v1) {};
\path (1,0cm) node (v2) {};
\path (1.5,0cm) node (v3) {};
\path (0.75,-1.5cm) node (v4) {$v_1$};
\path (0,-3cm) node (v5) {};
\path (0.5,-3cm) node (v6) {};
\path (1,-3cm) node (v7) {};
\path (1.5,-3cm) node (v8) {};

\path (2.5,0cm) node (v9) {};
\path (3,0cm) node (v10) {};
\path (3.5,0cm) node (v11) {};
\path (4,0cm) node (v12) {};
\path (3.25,-1cm) node (v13) {$v_2$};
\path (3.25,-2cm) node (v14) {$v_3$};
\path (2.5,-3cm) node (v15) {};
\path (3,-3cm) node (v16) {};
\path (3.5,-3cm) node (v17) {};
\path (4,-3cm) node (v18) {};

\path (0.75,-3.5cm) node[draw=none] (va) {\textbf{(a)}};
\path (3.25,-3.5cm) node[draw=none] (va) {\textbf{(b)}};

\draw (va);

\draw
	  (v0) -- (v4) -- (v1) (v2) -- (v4) -- (v3)
	  (v5) -- (v4) -- (v6) (v7) -- (v4) -- (v8)	
	
	  (v9) -- (v13) -- (v10) (v11) -- (v13) -- (v12)
	  (v13) -- (v14)
	  (v15) -- (v14) -- (v16) (v17) -- (v14) -- (v18)	
	  ;

\end{tikzpicture}
\caption{The example demonstrating that classical centrality measures do not provide any insight into synergies among the nodes. In network~\textbf{(a)}, betweenness centrality identifies $v_1$ as the most important node. Indeed, $v_1$ controls all $28$ shortest paths between other nodes. Similarly, in network~\textbf{(b)}, nodes $v_2$ and $v_3$ are central as they control $26$ shortest paths each. Now, let us consider $v_2$ and $v_3$ as a group $\{v_2,v_3\}$. In terms of network flow, the role of this group in network~\textbf{(b)} is exactly the same as the role of $v_1$  in network~\textbf{(a)}. However, it is unknown how the role of this group could be derived from the classical betweenness centrality (of both individual nodes). For instance, one natural approach to quantify the role of a group of nodes would be to sum up their individual centralities. But,  for $\{v_2,v_3\}$, this would yield $26+26=52>28$ which clearly exaggerates their joint value. Another simple approach would be to average the classical centrality $v_2$ and $v_3$ \citep{Everett:Borgatti:1999}. This would however understate the role of $\{v_2,v_3\}$ since $\frac{26+26}{2}=26<28$. Now, the group betweenness centrality assigns to $\{v_2,v_3\}$ value of $28$. Thus, unlike classical betweenness centrality, the group betweenness centrality is able to capture the negative synergy between $v_2$ and $v_3$.}
\label{fig:intro}
\end{figure}

To address this problem, \cite{Everett:Borgatti:1999} introduced \textit{group centrality} (see Definition~\ref{def:group:bet}). While group centrality measures have the advantage of taking synergies into account, these measures have their own shortcomings. Specifically, they produce a ranking of an exponential number of groups, which is hard to maintain in memory, and hard to reason about, as it is not obvious how to extract meaningful conclusions from so many groups. What would be more manageable is to produce a ranking of the nodes themselves, based on the ranking of all possible groups. The answer to it was given in the previous chapter where game-theoretic centrality measures was introduced (Definition~\ref{def:gt:cen}).

A potential downside of \textit{game-theoretic network centralities} is that many methods from cooperative game theory are computationally challenging. For instance, given a coalitional game defined over a network of $|V|$ nodes, a straight-forward computation of the Shapley value requires considering all possible $2^{|V|}$ coalitions (i.e., groups) of nodes. This is clearly prohibitive for networks with hundreds, or even tens, of nodes. Indeed, it has been shown that in some cases the exponential number of computations cannot be avoided, i.e., it is impossible to compute particular game-theoretic network centralities in time polynomial in the size of the network.

Fortunately, the positive computational results have been given in the previous chapter where we analysed various game-theoretic extensions of degree and closeness centrality and showed that sometimes it is possible to leverage the fact that the synergies in coalitions depend on the topology of the network. As a result, it was showed that some game-theoretic centrality measures are indeed computable in polynomial time.

Nevertheless, the aforementioned positive results concerned solely the extensions of degree and closeness centralities which are, in general, computationally less challenging than betweenness centrality. In fact, no game-theoretic extension of betweenness centrality has been developed to date.

In this chapter, we fill the above gap in the literature by proposing and analysing the computational aspects of the \emph{game-theoretic betweenness centrality}. This is the first game-theoretic centrality that is based on \textit{Semivalues}---a family of generalisations of the \textit{Shapley value} that offer a wider spectrum of ways in which the role of nodes within various groups can be evaluated.

As our main contribution, we propose polynomial time algorithms to compute both the Shapley value-based and Semivalue-based betweenness centralities, for weighted graphs as well as unweighted graphs. Surprisingly, as shown in Table~\ref{tab:summary:comlexity}, for the unweighted case, our algorithm for computing the Shapley value-based centrality has the same complexity as the best known algorithm due to \cite{Brandes:2001} for computing the standard betweenness centrality, i.e., $O(|V||E|)$, where $V$ is the set of nodes and $E$ set of edges in the network. Moreover, for Semivalues which are, in general, more computationally challenging than the Shapley value, the running time of our algorithm is $O(|V|^2|E|)$ (for the unweighted graphs).

We consider three applications of the new centrality measures. The first one concerns the protection of computer network against Spyware---the malicious computer programs, the second example is about the issue of measuring systemic importance of financial institutions \citep{ECB:2010}. Under certain assumptions, the Shapley value- and the Semivalue-based centrality measures turn out to be the exact solution to this problem. In the third application we simulate simultaneous node failures in a network. Following \cite{Albert:et:al:2000} and \cite{Holme:et:al:2002}, among others, we quantify the functionality of the network based on the average inverse geodesic measure. The simulation results show that compared to the standard measure, the cardinal ranking obtained by our measures reflects more accurately the influence that different nodes have on the network functionality.

\section{The Overview of Computational Results}\label{chap:bet:related}

\noindent The literature review on centralities and game-theoretic centrality in particular, was already presented in Chapter~\ref{chap:related}. Now, we discuss the computational properties of the classical betweenness centrality. In particular, the fastest known dedicated algorithm to compute it is due to \cite{Brandes:2001}. The algorithm works in $O(|V||E|)$ time for unweighted graphs, and in $O(|V||E| + |V|^2 \log(|V|) )$ time for weighted graphs, where $V$ is the set of nodes and $E$ set of edges in the network. Furthermore, group betweenness centrality for a pre-defined group of nodes can be computed in $O(|V||E|)$ time as shown by \cite{Brandes:2008}. \cite{Puzis:et:al:2007} developed an algorithm to compute group betweenness centrality of a given subset $S$. The algorithm starts with a preprocessing step (which alone takes $O(|V|^3)$ time), and then computes group betweenness centrality in $O(|S|^3)$ time.
\def\arraystretch{1.3}
\begin{table}[h!]
\begin{center}
\footnotesize
\begin{tabular}{| c | c | c |}
\hline
\textbf{Centrality}  & \textbf{Unweighted Graphs} & \textbf{Weighted Graphs} \\
\hline
\hline
standard             &     $O(|V||E|)$        &   $O(|V||E| + |V|^2 \log(|V|) )$  \\
betweenness          & \cite{Brandes:2001}    &  \cite{Brandes:2001}               \\
\hline
group                &     $O(|V||E|)$        &   $O(|V||E| + |V|^2 \log(|V|) )$    \\
betweenness          &   \cite{Brandes:2008}  &   \cite{Brandes:2008}               \\
\hline
Shapley value-based  &      $O(|V||E|)$       &      $O(|V|^2|E|) + |V|^2 log |V| $  \\
betweenness          &  \textbf{this dissertation}          &   \textbf{this dissertation}     \\
\hline
Semivalue-based      &        $O(|V|^2|E|)$   &      $O(|V|^3|E|) + |V|^3 log |V| $ \\
betweenness          &   \textbf{ this dissertation}          &     \textbf{this dissertation}             \\
\hline
\end{tabular}
\caption{Summary of the computational results obtained in this chapter.}
\label{tab:summary:comlexity}
\end{center}
\end{table}

\def\arraystretch{1}
Table~\ref{tab:summary:comlexity} summarises the above complexity results (for betweenness centrality and group betweenness centrality) as well as our computational results (for Shapley-based and Semivalue-based betweenness centralities). As can be seen, our algorithms also run in  polynomial time, both for unweighted and weighted graphs. Observe that our algorithm for computing the Shapley value-based betweenness centrality for unweighted graphs has the same complexity as the aforementioned algorithm by \cite{Brandes:2001} for standard betweenness centrality, which is surprising given the much more complicated nature of our measure.

\section{The New Centrality and Its Properties}\label{chap:bet:our_centrality}

\noindent In this section we introduce a family of game-theoretic betweeneness centrality measures and discuss their basic properties. We start in Section~\ref{subsection:failureOfStandardBetweenness}
with a few motivating scenarios that demonstrate the need for a more involved centrality measure than standard betweenness centrality. Sections \ref{subsection:ShapleyBetweenness} and \ref{subsection:semivalueBetweenness} introduce our Shapley value-based and Semivalue-based betweenness centralities, respectively. Finally, sections \ref{subsection:MCforShapley} and \ref{Subsection:ProbabilisticSV} analyse the marginal contributions of nodes, in the context of the Shapley value-based and Semivalue-based betweenness centralities, respectively.

\subsection{Motivating Scenarios}\label{subsection:failureOfStandardBetweenness}

\noindent %In this section we explain in more detail why standard betweenness centrality is unable to account for synergies.
We will consider three motivating scenarios: protect computer network against malicious programs (Section~\ref{section:network:intrusion}), ranking institutions in a financial system according to their systemic importance (Section~\ref{section:network:financial}), and analysing network vulnerability (Section~\ref{section:network:vulnerability}).

\subsubsection{Detecting Intrusion in a Network}\label{section:network:intrusion}
\noindent One of the scenarios on detecting intrusion in a network considered in the literature is \textit{distributed
network intrusion detection} \citep{Valdes:Skinner:2001}. In a nutshell, the aim here is to filter the content of a network by installing some anti-spyware software on a chosen set of $k$ nodes. Now, assuming that the majority of the network flow is transmitted through shortest paths, and ignoring the capacities of cables and routers, the best choice is to install the anti-spyware on the set of nodes $S \subseteq V, |S|=k$ whose group betweenness centrality $c_{gb}(S)$ (Definition~\ref{def:group:bet}) is maximal. Unfortunately, the problem is NP-hard; nonetheless, there exists approximation algorithms to solve it \citep{Puzis:et:al:2007b}.

The solution to the above problem can be used to guide the efforts to protect the most important nodes against a potential spyware attack. However, how to choose the optimal points to protect if we have limited resources and do not known in advance the size of the attack? To model such a situation, let us consider the following scenario. Suppose that nodes in a communication network have been exposed to the possibility of a random spyware infection. The more shortest paths there are that pass through the infected group, the worse it is for the network, because the malicious software may control or monitor the flow. Furthermore, suppose that neither the subset of infected nodes is known nor the size of the attack (i.e., the number of infected nodes). Assuming that the administrators' resources are limited and that the malicious software can be located at any workstation, a decision has to be made on which node to examine first. In other words, we need to identify a node that makes the biggest contribution to the spyware's ability to monitor the flow going through the infected nodes. Hence, we formalize this problem as follows:

\begin{problem}\label{def:pro0}
Given a discrete probability distribution $D:\{1,\cdots,n\}\to[0,1]$ and a random subset $S: |S|\sim D$, which node $v \in V$ in the graph $G=(V,E)$ maximizes the expression: $c_{gb}(S) - c_{gb}(S \setminus \{v\})$.
\end{problem}

Now if it was known that exactly one node has been affected by the spyware attack ($|S|=1$) and that the single target was chosen uniformly at random, then the standard betweenness centrality would solve the above problem; the node with the highest such centrality is the solution. However, if there is a possibility that multiple nodes have been infected ($|S|>1$), then the standard betweenness centrality becomes insufficient. This is because betweenness centrality of a group of nodes is not necessarily equal to the sum of betweenness centralities of individual nodes within this group (see Figure~\ref{Fig:difference} for an example). In other words, the following does not necessarily hold for every $S \subseteq V$:
\[
c_{gb}(S) = \sum_{v\in V} c_b(v).
\]
Hence, a node $v \in V$ that maximizes the expression $c_{gb}(S) - c_{gb}(S \setminus \{v\})$ for some $S \subseteq V$ does not have to be the one with the highest standard betweenness centrality.

In this chapter we propose the polynomial-time exact solution to Problem~\ref{def:pro0} for a family of probability distributions of choosing the set $S$.

\subsubsection{Measuring Systemic Importance of Financial Institutions}\label{section:network:financial}

\noindent The crisis of year 2007 and 2008 clearly showed the need for more stringent supervision of financial systems. As a result, a number of systemic risk monitoring and supervisory bodies have been created worldwide: the Financial Stability Oversight Council (USA), the European Systemic Risk Board and the Financial Stability Board. Furthermore, new prudential policies and regulations were laid out. Most notably, the Basel III  accords aim at introducing micro- and macro-prudential regulations such as stricter capital reserve ratios, increased audit transparency and more thorough risk management by banks and other financial institutions.

At a discretion of supervising authorities, systemically important institutions may be asked to follow additional macro-prudential policies including a combination of capital surcharges, contingent capital and bail-in debt. As
stated in the Financial Stability Board's interim report to G20 leaders, \emph{``Financial institutions
should be subject to requirements commensurate with the risks they pose to the financial
system''} \citep{FSB:2010}.

In this context, betweenness is considered to be the most suitable centrality to measure systemic importance of institutions in the financial system \citep{ECB:2010}. In more detail, it has the following natural interpretation: ``\emph{the betweenness of a bank $A$ connecting pairs of nodes in the network is a measure of the dependence of these other banks from $A$ to transfer the loans.}'' \citep{Gabrieli:2011}. As a result, the higher betweenness centrality of an institution is, the more attention and scrutiny it requires. It is feared that in the worst case, even one of the institution with the highest betweenness centrality may be able to bring down the entire system, due to contagion of illiquidity \citep{Freixas:et:al:1998}.

However, imposing additional capital requirements on financial institutions according to their standard betweenness centrality ignores the fact that financial crises are often ignited not by factors specific to a particular bank but by (combinations of) factors that simultaneously affect possibly many banks in the system. The infamous financial crises of 2009 started with the collapse of the subprime residential mortgage market in the United States and then spread to the rest of the world, due to the exposure to American real estate assets via a complex array of financial derivatives. As a result, many financial institutions in various countries asked for urgent state interventions \citep{Laeven:Valencia:2010}.

Consequently, it would be desirable to extend betweenness centrality so that it quantifies the betweenness of an individual institution but taking into account its role in various groups of institutions, as they can be affected simultaneously. We formalize this problem as follows:

\begin{problem}\label{def:pro1}
Given graph $G=(V,E)$, what is the (cardinal) ranking of nodes that reflects the marginal contribution of an individual node to group betweenness centrality  of a random subset of nodes $S: |S|\sim D$, i.e.,  $c_{gb}(S) - c_{gb}(S \setminus \{v\})$, where $D:\{1,\cdots,n\}\to[0,1]$ is a discrete probability distribution.
\end{problem}

This problem extends the Problem~\ref{def:pro0} in the way, that we consider the whole cardinal ranking of nodes, instead of looking only for the most important individual. In this chapter, we propose the polynomial-time exact solution to Problem~\ref{def:pro1} for a family of probability distributions $D$ of choosing set $S$.

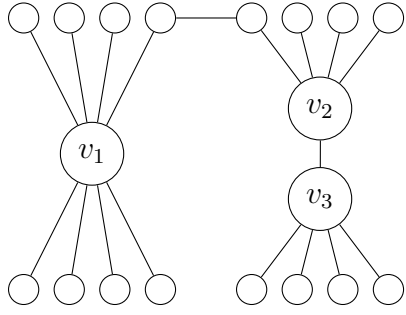
\begin{figure}[t!]
\begin{center}
\begin{tikzpicture}[scale=1.2]
\tikzstyle{every node}=[draw,shape=circle];

\path (10,0cm) node (v99) {};
\path (10.5,0cm) node (v109) {};
\path (11,0cm) node (v119) {};
\path (11.5,0cm) node (v129) {};
\path (10.75,-1cm) node (v139) {$v_2$};
\path (10.75,-2cm) node (v149) {$v_3$};
\path (10,-3cm) node (v159) {};
\path (10.5,-3cm) node (v169) {};
\path (11,-3cm) node (v179) {};
\path (11.5,-3cm) node (v189) {};

\draw (v99) -- (v139) -- (v109) (v119) -- (v139) -- (v129)
	  (v139) -- (v149)
	  (v159) -- (v149) -- (v169) (v179) -- (v149) -- (v189)
	  ;

\path (7.5,0cm) node (v09) {};
\path (8,0cm) node (v19) {};
\path (8.5,0cm) node (v29) {};
\path (9,0cm) node (v39) {};
\path (8.25,-1.5cm) node (v49) {$v_1$};
\path (7.5,-3cm) node (v59) {};
\path (8,-3cm) node (v69) {};
\path (8.5,-3cm) node (v79) {};
\path (9,-3cm) node (v89) {};

\draw (v09) -- (v49) -- (v19) (v29) -- (v49) -- (v39)
	  (v59) -- (v49) -- (v69) (v79) -- (v49) -- (v89)	
	
	  	  (v99) -- (v39)
	  ;	

% \path (9.5,-3.5cm) node[draw=none] (va) {\textbf{(d)}};-
	
\end{tikzpicture}
\end{center}
\caption{In the above sample network, both $v_1$ and $v_{2}$ are ranked equally by standard betweenness centrality; $c_b(v_1) = c_b(v_{2}) = 98$. Indeed, either $v_1$ or $v_2$ control the same number of shortest paths in this network. However, if we consider the role of nodes played within various groups---as is the case in our scenario where financial difficulties may affect multiple institutions simultaneously---then $v_1$ and $v_2$ should not be evaluated equally. To see why this is the case, let us consider nodes $v_2$ and $v_3$. A significant percentage of the shortest paths controlled by $v_2$ are also controlled by $v_3$. Consequently, the ability to control financial flows by the group of banks $\{v_2,v_3\}$ is smaller than  by the group $\{v_1,v_3\}$. As a result, $v_1$ is more systemically important than $v_2$.
}
\label{Fig:difference}
\end{figure}

\subsubsection{Analysis of Network Vulnerability}\label{section:network:vulnerability}

\noindent One of the extensive research directions on betweenness centrality involves the analysis of network vulnerability \citep{Bompard:et:al:2010,Holme:et:al:2002}. The aim is to identify (and possibly protect) the most critical nodes, whose removal degrades the functionality of the network the most. In this context, a standard measure used to assess the condition of the network is called  \textit{the average inverse geodesic measure} which, essentially, is a sum of inverse distances between any pair of nodes in the network \citep{Holme:et:al:2002,Albert:et:al:2000}. Formally:
\[
\text{IGM}(V,E) = \sum_{v \in V} \sum_{v \neq u \in V} \frac{1}{d(v,u)}
\]
\noindent where $d(u,v)$ is the distance, i.e., the length of the shortest path, between nodes $u$ and $v$ \citep{Freeman:1979}.

In order to formalize the problem of network vulnerability, let us first note that it is possible to reformulate Problem~\ref{def:pro0} by assuming that we are looking for a set of nodes $C \subseteq V$ of a pre-defined size $|C|=c$. This set should maximize the expression: $c_{gb}(S) - c_{gb}(S \setminus C)$, where $S$ is a random subset of nodes $S \subseteq V$ such that $l \leq |S| \leq k$.

Now, we define the problem of network vulnerability as follows:

\begin{problem}\label{def:pro3}
Given a discrete probability distribution $D:\{1,\cdots,n\}\to[0,1]$ and a random subset $S: |S|\sim D$, which set of nodes $C \subseteq V$ of a given size $|C|=c$ in the graph $G=(V,E)$  maximizes the expression: $IGM(V\setminus S) - IGM(V \setminus (S\setminus C))$.
\end{problem}
\noindent where $IGM$ stands for the average inverse geodesic measure.

To date, standard betweenness centrality has been advocated as an effective approximation that solves Problem~\ref{def:pro3} \citep{Holme:et:al:2002}. However, in Section~\ref{chap:bet:simulations}, we experimentally show that the game-theoretic betweenness centrality achieves consistently better results than standard betweenness centrality.

\subsection{The Shapley Value-based Betweenness Centrality}\label{subsection:ShapleyBetweenness}

\noindent In this subsection we will define group betweeness centrality. Firstly, recall that in our approach, we cast a network into the combinatorial structure of a coalitional game. That is, we consider all subsets of the set of nodes $V(G)$ and we define a function that assigns to every such subset a numerical value; this value is in fact the subset's group betweenness centrality. Formally, borrowing notation from coalitional games, we define a tuple $(V(G),\nu)$, where function $\nu:2^{V(G)}\to\mathbb{R}$ assigns to each $S \subseteq V(G)$ the value $c_{gb}(S)$, and $\nu(\emptyset) = 0$. This means that $\nu$ returns the group betweenness centrality of any given subset of nodes. Now since $(V(G),\nu)$ has the same combinatorial structure as that of a coalitional game, we can use the Shapley value to quantify the importance of nodes based on the impact their membership makes on group betweenness centrality. In other words, we obtain a measure of centrality of individual nodes which aggregates information about their role in group betweenness centrality across various groups. Formally, the Shapley value-based betweenness centrality measure is defined as follows:

\begin{definition}[Shapley value-based betweenness centrality measure]\label{def:gt:bet}
\textit{Shapley-value based betweenness centrality} is a pair~$(\psi_B, \phi^{\SV})$ consisting of a representation function $\psi_B(G) = c_{gb}$ that maps each graph to the game with group betweenness centrality and a Shapleu value solution concept $\phi^{\SV}$.
\end{definition}

\begin{example} Consider again nodes $v_1$ and $v_2$ from Figure~\ref{Fig:difference}. Here:
\[\phi^{\SV}_{1}(V(G),\psi_B(G)) = 18.2 ~~\text{ and }~~ \phi^{\SV}_{2}(V(G),\psi_B(G)) = 16.0833.\]
 Unlike the standard betweenness centrality, the ranking produced by our Shapley-value based betweenness centrality reflects the fact that node $v_2$ has many common shortest paths with $v_3$, while $v_1$ does not.

\end{example}
Intuitively, the Shapley value-based betweenness centrality measure represents the load placed on a given node, taking into account the role it plays in all possible groups of nodes in the network. Moreover, the Shapley value-based betweenness centrality satisfies the same properties as the Shapley value (see Section~\ref{sec:SV}). Let us translate them to the network context:

\begin{enumerate} \itemsep0.5em
\item \textbf{Symmetry}---any two symmetric nodes $v_i$ and $v_j$ are assigned the same centrality, where symmetry refers here to the fact that the two nodes can be interchanged in any group of nodes without affecting the group's betweenness. Formally, if $v(S \cup \{v_i\}) = v(S \cup \{v_j\})\ \forall S \subseteq V\setminus\{v_i,v_j\}$, then: $ \phi^{\SV}_{i} = \phi^{\SV}_{j}$;
\item \textbf{Marginality}---if a node's contribution to the group betweenness centrality of each subset of nodes in  network $G = (V,E)$ is no less than its contribution to the same subset in another network $G' = (V,E')$, then its centrality in the former network should be no less than in the latter one. Formally, for every pair of networks, $G = (V,E)$ and $G' = (V,E')$, and functions $\nu_G$ and $\nu_{G'}$, defined as in equation \eqref{def:group:bet}, and for every node $v_i\in V$, if
   \[
	\nu_G(S \cup \{v_i\})-\nu_G(S) \geq    	\nu_{G'}(S \cup \{v_i\})-\nu_{G'}(S), \ \ \ \forall S \subseteq V\setminus \{v_i\},
   \]
   then $\phi^{\SV}_{v}$ in $G$ is not smaller than $\phi^{\SV}_{v}$ in $G'$.
\item \textbf{Efficiency}---this property simply says that the sum of the Shapley value-based betweenness centralities for all the nodes in the network is normalised to $0$, i.e.,
$\sum_{v_i \in V} \phi^{\SV}_i = 0$.\footnote{\footnotesize Note that the group betweenness centrality of the entire network, i.e., of the group $S = V$, equals $0$.}
Naturally, it is also possible to normalize the Shapley value-based betweenness centrality so it ranges from 0 to 1. To this end, one should use $\frac{\phi^{\SV}_i}{2\max_{ v_i \in V }c_b(v_i)}+\frac{1}{2}$. This is because $\forall_{ v_i \in V }\forall_{C \subseteq V} ( c_{gb}(C \cup \{v_i\}) - c_{gb}(C)) \leq c_b(v_i)$.
\end{enumerate}

Observe that the new centrality solves Problem~\ref{def:pro0} and Problem~\ref{def:pro1} exactly for a particular probability distribution $\textit{PD}$.

\begin{proposition}\label{Prop:SV0}
The Shapley value-based betweenness centrality is the exact solution to Problem~\ref{def:pro0} and Problem~\ref{def:pro1} in a particular case in which each set $S \subseteq V$ of size $|S|=k>0$ is chosen with probability $\frac{1}{|V|\binom{|V|}{k}}$.
\end{proposition}
\begin{proof}
To show it, it is enough to rewrite formula~\eqref{eq:sv:original} as follows:
\begin{equation}\label{Eq:centrality:SV}
\phi^{\SV}_i(V,\nu) = \frac{1}{|V|} \sum_{1 \leq k \leq |V|}\frac{1}{\binom{|V|-1}{k-1}} \sum_{\substack{S \subseteq V \\ |S|=k \\ v_i \in S}} (\nu(S) - \nu(S \setminus \{v_i\} ))
\end{equation}
The inner sum in the above equation runs through all subsets of nodes containing $v_i$ of a given size~$k$, and $\binom{|V|-1}{k-1}$ is the number of such subsets. The first fraction implies that the size $k\in\{1,\ldots,n\}$ is chosen uniformly at random. The second fraction implies that out of all subsets of size $k$, one is chosen uniformly at random. %Thus, the node with the highest $SV_i(V,\nu)$ will contribute the most to the random set $v_i \in S \subseteq (V\setminus \{v_i\})$ chosen with probability $\frac{1}{|V|\binom{|V|-1}{k-1}}$. Finally, for any non-empty random set $S$ chosen with probability $\frac{1}{|V|\binom{|V|}{k}}$ the aforementioned contribution would also be maximized. This concludes our proof.
\end{proof}

Formula~\eqref{Eq:centrality:SV} has the following interpretation in the context of measuring systemic importance of financial institutions:
\begin{enumerate} \itemsep0.5em
\item Fraction $\frac{1}{|V|}$ means that the probability that the financial crises starts with a group of financial institutions of size $|S|=k$ is the same for any $k$ (where $k\in\{1, \ldots, |V|\}$). We will denote this probability $\textit{PD}(V,k)$.
\item Fraction $\frac{1}{\binom{|V|-1}{k-1}}$ means that the probability that each particular $k$-sized subset of institutions containing $v_i$ bankrupts is the same.
\end{enumerate}

While the latter implication seems natural, the former implication seems to be a special case. Arguably, it does not seem natural to assume that the probability of a single institution bankruptcy is exactly the same as the probability of having all institutions bankrupted (especially if the network is very large). To relax this assumption, we need to change the fraction $\frac{1}{|V|}$ in equation~\eqref{Eq:centrality:SV}. However, this fraction is inherent in the Shapley value. As such, modifying it requires adopting Semivalues as a centrality measure. In the next section, we show how this can be done.

%\begin{remark}\label{Prop:SV0}
%Let $G=(V,E)$ be a network, and assume that the simultaneous spyware infection of $1 \leq k \leq |V|$ nodes occurs with uniform probability, i.e., $\textit{PD}(V,k) = \frac{1}{k}$, and that groups of equal size become infected with equal probability, i.e., $\textit{PD}(V,k)\cdot\frac{1}{\binom{|V|}{k}}$. Then, the node with highest Semivalue-based betweenness centrality is the node whose membership in an infected group is expected to be the most harmful, in terms of group betweenness centrality. In other words, this measure is a solution to Problem~\ref{def:pro1} for $\textit{PD} = \textit{PD}(V,k)\cdot\frac{1}{\binom{|V|}{k}}$.
%\end{remark}

% In the next subsection, we show that the Semivalue allows us to generalize $\textit{PD}$ to an entire family of probability distributions.

\subsection{Semivalue-based Betweenness Centrality}\label{subsection:semivalueBetweenness}
 %Let us begin by considering the difference between the Shapley value and Semivalues in more detail.
%They are essentially generalisations of the Shapley value and they provide the exact solution the Problem~\ref{def:pro1} and good approximation to the Problem~\ref{def:pro3}.
%To this end, recall that the Shapley value can be computed as in \eqref{Eq:centrality:SV}, where the outer sum (with $k$ running from 1 to $|V|$) is multiplied by $\frac{1}{|V|}$. This can be interpreted as a uniform distribution over the set $\{1,\ldots,|V|\}$. Since
\noindent Recall that a Semivalues defined in equation~\eqref{eq:semi} generalises the Shapley value by allowing for different sequences of weights $p_{|S|}$ for each set $S$. Building upon Semivalues, we will now introduce a centrality measure that generalizes the Shapley value-based betweenness centrality.

\begin{definition}[Semivalue-based betweenness centrality measure]\label{de:gt:semi}
\textit{Semivalue-based betweenness centrality measure} is a pair~$(\psi_B, \phi^{\SEMI})$ consisting of a representation function $\psi_B(G) = c_{gb}$ that maps each graph to the game with group betweenness centrality and a Semivalue  $\phi^{\SEMI}$.
\end{definition}

We observe that the above family of Semivalue-based betweenness centrality measures inherits from Semivalues the properties of symmetry and marginality. Naturally, particular centralities may satisfy additional properties inherited from the corresponding solution concepts in coalitional games. For instance, it is not difficult to see that, similarly to the Banzhaf index \citep{Casajus:2011}, the Banzhaf index-based betweenness centrality measure satisfies the ``2-Efficiency'' property. In essence, this property says that, for each pair of nodes $u, v \in V$,  if you consider them to be a single node $vu$ (note that $v$ and $u$ are not merged, but they appear in each coalition together) then the Banzhaf index-based betweenness centrality of $vu$ equals the sum of the Banzhaf index-based betweenness centralities of $u$ and $v$.

\begin{proposition}\label{Prop:SV1}
The Semivalue-based betweenness centrality is the exact solution to Problem~\ref{def:pro0} and Problem~\ref{def:pro1}, where each set $S \subseteq V$ of size $|S|=k>0$ is chosen with probability $\textit{PD}(V,k)\frac{1}{\binom{|V|}{k}}$.
\end{proposition}

\begin{proof}
To observe it, we transform formula~\eqref{eq:semi} as follows:
\begin{equation}\label{Eq:centrality_prob}
\phi^{\SEMI}(V,\nu) = \sum_{1 \leq k \leq |V|} \textit{PD}(V,k)  \frac{1}{\binom{|V|-1}{k-1}} \sum_{\substack{S \subseteq V \\ |S|=k \\ v_i \in S}} (v(S) - v(S \setminus \{v_i\} )),
\end{equation}
\noindent where $\textit{PD}(V,k) = p_{k-1} \binom{|V|-1}{k-1}$.

The analysis of the above equation leads to the following. Each set $v_i \in S \subseteq (V\setminus \{v_i\})$ such that $|S|=k$ is chosen with probability $\textit{PD}(V,k)$, and each group of equal size is chosen with equal probability, i.e., $\textit{PD}(V,k)\cdot\frac{1}{\binom{|V-1|}{k-1}}$. %
\end{proof}

In our motivating scenario of measuring systemic importance of financial institutions, with the Semivalue-based betweenness centrality we are able to assume an arbitrary probability distribution on the size of the bankrupting institution.

Also the Semivalue-based between centrality can be normalized to the interval $[0,1]$ in the same manner as Shapley value-besed betweenness centrality. To this end, one needs to compute:
\[
\frac{c_{S}(v_i)}{2\max_{ v_i \in V }c_b(v_i)}+\frac{1}{2}
\]

In the following subsections, we provide formal analysis of the Shapley value-based betweenness centrality and the more-general Semivalue-based betweenness centrality.

\subsection{A Look at Marginal Contribution for the Shapley Value}  \label{subsection:MCforShapley}

\noindent The centrality of a node is measured by computing a weighted average of its marginal contribution to the group-betweenness centrality of all subsets of nodes in the graphs. In this subsection, we show how to compute this weighted average efficiently.

To this end, we will use Definition~\ref{def:sv:expected}. More specifically, given a graph $G$ and some vertex $v \in V(G)$, we would like to compute the expected marginal contribution of this vertex to the set of vertices $P_{\pi}(v)$ that precede $v$ in a random permutation of all vertices of the graph. We focus in our analysis on three exhaustive cases, where the marginal contribution is positive, negative, or neutral, respectively.

First, let us consider positive contributions, and let us focus on some particular shortest path $p$ which contains vertex $v$. Recall that we denote by $\sigma_{st}$ the number of shortest paths between vertices $s$ and $t$. Furthermore, we denote by $\sigma_{st}(v)$ the number of shortest paths between vertices $s$ and $t$ that each pass through vertex $v$, where $ t \neq v \neq s$ (recall that every such path is said to be \emph{controlled} by $v$). Every path in $\sigma_{st}(v)$ has a positive contribution to the coalition $P_{\pi}(v)$ via vertex $v$ if and only if the path is not controlled by (i.e., does not pass through) any vertex from the set $P_{\pi}(v)$. In this case, the positive contribution equals $\frac{1}{\sigma_{st}}$. The necessary and sufficient condition for this to happen can be expressed by $\Psi(p) \cap P_{\pi}(v) = \emptyset $, where $\Psi(p)$ is the set of all vertices lying on the path $p$, including endpoints.

Now, let us introduce a Bernoulli random variable $B^+_{v,p}$ which indicates whether vertex $v$ makes a positive contribution to the set $P_{\pi}(v)$ through shortest path $p$. Thus, we have:
\begin{equation}
	 \mathbb{E}[\frac{1}{\sigma_{st}}B^+_{v,p}] = \frac{1}{\sigma_{st}}P[\Psi(p) \cap P_{\pi}(v) = \emptyset] \nonumber,
\end{equation}
\noindent where $P[\cdot]$ denotes probability, and $\mathbb{E}[\cdot]$ denotes expected value. In other words, we need to know the probability of having $v$ precede all vertices from $\Psi(p) \setminus \{v\}$ in a random permutation of all vertices in the graph.

\begin{theorem}\label{perm_the}
Let $K$ be a set of elements such that $|K|=k$. Furthermore, let $L$ and $R$ be two disjoint subsets of $K$, and let $|L|=l$ and $|R|=r$. Now, given some element $x \in K$, where $x \notin L \cup R$, and given a random permutation $\pi \in \Pi(K)$, the probability of having every element in $L$ before $x$, and every element in $R$ after $x$, is:
\begin{equation}\label{eqn:theorem:K}
\textstyle P[\forall_{e \in L}{\pi(e) < \pi(x) } \land \forall_{e \in R}{\pi(e) > \pi(x)}] = \frac{1}{(l+1)\binom{l+r+1}{r}} \nonumber
\end{equation}

\end{theorem}

\begin{proof}
Let us first count the number of ways in which a permutation $\pi \in \Pi(K)$ can be constructed such that the following condition holds:
$\forall_{e \in L}{\pi(e) < \pi(x) } \land \forall_{e \in R}{\pi(e) > \pi(x)}$. Specifically:\footnote{\footnotesize We thank the anonymous reviewer for suggesting this formulation of the proof.}

	\begin{itemize} \itemsep0.5em
	\item Let us choose $ l+r+1 $ positions in the permutation. There are $ \binom{k}{l+r+1}$ such choices.
	\item Now, in the first $l$ chosen positions, place all the elements of $L$. Directly after those, place the element $x$. Finally, in the remaining $r$ chosen positions, place all the elements of $R$. The number of such line-ups is: $l!r!$.
	\item As for the remaining elements (i.e., the elements of $K\setminus (L\cup \{x\}\cup R)$), they can be arranged in the permutation in $ (k - (l+r+1))! $ different ways.
	\end{itemize}

Thus, the number of permutations satisfying our condition is:
\begin{equation}
\textstyle\binom{k}{l+r+1}l!r!(k - (l+r+1))! =  \frac{k!}{(l+1)\binom{l+r+1}{r}}\ . \nonumber
\end{equation}

Finally, since we have a total of $k!$ permutations, the probability that one of them satisfies our condition is given by equation~\eqref{eqn:theorem:K}.
\end{proof}

Based on Theorem \ref{perm_the}, we can obtain the probability of an event in which the vertex $v$ lies on the path $p$ and precedes all the other vertices from
this path in a random permutation of all vertices in the graph $G$. Now, by setting $K=V(G)$, $L = \emptyset$ and $R = \Psi(p) \setminus \{v\}$,
we obtain the desired probability: $\frac{1}{|\Psi(p)|}$. Thus:
\begin{equation}\label{bern_pos}
	 \mathbb{E}[\frac{1}{\sigma_{st}}B^+_{v,p}] = \frac{1}{\sigma_{st}|\Psi(p)|}.
\end{equation}

Having discussed the first case, where marginal contributions are positive, let us now discuss the second case, where the marginal contribution of vertex $v$ to set $P_{\pi}(v)$ is negative. This happens when path $p$ ends or starts with $v$ (w.l.o.g. we consider $v$ as the end point). Specifically, if coalition $P_{\pi}(v)$ already controls path $p$ (which includes vertex $v$), then not only is there no value added from $v$ becoming a member of this coalition, but there is a negative effect of this move. In particular, the group betweenness centrality assumes that a set of vertices $S$ controls only those paths with both ends not belonging to $S$. Therefore, when $v$ becomes a member of coalition $P_{\pi}(v)$, its negative contribution through path $p$ equals $-{\frac{1}{\sigma_{sv}}}$, where $\sigma_{sv}$ denotes the number of paths that start with $s$ and end with $v$.

Next, we analyse the probability of such a negative contribution of vertex $v$. Observe that if the marginal contribution is negative, then $p$ must end with $v$. However, the opposite does not necessarily hold, i.e., if $p$ ends with $v$ then the marginal contribution of $v$ is not necessarily negative; it could also be neutral. Based on this, in order to compute the probability of a negative contribution, we will start by focusing  on the complementary event where $v$ makes a neutral contribution to the set $P_{\pi}(v)$, as this simplifies the computational analysis. This complementary event happens if and only if either vertex $s$---the starting point of path $p$---belongs to set $P_{\pi}(v)$, or the path $p$ is not controlled by any of the vertices in $P_{\pi}(v)$. Formally:
$$
s \in P_{\pi}(v) \vee (\Psi(p) \cap P_{\pi}(v)) = \emptyset.
$$
Based on this, we can introduce a Bernoulli random variable $B^-_{v,p}$ indicating whether vertex $v$ makes a negative contribution through path $p$ to set $P_{\pi}(v)$ as follows:
\begin{equation}
	 \mathbb{E}[-{\frac{1}{\sigma_{sv}}}B^-_{v,p}] = -{\frac{1}{\sigma_{sv}}}(1-P[ s \in P_{\pi}(v) \vee (\Psi(p) \cap P_{\pi}(v)) = \emptyset ]) \nonumber.
\end{equation}

Again, one can show with the combinatorial arguments presented in Theorem~\ref{perm_the} that this probability is  $P[ \Psi(p) \cap P_{\pi}(v) = \emptyset] = \frac{1}{|\Psi(p)|}$. Furthermore, due to symmetry, we have:  $P[ s \in P_{\pi}(v)] = \frac{1}{2}$. Finally, from the disjointness of the aforementioned complementary events, we have:
\begin{equation}\label{bern_neg}
	 \mathbb{E}[-{\frac{1}{\sigma_{sv}}}B^-_{v,p}] = \frac{2-|\Psi(p)|}{2\sigma_{sv}|\Psi(p)|}.
\end{equation}

Recall that our analysis was divided into three exhaustive cases, where the marginal contribution is positive, negative, or neutral. Having dealt with the first two cases, it remains to deal with the case of neutral marginal contribution. Since this case does not affect the expected marginal contribution of $v$ to $P_{\pi}(v)$, we simply disregard it during the computation.

Next, let us compute the expected marginal contribution, by aggregating the cases of positive and negative contributions. To this end, let $\partial_{st}$ be the set of all shortest paths from $s$ to $t$, and analogously let $ \partial_{st}(v) \subseteq \partial_{st}$ be the set of shortest paths from $s$ to $t$ that pass through vertex $v$.\footnote{\footnotesize Note that $\sigma_{st} = |\partial_{st}|$ and $\sigma_{st}(v) = |\partial_{st}(v)|$.} Now, using the expected value of the Bernoulli random variables \eqref{bern_pos} and \eqref{bern_neg} we are able to compute the Shapley value of vertex $v$, which is the expected marginal contribution of $v$ to $P_{\pi}(v)$, as:

\begin{align}\label{shap_eq}
\hspace*{-0.3cm} \phi^{\SV}_v(V(G),\nu) &= \hspace{-0.25cm} \sum_{s \neq v \neq t} \sum_{ p \in \partial_{st}(v)}  \hspace{-0.3cm}\mathbb{E}[\frac{1}{\sigma_{st}}B^+_{v,p}]
+ \sum_{s \neq v} \sum_{ p \in \partial_{sv}}  \hspace{-0.15cm} \mathbb{E}[-\frac{1}{\sigma_{st}}B^-_{v,p}] \nonumber \\
&= \hspace{-0.25cm} \sum_{s \neq v \neq t} \sum_{ p \in \partial_{st}(v)}  \hspace{-0.15cm}\frac{1}{\sigma_{st}|\Psi(p)|}
+ \sum_{s \neq v} \sum_{ p \in \partial_{sv}} \hspace{-0.1cm}\frac{2-|\Psi(p)|}{2\sigma_{sv}|\Psi(p)|}.
\end{align}

Interestingly, the above equation provides insight into the Shapley value-based betweenness centrality, which is a combination of two factors. Firstly, the left sum resembles the standard betweenness centrality, but scaled by the number of vertices that belong to each shortest  path. Secondly, the right sum resembles the closeness centrality,\footnote{\footnotesize Recall that closeness centrality assigns to each vertex $v$ the value of $\sum_u\frac{1}{d(u,v)}$.} but with distances measured as the number of vertices on the shortest paths.

\subsection{A Look at Marginal Contribution for the  Semivalue}\label{Subsection:ProbabilisticSV}

\noindent We start our analysis by rewriting equation~\eqref{Eq:centrality_prob} as:
\begin{equation}\label{Eq:centrality_prob2}
\phi^{\SEMI}_v(V,\nu) = \sum_{1 \leq k \leq |V|} \textit{PD}(V,k) \mathbb{E}_{k-1,v}[\cdot] ,
\end{equation}

\noindent where $ \mathbb{E}_{k-1,v}[\cdot] $ is the expected marginal contribution of vertex $v$ to a random coalition of size $k-1$. So, analogously to the previous subsection, we need to analytically compute the probability that the node $v$ contributes to a coalition through a shortest path $p$. Again, here we only focus on the case of positive contribution and the case of negative contribution, since the neutral case has no impact on the expected marginal contribution. We will denote by $S_{k-1} \subseteq V$ a random coalition of size $k-1$ uniformly drawn from the set $\{S\subseteq V\setminus \{v\} : |S|=k-1\}$.

Firstly, we consider positive contributions. In what follows, let us focus on some particular shortest path $p$ which contains vertex $v$. This path has a positive contribution to the coalition $S_{k-1}$ through $v$ if and only if it is not controlled by (i.e., does not pass through) any vertex from set $S_{k-1}$. In this case, the positive contribution equals $\frac{1}{\sigma_{st}}$. The necessary and sufficient condition for this to happen can be expressed by $ \Psi(p) \cap S_{k-1} = \emptyset $, where $\Psi(p)$ is the set of all vertices lying on the path $p$ including endpoints.

Now, let us introduce a Bernoulli random variable $\textit{BS}^+_{k-1,v,p}$ which indicates whether vertex $v$ makes a positive contribution through path $p$ to set $S_{k-1}$. Note that, if $|V|-\Psi(p)< k-1$ then the probability of this event equals $0$. Otherwise, we have:

\begin{align}\label{Eq:centrality_prob_plus}
\mathbb{E}[\frac{1}{\sigma_{st}}\textit{BS}^+_{k-1,v,p}] = &  \frac{1}{\sigma_{st}}P[\Psi(p) \cap S_{k-1} = \emptyset] =  \frac{1}{\sigma_{st}} \frac{\binom{|V|-\Psi(p)}{k-1}}{\binom{|V|-1}{k-1}}  \nonumber \\
=&  \frac{(|V|-\Psi(p))!(|V|-k)!}{\sigma_{st}(|V|-\Psi(p) -k+1)!(|V|-1)!},
\end{align}

\noindent where $P[\cdot]$ denotes probability, and $\mathbb{E}[\cdot]$ denotes expected value. In other words, we need to know the probability of having random set $S_{k-1}$ containing $v$ and not containing any vertex from $\Psi(p) \setminus \{v\}$.

Having analysed the case of positive contribution, we now move on to the case of negative contribution. This happens when the path $p$ ends with $v$ and is controlled by the set $S_{k-1}$, i.e., $p$ passes through at least one of the nodes in $S_{k-1}$, where $v\notin S_{k-1}$. As in the previous subsection, we will simplify the analysis by focusing on the complementary event in which the marginal contribution is neutral. Observe that $v$ makes a neutral contribution to set $S_{k-1}$ through a path $p \in \partial_{sv}$ if and only if either vertex $s$---the end point of path $p$---belongs to the set $S_{k-1}$, or the path is not controlled by any of the vertices in $S_{k-1}$. Formally:
$$
s \in S_{k-1} \vee (\Psi(p) \cap S_{k-1}) = \emptyset.
$$
Based on this, we can introduce a Bernoulli random variable $\textit{BS}^-_{k-1,v,p}$ indicating whether that vertex $v$ makes a negative contribution through
path $p$ to set $S_{k-1}$. Now if $|V|-\Psi(p) \geq k-1$, we obtain the following expression:

\begin{align}\label{Eq:centrality_prob_minus}
\mathbb{E}[-{\frac{1}{\sigma_{sv}}}\textit{BS}^-_{k-1,v,p}] = &  -{\frac{1}{\sigma_{sv}}}(1-P[ s \in S_{k-1} \vee (\Psi(p) \cap S_{k-1}) = \emptyset ]) \nonumber \\
=&   - {\frac{1}{\sigma_{sv}}}\Big(1 \hspace{-0.1cm}-\hspace{-0.1cm} \big(\frac{k-1}{|V|-1} \hspace{-0.1cm} +\hspace{-0.1cm}  \frac{(|V|\hspace{-0.05cm}-\hspace{-0.05cm}\Psi(p))!(|V|-k)!}{(|V|\hspace{-0.05cm}-\hspace{-0.05cm}\Psi(p)\hspace{-0.05cm} -\hspace{-0.05cm}k\hspace{-0.05cm}+\hspace{-0.05cm}1)!(|V|\hspace{-0.05cm}-\hspace{-0.05cm}1)!}\big)\Big) \nonumber \\
=&  \frac{k-1}{{\sigma_{sv}}(|V|\hspace{-0.05cm}-\hspace{-0.05cm}1)} \hspace{-0.05cm} +\hspace{-0.05cm}  \frac{(|V|\hspace{-0.05cm}-\hspace{-0.05cm}\Psi(p))!(|V|-k)!}{\sigma_{sv}(|V|\hspace{-0.05cm}-\hspace{-0.05cm}\Psi(p)\hspace{-0.05cm} -\hspace{-0.05cm}k\hspace{-0.05cm}+\hspace{-0.05cm}1)!(|V|\hspace{-0.05cm}-\hspace{-0.05cm}1)!})\hspace{-0.05cm}-\hspace{-0.05cm}{\frac{1}{\sigma_{sv}}}.
\end{align}

\noindent On the other hand, if $|V|-\Psi(p) < k-1$, we obtain the following expression:
\begin{align}\label{Eq:centrality_prob_minus2}
\mathbb{E}[-{\frac{1}{\sigma_{sv}}}\textit{BS}^-_{k-1,v,p}] = \frac{k-1}{{\sigma_{sv}}(|V|\hspace{-0.05cm}-\hspace{-0.05cm}1)} \hspace{-0.05cm}-\hspace{-0.05cm}{\frac{1}{\sigma_{sv}}}.
\end{align}

Now, let us compute the expected marginal contribution, by aggregating the cases of positive and negative contributions. To this end, recall that $ \partial_{st}$ is the set of all shortest paths from $s$ to $t$, and, analogously, $ \partial_{st}(v) \subseteq \partial_{st}$ is the set of shortest paths from $s$ to $t$ passing through vertex $v$. Now, using the expected value of Bernoulli random variables \eqref{Eq:centrality_prob_plus}, \eqref{Eq:centrality_prob_minus} and \eqref{Eq:centrality_prob_minus2} we are able to compute the Semivalue of vertex $v$ by using equation~\eqref{Eq:centrality_prob2}:

\begin{align}\label{Eq:PSV_formula}
\hspace*{-0.3cm} \phi^{\SEMI}_v & \hspace{-0.12cm}= \hspace{-0.65cm}\sum_{0 \leq k \leq |V|-1} \hspace{-0.45cm}\textit{PD}(V,\hspace{-0.05cm}k\hspace{-0.05cm}+\hspace{-0.05cm}1) \Big( \hspace{-0.25cm} \sum_{s \neq v \neq t} \sum_{ p \in \partial_{st}(v)}  \hspace{-0.3cm}\mathbb{E}[\frac{1}{\sigma_{st}}\textit{BS}^+_{k,v,p}] \hspace{-0.15cm}+\hspace{-0.15cm} \sum_{s \neq v} \sum_{ p \in \partial_{sv}}  \hspace{-0.15cm} \mathbb{E}[-\frac{1}{\sigma_{st}}\textit{BS}^-_{k,v,p}] \Big) \nonumber \\
& \hspace{-0.15cm}=  \hspace{-0.35cm}\sum_{1 \leq k \leq |V|} \hspace{-0.25cm}\textit{PD}(V,k) \hspace{-0.05cm} \Bigg(   \sum_{s \neq v \neq t}\hspace{-0.25cm} \sum_{ \substack{p \in \partial_{st}(v) \\ |V|-\Psi(p) \geq k-1 } }  \hspace{-0.5cm}  \frac{(|V|\hspace{-0.055cm} - \hspace{-0.05cm}  \Psi(p))!(|V|\hspace{-0.05cm} -\hspace{-0.05cm} k)!}{\sigma_{st}(|V|\hspace{-0.05cm} -\hspace{-0.05cm} \Psi(p)\hspace{-0.05cm}  -\hspace{-0.05cm} k\hspace{-0.05cm} +\hspace{-0.05cm} 1)!(|V|\hspace{-0.05cm} -\hspace{-0.05cm} 1)!} \nonumber \\
&+ \hspace{-0.15cm} \sum_{s \neq v} \Big( \hspace{0.05cm}  \frac{k-|V|}{|V|\hspace{-0.05cm}-\hspace{-0.05cm}1} \hspace{-0.05cm}+\hspace{-0.6cm} \sum_{  \substack{p \in \partial_{sv} \\ |V|-\Psi(p) \geq k-1 }} \hspace{-0.6cm}  \frac{(|V|\hspace{-0.05cm}-\hspace{-0.05cm}\Psi(p))!(|V|-k)!}{\sigma_{sv}(|V|\hspace{-0.05cm}-\hspace{-0.05cm}\Psi(p)\hspace{-0.05cm} -\hspace{-0.05cm}k\hspace{-0.05cm}+\hspace{-0.05cm}1)!(|V|\hspace{-0.05cm}-\hspace{-0.05cm}1)!}\hspace{0.05cm} \Big)  \Bigg).
\end{align}

While the closed-form formulas \eqref{Eq:PSV_formula} and \eqref{shap_eq} allow us to efficiently compute in polynomial time the Shapley value-based and Semivalue-based betweenness centralities, respectively, in the following section we introduce algorithms that perform this computation even more efficiently.

\section{Algorithms to Compute the Shapley Value and Semivalue based Betweenness Centralities}\label{chap:bet:algorithms}

\noindent Generally speaking, given a graph $G$, computing the Shapley value takes $O(2^{|V(G)|})$ time (to consider all subsets of $V(G)$). To circumvent this major obstacle, we proposed formulas \eqref{shap_eq} and \eqref{Eq:PSV_formula} to compute in polynomial time the Shapley Value and the Semivalue, respectively. To speed up the computation even further, we propose in this section four algorithms:
\begin{enumerate} \itemsep0.5em
\item Algorithm~\textbf{SVB} computes the \emph{Shapley value}-based betweenness centrality given an \emph{unweighted} graph;
\item Algorithm~\textbf{SB} computes the \emph{Semivalue}-based betweenness centrality given an \emph{unweighted} graph;
\item Algorithm~\textbf{WSVB}  computes the \emph{Shapley value}-based betweenness centrality given a \emph{weighted} graph;
\item Algorithm~\textbf{WSB} computes the \emph{Semivalue}-based betweenness centrality given a \emph{weighted} graph.
\end{enumerate}

We also show in this section how the above algorithms can be easily adapted to work on directed graphs. Interestingly, we show that Algorithm~\textbf{SVB} has the same complexity as the best known algorithm to compute the standard betweenness centrality (due to Brandes, 2001). Our algorithms are based on a general framework proposed in this chapter, which generalises Brandes' framework \citep{Brandes:2001}.

This section is structured as follows. Section~\ref{framework:unweighted} introduces our generalisation of Brandes' framework for unweighted graphs. Building upon it, Sections \ref{framework:unweighted:SV} and \ref{framework:unweighted:S} introduce Algorithm~\textbf{SVB} and Algorithm~\textbf{SB}, respectively. After that, in Section~\ref{framework:weighted}, we introduce our generalisation of Brandes' framework for weighted graphs. Finally, Sections \ref{framework:weighted:SV} and \ref{framework:weighted:S} introduce Algorithm~\textbf{WSVB} and Algorithm~\textbf{WSVB}, respectively.
\subsection{The Framework for Unweigthed Graphs}\label{framework:unweighted}

\noindent In order to compute the standard betweenness centrality for all nodes, a na\"{\i}ve algorithm would first compute the number of shortest paths between all pairs of nodes, and then for each node $v$ sum up all pair-dependencies, which are defined as $\frac{\sigma_{st}(v)}{\sigma_{st}}$. This process takes $O(|V|^3)$ time. \cite{Brandes:2001} proposed an algorithm to improve this complexity by using some recursive relation. This algorithm runs in $O(|V|\cdot|E|)$ time, and requires $O(|V| + |E|)$ space.

Next, we propose a framework that generalises Brandes' algorithm. We start by defining pair-dependency as:
\begin{equation}\label{eq:pair_dep}
\delta_{s,t}(v) = \frac{\sigma_{st}(v)}{\sigma_{st}}f_{\delta}(d(s,t)),
\end{equation}
\noindent which is the positive contribution that all shortest paths between $s$ and $t$ make to the assessment of vertex $v$. The value of the function $f_{\delta}$ depends solely on the distance between $s$ and $t$. For instance, if we set $f_{\delta} = \frac{1}{d(s,t)}$ we obtain the distance-scaled betweenness centrality introduced by \cite{Borgatti:Everett:2006}.

Next, we define one-side dependency as:
\begin{equation}\label{eq:onde_side_dep}
\delta_{s,\cdot}(v) = \sum_{t \in V}{\delta_{s,t}(v)} ,
\end{equation}

\noindent which is the positive contribution that all shortest paths starting in vertex $s$ make to the assessment of vertex $v$. Now, building upon the recursive equation proposed by \cite{Brandes:2001}, we obtain the following:
\begin{equation}\label{eq:recursion}
\delta_{s,\cdot}(v) = \sum_{ \substack{ w:~(v,w) \in E \\ d(s,w) = d(s,v) + 1} } \frac{\sigma_{sv}}{\sigma_{sw}} \bigg(f_{\delta}(d(s,w)) + \delta_{s,\cdot}(w)\bigg).
\end{equation}

Using the above equation, we can compute what we call the \emph{parameterised  betweenness centrality} $c_{fg}$ for a vertex $v$ by iterating over all other vertices and summing their contributions as follows:
\begin{equation}\label{eq:general_cen}
c_{fg}(v) = \sum_{s \neq v}\bigg( \delta_{s,\cdot}(v) + g_{\delta}(d(s,v)) \bigg),
\end{equation}

\noindent where the value of the function $g_{\delta}$ depends solely on the distance between $s$ and~$v$.

Now, we are ready to introduce our generalised framework, \textbf{PBC} (see Algorithm~\ref{algo:pbc}) for computing the \emph{parametrised betweenness centrality} in unweighted graphs. More specifically, this framework modifies Brandes' approach and computes the betweenness centrality parametrised by two functions: $f_{\delta}$ and $g_{\delta}$. All lines that differ from the original Brandes' algorithm were highlighted.

%%%%%%%%%%%%%%%%%%%%%%%%%%%%%%%%%%%%%%%%%%%%%%%%%%%%%%%%%%%%%%%%%%%%%%%%%%%%%%%%%%%%%%%%%%%%%%%%%%%%
\RestyleAlgo{ruled} %The old version of this is: \restylealgo

\begin{algorithm}[t!]
\SetAlgoVlined %The old version of this is: \SetVline
\LinesNumbered %The old version of this is: \linesnumbered
\caption{\textbf{PBC}---a general framework to compute the \emph{parametrised betweenness centrality} in \emph{unweighted} graphs}
\label{algo:pbc}
\KwIn {Graph $G = (V,E)$, functions: $f_{\delta}$ and $g_{\delta}$}
\KwData {queue $\mathcal{Q}$, stack $\mathcal{S}$ for each $v \in V$ and some source $s$ \\
$d(s,v):$ distance from $v$ to the source $s$ \\
$Pred_{s}(v):$ list of predecessors of $v$ on the shortest paths from source $s$ \\
$\sigma_{sv}:$ the number of shortest paths from $s$ to $v$ \\
$\delta_{s,\cdot}(v):$ the one-side dependency of $s$ on $v$\\ }
\KwOut {$c_{fg}(v)$ parametrised betweenness centrality for each vertex $v\in V$}

\lForEach{$v \in V$} {
	$c_{fg}(v) \gets 0;$
}

\ForEach{$s \in V$} {
	\ForEach{$v \in V$} {
		$Pred_{s}(v) \gets empty~list;~d(s,v) \gets \infty;~\sigma_{sv} \gets 0;$
	}
	$d(s,s) \gets 1;~\sigma_{ss} \gets 1;~$
	$enqueue~s \rightarrow \mathcal{Q};$ \\
	\While{$\mathcal{Q}$ is not empty }{ \nllabel{sch:first_start}
		$dequeue~v \gets \mathcal{Q};~push~v \rightarrow \mathcal{S};$ \\
		\ForEach{$w$ such that $(v,w)\in E $} {
			\If{$d(s,w) = \infty$}{
				$d(s,w) \gets d(s,v)+1$; $enqueue~ w \rightarrow \mathcal{Q}$;
			}
			\If{$d(s,w) = d(s,v) + 1$} {
				$\sigma_{sw} \gets \sigma_{sw} + \sigma_{sv}$; $append~v \rightarrow Pred_{s}(w);$ \\   \nllabel{sch:first_end}
			}
		} 	
	}
	\ForEach{$v \in V$}{$\delta_{s,\cdot}(v) \gets 0;$}
	\While{$\mathcal{S}$ is not empty} {
		$pop~ w \gets \mathcal{S};$\\
		\ForEach{$v \in Pred_{s}(w)$}{
			$\highlight{\delta_{s,\cdot}(v) \gets \delta_{s,\cdot}(v) + \frac{\sigma_{sv}}{\sigma_{sw}}(f_{\delta}(d(s,w))+\delta_{s,\cdot}(w));}$ \nllabel{sch:rec_impl}
		}
		\If{$w \neq s$}{
			$\highlight{c_{fg}(w) \gets c_{fg}(w) + \delta_{s,\cdot}(w) + 2g_{\delta}(d(s,w));}$ \nllabel{sch:rest_impl}
		}
	}
}
\ForEach{$v \in V$} { \nllabel{sch:del_cont}
	$\highlight{c_{fg}(v) = \frac{c_{fg}(v)}{2};}$
}

\end{algorithm}
%%%%%%%%%%%%%%%%%%%%%%%%%%%%%%%%%%%%%%%%%%%%%%%%%%%%%%%%%%%%%%%%%%%%%%%%%%%%%%%%%%%%%%%%%%%%%%%%%%%%

Firstly, in lines \ref{sch:first_start} - \ref{sch:first_end}, the algorithm calculates both the distance and the number of shortest paths from a source $s$ to each vertex. These lines are exactly the same as in the original algorithm introduced by Brandes. While executing these lines, for each vertex $v$, all directly preceding vertices occurring on shortest paths from $s$ to $v$
are stored in memory. This process uses Breadth-First Search \citep{Cormen:2001} which takes $O(|V|)$ time and $O(|V|+|E|)$ space. In the second step (lines \ref{sch:rec_impl} and \ref{sch:rest_impl}), the algorithm uses formula \eqref{eq:general_cen} to calculate the contribution of the source $s$ to the value of our betweenness centrality for each vertex that is reachable from the source. This step also takes $O(|V|)$ time and $O(|V|+|E|)$ space. Only lines \ref{sch:rec_impl} and \ref{sch:rest_impl} and additionally line \ref{sch:del_cont} differ from the original Brandes algorithm.

As visible in formula \eqref{eq:general_cen}, in an undirected graph, each path is considered twice. Thus, at the end of the algorithm, in line \ref{sch:del_cont}, we halve the accumulated result. In line \ref{sch:rest_impl}, we multiply the influence of the function $g_{\delta}$ by two, because the influence of each source even in undirected graphs is considered only once. Finally, we note that it is very easy to adopt Algorithm \ref{algo:pbc} to directed graphs. To this end, we remove the loop from line \ref{sch:del_cont} and halve the contribution of the vertex $s$ from line \ref{sch:rest_impl}, which now should look as follows: \\

$ \small{ \textbf{\ref{sch:rest_impl}}}: ~~ c_{fg}(w) \gets c_{fg}(w) + \delta_{s,\cdot}(w) + g_{\delta}(d(s,w));$ \\

So, for directed and undirected graphs, the algorithm runs in $O(|V|\cdot|E|)$ time, and requires $O(|V| + |E|)$ space.

\subsection{The Algorithm for the Shapley Value-based Betweenness Centrality for Unweighted Graph}\label{framework:unweighted:SV}

\noindent In this subsection we will construct an efficient algorithm for computing the Shapley value-based betweenness centrality for unweighted graphs. Specifically, in such graphs, the number of vertices on the shortest path between $s$ and $t$ is equal to the distance between $s$ and $t$, denoted as $d(s,t)$.\footnote{\footnotesize For notational convenience, we assume the distance between two vertices to be the number of vertices on the shortest path between them (not the number of edges), e.g., $d(s,s)=1$.} In other words, we have: $|\Psi(p)| = d(s,t)$. Based on this, it is possible to simplify \eqref{shap_eq} as follows:
\hspace{-0.2cm}
\begin{align}\label{shap_eq_simpl}
\phi^{\SV}_v(V(G),\nu) &=  \hspace{-0.25cm}\sum_{s \neq v \neq t} \sum_{ p \in \partial_{st}(v)}  \hspace{-0.1cm}\frac{1}{\sigma_{st}d(s,t)}
+ \sum_{s \neq v} \sum_{ p \in \partial_{sv}}\hspace{-0.1cm}\frac{2-d(s,v)}{2\sigma_{sv}d(s,v)} \nonumber \\
&=  \hspace{-0.10cm}\sum_{s \neq v} \bigg(  \sum_{t \neq v}{ \frac{\sigma_{st}(v)}{\sigma_{st}d(s,t)}   } +
\frac{2-d(s,v)}{2d(s,v)} \bigg).
\end{align}

The above equation provides the same insight as equation~\eqref{shap_eq} but for specific unweighted graphs. By transforming the second element of the inner sum $\frac{2-d(s,v)}{2d(s,v)} = \frac{1}{d(s,v)} - \frac{1}{2} $ we find that, in unweighted graphs, the Shapley value using group betweenness centrality as a characteristic function is in fact the sum of \emph{the distanced scaled betweenness centrality} (introduced by Borgatti and Everett, 2006) and the \emph{closeness} centrality, shifted by half. Additionally, the above equation allows us to compute the Shapley value-based betweenness centrality in $O(|V|^3)$, but this result will be improved further.

Now, we adopt the framework presented in the previous subsection in order to accommodate equation~\eqref{shap_eq_simpl}. We simply need to set $f_{\delta} = \frac{1}{d(s,t)}$ and set $g_{\delta} = \frac{2-d(s,v)}{2d(s,v)} $. This way, we can use our general framework---\textbf{PBC}, see Algorithm~\ref{algo:pbc}---to compute the Shapley value-based betweenness centrality in $O(|V||E|)$. This method is presented in Algorithm~\ref{algorithms:SV}.

%%%%%%%%%%%%%%%%%%%%%%%%%%%%%%%%%%%%%%%%%%%%%%%%%%%%%%%%%%%%%%%%%%%%%%%%%%%%%%%%%%%%%%%%%%%%%%%%%%%%
\RestyleAlgo{ruled} %The old version of this is: \restylealgo

\begin{algorithm}[t!]
\SetAlgoVlined %The old version of this is: \SetVline
\LinesNumbered %The old version of this is: \linesnumbered
\caption{\textbf{SVB} SV-based betweenness centrality}
\label{algorithms:SV}
\KwIn {Graph $G = (V,E)$}
\KwOut {$\phi^{\SV}_v$ SV-based betweenness centrality for each vertex $v\in V$}

$f_{\delta}(d) \gets \frac{1}{d}$; \\
$g_{\delta}(d) \gets \frac{2-d}{2d}$; \\
$\textbf{PBC}(G,f_{\delta},g_{\delta})$ \\

\end{algorithm}
%%%%%%%%%%%%%%%%%%%%%%%%%%%%%%%%%%%%%%%%%%%%%%%%%%%%%%%%%%%%%%%%%%%%%%%%%%%%%%%%%%%%%%%%%%%%%%%%%%%%

\subsection{The Algorithm for the Semivalue-based Betweenness Centrality for Unweighted Graphs}\label{framework:unweighted:S}

\noindent We will follow the same reasoning as in the previous subsection. Specifically, since in unweighted graphs it holds that $|\Psi(p)| = d(s,t)$, we can transform equation~\eqref{Eq:PSV_formula} into:

\begin{align}\label{shap_eq_prob}
\hspace*{-0.3cm} \phi^{\SEMI}_v(V(G),\nu) & \hspace{-0.12cm}=  \hspace{-0.35cm}\sum_{1 \leq k \leq |V|} \hspace{-0.25cm}\textit{PD}(V,k) \hspace{-0.05cm} \Bigg(   \sum_{s \neq v \neq t}\hspace{-0.3cm} \sum_{ \substack{p \in \partial_{st}(v) \\ |V|-d(s,t) \geq k-1 } }  \hspace{-0.55cm}  \frac{(|V| d(s,t))!(|V|\hspace{-0.05cm} -\hspace{-0.05cm} k)!}{\sigma_{st}\hspace{-0.03cm}(|V|\hspace{-0.055cm} -\hspace{-0.055cm} d(s,t)\hspace{-0.055cm}  -\hspace{-0.055cm} k\hspace{-0.055cm} +\hspace{-0.055cm} 1)!(|V|\hspace{-0.055cm} -\hspace{-0.055cm} 1)!} \nonumber \\
&+ \hspace{-0.15cm} \sum_{s \neq v} \Big( \hspace{0.05cm}  \frac{k-|V|}{|V|\hspace{-0.05cm}-\hspace{-0.05cm}1} \hspace{-0.05cm}+\hspace{-0.6cm} \sum_{  \substack{p \in \partial_{sv} \\ |V|-d(s,v) \geq k-1 }} \hspace{-0.6cm}  \frac{(|V|\hspace{-0.05cm}-\hspace{-0.05cm}d(s,v))!(|V|-k)!}{\sigma_{sv}(|V|\hspace{-0.05cm}-\hspace{-0.05cm}d(s,v)\hspace{-0.05cm} -\hspace{-0.05cm}k\hspace{-0.05cm}+\hspace{-0.05cm}1)!(|V|\hspace{-0.05cm}-\hspace{-0.05cm}1)!}\hspace{0.05cm} \Big)  \Bigg) \nonumber \\
&=  \hspace{-0.4cm}\sum_{1 \leq k \leq |V|} \hspace{-0.25cm}\textit{PD}(V,k) \Bigg(   \sum_{ \substack{s \neq v \neq t \\ |V|-d(s,t) \geq k-1 }} \hspace{-0.4cm} \frac{\sigma_{st}(v)(|V|\hspace{-0.05cm} -\hspace{-0.05cm}  d(s,t) )!(|V|-k)!}{\sigma_{st}(|V|\hspace{-0.05cm} -\hspace{-0.05cm} d(s,t) \hspace{-0.05cm}  -\hspace{-0.05cm} k\hspace{-0.05cm} +\hspace{-0.05cm} 1)!(|V|-1)!} \nonumber \\
& + \hspace{-0.6cm} \sum_{ \substack{s \neq v \\ |V|-d(s,v)  \geq k-1 } }  \hspace{-0.6cm} \Big( \frac{(|V|\hspace{-0.05cm}-\hspace{-0.05cm}d(s,v) )!(|V|-k)!}{(|V|\hspace{-0.05cm}-\hspace{-0.05cm}d(s,v) \hspace{-0.05cm} -\hspace{-0.05cm}k\hspace{-0.05cm}+\hspace{-0.05cm}1)!(|V|\hspace{-0.05cm}-\hspace{-0.05cm}1)!} \Big) + k-|V|\hspace{0.05cm} \Bigg).
\end{align}

Our framework can be easily adopted to deal with the Semivalue; all we need to do is to set:
\renewcommand{\arraystretch}{1.5}
\begin{equation}\label{eq:pos}	
f_{\delta} = \left\{
 \begin{array}{l l}
   \frac{(|V|-d )!(|V|-k)!}{(|V|-d -k+1)!(|V|-1)!} &  \text{if $|V|-d(s,t) \geq k-1$}\\
   0 &  \text{if $|V|-d(s,t) < k-1$}
 \end{array} \right.	 \nonumber
\end{equation}
\renewcommand{\arraystretch}{1}

\noindent and set $g_{\delta} = f_{\delta}  + \frac{k-|V|}{|V|-1}$. This way, we can use equation~\eqref{shap_eq_prob} along with our general framework---\textbf{PBC}, see Algorithm~\ref{algo:pbc}---to compute the Semivalue-based betweenness centrality for unweighted graphs in $O(|V|^2|E|)$ time. The pseudo code is provided in Algorithm~\ref{algorithms:PSV}.

%%%%%%%%%%%%%%%%%%%%%%%%%%%%%%%%%%%%%%%%%%%%%%%%%%%%%%%%%%%%%%%%%%%%%%%%%%%%%%%%%%%%%%%%%%%%%%%%%%%%
\RestyleAlgo{ruled} %The old version of this is: \restylealgo

\begin{algorithm}[t!]
\SetAlgoVlined %The old version of this is: \SetVline
\LinesNumbered %The old version of this is: \linesnumbered
\caption{\textbf{SB}  Semivalue-based betweenness centrality }\label{algorithms:PSV}
\KwIn {Graph $G = (V,E)$}
\KwOut {$\phi^{\SEMI}_v$ Semivalue-based betweenness centrality for each vertex $v\in V$}
$f_{\delta}(d) \gets $ \lIf{$|V|-d \ge k-1$} { $\frac{(|V|-d )!(|V|-k)!}{(|V|-d -k+1)!(|V|-1)!}$} \lElse {0}
$g_{\delta}(d) \gets f_{\delta}(d) + \frac{k-|V|}{|V|-1}$; \\
\For{$k \gets 1$ \KwTo $|V|$} {
	$\textbf{PBC}(G,f_{\delta},g_{\delta})$ ; \\
	\ForEach{$v \in V$} {
		$\phi^{\SEMI}_v \gets \textit{PD}(V,k)c_{fg}(v);$
	}
}
\end{algorithm}
%%%%%%%%%%%%%%%%%%%%%%%%%%%%%%%%%%%%%%%%%%%%%%%%%%%%%%%%%%%%%%%%%%%%%%%%%%%%%%%%%%%%%%%%%%%%%%%%%%%%

\subsection{The Framework for Weighted Graphs}\label{framework:weighted}

\noindent While the focus of the previous subsections was on unweighted graphs, in this subsection we show how to generalise our framework to deal with weighted graphs. In particular, we consider one of the most popular semantics of weighted graphs, where the weight $\lambda(v,u)$ of the edge between $v$ and $u$ is interpreted as the distance between $v$ and $u$ . Thus, unlike the case with unweighted graphs, where $|\Psi(p)| = d(s,t)$, in the case  of weighted graphs this equality does not necessarily hold.

To introduce our generalised framework, we need additional notation. Let $T_{st}[i] : i \in\{1,\dots, |V|\}$ be the number of shortest paths between $s$ and $t$ that contain exactly $i$ vertices. The array $T_{st}$ uniquely determines the polynomial $W_{st}$ with terms $T_{st}[i]x^i $. We define the following operations on $T_{st}$:

\begin{itemize}

\item \textbf{Shifting}: $T_{st}^{\rightarrow}$ and $T_{st}^{\leftarrow}$ increase and decrease the indices of all values of the array by one, respectively. This takes $O(|V|)$ time.
\item \textbf{Adding}: $ T_{sv} \oplus T_{su} $ is the operation of adding two polynomials $W_{sv}$ and $W_{su}$. It takes $O(|V|)$ time. We will denote by $\bigoplus$ the sum of a series of polynomials.
\item \textbf{Multiplying}: $ T_{sv} \otimes T_{vt} $ is the operation of multiplying two polynomials $W_{sv}$ and $W_{vt}$. This takes $O(|V|\log{|V|})$ time using the polynomial multiplying algorithm from \citep{Cormen:2001}.
\item \textbf{Resetting}: $ T_{sv} \gets 0 $ is an operation that assigns $0$ to each cell in $T_{sv}$.
\end{itemize}

Next, we will show how the above operations allow us to tackle two algorithmic problems: (i) how to count all shortest paths and (ii) how to derive the recursive relation in order to compute one-side dependency. We start by considering the first of those problems, i.e., counting the shortest paths. Here, we will use the following relation:
\begin{equation}\label{weighted:counting}
	T_{sv} = \hspace{-0.6cm}\bigoplus_{u:~ d(s,u) + \\\lambda(u,v) = d(s,v)} \hspace{-0.6cm} {T_{su}^{\rightarrow} }.
\end{equation}

Using Dijkstra's algorithm \citep{Cormen:2001}, as well as equation~\eqref{weighted:counting}, we can compute $T_{st}$ for every pair of nodes $s$ and $t$. If vertex $u$ immediately precedes vertex $v$ on some shortest path from source $s$, all shortest paths stored in $T_{su}$ extended by vertex $v$ are part of the set of shortest paths stored in $T_{sv}$. This procedure takes $O(|V|^2|E| + |V|^2\log{|V|})$ time.

Secondly, in order to count all paths passing through a given node $v$, we will introduce the following relationship, where $T_{st}(v)$ is an array defined just like $T_{st}$ except that it only counts the shortest paths between vertices $s$ and $t$ that pass through $v$:
\begin{equation}\label{paths_multi}
	T_{st}(v) = (T_{sv} \otimes T_{vt})_{st}^{\leftarrow} = T_{sv} \otimes T_{vt}^{\leftarrow}.
\end{equation}

The above equation can be interpreted as follows. Every path stored in the array $T_{sv}$ can be extended by every path stored in the array  $T_{vt}$. The outcome of this operation, which is in fact the multiplication of the two polynomials $W_{sv}$ and $W_{vt}$, gives us information about all shortest paths from $s$ to $t$ passing through $v$. The vertex $v$ is counted twice. To avoid duplicate counting, we shift the result of multiplication to the left, which effectively reduces the number of nodes in each path by one.

In weighted graphs the pair-dependency limited to the shortest paths consisting of $i$ nodes is defined as follows:
\begin{equation}\label{eq:pair_dep:i:weighted}
\delta^*_{s,t}(v)[i] = \frac{T_{st}(v)[i]}{\sigma_{st}}f_{\delta^*}(i).
\end{equation}

\noindent This limited pair-dependency measures the influence of all shortest paths between $s$ and $t$ consisting of $i$ nodes on the evaluation of vertex $v$. The value of the function $f_{\delta^*}$ depends solely on the number of nodes lying on the shortest paths between $s$ and $t$.

Taking all possible values of $i$ into consideration, we obtain the pair-dependency in weighted graphs:
\begin{equation}\label{eq:pair_dep:weighted}
\delta^*_{s,t}(v) = \sum_{i=1}^{|V|} \delta^*_{s,t}(v)[i],
\end{equation}

\noindent which is the positive contribution that all shortest paths between $s$ and $t$ make to the assessment of vertex $v$,

The definition of one-side dependency is similar to the one presented earlier in Subsection~\ref{framework:unweighted}, except that we replace $\delta$ with $\delta^*$:
\begin{equation}\label{eq:onde_side_dep:weighted}
\delta^*_{s,\cdot}(v) = \sum_{t \in V}{\delta^*_{s,t}(v)}.
\end{equation}

\noindent This is the positive contribution that all shortest paths starting with vertex $s$ make to the evaluation of vertex~$v$. The version of equation \eqref{eq:onde_side_dep:weighted} that only considers the shortest paths containing exactly $i$ nodes each is:
\begin{equation}\label{eq:onde_side_dep:i:weighted}
\delta^*_{s,\cdot}(v)[i] = \sum_{t \in V}{\delta^*_{s,t}(v)[i]} ,
\end{equation}

Now, we are able to infer the recursive relation for weighted graphs:
\begin{equation}\label{eq:recursion:weighted}
\delta^*_{s,\cdot}(v)[i] = \sum_{ \substack{ w:~(v,w) \in E \\ d(s,w) = d(s,v) + 1} }\frac{T_{sv}[i]}{\sigma_{sw}} \bigg(f_{\delta^*}(i+1) + \delta^*_{s,\cdot}(w)[i+1]\bigg).
\end{equation}

Using the above equation, we are able to compute a parametrised betweenness centrality for weighted graphs $c^*_{fg}$. More specifically, for a vertex $v$, we compute  $c^*_{fg}(v)$ by iterating over all other vertices and summing their contributions. More formally:
\begin{equation}\label{eq:general_cen:weighted}
c^*_{fg}(v) = \sum_{s \neq v}\sum_{i=1}^{|V|} \bigg( \delta^*_{s,\cdot}(v)[i] + \frac{T_{sv}[i]}{\sigma_{sv}}g_{\delta^*}(i) \bigg),
\end{equation}

\noindent where the value of the function $g_{\delta^*}$ depends solely on the number of nodes lying on the shortest paths between $s$ and $v$.

Now, we are ready to introduce our general framework, namely \textbf{WPBC} (see Algorithm \ref{algorithms:scheme:weighted}) to compute the \emph{parametrised betweenness centrality} in weighted graphs. All lines that differ from the original Brandes' algorithm were highlighted. In lines \ref{sch:w:first_start} - \ref{sch:w:first_end} the algorithm uses Dijkstra's algorithm \citep{Cormen:2001} to traverse the graph and count all shortest paths from the source $s$. In these lines it essentially differs from Brandes algorithm, because it uses polynomial arithmetic to count paths and number of nodes lying on these paths. Next, in lines \ref{sch:w:rec_impl} and \ref{sch:w:rest_impl}, which also differ from Brandes' framework, we implemented the recursive formula~\eqref{eq:recursion:weighted}.

Algorithm \ref{algorithms:scheme:weighted} runs in $O(|V|^2|E| + |V|^2\log{|V|})$  time and requires $O(|V|^2)$ space. Furthermore, this algorithm (just like Algorithm~\ref{algo:pbc}) can be easily adapted to directed graphs.

%%%%%%%%%%%%%%%%%%%%%%%%%%%%%%%%%%%%%%%%%%%%%%%%%%%%%%%%%%%%%%%%%%%%%%%%%%%%%%%%%%%%%%%%%%%%%%%%%%%%%
\RestyleAlgo{ruled} %The old version of this is: \restylealgo
\begin{algorithm}[thp!]
\SetAlgoVlined %The old version of this is: \SetVline
\LinesNumbered %The old version of this is: \linesnumbered
\caption{\textbf{WPBC}---a general framework to compute the \emph{parametrised betweenness centrality} in \emph{weighted} graphs}
\label{algorithms:scheme:weighted}
\KwIn {weighted graph $G = (V,E)$, with weight function $\lambda: E \rightarrow \mathbb{R^+}$}
\KwData {priority queue $\mathcal{Q}$ with key $d()$, stack $\mathcal{S}$ \\
$d(s,v):$ the distance from $s$ to $v$ \\
$Pred_{s}(v):$ the list of predecessors of $v$ on the shortest paths from source $s$ \\
$\sigma_{sv}:$ the number of shortest paths from $s$ to $v$ \\
$\delta_{s,\cdot}^*(v)[i]:$ one-side dependency of $s$ on $v$ \\
$T[i]_{sv}:$ the number of shortest paths from $s$ to $v$ containing $i$ vertices  }
\KwOut {$c^*_{fg}(v)$ parametrised betweenness centrality}
%\lForEach{$v \in V$} {
%	$c^*_{fg}(v) \gets 0;$
%}

\ForEach{$s \in V$} {
	\lForEach{$v \in V$} {
		$Pred_{s}(v) \gets empty~list;~d(s,v) \gets \infty;~\sigma_{sv} \gets 0;$
	}	
	$d(s,s) \gets 1;~\sigma_{ss} \gets 1;~$
	$\highlight{T_{ss}[1] \gets 1;}~$ $enqueue~s \rightarrow \mathcal{Q};$\\
	\While{$\mathcal{Q}$ is not empty }{\nllabel{sch:w:first_start}
		$extract~v \gets \mathcal{Q} $ with minimal $d(s,v);$ $~push~v \rightarrow \mathcal{S};$ \\
		\ForEach{$w$ such that $(v,w)\in E $} {
			\If{$d(s,w) > d(s,v) + \lambda(v,w) $}{
				$d(s,w) \gets d(s,v) + \lambda(v,w);$\\
				 $ insert/update~w \rightarrow \mathcal{Q} $ with $d(s,w);$ \\
				$\sigma_{sw} \gets 0;~\highlight{T_{sw} \gets 0;}$ $Pred_{s}(w) \gets empty~list;$
			}
			\If{$d(s,w) = d(s,v) + \lambda(v,w)$} {
				$\sigma_{sw} \gets \sigma_{sw} + \sigma_{sv};$  $\highlight{T_{sw} = T_{sw} \oplus T_{sv}^{\rightarrow};}$ \\
				$append~v \rightarrow Pred_{s}(w);$ \\ \nllabel{sch:w:first_end}
				
			}\vspace*{-0.1cm}
		} 	\vspace*{-0.1cm}
	}\vspace*{-0.1cm}
	\ForEach{$v \in V$}{$\delta_{s,\cdot}^*(v) \gets 0;$}
	\While{$\mathcal{S}$ is not empty} {
		$pop~ w \gets \mathcal{S};$\\
		\lForEach{$v \in Pred_{s}(w)$}{
			\For{$i \gets 1$ \KwTo $|V|-1$}{

			$\highlight{ \delta^*_{s,\cdot}(v)[i] \gets \delta^*_{s,\cdot}(v)[i] + \frac{T_{sv}[i]}{\sigma_{sw}} \bigg(f_{\delta^*}(i+1) + \delta^*_{s,\cdot}(w)[i+1]\bigg);}$ \nllabel{sch:w:rec_impl}
			\vspace*{-0.1cm}
			 }
		
		}
		\lIf{$w \neq s$}{
			\For{$i \gets 1$ \KwTo $|V|$}{
				$\highlight{c^*_{fg}(w) \gets c^*_{fg}(w) + \delta^*_{s,\cdot}(w)[i] + 2\frac{T_{sw}[i]}{\sigma_{sw}}g^*_{\delta}(i) ;}$ \nllabel{sch:w:rest_impl}
			}
		}
	}
}

\vspace*{-0.4cm}
\lForEach{$v \in V$} {  \nllabel{sch:w:del_cont}
	$\highlight{c^*_{fg}(v) = \frac{c^*_{fg}(v)}{2};}$
}

\end{algorithm}
%%%%%%%%%%%%%%%%%%%%%%%%%%%%%%%%%%%%%%%%%%%%%%%%%%%%%%%%%%%%%%%%%%%%%%%%%%%%%%%%%%%%%%%%%%%%%%%%%%%%%

\subsection{The Algorithm for Shapley Value-based Betweenness Centrality in Weighted Graphs}\label{framework:weighted:SV}

\noindent In this subsection we will construct an efficient algorithm for computing the Shapley Value-based betweenness centrality for weighted graphs. We will use the notation introduced in the previous subsection, and transform equation \eqref{shap_eq} into:

\begin{align}\label{shap_eq_simpl:weighted}
\phi^{\SV}_v(V(G),\nu) &=  \hspace{-0.10cm}\sum_{s \neq v} \bigg(  \sum_{t \neq v}{ \sum_{i=1}^{|V|} \frac{T_{st}(v)[i]}{\sigma_{st}~i}   } + \sum_{i=1}^{|V|}
\frac{T_{sv}[i](2-i)}{2\sigma_{sv}~i} \bigg).
\end{align}

The framework for weighted graphs (introduced in Section~\ref{framework:weighted}) can be used to compute \eqref{shap_eq_simpl:weighted}. All we need is to define $f^*_{\delta}(i) = \frac{1}{i}$ and $g^*_{\delta}(i) = \frac{2-i}{2i} $. By doing so, we can use our general framework---\textbf{WPBC}, see Algorithm~\ref{algorithms:scheme:weighted}---to compute the weighted Shapley value-based betweenness centrality for all nodes in $O(|V|^2|E| + |V|^2\log{|V|})$. This method is presented in Algorithm~\ref{algorithms:SV:weighted}.

%%%%%%%%%%%%%%%%%%%%%%%%%%%%%%%%%%%%%%%%%%%%%%%%%%%%%%%%%%%%%%%%%%%%%%%%%%%%%%%%%%%%%%%%%%%%%%%%%%%%
\RestyleAlgo{ruled} %The old version of this is: \restylealgo

\begin{algorithm}[t!]
\SetAlgoVlined %The old version of this is: \SetVline
\LinesNumbered %The old version of this is: \linesnumbered
\caption{\textbf{WSVB} Weighted SV-based betweenness centrality}\label{algorithms:SV:weighted}
\KwIn {weighted graph $G = (V,E)$, with weight function $\lambda: E \rightarrow \mathbb{R^+}$}
\KwOut {$\phi^{\SV}_v$ weighted SV-based betweenness centrality for each $v\in V$}

$f^*_{\delta}(i) \gets \frac{1}{i}$; \\
$g^*_{\delta}(i) \gets \frac{2-i}{2i}$; \\
$\textbf{WPBC}(G,\lambda, f^*_{\delta},g^*_{\delta})$ \\

\end{algorithm}
%%%%%%%%%%%%%%%%%%%%%%%%%%%%%%%%%%%%%%%%%%%%%%%%%%%%%%%%%%%%%%%%%%%%%%%%%%%%%%%%%%%%%%%%%%%%%%%%%%%%

\subsection{The Algorithm for Semivalue-based Betweenness Centrality in Weighted Graphs}\label{framework:weighted:S}

\noindent In this subsection we use our framework for weighted graphs, namely \textbf{WPBC} (see Algorithm~\ref{algorithms:scheme:weighted}), to compute the Semivalue-based betweenness centrality for weighted graphs. To this end, we can transform equation~\eqref{Eq:PSV_formula} into:

\begin{align}\label{Eq:PSV_formula:weighted}
\hspace*{-0.3cm} \phi^{\SEMI}_v(V(G),\nu) & \hspace{-0.12cm}=  \hspace{-0.35cm}\sum_{1 \leq k \leq |V|} \hspace{-0.25cm}{PD}(V,k) \hspace{-0.05cm} \Bigg(   \sum_{s \neq v \neq t}\hspace{-0.25cm} \sum_{i=1}^{|V|-i \geq k-1} \hspace{-0.2cm}  \frac{T_{st}(v)[i](|V|\hspace{-0.055cm} - \hspace{-0.05cm}  i)!(|V|\hspace{-0.05cm} -\hspace{-0.05cm} k)!}{\sigma_{st}(|V|\hspace{-0.05cm} -\hspace{-0.05cm} i\hspace{-0.05cm}  -\hspace{-0.05cm} k\hspace{-0.05cm} +\hspace{-0.05cm} 1)!(|V|\hspace{-0.05cm} -\hspace{-0.05cm} 1)!} \nonumber \\
&+ \hspace{-0.15cm} \sum_{s \neq v} \Big( \hspace{0.03cm}  \frac{k-|V|}{|V|\hspace{-0.05cm}-\hspace{-0.05cm}1} \hspace{-0.05cm}+\hspace{-0.3cm} \sum_{i=1}^{|V|-i \geq k-1} \hspace{-0.3cm}  \frac{T_{sv}[i](|V|\hspace{-0.05cm}-\hspace{-0.05cm}i)!(|V|-k)!}{\sigma_{sv}(|V|\hspace{-0.05cm}-\hspace{-0.05cm}i\hspace{-0.05cm} -\hspace{-0.05cm}k\hspace{-0.05cm}+\hspace{-0.05cm}1)!(|V|\hspace{-0.05cm}-\hspace{-0.05cm}1)!}\hspace{0.05cm} \Big)  \Bigg).
\end{align}

Our framework can be easily adopted to deal with Semivalues; all we need is to set:
\renewcommand{\arraystretch}{1.5}
\begin{equation}\label{eq:pos:weighted}	
f^*_{\delta}(i) = \left\{
 \begin{array}{l l}
   \frac{(|V|-i )!(|V|-k)!}{(|V|-i -k+1)!(|V|-1)!} &  \text{if $|V|-i \geq k-1$}\\
   0 &  \text{if $|V|-i < k-1$}
 \end{array} \right.	\nonumber
\end{equation}
\renewcommand{\arraystretch}{1}

\noindent and set $g^*_{\delta} = f^*_{\delta}  + \frac{k-|V|}{|V|-1}$. This way, we can use our general framework for weighted graphs, namely \textbf{WPBC} (see Algorithm~\ref{algorithms:scheme:weighted}) to compute the Semivalue-based betweenness centrality in $O(|V|^3|E| + |V|^3\log{|V|})$ time. The pseudo code is presented in Algorithm~\ref{algorithms:PSV:weighted}.

%%%%%%%%%%%%%%%%%%%%%%%%%%%%%%%%%%%%%%%%%%%%%%%%%%%%%%%%%%%%%%%%%%%%%%%%%%%%%%%%%%%%%%%%%%%%%%%%%%%%
\RestyleAlgo{ruled} %The old version of this is: \restylealgo

\begin{algorithm}[t!]
\SetAlgoVlined %The old version of this is: \SetVline
\LinesNumbered %The old version of this is: \linesnumbered
\caption{\textbf{WSB} Weighted Semivalue-based betweenness centrality }\label{algorithms:PSV:weighted}
\KwIn {weighted graph $G = (V,E)$, with weight function $\lambda: E \rightarrow \mathbb{R^+}$}
\KwOut {$\phi^{\SEMI}_v$ weighted Semivalue-based betweenness centrality for $v\in V$}

$f^*_{\delta}(i) \gets $ \lIf{$|V|-i \ge k-1$} { $\frac{(|V|-i )!(|V|-k)!}{(|V|-i -k+1)!(|V|-1)!}$} \lElse {0}
$g^*_{\delta}(i) \gets f^*_{\delta}(i) + \frac{k-|V|}{|V|-1}$; \\
\For{$k \gets 1$ \KwTo $|V|$} {
	$\textbf{WPBC}(G,\lambda,f^*_{\delta},g^*_{\delta})$ ; \\
	\ForEach{$v \in V$} {
		$\phi^{\SEMI}_v \gets \textit{PD}(V,k)c^*_{fg}(v);$
	}
}
\end{algorithm}
%%%%%%%%%%%%%%%%%%%%%%%%%%%%%%%%%%%%%%%%%%%%%%%%%%%%%%%%%%%%%%%%%%%%%%%%%%%%%%%%%%%%%%%%%%%%%%%%%%%%

\section{Simulations}\label{chap:bet:simulations}

\noindent So far, we argued that centrality measures based on Semivalues in general, and on the Shapley value in particular, can help analyse networks in situations where simultaneous incidents may occur. Furthermore, we presented polynomial time algorithms to compute both the Shapley value-based and Semivalue-based betweenness centralities, for weighted graphs as well as unweighted graphs. This section provides an empirical evaluation of the above contributions. In particular, Section~\ref{sec:simulation:measure} compares our centrality measures against standard betweenness centrality, in a scenario where simultaneous node failures are simulated. In particular, we compare the effectiveness of both measures when tackling Problem~\ref{def:pro3}. Section~\ref{sec:simulation:algorithm} evaluates the running time of our algorithms, and compares them against the Monte Carlo method---the only available alternative in the literature.

We carry out our experiments on weighted and unweighted random scale-free graphs, which are created using the \emph{preferential attachment} mechanism introduced by \cite{Barabasi:Albert:1999}. We also carry out experiments on two real-life network: (1) the email communication networks introduced by \cite{Guimer:et:al:2009}, which contains $1133$ nodes and $5451$ edges, and (2) the neural system of C. elegans studied by \cite{Watts:Strogatz:1998}, which contains $297$ nodes and $2359$ edges.

\subsection{Evaluating our Centrality Measure}~\label{sec:simulation:measure}

\noindent In this subsection, we compare the node ranking obtained by the Semivalue-based betweenness centrality against the node ranking obtained by standard betweenness centrality. Note that, by this evaluation, we also implicitly evaluate our Shapley value-based betweenness centrality, as it is a special case of the Semivalue-based one.

To this end, let us consider the scenario of network vulnerability (see Section~\ref{section:network:vulnerability} for more details). The comprehensive work by \cite{Holme:et:al:2002} examines how different strategies of attack on the network affect its performance. In more detail, the authors studied how removing the nodes with the highest betweenness centrality affects communication throughout the network. They found that, in random scale-free networks, removing the nodes with the highest betweenness centrality causes an exponential decay in network performance. Here, performance is measured as the average distance between nodes---an important measure used, e.g., in the seminal paper of \cite{Albert:et:al:2000} on network vulnerability to attacks and random node failures. %The aim of that research was to find the most important nodes that should be protected, in order to counteract the potential simultaneous attack on the network structure.

The literature closely related to the work done by \cite{Holme:et:al:2002} is about network vulnerability analysis based on three-stage Stackelberg games \citep{Stackelberg:1952}. In a nutshell, two opponents, an \emph{interdictor} and an \emph{operator}, try to maximize their gains. In the first stage, the operator fortifies the network by choosing the nodes that should be fully protect. Next, the interdictor performs an attack, which involves removing a limited number of unfortified nodes. Finally, the operator acts on the network in order to fulfill its objectives, e.g., he can travel between two nodes, or send some load between nodes. The interdictor can be human or can be of a different nature such as a  natural disaster or a random failure of nodes. In such cases one can assume worst-case scenarios (Murphy's Law) \citep{Smith:Lim:2008}, or some stochastic model of interdictor attacks \citep{Cormican:et:al:1998}. The instance of Stackelberg games closely related to this problem is called \emph{Shortest Path Interdiction}. In such games the \emph{interdictor} attempts to remove $k$ edges in order to maximize the distance between two nodes $s$ and $t$. This problem is proven to be NP-hard \citep{Ball:et:al:1989}. However, \cite{Malik:et:al:1989} proposed an effective approximation algorithm and \cite{Corley:Sha:1982} solved this problem in polynomial time for the special case where $k=1$.	

In this subsection  we follow the Holme's line of research. We study the resilience of a network against random failures, where we assume that the level of protection of each node is determined based on its ranking. In order to carry out our experiment, we first need to agree on the way in which the functionality of the network will be measured. To this end, we use the average inverse geodesic measure (IGM), which was used by \cite{Albert:et:al:2000} and \cite{Holme:et:al:2002}. This measure captures the notion of average distance between nodes in a network, and is given by the following formula: 	

\[
\text{IGM}(V,E) = \sum_{v \in V} \sum_{v \neq u \in V} \frac{1}{d(v,u)}.
\]

Our experiment consists of two phases. The first involves generating random graphs and computing node cardinal rankings. This phase involves the following steps:

\begin{enumerate} \itemsep0.5em
\item Generate an unweighted random scale-free graph containing $160$ nodes, in which the average node degree is $\sqrt{160}$.\footnote{\footnotesize We also experimented with alternative values; the results exhibited similar trends.}
\item Set the interval $[a,b)$, which means that the number of nodes that will be simultaneously exposed to failure is between $a$ (inclusive) and $b$ (exclusive).
\item Compute two rankings of nodes. The first is based on standard betweenness centrality. The second is based on Semivalue-based betweenness centrality with probability distribution:
\[
PD(V,k) = \hspace{-0.1cm} \left\{
\begin{array}{l l}
 \frac{1}{b-a}  &  \text{if $k \in [a,b)$}\\
 0      &   \text{otherwise} .
\end{array} \right.
\]
\end{enumerate}

\begin{figure}[t!]
\includegraphics[width=1\textwidth]{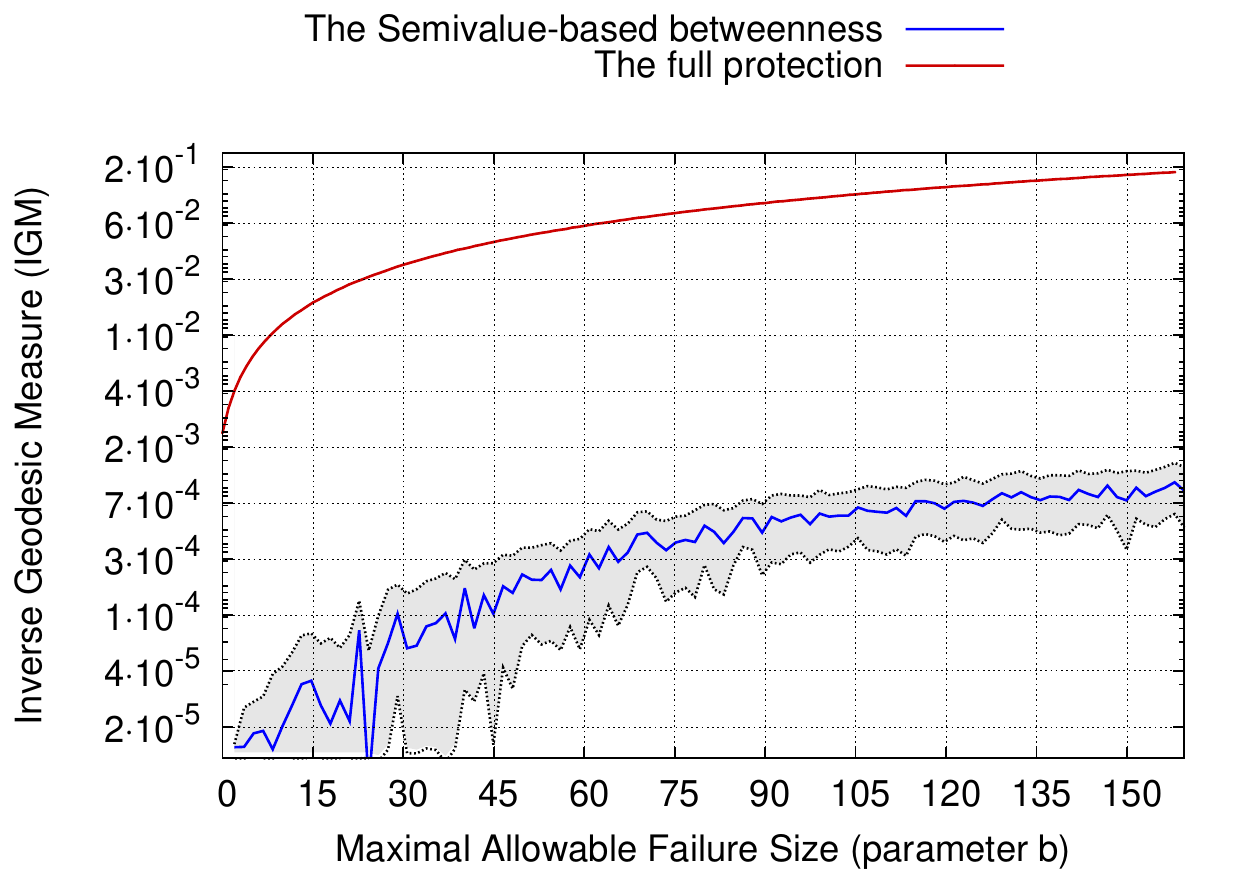}
\caption{The blue line plots the difference in performance between the Semivalue-based measure and the standard betweenness centrality on random graphs, where $a=1$ and $b =2,\ldots,160$, and where the protection strategy is set to protect nodes with probability $\frac{1}{r(v)^2}$. The confidence interval is set to $75\%$. The rightmost point of the plot (where the interval is $[1,160)$) represents the special case where the Semivalue-based measure happens to be identical to the Shapely value-based measure. The red line depicts the performance of the full protection strategy. The y-axis is in logscale.}
\label{fig:random:perf}
\end{figure}

The second phase of the experiment involves the following steps:

\begin{enumerate} \itemsep0.5em
\item Protect each node $v \in V$. The level of protection of each node is determined based on its position in the ranking, e.g., if a node is ranked first, then it is the most important node and so will have the highest level of protection. We study two approaches to nodes protection. The first one assigns to each node in the ranking the specific power of protection. More formally, if we denote by $r(v)$ the position of node $v$ in ranking $r$, then node $v$ is protected with probability $\frac{1}{r(v)^2}$. The second one is to fully protect the top $10\%$ of nodes in the ranking.
\item Choose randomly a set of nodes $S \subseteq V$, which will be exposed to failure.\footnote{\footnotesize In our setting, being \emph{exposed} to failure does not necessarily imply that the node will fail; this depends on the level of protection that this node has against failure.}  Here, the number of nodes that are exposed to failure is chosen uniformly at random from the range $[a,b)$. Likewise, out of all subsets that have exactly this number of nodes, the subset $S$ is chosen uniformly at random.
\item Using the first approach implies that $v$ will fail with probability $1-\frac{1}{r(v)^2}$. Using the second one indicates that each node $v\in S$ will fail if $r(v) < 0.1*|V|$.
\end{enumerate}

Having explained each phase of the experiment, we now explain how the entire experiment works:
\begin{enumerate} \itemsep0.5em
\item Generate $30$ random graphs.
\item Set the size of the failure to $[a,b)$, where $a=1$ and $b$ is running form $2$ to $160$.
\item For each $[a,b)$, simulate $5000$ simultaneous node failures.
\item Protect nodes depending on the two rankings: standard betweenness centrality and Semivalue-based betweenness centrality.
\item Compare the average IGM that results from each measure together with confidence interval of $75\%$.
\item Compare results to the full protection strategy that will protect all nodes from being removed.
\end{enumerate}

\begin{figure}[t!]
\includegraphics[width=1\textwidth]{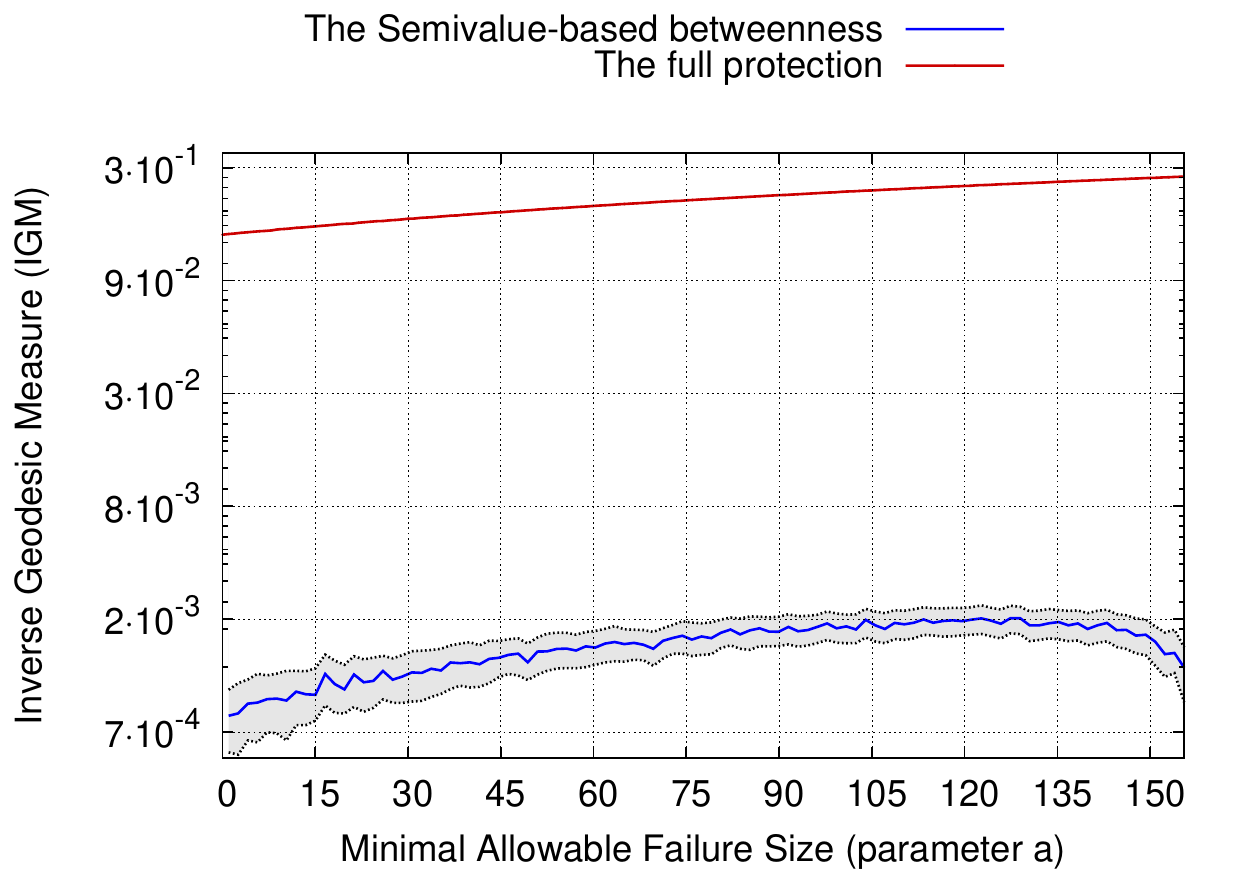}
\caption{The blue line is the difference in performance between the Semivalue-based measure and the standard betweenness centrality on random graphs, where $a=1,\ldots,159$ and $b=160$, and where the protection strategy is set to protect nodes with probability $\frac{1}{r(v)^2}$. The confidence interval is set to $75\%$. The leftmost point of the plot (where the interval is $[1,160)$) represents the special case where the Semivalue-based measure happens to be identical to the Shapely value-based measure. The red line depicts the performance of the full protection strategy. The y-axis is in logscale.}
\label{fig:random:perf:2}
\end{figure}

Figure~\ref{fig:random:perf} depicts the average difference in IGM between the two centrality measures. On the plot the blue line is $IGM_{S} - IGM_{B}$, where $IGM_{S}$ is the functionality of the network when protected according to our centrality measure, while $IGM_{B}$ is the functionality of the network when protected according to the standard betweenness centrality. The protection strategy is set to protect each node with probability $\frac{1}{r(v)^2}$. The values below zero indicate that standard betweenness centrality gives better protection to the network, whereas those above zero indicate that our Semivalue-based centrality gives better protection. As can be seen, when $b$ is small, the performance of the Semivalue-based betweenness centrality is almost identical to the standard betweenness centrality. This is expected because our centrality measure, given the failure interval $[1,2)$, is exactly the same as the standard measure. However, as we increase $b$, more nodes can fail simultaneously, and so the difference between our measure and the standard one become more evident. As can be seen, in such cases, our measure outperforms the standard one in terms of network protection, as more influential nodes (in terms of maintaining a low average distance throughout the network) are preserved. Note that for the interval $[1,160)$ the Semivalue-based centrality is exactly the same as the Shapley value-based centrality. Additionally, the red line in Figure~\ref{fig:random:perf} depicts the performance of the full protection strategy that protects each node from being removed. It is $IGM_{O} - IGM_{B}$, where $IGM_{O}$ is the functionality of the network when protected according to the full protection strategy, and $IGM_{B}$ is the functionality of the network when protected according to the standard betweenness centrality.

Figure~\ref{fig:random:perf:2} depicts the same experiment, except that  $b=160$ is now constant and $a =1,\ldots,159$. As can be seen, for all values of $a$ our measure gives better protection. Again note that for the interval $[1,160)$ the Semivalue-based centrality is exactly the same as the Shapley value-based centrality.

\begin{figure}[t!]
\begin{minipage}[b]{0.47\linewidth}\centering
\includegraphics[natheight=3.3cm, natwidth=8cm, height=5cm, width=1\textwidth]{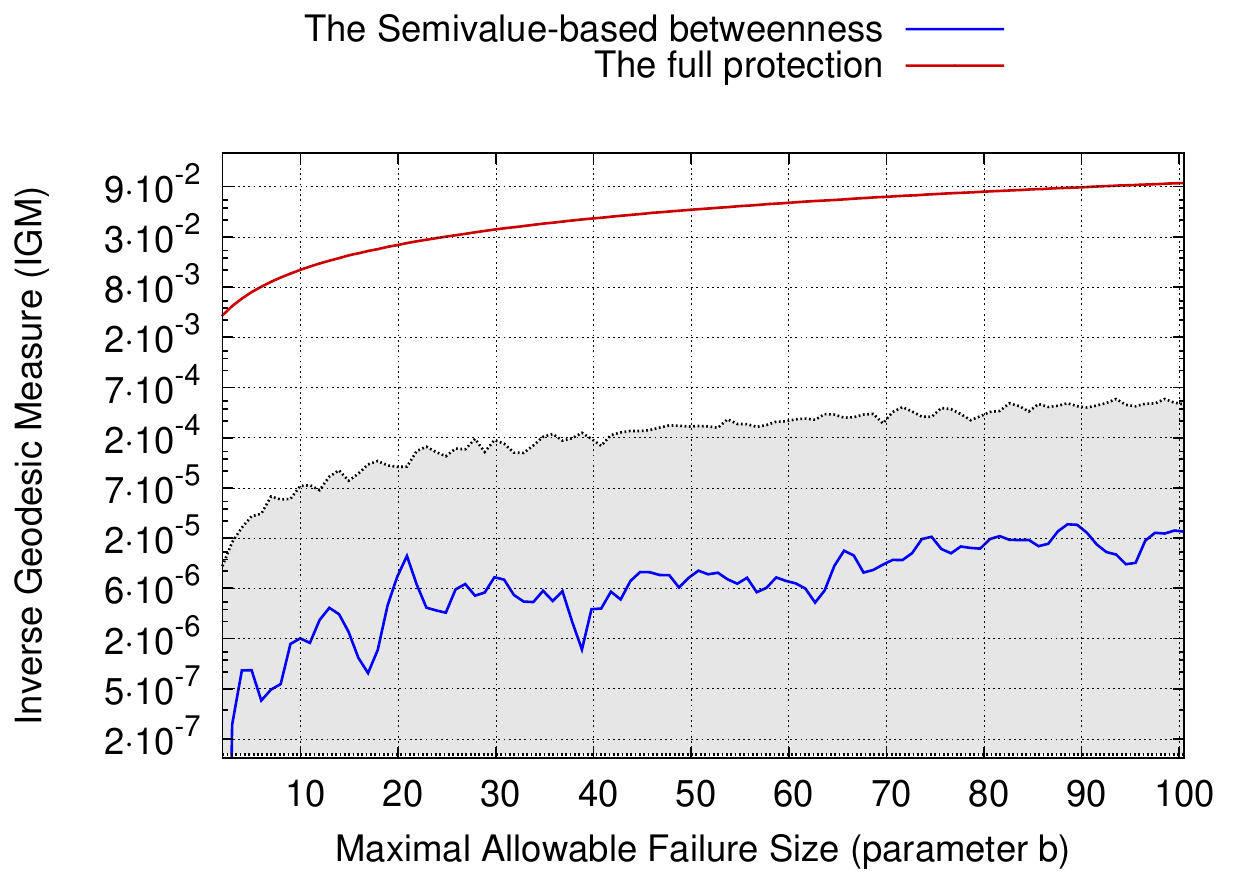}
\caption{The difference in performance between the Semivalue-based measure and the standard betweenness centrality, where $a=1$ and $b =2,\ldots,100$, and the strategy is to fully protect the top $10\%$ of nodes.}
\label{fig:random:perf:3}
\end{minipage}
\begin{minipage}[b]{0.05\linewidth}\centering
\hspace{0.05cm}
\end{minipage}
\begin{minipage}[b]{0.47\linewidth}\centering
\includegraphics[natheight=3.3cm, natwidth=8cm, height=5cm, width=1\textwidth]{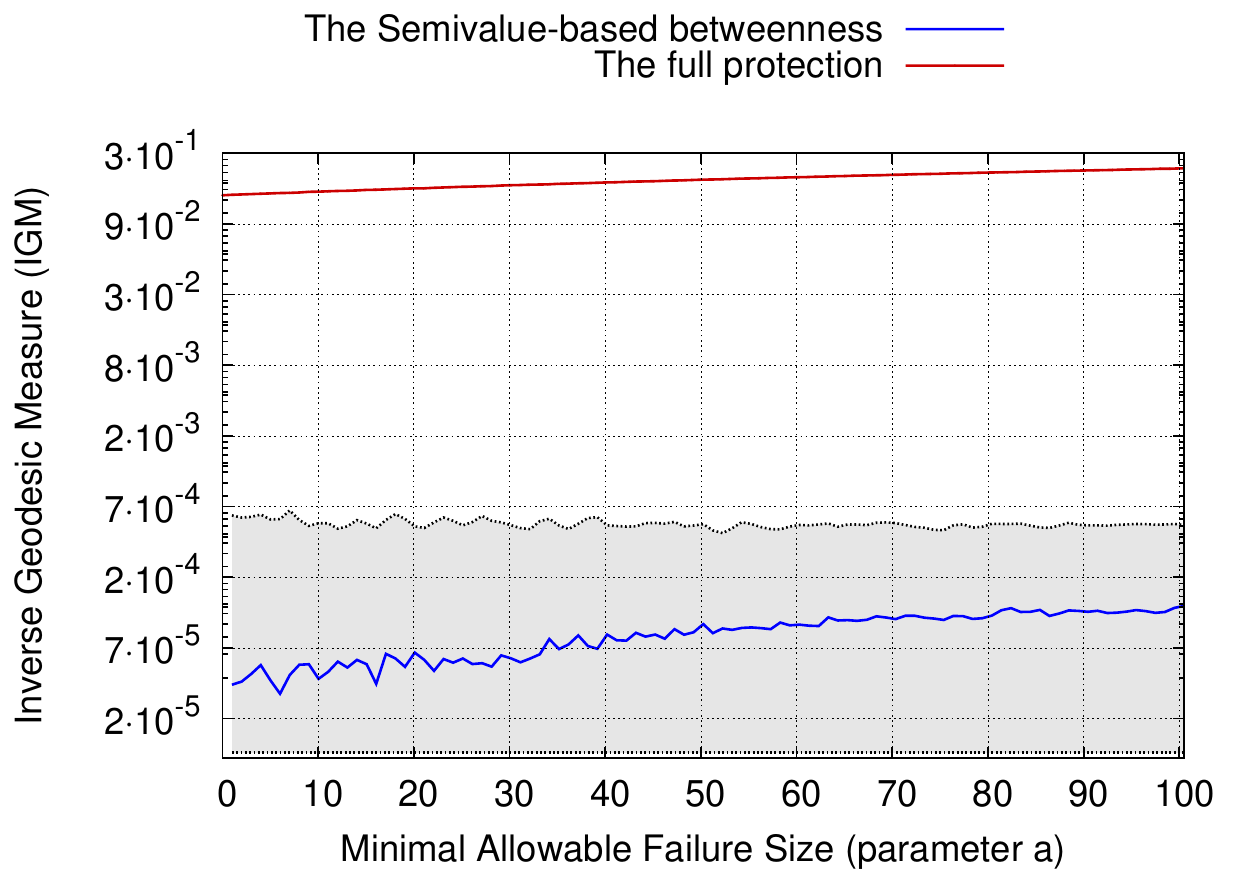}
\caption{The difference in performance between the Semivalue-based measure and the standard betweenness centrality, where $a=1,\ldots,99$ and $b=100$, and the strategy is to fully protect the top $10\%$ of nodes.}
\label{fig:random:perf:4}
\end{minipage} \\
\end{figure}

Figures~\ref{fig:random:perf:3}~and~\ref{fig:random:perf:4} present the repetition of the above two experiments but for the other protection strategy, where we fully protect top $10\%$ of nodes. The results are somewhat similar, but the advantage of our measure is less evident. For instance, given $a=30$ in Figure~\ref{fig:random:perf:4}, the advantage of our measure is below $7\cdot 10^{-5}$, whereas for the first protection strategy it is above $7\cdot 10^{-4}$. This stems from the fact that the top nodes of the two rankings (i.e., the Semivalue-based and standard betweenness) are less different than the middle nodes or the bottom nodes. This is the consequence of the scale-free distribution and existence of powerful \emph{hub} nodes that dominate two rankings. Thus, if we protect only the top $10\%$ of nodes, the difference in performance will be less visible, but our new measure still outperforms standard betwenneess centrality.

Finally, Figure~\ref{fig:random:perfSVB} demonstrates that the performance between two specific Semivalue-based centralities: (1) the Shapley value-based and (2) the Banzhaf-based does not differ much. That is to say, the two measures can be used interchangeably in this application.

\begin{figure}[t!]
\includegraphics[width=1\textwidth]{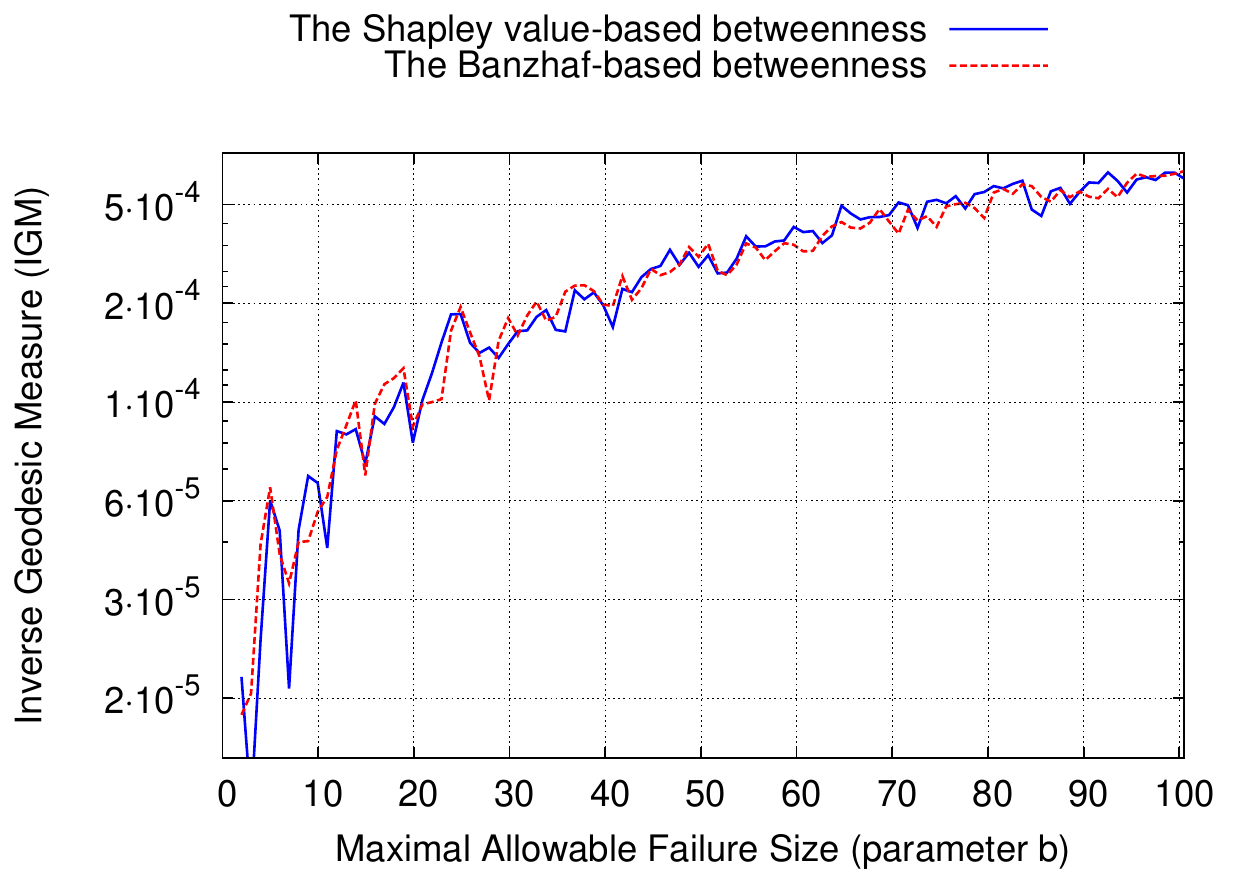}
\caption{The blue line plots the difference in performance between the Shapley value-based measure and the standard betweenness centrality, and the red one the difference in performance between the Baznaf-based measure and the standard betweenness centrality on random graphs, where $a=1$ and $b =2,\ldots,100$, and where the protection strategy is set to protect nodes with probability $\frac{1}{r(v)^2}$. The y-axis is in logscale.}
\label{fig:random:perfSVB}
\end{figure}

\subsection{Evaluating our Algorithms}~\label{sec:simulation:algorithm}

\noindent In this section, the running time of our exact algorithms is compared against that of the Monte Carlo (MC) sampling algorithm for computing the Shapley value.\footnote{\footnotesize For a formalisation and a discussion of this algorithm, see \citep{Castro:et:al:2009} and \citep{Michalak:et:al:2013}, respectively.} Note that, due to the exponential number of coalitions in a coalitional game, sampling (as in the MC algorithm) is the only general-purpose method available in the literature for computing the Shapley value in reasonable time. Generally speaking, the MC algorithm works by iteratively sampling permutations of players (or nodes in our case). For each such permutation, $\pi$, and each node $v$, the algorithm computes the marginal contribution of $v$ to the coalition consisting of every node that precedes $v$ in $\pi$. Finally, after several iterations, the algorithm divides the aggregate sum of all marginal contributions for each node by the the number of iterations performed. The time complexity of this algorithm is  $O(Iter\times con)$, where $Iter$ denotes the number of iterations and $con$ denotes the number of operations necessary for computing each marginal contribution. To optimise $con$, we implement MC sampling using the fast algorithm for computing group betweenness centrality presented by \cite{Puzis:et:al:2007}. This method guarantees that each MC iteration takes $O(|V|^3)$ time.\footnote{\footnotesize Note that if we have oracle access to the coalition values (which is a common assumption in the literature), then $con=1$. However, in our setting we do not make such an assumption. Instead, to obtain the value of a coalition, we need to compute its group betweenness centrality.} Clearly, the MC sampling is never guaranteed to yield the exact Shapley values and its performance should be measured by an approximation error. In our simulations, we compute the average relative error ($|1-\frac{\text{approx. value}}{\text{exact value}}|$) for the value of each node in the network.

\begin{figure}[t!]
\begin{minipage}[b]{0.5\linewidth}\centering
\includegraphics[natheight=3.3cm, natwidth=8cm, height=5cm, width=1\textwidth]{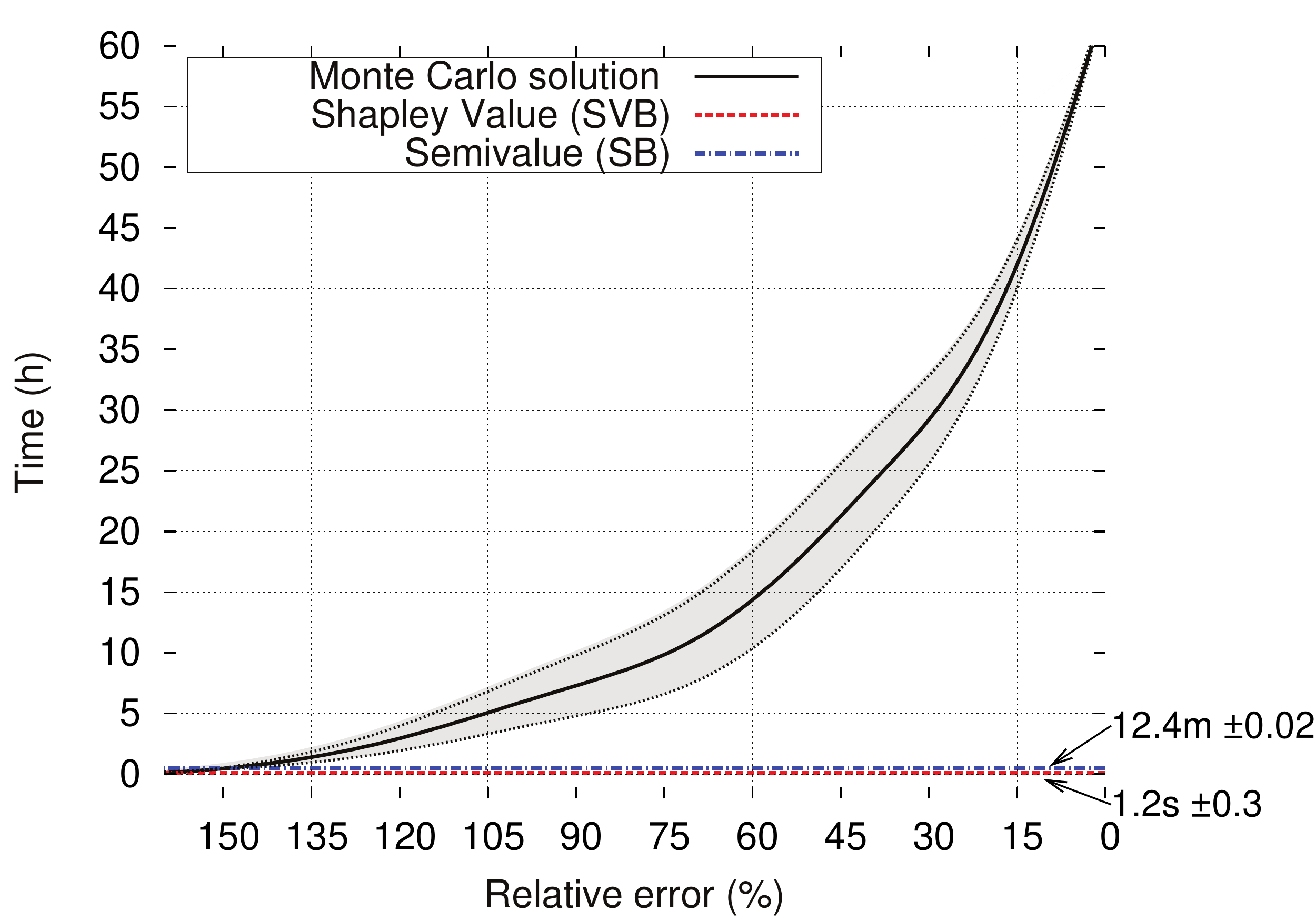}
\caption{Random graphs with $500$ nodes.}
\label{fig:random}
\end{minipage}
\begin{minipage}[b]{0.5\linewidth}\centering
\includegraphics[natheight=3.3cm, natwidth=8cm, height=5cm, width=1\textwidth]{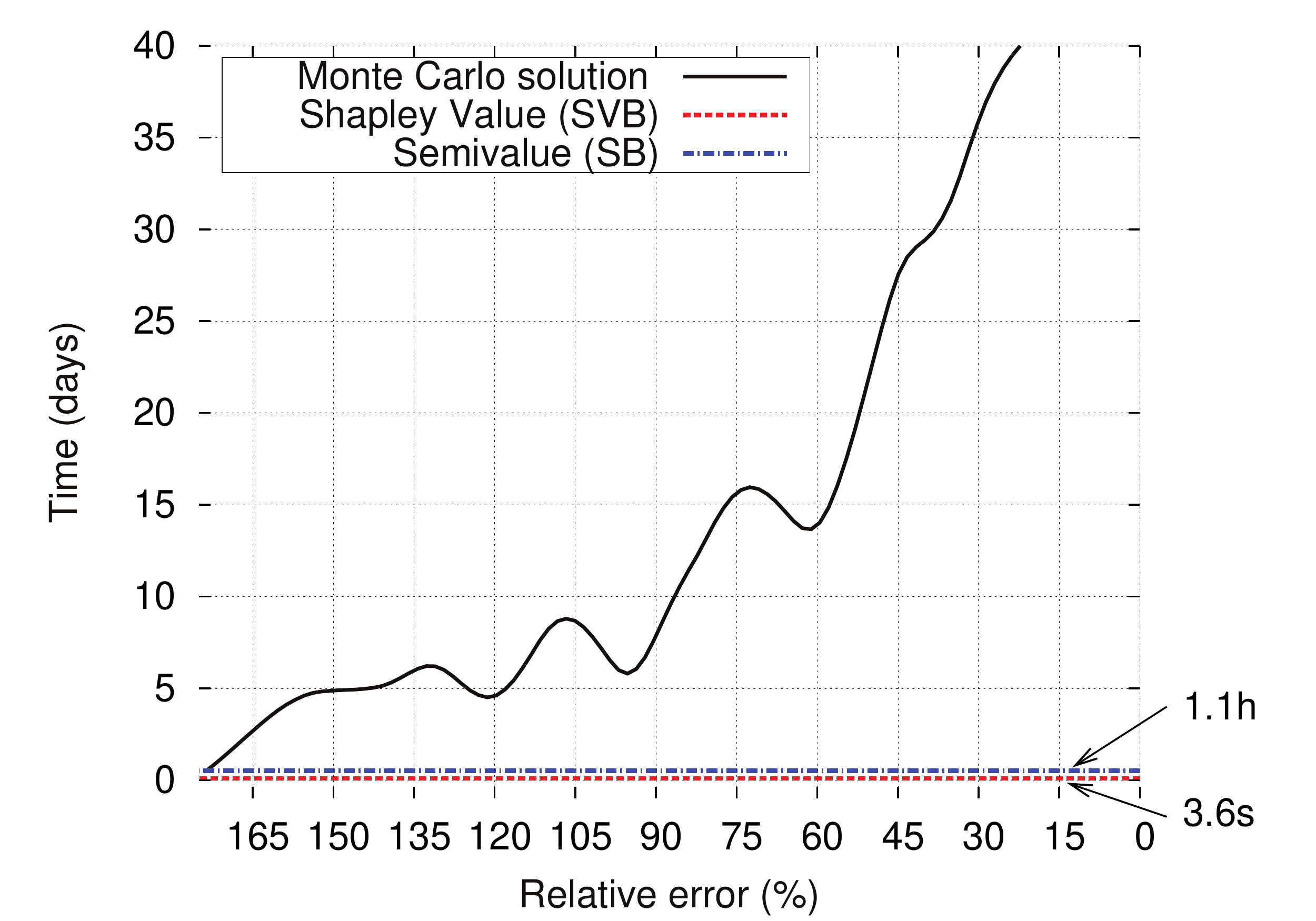}
\caption{The communication network.}
\label{fig:real}
\end{minipage} \\
\end{figure}

Starting with unweighted graphs, Figure~\ref{fig:random} and Figure~\ref{fig:real} compare the running time of our exact algorithms against the MC sampling given random networks, and given a real communication network, respectively. In more detail, Figures~\ref{fig:random} depicts the average result taken over $30$ randomly-generated graphs consisting of $500$ nodes each, with an average degree of $25$. The shaded area depicts the confidence interval of $75\%$. As can be seen, the Shapley value-based betweenness centrality is computed (using Algorithm~\ref{algorithms:SV}) in $1.2$ seconds, while and Semivalue-based betweenness centrality is computed (using Algorithm~\ref{algorithms:PSV}) in $12.4$ minutes. On the other hand, in order to obtain an admissible average error of $10\%$, the MC sampling method takes around $45$ hours. Furthermore, obtaining smaller average errors requires an exponential number of additional operations. Moving on to Figure~\ref{fig:real}, it depicts the results of the experiment performed on a real email communication network \citep{Guimer:et:al:2009}. As can be seen, Algorithms \ref{algorithms:SV} and \ref{algorithms:PSV} take $3.6$ seconds and $1.1$ hour to compute Shapley value-based and Semivalue-based betweenness centrality, respectively. In contrast, the MC sampling is unable to guarantee a better approximation than $15\%$ even after 40 hours.

\begin{figure}[t!]
\begin{minipage}[b]{0.5\linewidth}\centering
\includegraphics[natheight=3.3cm, natwidth=8cm, height=5cm, width=1\textwidth]{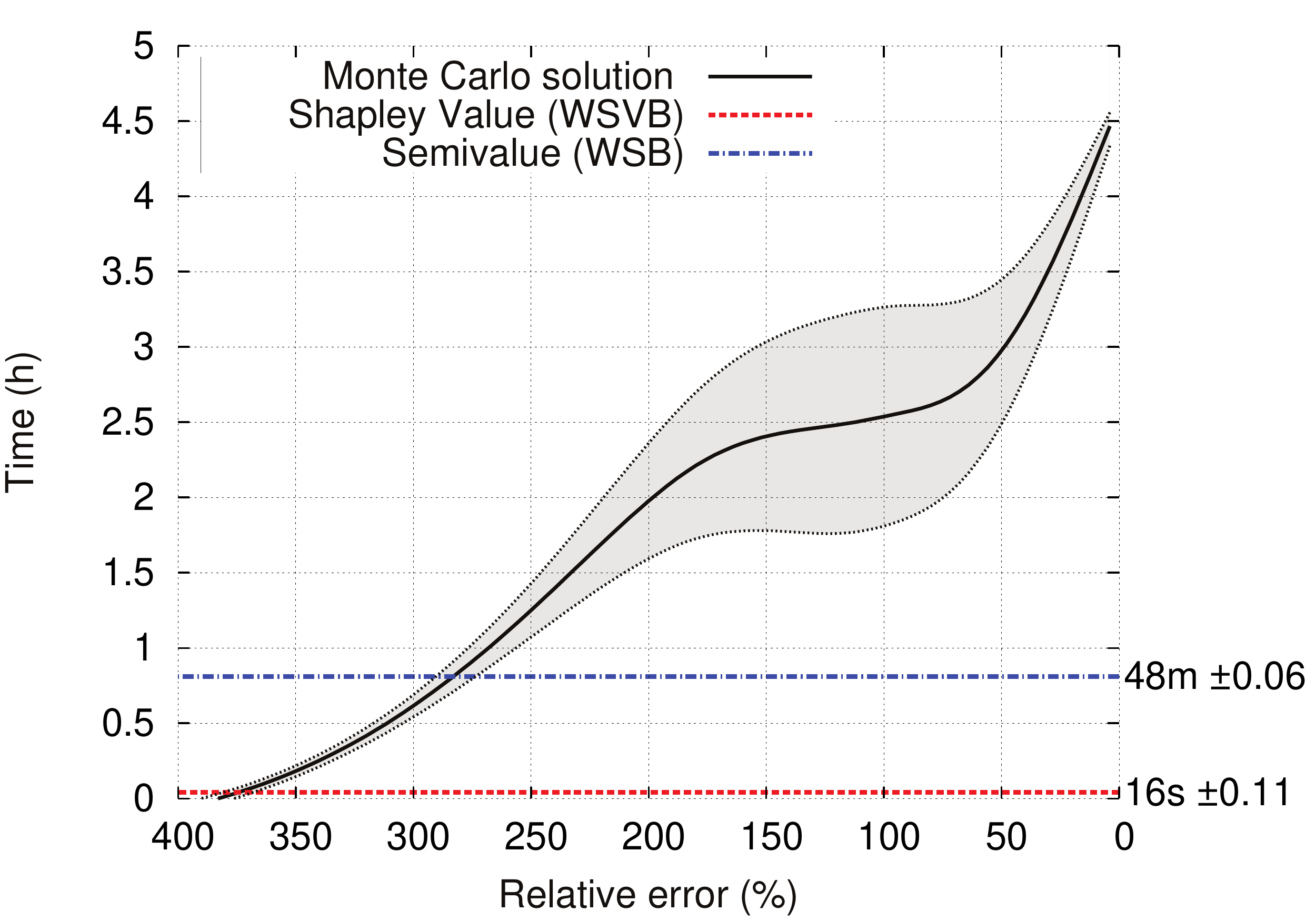}
\caption{Random graphs with $200$ nodes.}
\label{fig:random:weighted}
\end{minipage}
\begin{minipage}[b]{0.5\linewidth}\centering
\includegraphics[natheight=3.3cm, natwidth=8cm, height=5cm, width=1\textwidth]{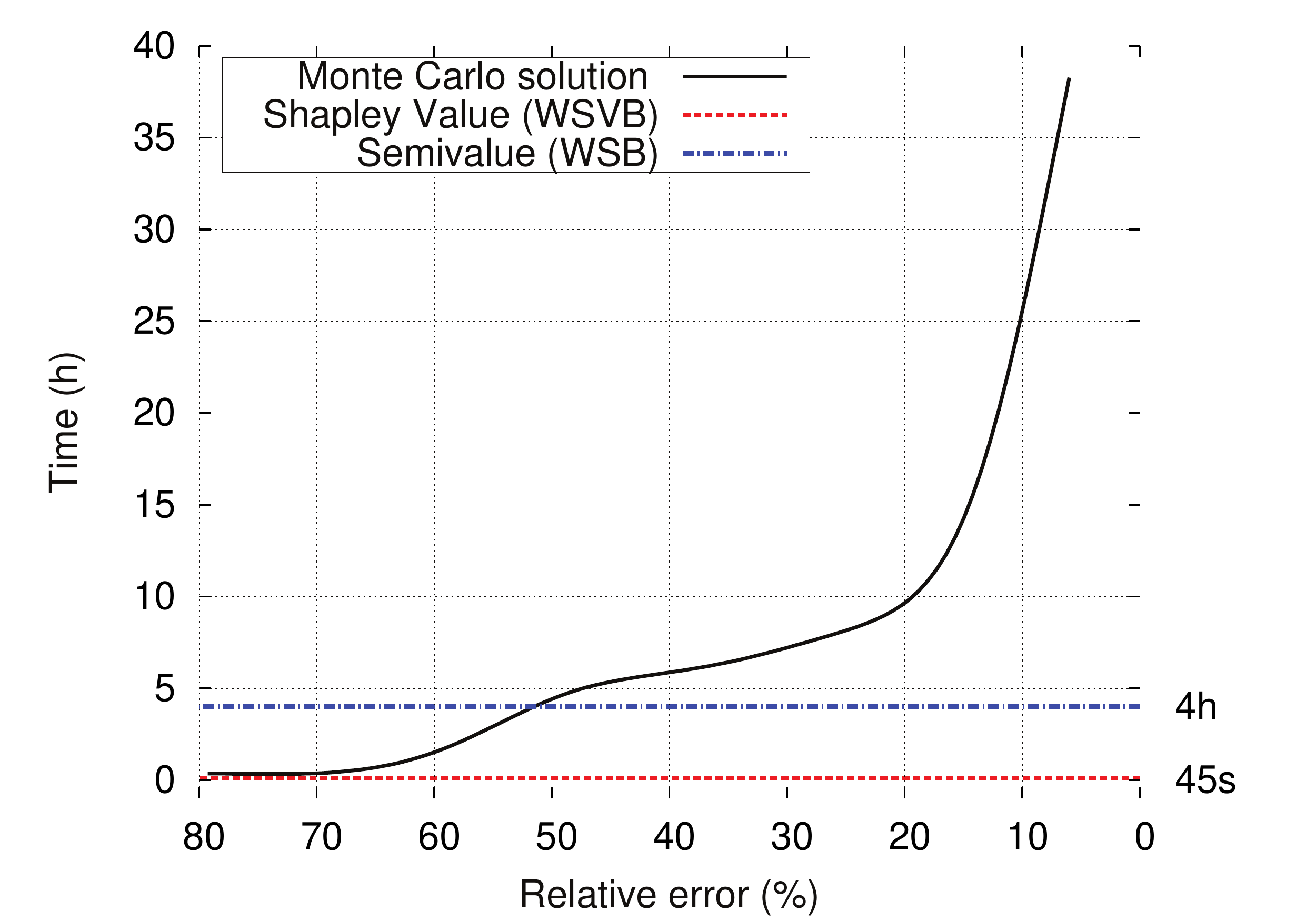}
\caption{The neural system network.}
\label{fig:real:weighted}
\end{minipage} \\
\end{figure}

Having presented our results for unweighted graphs, we now move on to weighted graphs. Figures \ref{fig:random:weighted} and \ref{fig:real:weighted} compare the performance of our algorithms against the MC sampling, given random networks, and a real neural network, respectively. In more detail, Figure~\ref{fig:random:weighted} depicts the average result taken over $30$ randomly-generated graphs consisting of $200$ nodes each, with an average degree of $15$. Weights of edges were chosen uniformly at random from the range $[1,100]$. The shaded area depicts the confidence interval of $75\%$. As can be seen, the Shapley value-based betweenness centrality is computed (using Algorithm~\ref{algorithms:SV:weighted}) in $16$ seconds, while the Semivalue-based betweenness centrality is computed (using Algorithm~\ref{algorithms:PSV:weighted}) in $48$ minutes. On the other hand, in order to obtain an admissible average error of $10\%$, the MC sampling takes around $4$ hours. Moving on to Figure~\ref{fig:real:weighted}, it depicts the experiment performed on a real neural network \citep{Watts:Strogatz:1998}. As can be seen, Algorithms \ref{algorithms:SV:weighted} and \ref{algorithms:PSV:weighted} take $45$ seconds and $4$ hours to compute Shaple-value-based and Semivalue-based betweenness centralities, respectively.

\section{Conclusions}\label{chap:bet:summary}
%two real-life networks ---> in the intro

\noindent In this chapter, was proposed an extension of betweenness centrality, based on Semivalues---a wide and flexible family of solution concepts from coalitional game theory. Our new measure provides a ranking of individual nodes while taking into consideration the synergies that may exist in different groups of nodes. The main contributions of this chapter are summarised in contrast to the literature in Table~\ref{tab:summary}.

\renewcommand{\arraystretch}{1.3}
\setlength{\tabcolsep}{.2em}
\begin{table}[t!]
\begin{center}
\footnotesize
\begin{tabular} { c  c  c  c c c}
\hline
{Standard}  & {Group}     & {SV-based}  & {Efficient}  & Semivalue-based & Efficient\vspace*{-0.07cm}\\
{centrality}& {centrality}& {centrality}& {computation} & centrality & computation\\
\hline\hline
\emph{node}        &  Everett  and    & Narayanam and   & \textbf{Chapter~\ref{chap:gtcm}}  &  \multicolumn{2}{c}{\citep{Szczepanski:et:al:2015b} } \\ 
		           &  Borgatti [1999]   & Narahari (2010)     &    & & \\ \hline
\emph{closeness}   & Everett  and      & \textbf{Chapter~\ref{chap:gtcm}}  &  \textbf{Chapter~\ref{chap:gtcm}} &  \multicolumn{2}{c}{\citep{Szczepanski:et:al:2015b} }  \\
		           &  Borgatti [1999]   &     &    &  & \\ \hline
\emph{betweenness} & Everett  and      & \textbf{this chapter}         & \textbf{this chapter} &  \textbf{this chapter}   &  \textbf{this chapter}  \\
		           &  Borgatti [1999]   &      &   & & \\ \hline
\end{tabular}
\caption{Summary of the results obtained in this chapter vs. other contributions in the literature.}
\label{tab:summary}
\end{center}
\end{table}
\renewcommand{\arraystretch}{1}

Despite the fact that Semivalues are, in general, computationally challenging, we showed that the new measure can be computed efficiently. More specifically, we proposed polynomial-time algorithms to compute all Semivalue-based betweenness centralities for weighted graphs as well as unweighted graphs. These include both the Shapley value-based and Banzhaf power index-based betweenness centralities. Interestingly, our algorithm for computing the Shapley value-based centrality for the unweighted networks has the same time complexity as the best known algorithm due to \citep{Brandes:2001} for computing the standard betweenness centrality.

We evaluated our algorithms empirically on weighted and unweighted random scale-free graphs \cite{Barabasi:Albert:1999} as well as on two real networks. Our algorithms turned out to be significantly faster than the Monte Carlo sampling \citep{Castro:et:al:2009}---the fastest alternative from the literature.

We also evaluated our centrality measures in a simulated scenario in which simultaneous node failures occur. Here, following Albert et al.,~\citeyear{Albert:et:al:2000}, and Holme et al.,\citeyear{Holme:et:al:2002}, we quantified the functionality of the network based on the average inverse geodesic measure. The results showed that, compared to the standard measure, the ranking obtained by our game-theoretic network centrality reflects more accurately the influence that different nodes have on the network functionality.

\chapter{A Centrality Measure for Networks With Community Structure Based on a Generalization of the Owen Value}\label{chap:OV}

\noindent In the previous chapters we presented the game-theoretic centrality measures based on Shaply value and Semivalues. In this chapter we propose the first measure of node centrality that takes into account the \emph{community structure} of the underlying network. To allow for flexible modeling of community structures, we propose a generalization of the Owen value---a well-known solution concept from cooperative game theory to study games with \textit{a priori}-given unions of players. As a result we obtain the first measure of centrality that accounts for both the value of an individual node's relationships within the network and the quality of the community this node belongs to. We use our measure to evaluate citation networks, where all publications belong to naturally defined science communities united under journal titles and conference venues.

All technical contributions presented in this chapter were made exclusively by the author of this dissertation.

\section{Introduction}

Real-world networks frequently have highly complex structures. They can often be characterized by properties such as heavy-tailed degree distributions, clustering, the small-world property, etc. Another important characteristics that many real-life networks have in common is their \textit{community structure} \citep{Girvan:Newman:2002,Newman:2006} (see Section~\ref{sec:communities}). Communities are usually composed of nodes that are more densely connected internally than with other nodes in the network. For instance, the teachers from a particular secondary school may form a community within the social network of all teachers in the city. Similarly, trade links among the European Union countries are usually more intense than their links with the rest of the world. In addition, certain communities may be considered to be stronger than others. Secondary schools may vary in reputation, and some trade blocks may be more important to the global economy than others.

The importance of a community is usually increased when a new,
powerful individual joins it. Conversely, membership in a strong
community may boost the importance of an otherwise weak
individual. Quantifying this latter effect is the primary goal of this
paper. In other words, we are concerned with the problem of analysing
the importance (the \textit{centrality}) of individual nodes given the
underlying community structure of the network.

In this chapter, we use the flexibility offered by the game-theoretic approach to centrality measures and constructed
first centrality measure in the literature that is able to account for
complex community structures in networks. To this end, we model the
community structure as the \textit{a priori} given coalition structure
of a cooperative game. By doing so, we are able to build a centrality
metric by generalizing the Owen value~\citep{Owen:1977}---a well-known
solution concept for cooperative games in which players are
partitioned into pre-defined groups.

In our approach, the computation of a node's power is a two-step
process. First, we compute the importance of the community (if any)
that this node belongs to. Next, we compute the power of the given
node within this community. Our generalization of the Owen value,
which we call \textit{coalitional semivalues}, is a much broader
solution concept. In fact, coalitional semivalues encompass the Owen
value as well as all other solution concepts in the literature that
were developed for games with an \textit{a priori} defined coalition
structure of players: the Owen-Banzhaf value~\citep{Owen:1982}, the
symmetric coalitional Banzhaf value, and p-binomial
semivalues~\citep{Carreras:Puente:2012}.

 Although, in general, the new centrality introduced in this chapter is \#P-complete (and hence
NP-hard), we are able to give a polynomial algorithm to compute it for
problem instances where the value of any group of nodes is determined
by their degree centrality~\citep{Freeman:1979} (see Definition~\ref{def:cen:degree}). We verify the
practical aspects of our algorithm on a large citation network that
contains more than 2 million nodes and links. Our experiments compare
three different degree centralities: group degree centrality \citep{Everett:Borgatti:1999}, the
Shapley value-based degree centrality (see Definition~\ref{def:gt:deg:sv}), and our new centrality. We show
that, unlike others, our new centrality produces a ranking in which
the power of the top nodes significantly differs, depending on the
power of the communities that these nodes belong to.

\section{Preliminaries}\label{section:model}

\noindent In this section we introduce the basic notation for coalitional games with coalitional structure and also we present the Owen value~\citep{Owen:1977}---the most popular extension of the Shapley value to games with coalitional structure.

\subsection{The Coalitional Structure}

\noindent We already know that the \emph{coalition} is a group of nodes in coalitional game. The coalitions can be non-overlapping, which mean that player can belong to only one coalition, and overlapping. In this dissertation we focus on non-overlapping coalitions. The game divided into fixed coalitions forms a \emph{coalitional structure}.

\begin{definition}[Coalitional structure]\label{def:CS:games}
For a given coalitional game $ g =(A,\nu) $ the coalitional structure $\mathit{CS}=\set{C_1, C_2, \ldots, C_m}$ is the partition of the set $A$, which implies that $\emptyset \notin \mathit{CS}$, $\sum_{C \in \mathit{CS}} C = A$ and  $\forall_{C_1,C_2 \in \mathit{CS}}{C_1 \neq C_2 \implies C_1 \cap C_2 = \emptyset}$.
\end{definition}

One of the most important research topics in coalitional games is to find optimal coalitional structure \citep{Sandholm:1999,Rahwan:et:al:2012}. More formally, this problem known as optimal coalition structure generation involves solving the following equation: $\argmax_{\mathit{CS}}\sum_{C \in \mathit{CS}} \nu(C)$. Unfortunately, this problem is NP-hard \citep{Sandholm:1999}.

Coalitional structure is a very similar to community structure defined on graphs (Definition~\ref{def:CS}). In this chapter we use this analogy and propose the first measure of node centrality that takes into account the \emph{community structure} and that is based on coalitional games with coalitional structure.

\subsection{The Owen Value}

\noindent Coalitional games can be analysed from both the \textit{ex ante} and \textit{ex post} perspectives. In the \textit{ex ante} perspective, it is not known which coalition will actually form. The Semivalues (Definition~\ref{def:semivalue}) are \textit{ex ante} since they consider the sum of the marginal contributions of a player to all possible coalitions without any additional assumptions.

Another approach is to consider a coalitional game from the \textit{ex post} perspective. In this case it is already known which coalitions form by the end of the game. In other words, it is known how the agents have partitioned themselves into a coalition structure. This \textit{ex post} perspective is especially appealing for our model where networks are or can be divided into communities.

%In the remainder of this section we will propose an answer to the following question: how should we evaluate the importance of individual nodes within a network so that the relative importance of their communities is taken into account?  To this end,  we will consider how to adapt game-theoretic centralities so that they allow for an \textit{a priori} given community structure of the social network.

%The semivalues consider marginal contributions of players to all possible coalitions. In other words, they look at the coalitional game problem from the \emph{ex ante} perspective, assuming that we do not know beforehand who will form a coalition with who. However, it is also possible to consider the coalition game problem from the \textit{ex post} perspective.  \textit{a priori} given coalitional structure $\mathit{CS}$, which is an exhaustive and disjoint partition of the set $N$. Each player $i$ belongs to some coalition $C_j \in \mathit{CS}$, and the performance of this coalition influences the importance of player $i$.

%The literature on game Determining the importance of player in this game is called the \emph{ex post} coalitional game problem, since the coalitional structure that would be formed by the players is known.~\cite{Casajus:2009}

% We will now define it using terminology appropriate to game-theoretic network centrality.

The most popular extension of the Shapley value (Definition~\ref{def:sv}) to ``\textit{ex post}''-like situations was proposed by \cite{Owen:1977}. To this end, let us first introduce the concept of the \emph{quotient game} $\nu^Q$. Given the coalitional game $(A,\nu)$ and coalitional structure $\mathit{CS}=\set{C_1, C_2, \ldots, C_m}$, we define a new coalitional game, where the coalitions are considered to be individual players:
\[
\nu^Q(R) = \nu \bigg( \bigcup_{r \in R} C_r \bigg)~\text{for all}~R \subseteq M,
\]

\noindent where the set $M=\set{1,2,\ldots,m}$ represents coalitions' numbers. Note that $\bigcup_{r \in R} C_r$ is a \textit{coalition of coalitions}. We will denote such coalition by $Q_R$.

We are now ready to define the Owen value. Given $(A,\nu)$ and $\mathit{CS}=\set{C_1, C_2, \ldots, C_m}$, the share of the grand coalition's payoff, $\nu(A)$, given to player $a_i \in C_j \in \mathit{CS}$ according to the Owen value is given by:

\begin{definition}[The Owen Value]\label{def:OV}
In a coalitional game $(A,\nu)$ with coalition structure $\mathit{CS}$  the Owen value for a player $a_i$ is given by:
\[
\phi^{\OV}_i(A, \nu, \mathit{CS})  =   \sum_{R \subseteq M \setminus \{j\}} \frac{1}{|M|  \binom{|M|- 1}{|R|}}  \sum_{C \subseteq C_j \setminus \{a_i\}}   \frac{1}{|C_j|  \binom{|C_j|  - 1}{|C|}}    \text{MC}(  Q_R \cup C,i) .
\]
\end{definition}

Let us examine the above formula more closely. The computation of the Owen value can be thought of as a two-step process. In the first step, coalitions play the game $(M,\nu^Q)$ between themselves and receive their Shapley values. In the second step, the values of these communities are, in turn, divided among their members according to the Shapley value of the members.

The Owen value is the unique division scheme that satisfies the following five often desirable properties: Efficiency, Symmetry, Null~player, Additivity, and Component Symmetry \citep{Owen:1977}, which will be discussed more extensively in the next section.

\section{The Coalitional Semivalue}

\noindent In this chapter we introduce a generalization of the Owen value, where more general division schemes---Semivalues---are used as opposed to the Shapley value within the definition of Owen value. Specifically, we combine formula for Semivalue (Definition~\ref{def:semivalue}) with the formula for the Owen value (Definition~\ref{def:OV}) and propose \emph{coalitional semivalues}.

\begin{definition}[The Coalitional Semivalue]\label{def:coalitional:semivalue}
In a coalitional game $(A,\nu)$ with coalition structure $\mathit{CS}$  the Coalitional Semivalue for a player $a_i$ is given by:
\[
\phi^{\CSEMI}_i(\nu,\mathit{CS})  = \   \sum_{0 \leq k < |M|}  \beta(k)  \sum_{0 \leq l < |C_j|}   \alpha_j(l)     \mathbb{E}_{T^k,C^l}[ \text{MC}( Q_{T^k} \cup C^l,i)] .
\]
\end{definition}

\noindent where $T^k$ is a random set of size $k$ drawn uniformly from the set $M \setminus \set{j}$, and $C^l$ a the random set of size $l$ drawn uniformly from the set $C_j \setminus \set{i}$. The function $\beta: \{0,1,\ldots,|M|-1\} \to [0,1]$ is a function such that $\sum_{k=0}^{|M|-1}\beta(k) = 1$. $\set{\alpha_j}_{j \in \set{1,\ldots,|M|}}$ is a family of functions such that $\alpha_j: \{0,1,\ldots,|C_j|-1\} \to [0,1]$ and $\sum_{k=0}^{|C_j|-1}\hspace{-0.1cm}\alpha_j(k)=1$.

Intuitively, $\beta$ is a probability distribution used to compute $\phi_j(M,\nu^Q)$, and $\alpha_j$ is the probability distribution used to evaluate the players inside a coalition $\phi_i(C_j,\nu)$. Importantly, as shown in Table~\ref{table:solutions}, by adopting various probability distributions, we can obtain the Owen value~\citep{Owen:1977}, as well as all of its modifications proposed to date in the literature: \textit{Owen-Banzhaf value}~\citep{Owen:1982}, \emph{symmetric coalitional Banzhaf value}~\citep{Amer:et:al:2002}, and \emph{symmetric coalitional p-binomial semivalues}~\citep{Carreras:Puente:2012}.\footnote{We refrain from the axiomatic characterization of the new solution concept as being out of the scope of this dissertation.}
\def\arraystretch{1}

\begin{table}
\centering

\caption{Values of $\alpha$ and $\beta$ for the Owen value and its various extensions.}
\label{table:solutions}
\scalebox{1}{
\begin{tabular}{l@{\;}c@{\;\;}cl}
	\toprule
	{\bf Solution name}	&		$\beta(k)$	&	$\alpha_j(l)$		\\\toprule
		
		{Owen value \citep{Owen:1977}}
		&	$\frac{1}{|M|}$	&	$\frac{1}{|C_j|}$	 \\	\midrule 			
		
		{Owen-Banzhaf value \citep{Owen:1982}}
		&	\scalebox{0.8}[0.8]{$\frac{\binom{|M|-1}{k}}{2^{|M|-1}}$}	&	\scalebox{0.8}[0.8]{$\frac{\binom{|C_j|-1}{l}}{2^{|C_j|-1}}$}	 \\		 \midrule 		
		{\parbox[c][2.5em][t]{20em}
		{symmetric coalitional  Banzhaf value \\ \citep{Amer:et:al:2002}}}
		&	\scalebox{0.8}[0.8]{$\frac{\binom{|M|-1}{k}}{2^{|M|-1}}$} &	$\frac{1}{|C_j|}$	\\\midrule
		{\parbox[c][2.5em][t]{20em}
		{symmetric coalitional p-binomial semivalue \\ \citep{Carreras:Puente:2012}}}
		&	{\parbox[c][2.5em][t]{8em}
		{$p^k(1-p)^{|M|-1-k}$ \\ $p \in [0,1]$  }}	&	$\frac{1}{|C_j|}$\\
	\bottomrule
\end{tabular}
}
\end{table}

\section{The New Centrality Measure and Its Properties}

\noindent Let us now introduce the game-theoretic network centrality measure based on coalitional semivalues:
\begin{definition}[The Coalitional Semivalue-based Centrality]
  The game-theoretic network centrality for the graph $G$ with community structure $\mathit{CS}$ is a quadruple $(G, \mathit{CS}, \psi, \phi^{\CSEMI})$, where the value of each node $v \in V$ is given by $\phi^{\CSEMI}_v(\nu_G,\mathit{CS})$, where $\psi(G) = \nu_G$.
\end{definition}
This is the first centrality measure that evaluates nodes by taking into account the community structure of the network. In the next section, we will consider various properties of this new measure.

The aim of the rest of this section is to translate the properties of various instances of coalitional semivalues into the properties of the resulting centrality measure. The first three properties are derived from Null Player, Additivity and Symmetry, respectively. Let $(G, \mathit{CS},  \psi, \phi)$ be a game-theoretic network centrality for the graph $G$ with community structure $\mathit{CS}$ and $v \in C_j \in \mathit{CS}$ some node from community $C_j$.
\begin{property}
If a node makes no contribution to any community then its value is zero: $\forall_{C \subseteq V \setminus \{v\}} \text{\textnormal{MC}}(C,v)= 0 \implies \phi^{\CSEMI}_v(\nu_G,\mathit{CS})= 0$.
\end{property}

\begin{property}
If two group centralities are combined into one $\nu_G$ = $\nu'_G + \nu''_G$ then $\phi^{\CSEMI}_v(\nu_G,\mathit{CS}) = \phi^{\CSEMI}_v(\nu'_G,\mathit{CS}) + \phi^{\CSEMI}_v(\nu''_G,\mathit{CS})$.
\end{property}

\begin{property}
If two nodes from the same community $v,u \in C_j$ contribute the same value to all possible communities then they are equally important: $\forall_{C \subseteq V \setminus \{v,u\}}\text{\textnormal{MC}}(C,v)   =   \text{\textnormal{MC}}(C,u) \implies    \phi^{\CSEMI}_v(\nu_G,\mathit{CS})= \phi^{\CSEMI}_u(\nu_G,\mathit{CS})$.
\end{property}

The next property involves the Quotient game:

\begin{property}
The power of a community is the aggregation of the power of nodes comprising this community. Formally: $\phi^{\SEMI}_j(M,\nu_G^Q) = \sum_{v \in C_j} \phi^{\CSEMI}_v(\nu_G,\mathit{CS}).$
\end{property}
All the solutions, where the power inside the communities is computed using the Shapley value (due to the Efficiency), have the above property. In the same spirit, it can required that all the power of the whole network $\nu_G(V)$ is distributed among the nodes:
\begin{property}
The value of the whole network $\nu(V)$ is the aggregation of the power of nodes comprising this network: $\nu_G(V) = \sum_{v \in V}  \phi^{\CSEMI}_v(\nu_G,\mathit{CS})$.
\end{property}
Our final property is the translation of Component Symmetry. If we define the marginal contribution of the coalition $C$ to the set of nodes $Q_T$ as $\text{MC}(Q_T,C)=\nu(Q_T \cup C)-\nu(Q_T)$, we get:

\begin{property}
If two communities contribute the same value to all possible groups of communities then their evaluation is the same:
$\forall_{T \subseteq M \setminus \{i, j\}}\text{\textnormal{MC}}(Q_T,C_i)  =  \text{\textnormal{MC}}(Q_T,C_j)  \implies  \phi^{\SEMI}_i(M,\nu_G^Q)  =  \phi^{\SEMI}_j(M,\nu_G^Q)$.
\end{property}

Table~\ref{table:axioms} summarizes properties of the coalitional semivalues.

\def\arraystretch{1}

\begin{table}
\centering

\caption{The properties of coalitional semivalue and its various instances.}

\label{table:axioms}
\scalebox{1}{
\begin{tabular}{l@{\;\;\;\;}c@{\;\;}c@{\;\;}c@{\;\;}c@{\;\;}c@{\;\;}c}
	\toprule
	{\bf Solution name}	&		$P1$	&	$P2$	&	$P3$&	$P4$&	$P5$&	$P6$		\\\toprule
	
		{\textbf{coalitional semivalue}~\textbf{[this chapter]}}
		&	\checkmark	&	\checkmark & \checkmark	& $\times$	& $\times$	&	$\times$ \\	\toprule 		
		
		{\ \ \ Owen value~\citep{Owen:1977}}
		&	\checkmark	&	\checkmark &\checkmark	& \checkmark	& \checkmark	&	\checkmark  \\	\midrule 			
		
		{\ \ \ Owen-Banzhaf value~\citep{Owen:1982}}
		&	\checkmark	& \checkmark	&	\checkmark & $\times$ 	& $\times$ 	&	$\times$  	 \\		\midrule 		
		{\parbox[c][2.5em][t]{21em}
		{\ \ \ symmetric coalitional  Banzhaf value \\  \hspace*{0.12cm}  \citep{Amer:et:al:2002}}}
		&	\checkmark & \checkmark	&	\checkmark&	\checkmark &	$\times$ &	\checkmark	  \\\midrule
		{\parbox[c][2.5em][t]{21em}
		{\ \ \ symmetric coalitionalp-binomial semivalue \\  \hspace*{0.12cm}  \citep{Carreras:Puente:2012}}}
		&	\checkmark &\checkmark	&\checkmark	&	\checkmark &	$\times$ &	\checkmark 	 \\
	\bottomrule
\end{tabular}
}
\end{table}

\section{Computational Analysis}

\noindent For many succinctly represented coalitional games, computing the Shapley value is NP-hard (in fact, it is often \#P-complete~(see Section~\ref{sec:rw:representation} or \citep{Chalkiadakis:et:al:2011}). Naive algorithms to compute the Shapley value (exhaustively computing the average marginal contribution over all orderings of players) have exponential running time. Given this, there are two possible research directions. Firstly, efficient approximate algorithms can be developed. Secondly, classes of centralities, that have real-life applications and can be computed in polynomial time, can be defined. In this chapter, we take the latter approach and propose a polynomial time algorithm for computing coalitional semivalue-based centralities, where the characteristic function---the value of any group of nodes---is based on their degree. Thus, our method is build upon degree centrality --- an important method of evaluating nodes in social networks analysis \citep{Freeman:1979,Everett:Borgatti:1999,Michalak:et:al:2013}.
\subsection{Weighted Degree Centrality Measure}
\noindent In this subsection we define a class of cooperative games such that a node's value is based on its degree. The \textit{weighted group degree centrality} of a community $C$ in graph $G$ is defined as follows:
\begin{equation}\label{deg:cen}
\psi_{WD}(G) =  \nu^{WD}_{G}(C) = \sum_{ v \in N(C)} f(v),
\end{equation}
\noindent where $N(C)$ is the set of neighbours of $C$, and $f$ is a parameter that is a polynomially computable function.

\begin{definition}[The Coalitional Semivalue-based Degree Centrality]
A game-theoretic network weighted degree centrality for the graph $G$ with the community structure $\mathit{CS}$ is a quadruple $(G, \mathit{CS}, \psi_{WD}, \phi^{\CSEMI})$, where the value of each node $v \in V$ is given by $\phi^{\CSEMI}_v(\nu^{WD}_{G}, \mathit{CS})$.
\end{definition}

In the next section we will look more closely at the marginal contributions of nodes in order to effectively compute their expected value. This, in turn, will let us compute coalitional semivalues based on weighted group degree centrality in polynomial time.

\subsection{The Marginal Contribution Analysis}

\noindent In this subsection we lay the groundwork for the efficient algorithm that will compute coalitional semivalues for weighted group degree centrality in polynomial time. To this end, we will use equation from Definition~\ref{def:coalitional:game}. For a given node, $v \in C_j \in \mathit{CS}$, the focus will be on computing the expected value of its marginal contribution: $\mathbb{E}_{T^k,C^l}[ \text{MC}( Q_{T^k} \cup C^l,v)]$.

We must consider the value of the expected marginal contribution of a node $v$ to the set $Q_{T^k} \cup C^l$, where $T^k$ is a random set of size $k$, and $C^l$ is a random set of size $l$. Both sets are drawn uniformly from the sets $M \setminus \set{j}$ and $C_j \setminus \set{v}$, respectively. We will construct the effective algorithm in two steps. First, we will decide under what conditions $v$ makes a contribution to the set $Q_{T^k} \cup C^l$. Second, we will use a combinatorial argument to compute this contribution for each of the cases distinguished in the first step.

Before we start, we need to introduce the following notation: for the node $v \in C_j \in \mathit{CS}$ we define the set of adjacent communities as $N_{\mathit{CS}}(v) = \set{C_i \in \mathit{CS} \setminus C_j \mid C_i~\cap~N(v) \neq \emptyset}$, inter-community degree as $deg_{\mathit{CS}}(v) = |N_{\mathit{CS}}(v)|$, the set of neighbours within a community as $N_{j}(v) = N(v) \cap C_j$, and intra-community degree as $deg_{j}(v) = |N_{j}(v)|$.

\begin{theorem}\label{theorem:1}
  The game-theoretic network degree centrality for graph $G$ with
  community structure $\mathit{CS}$ for the node $v \in V$ can be
  computed in time polynomial in $|V|$.
\end{theorem}

\begin{proof}

In our proof we will use concepts from probability theory. Thus, we will first define the probability space, which is a triple $(\Omega,\mathcal{F},P)$, where $\Omega$ is a sample space containing sets $\set{Q_{T^k} \cup C^l}$, $T^k \subseteq 2^{|M|-1}$ is such that $|T^k|=k$ and $C^l \subseteq 2^{|C_j|-1}$ is such that $|C^l|=l$. The important observation is that $|\Omega| = \binom{|M|-1}{k}\binom{|C_j|-1}{l}$. In our model, $\mathcal{F}$ is the set of elementary events ($\mathcal{F} = \Omega$), and $P: \mathcal{F} \to [0,1]$ is a probability distribution function such that for each event $A \in \mathcal{F}$ we have $P(A) = \frac{1}{|\Omega|}$.

There are two types of marginal contribution that a node can make. For the first, let us consider the marginal contribution of a single vertex $v$ to the random set $Q_{T^k} \cup C^l$. When $v$ joins a coalition $C$, it can contribute to its value with the help of any vertex $u \in N(v)$ if and only if $u$ is not in $C$ and $u$ is not already directly connected to $C$. Let us introduce the Bernoulli random variable $B^{[1]}_{v,u,k,l}$, which will indicate whether the vertex $v$ makes a contribution through vertex $u$ to the random set $Q_{T^k} \cup C^l$. Equation~\eqref{deg:cen} tells us that this contribution will be $f(u)$. Thus, we have:
\begin{equation}
	 \mathbb{E}[f(u)B^{[1]}_{v,u,k,l}] = f(u)P[ (N(u) \cup \set{u}) \cap (Q_{T^k} \cup C^l) = \emptyset] \nonumber,
\end{equation}

\noindent where $P[\cdot]$ denotes probability, and $\mathbb{E}[\cdot]$ denotes expected value.

The second type of contribution takes place when vertex $v$ joins a coalition $C$ and takes away the value $f(v)$. Such a contribution happens when vertex $v$ is directly connected to the coalition $C$. In particular, weighted group degree centrality $\nu^{D}_{G}$ assumes that the value of a set of vertices depends only on nodes directly connected to this set, ignoring nodes already inside it. Therefore, when the node $v$ becomes a member of $C$ and it is not any more directly connected with it, the value of $C$ is reduced by $f(v)$. Let us introduce the Bernoulli random variable $B^{[2]}_{v,u,k,l}$, which will indicate whether vertex $v$ makes a contribution through itself to the random set $Q_{T^k} \cup C^l$. More formally, we have:
\begin{equation}
	 \mathbb{E}[-f(v)B^{[2]}_{v,k,l}] = -f(v)P[ v \in N(Q_{T^k} \cup C^l)] \nonumber.
\end{equation}
Now, we will move on to the second step of the proof and use a combinatorial argument to compute $P^{[1]}=P[ (N(u) \cup \set{u}) \cap (Q_{T^k} \cup C^l) = \emptyset]$ and $P^{[2]}=P[ v \in N(Q_{T^k} \cup C^l)]$.

Recall that there are exactly $\binom{|M|-1}{k}\binom{|C_j|-1}{l}$ sets $Q_{T^k} \cup C^l$. This is the size of the sample space. With this in mind, the probability $P^{[1]}$ for $u \in N(v)$, if $u,v \in C_j$ can be computed as follows:
\begin{align}
P^{[1.1]} &=\frac{\binom{|M|-1- deg_{\mathit{CS}}(u)}{k}\binom{|C_j|-1-deg_j(u)}{l}}{\binom{|M|-1}{k}\binom{|C_j|-1}{l}} \nonumber,
\end{align}

\noindent otherwise if  $v \in C_j$ and $u \in C_i$ and $i \neq j$ we have:
\begin{align}
P^{[1.2]} &=\frac{\binom{|M|-1- deg_{\mathit{CS}}(u)}{k}\binom{|C_j|-deg_j(u)}{l}}{\binom{|M|-1}{k}\binom{|C_j|-1}{l}} \nonumber
\end{align}

Finally, for $v \in C_j$ and $u \in N(v)$ we obtain:
\renewcommand{\arraystretch}{1}
\begin{equation}\label{contr:one}	
\mathbb{E}[f(u)B^{[1]}_{v,u,k,l}]\hspace{-0.1cm} =\hspace{-0.1cm} \left\{
  \begin{array}{l l}
      0  &  \text{if $u \in C_j$}\\
	     &  \text{and $\big(deg_{\mathit{CS}}(u)\hspace{-0.1cm}>\hspace{-0.1cm}|M|\hspace{-0.1cm}-\hspace{-0.1cm}1$}\\
         &  \text{or $deg_j(u) > |C_j|-1\big)$}\\
    f(u)P^{[1.1]}  &  \text{if $u \in C_j$}\\
            0  &  \text{if $u \notin C_j$} \\
           & \text{and $\big(deg_{\mathit{CS}}(u)\hspace{-0.1cm}>\hspace{-0.1cm}|M|\hspace{-0.1cm}-\hspace{-0.1cm}1$}\\
         & \text{or $deg_j(u) > |C_j|\big)$} \\
   f(u)P^{[1.2]}  &  \text{if $u \notin C_j$ }

  \end{array} \right.	
\end{equation}
\renewcommand{\arraystretch}{1}

In order to compute $P^{[2]}$ we consider a complementary event  $P^{[2]}=(1 - P[N(v) \cap  (Q_{T^k} \cup C^l)) = \emptyset])$ and using the same combinatorial argument as for computing  $P^{[1.1]}$, for $v \in C_j$ we get:
\begin{align}
P^{[2]} &=1 - \frac{\binom{|M|-1- deg_{\mathit{CS}}(v)}{k}\binom{|C_j|-1-deg_j(v)}{l}}{\binom{|M|-1}{k}\binom{|C_j|-1}{l}} \nonumber,
\end{align}

\noindent and consequently we obtain:
\renewcommand{\arraystretch}{1}
\begin{equation}\label{contr:two}	
\mathbb{E}[f(v)B^{[2]}_{v,k,l}]\hspace{-0.1cm} =\hspace{-0.1cm} \left\{
  \begin{array}{l l}
    -f(v)  &  \text{if $\big(deg_{\mathit{CS}}(u)\hspace{-0.1cm}>\hspace{-0.1cm}|M|\hspace{-0.1cm}-\hspace{-0.1cm}1$}\\
         &  \text{or $deg_j(u) > |C_j|-1\big)$} \\
  -f(v)P^{[2]}  & \text{otherwise} .
  \end{array} \right.
\end{equation}
\renewcommand{\arraystretch}{1}

The final formula combines equations \eqref{contr:one} and \eqref{contr:two} :
\begin{align}\label{eq:owen:con}
\mathbb{E}[ \text{MC}( Q_{T^k} \cup C^l,v)] & \hspace{-0.1cm} = \hspace{-0.3cm} \sum_{u \in N(v)} \Big( \mathbb{E}[f(u)B^{[1]}_{v,u,k,l}] \Big) + \mathbb{E}[f(v)B^{[2]}_{v,k,l}] \nonumber \\[-0.1cm]
&\hspace{-1.8cm}= \hspace{-0.5cm} \sum_{u \in N(v) \cap C_j} \hspace{-0.2cm} f(u)\Big(\frac{\binom{|M|-1- deg_{\mathit{CS}}(u)}{k}\binom{|C_j|-1-deg_j(u)}{l}}{\binom{|M|-1}{k}\binom{|C_j|-1}{l}}  \Big) \nonumber \\[-0.1cm]
&\hspace{-1.8cm}+ \hspace{-0.5cm} \sum_{u \in N(v) \setminus   C_j} \hspace{-0.2cm} f(u)\Big( \frac{\binom{|M|-1- deg_{\mathit{CS}}(u)}{k}\binom{|C_j|-deg_j(u)}{l}}{\binom{|M|-1}{k}\binom{|C_j|-1}{l}} \Big) \nonumber \\[-0.1cm]
 & \hspace{-1.8cm}-  f(v)\Big(1-\frac{\binom{|M|-1- deg_{\mathit{CS}}(v)}{k}\binom{|C_j|-1-deg_j(v)}{l}}{\binom{|M|-1}{k}\binom{|C_j|-1}{l}}\Big) .
\end{align}

The above formula can be used to compute $\mathbb{E}_{T^k,C^l}[ \text{MC}( Q_{T^k} \cup C^l,v)]$ in polynomial time. Therefore, the game-theoretic network degree centrality for graph $G$ with community structure $\mathit{CS}$ can be computed in polynomial time using equation from Definition~\ref{def:coalitional:semivalue} which ends our proof.
\end{proof}

\subsection{Algorithms}

\noindent Algorithm \ref{algo:semi} directly implements expression from Definition~\ref{def:coalitional:semivalue}. The expected value operator is computed using the final result of Theorem \ref{theorem:1}: equation \eqref{eq:owen:con}. It computes the game-theoretic network degree centrality for a given graph $G$ with community structure $\mathit{CS}$. For the sake of clarity, we assume in our algorithm that for $a<b$ we have $\binom{a}{b}=0$, and for any $a$ we have $\frac{a}{0}=0$.

%%%%%%%%%%%%%%%%%%%%%%%%%%%%%%%%%%%%%%%%%%%%%%%%%%%%%%%%%%%%%%%%%%%%%%%%%%%%%%%%%%%%%%%%%%%%%%%%%%%%
\RestyleAlgo{ruled}
\begin{algorithm}[h]
\SetAlgoVlined
\LinesNumbered
\caption{The Coalitional Semivalue-based weighted degree centrality}\label{algo:semi}
\KwIn {Graph $G = (V,E)$, node $v \in V$, coalition structure $\mathit{CS}$, functions $\beta$ and family of functions $\set{\alpha}$ }
\KwData {for each vertex $u \in V$ and the community $v \in C_j$:\\
$deg_{\mathit{CS}}(u)$ - the inter-community degree \\
$deg_j(u)$ - the intra-community degree\\
}
\KwOut {$\phi^{\CSEMI}_v$ coalitional semivalue-based degree centrality}
	$\phi^{\CSEMI}_v \gets 0$; \\
	\For{$k \gets 0$ \KwTo $|M|-1$}{  \nllabel{semi:1loop}
		\For{$l \gets 0$ \KwTo $|C_j|-1$}{  \nllabel{semi:2loop}
			$\text{MC}_{k,l} \gets 0$; \\		
			
			% \Comment{******************[ $B^{[1]}$ ]******************} \\
				\ForEach{$u \in N(v) \cap  C_j$} {
					$\text{MC}_{k,l} \hspace{-0.05cm} \gets \hspace{-0.05cm}  \hspace{-0.05cm}  \text{MC}_{k,l} \hspace{-0.03cm} +\hspace{-0.03cm}   \frac{f(u)\hspace{-0.05cm} \binom{|M|\hspace{-0.04cm}-\hspace{-0.04cm}1\hspace{-0.04cm}-\hspace{-0.04cm} deg_{\mathit{CS}}(u)}{k}\hspace{-0.05cm}  \binom{|C_j|\hspace{-0.04cm}-\hspace{-0.04cm}1\hspace{-0.04cm}-\hspace{-0.04cm}deg_j(u)}{l}}{\binom{|M|-1}{k}\binom{|C_j|-1}{l}}$
					
				}
				\ForEach{$u \in N(v) \setminus  C_j$} {
										$\text{MC}_{k,l} \hspace{-0.05cm} \gets \hspace{-0.05cm} \text{MC}_{k,l} +  \frac{f(u)\hspace{-0.05cm}\binom{|M|\hspace{-0.04cm}-\hspace{-0.04cm}1\hspace{-0.04cm}-\hspace{-0.04cm} deg_{\mathit{CS}}(u)}{k} \hspace{-0.05cm}\binom{|C_j|\hspace{-0.04cm}-\hspace{-0.04cm}deg_j(u)}{l}}{\binom{|M|-1}{k}\binom{|C_j|-1}{l}}$;\\
										
				}
			%\Comment{******************[ $B^{[2]}$ ]******************} \\
			$\text{MC}_{k,l}\hspace{-0.04cm} \gets \text{MC}_{k,l} -  f(v)$;\\
		$\text{MC}_{k,l}\hspace{-0.04cm} \gets \text{MC}_{k,l}\hspace{-0.04cm} +\hspace{-0.04cm}  \frac{f(v)\binom{|M|-1- deg_{\mathit{CS}}(v)}{k}\hspace{-0.04cm}\binom{|C_j|-1-deg_j(v)}{l}}{\binom{|M|-1}{k}\binom{|C_j|-1}{l}}$;\\
		$\phi^{\CSEMI}_v \gets \phi^{\CSEMI}_v + \beta(k)\alpha_j(l) \text{MC}_{k,l}$;
		}

	}
\end{algorithm}
%%%%%%%%%%%%%%%%%%%%%%%%%%%%%%%%%%%%%%%%%%%%%%%%%%%%%%%%%%%%%%%%%%%%%%%%%%%%%%%%%%%%%%%%%%%%%%%%%%%%

This algorithm requires some precomputations. For each node $u \in V$
we need to calculate $deg_{\mathit{CS}}(v)$ and $deg_j(v)$. We can
store these values using $O(|V|)$ space. Provided that it is possible
to check the community of a given node in constant time, we can
perform these precomputations in $O(|V|+|E|)$ time. In the worst case,
the main algorithm works in $O(|V|^3)$ time.

%%%%%%%%%%%%%%%%%%%%%%%%%%%%%%%%%%%%%%%%%%%%%%%%%%%%%%%%%%%%%%%%%%%%%%%%%%%%%%%%%%%%%%%%%%%%%%%%%%%%
\RestyleAlgo{ruled}
\begin{algorithm}[h]
\SetAlgoVlined
\LinesNumbered
\caption{The Owen value-based weighted degree centrality}\label{algo:OV}
\KwIn {Graph $G = (V,E)$, node $v \in V$, coalition structure $\mathit{CS}$, functions $\beta$ and family of functions $\set{\alpha}$ }
\KwData {for each vertex $u \in V$ and the community $v \in C_j$:\\
$deg_{\mathit{CS}}(u)$ - the inter-community degree \\
$deg_j(u)$ - the intra-community degree\\
}
\KwOut {$\phi^{\CSEMI}_v$ coalitional semivalue-based degree centrality}
	$\phi^{\CSEMI}_v \gets 0$; \\

			$\text{MC}\gets 0$; \\		
			
			% \Comment{******************[ $B^{[1]}$ ]******************} \\
				\ForEach{$u \in N(v) \cap  C_j$} {
				$\text{MC} \hspace{-0.05cm} \gets \hspace{-0.05cm}  \hspace{-0.05cm}  \text{MC} \hspace{-0.03cm} +\hspace{-0.03cm}   \frac{f(u)}{(1+ deg_{\mathit{CS}}(u))(1+deg_j(u))}$							
				}
				\ForEach{$u \in N(v) \setminus  C_j$} {
				$\text{MC} \hspace{-0.05cm} \gets \hspace{-0.05cm}  \hspace{-0.05cm}  \text{MC} \hspace{-0.03cm} +\hspace{-0.03cm}   \frac{f(u)}{(1+deg_{\mathit{CS}}(u))deg_j(u)}$	
										
				}
			%\Comment{******************[ $B^{[2]}$ ]******************} \\		
	$\text{MC}\hspace{-0.04cm} \gets \text{MC} + \frac{f(u)(1 - (1+deg_j(v))(1+deg_{\mathit{CS}}(u))}{ (1+deg_j(v)))(1+deg_{\mathit{CS}}(u))}$;\\	

		$\phi^{\CSEMI}_v \gets \phi^{\CSEMI}_v + \text{MC}$;

\end{algorithm}
%%%%%%%%%%%%%%%%%%%%%%%%%%%%%%%%%%%%%%%%%%%%%%%%%%%%%%%%%%%%%%%%%%%%%%%%%%%%%%%%%%%%%%%%%%%%%%%%%%%%

Our next observation is that for trivial coalition structures (such as $ \mathit{CS} = \{A\}$, or $ \mathit{CS} = \{ \{a_1\}, \{a_2\}, \ldots, \{a_n\} \}$) our algorithm computes any weighted degree-based semivalue in $O(|V|^2)$ time. Finally, we would like to note that this algorithm is easily adapted to directed networks. To this end, depending on the new definition of weighted group degree, we need to replace all instances of $deg_{\mathit{CS}}(u)$  and $deg_j(u)$ with their counterparts for directed networks: \emph{in} or \emph{out} degree.

 Algorithm~\ref{algo:OV} is the optimized version of  Algorithm~\ref{algo:semi} to compute Owen value for weighted degree centrality.

\section{Simulations}\label{section:simulation}

\noindent The main aim of this experiment is to compare three rankings created by three different methods: (i) one that uses weighted degree centrality and evaluates each node $v$ by the number of neighbours it has (we denote it by $\nu^{WD}_G(\{v\})$); (ii) one with the Shapley value-based degree centrality (denoted $\phi^{\SV}_v$); and, (iii) one with the Owen value-based degree centrality (denoted $\phi^{\OV}_v$), which evaluates nodes in the context of the communities they belong to and their respective power. Thus, the first two methods do not account for the existence of the community structure while the third one does.

%We will analyse the real-life data with game-theoretical centralities build upon the  group degree centrality  ----  the characteristic function that is a very important concept for social network analysis.

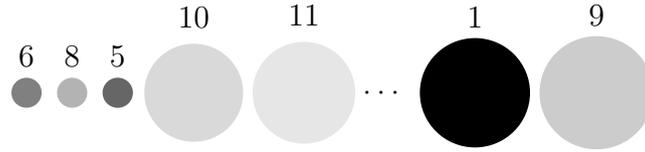
\begin{figure}[t]\centering
\begin{tikzpicture}[node distance = 2cm]
\tikzstyle{every node}=[draw,shape=circle,minimum size=1em];

\path (0, 0cm) node[label={[label distance=-0.1cm]90:$6$},fill=black!50,draw=none,minimum size=0.3cm] (r16) {};
\path (0.6, 0cm) node[label={[label distance=-0.1cm]90:$8$},fill=black!30,draw=none,minimum size=0.3cm] (r18) {};
\path (1.2, 0cm) node[label={[label distance=-0.1cm]90:$5$},fill=black!60,draw=none,minimum size=0.4cm] (r15) {};

\path (2.2,0cm) node[label={[label distance=-0.1cm]90:$10$},fill=black!15,draw=none,minimum size=1.3cm] (r15) {};
\path (3.65,0cm) node[label={[label distance=-0.1cm]90:$11$},fill=black!10,draw=none,minimum size=1.35cm] (r15) {};
\path (4.7,0cm)  node[draw=none] (d) {$\cdots$};
\path (5.9,0cm) node[label={[label distance=-0.1cm]90:$1$},fill=black!100,draw=none,minimum size=1.45cm] (r15) {};
\path (7.5,0cm) node[label={[label distance=-0.1cm]90:$9$},fill=black!20,draw=none,minimum size=1.5cm] (r15) {};

\end{tikzpicture}
\caption{The relative power of communities for the first top nodes from the $\nu^{WD}_{G}(\{v\})$ ranking. The power of the communities of nodes $5$, $6$ and $8$ is significantly smaller than the power of communities of the other top nodes. }
\label{fig:communities}
\end{figure}

The real-life network used for simulations is a citation network that consists of $2,084,055$ publications and $2,244,018$ citation relationships.\footnote{\footnotesize The database used for these experiments is available under the following link: http://arnetminer.org/citation.} This dataset is a list of publications with basic attributes (such as: title, authors, venue, or citations), and it is part of the project \emph{ArnetMiner}~being under development by \cite{Tang:et:al:2008}. All publications extracted from this dataset were categorized into $22954$ unique communities representing journals, conference proceedings or single book titles using basic text mining techniques. These communities can be interpreted as scientific groups united under the same topics of interests. In our experiment we use the directed version of our algorithm and assume that $f(v) = \frac{1}{\#\text{numer of articles citing}~v}$. The Shapley value-based centrality (the second method) is computed using the polynomial time algorithm introduced by \cite{Michalak:et:al:2013}. The Owen value-based centrality is computed with the modification of Algorithm~\ref{algo:semi}, in which thanks to the form of the $\alpha$ and $\beta$ (in Owen value these discrete probabilities are uniform) the complexity was reduced to $O(|V|+|E|)$.

In what follows we focus on the 11 top nodes from the basic ranking $\nu^{WD}_{G}(\{v\})$. Figure~\ref{fig:communities} shows the relative power of the communities to which these nodes belong. Nodes indexed 5, 6 and 8 belong to significantly less powerful communities than nodes 1, 2, 3, 4, 7, 9, 10 and 11.

\begin{figure}[t]\centering
\begin{tikzpicture}[node distance = 2cm]
\tikzstyle{every node}=[draw,shape=circle,minimum size=1em];

\path (0,0cm) node[draw=none] (r1) {$\nu^{WD}_{G}(\set{v})$};

\path (2,0cm) node[draw=none] (r2) {$\phi^{\SV}_v$};

\path (4,0cm) node[draw=none] (r2) {$\phi^{\OV}_v$};

\path (-1.5,-0.5cm) node[draw=none] (1) {$1.$};
\path (-1.5,-1cm) node[draw=none] (2) {$2.$};
\path (-1.5,-1.5cm) node[draw=none] (3) {$3.$};
\path (-1.5,-2cm) node[draw=none] (4) {$4.$};
\path (-1.5,-2.5cm) node[draw=none] (5) {$5.$};
\path (-1.5,-3cm) node[draw=none] (6) {$6.$};
\path (-1.5,-3.5cm) node[draw=none] (7) {$7.$};
\path (-1.5,-4cm) node[draw=none] (8) {$8.$};
\path (-1.5,-4.5cm) node[draw=none] (9) {$9.$};
\path (-1.5,-5cm) node[draw=none] (10) {$10.$};
\path (-1.5,-5.5cm) node[draw=none] (11) {$11.$};
\path (-1.5,-6cm) node[draw=none] (12) {$12.$};
\path (-1.5,-6.5cm) node[draw=none] (13) {$13.$};
\path (-1.5,-6.9cm) node[draw=none] (d) {$\vdots$};
\path (-1.5,-7.5cm) node[draw=none] (20) {$20.$};
\path (-1.5,-7.9cm) node[draw=none] (d) {$\vdots$};
\path (-1.5,-8.5cm) node[draw=none] (1000) {$>1000.$};

\path (0,-0.5cm) node[fill=black!100,draw=none] (r11) {};
\path (0,-1cm) node[fill=black!90,draw=none] (r12) {};
\path (0,-1.5cm) node[fill=black!80,draw=none] (r13) {};
\path (0,-2cm) node[fill=black!70,draw=none] (r14) {};
\path (0,-2.5cm) node[fill=black!60,draw=none] (r15) {};
\path (0,-3cm) node[fill=black!50,draw=none] (r16) {};
\path (0,-3.5cm) node[fill=black!40,draw=none] (r17) {};
\path (0,-4cm) node[fill=black!30,draw=none] (r18) {};
\path (0,-4.5cm) node[fill=black!20,draw=none] (r19) {};
\path (0,-5cm) node[fill=black!15,draw=none] (r110) {};
\path (0,-5.5cm) node[fill=black!10,draw=none] (r111) {};
\path (0,-6cm) node[draw=none] (r112) {};
\path (0,-6.5cm) node[draw=none] (r113) {};

\path (2,-0.5cm) node[fill=black!100,draw=none] (r21) {};
\path (2,-1.5cm) node[fill=black!90,draw=none] (r22) {};
\path (2,-1cm) node[fill=black!80,draw=none] (r23) {};
\path (2,-2cm) node[fill=black!70,draw=none] (r24) {};
\path (2,-4cm) node[fill=black!60,draw=none] (r25) {};
\path (2,-3.5cm) node[fill=black!50,draw=none] (r26) {};
\path (2,-3cm) node[fill=black!40,draw=none] (r27) {};
\path (2,-7.5cm) node[fill=black!30,draw=none] (r28) {};
\path (2,-6cm) node[fill=black!20,draw=none] (r29) {};
\path (2,-5cm) node[fill=black!15,draw=none] (r210) {};
\path (2,-2.5cm) node[fill=black!10,draw=none] (r211) {};

\path (4,-0.5cm) node[fill=black!100,draw=none] (r31) {};
\path (4,-1.5cm) node[fill=black!90,draw=none] (r32) {};
\path (4,-1cm) node[fill=black!80,draw=none] (r33) {};
\path (4,-2cm) node[fill=black!70,draw=none] (r34) {};
\path (4,-8.5cm) node[fill=black!60,draw=none]   (r35) {};
\path (4.4,-8.5cm) node[fill=black!50,draw=none]   (r36) {};
\path (4,-4cm) node[fill=black!40,draw=none] (r37) {};
\path (3.6,-8.5cm) node[fill=black!30,draw=none] (r38) {};
\path (4,-6.5cm) node[fill=black!20,draw=none] (r39) {};
\path (4,-3cm) node[fill=black!15,draw=none] (r310) {};
\path (4,-2.5cm) node[fill=black!10,draw=none] (r311) {};

\draw[solid,black!100] (r11) -- (r21)  -- (r31);
\draw[solid,black!90] (r12) -- (r22)  -- (r32);
\draw[solid,black!80] (r13) -- (r23)  -- (r33);
\draw[solid,black!70] (r14) -- (r24)  -- (r34);
\draw[solid,black!60] (r15) -- (r25)  -- (r35);
\draw[solid,black!50] (r16) -- (r26)  -- (r36);
\draw[solid,black!40] (r17) -- (r27)  -- (r37);
\draw[solid,black!40] (r18) -- (r28)  -- (r38);
\draw[solid,black!20] (r19) -- (r29)  -- (r39);
\draw[solid,black!15] (r110) -- (r210)  -- (r310);
\draw[solid,black!10] (r111) -- (r211)  -- (r311);

\draw[black,thick,dotted] ($(r38.north west)+(-0.2,0.3)$)  rectangle ($(r36.south east)+(0.2,-0.3)$);

\end{tikzpicture}
\caption{Three rankings of the top nodes. The $OV$ ranking radically decreases the positions of the nodes $5$, $6$ and $8$.  }
\label{fig:rankings}
\end{figure}
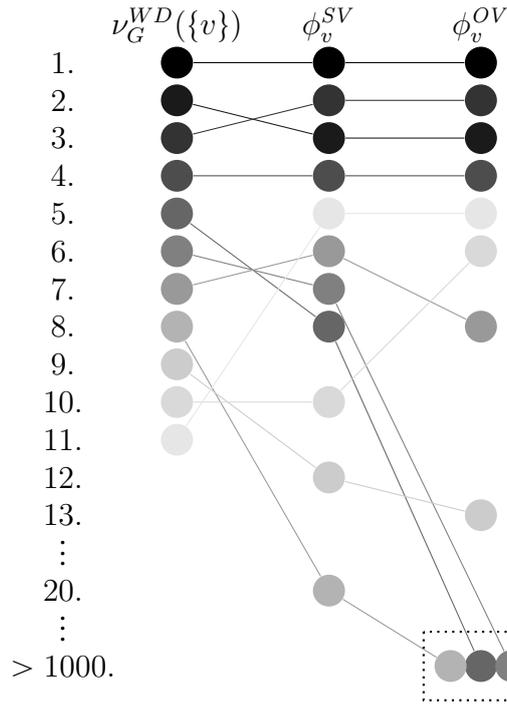

Figure~\ref{fig:rankings} shows how the position of top nodes selected using $\nu^D_G(\{v\})$ changes in the $\phi^{\SV}_v$ and $\phi^{\OV}_v$ rankings. While for most nodes the perturbations are not so intensive, we observe significant downgrade of the position of nodes 5, 6 and 8 in the $OV_v$ ranking. This demonstrates coalitional semivalues-based centrality (in this case the Owen value-based centrality) is able to recognize that these three nodes belong to much weaker communities.

\def\arraystretch{0.8}

\begin{table}
\centering

\caption{The values of different coalitional semivalues.}
\label{table:zachary}
\scalebox{1}{
\begin{tabular}{l@{\;\;\;\;}c@{\;\;}c@{\;\;}c@{\;\;}c@{\;\;}c@{\;\;}c}
	\toprule
	{\bf Solution name} / {\bf Nodes's degree}	&		$17$	&	$16$	&	$11$& $10$	&	$9$	\\\toprule		
		{\ \ \ Owen value}
		&	$3.51$	&	$2.68$ & $1.47$ 	& $1.37$ 	& $0.70$ 	  \\	\midrule 			
		
		{\ \ \ Owen-Banzhaf value}
		&	$2.28$	& $1.38$	&	$0.47$  & $0.88$	& $0.01$ 	 	 \\		\midrule 		
		{\parbox[c][2.5em][t]{11em}
		{\ \ \ symmetric coalitional \\ \hspace*{0.12cm}  Banzhaf value}}
		&	$3.51$ & $2.68$	&	$1.47$ &	$1.37$ &	$0.70$	  \\\midrule
		{\parbox[c][2.5em][t]{14em}
		{\ \ \ symmetric coalitional\\ \hspace*{0.12cm} p-binomial semivalue~($p=\frac{1}{4}$)}}
		&	$4.38$ & $3.15$	& $1.65$	&	$2.38$ &	$1.04$  	 \\
	\bottomrule
\end{tabular}
}
\end{table}

The rankings of nodes may differ depending which coalitional semivalue we choose. To illustrate this fact we evaluated top $5$ nodes with the highest degree centrality from Zachary Karate Club Network~\citep{Zachary:1977}. This network consists of $34$ nodes divided into two communities. We observe in Table~\ref{table:zachary} that the ranking created with the Owen value differs with the one created with Owen-Banzhaf value at the $3^{\text{rd}}$ and $4^{\text{th}}$ positions.

\section{Conclusions}\label{chap:OV:section:summary}

\noindent The centrality metric proposed in this chapter is the first tool that
evaluates individual nodes \emph{in the context of their
  communities}. This metric is based on the Owen value---a well-known
concept from coalitional game theory that we generalize by introducing
coalitional semivalues. Our experiments show that the rankings can
significantly differ if we account for the power of the relevant
communities that the nodes belong to. If the community of a node is
weak, it can significantly weaken the position of the node in the
ranking based on the coalitional semivalue. It also demonstrates that
our polynomial time algorithm is applicable to large data sets.

\chapter{Computational Analysis of a Connectivity Game on Networks}\label{chap:connectivity}

\noindent We study a recently developed centrality metric to identify key players in terrorist organisations due to \cite{Lindelauf:et:al:2013}. This metric, which involves computation of the Shapley value for \textit{connectivity games on graphs} proposed by \cite{Amer:Gimenez:2004}, was shown to produce substantially better results than previously used standard centralities. In this chapter, we present the first computational analysis of this class of coalitional games, and propose two algorithms for computing Lindelauf \textit{et al.}'s centrality metric. Our first algorithm is exact, and runs in time linear by number of connected subgraphs in the network. As shown in the numerical simulations, our algorithm identifies key players in the WTC 9/11 terrorist network, constructed of 36 members and 125 links, in less than 40 minutes. In contrast, a general-purpose Shapley value algorithm would require weeks to solve this problem. Our second algorithm is approximate and can be successfully used to study much larger networks.

The contribution of the author of this dissertation covers the complexity analysis and developing exact and approximate algorithms for computing Lindelauf \textit{et al.}'s centrality metric.

%This chapter  is organized as follows. In Section \ref{chap:ter:intro}, we introduce the subject of covert networks. Section~\ref{chap:ter:related} describes the related work. The connectivity games and Lindelauf's centrality are introduced in Section~\ref{chap:ter:basic}. Section \ref{chap:ter:algos} provides complexity analysis and introduces two algorithms to compute Shapley value for connectivity games. Finally, Section~\ref{chap:ter:simulations} provides extensive empirical evaluation on two real-life terrorists networks.

\section{Introduction}\label{chap:ter:intro}

\noindent Despite enormous efforts to tackle terrorism in the aftermath of 9/11, many terrorist networks are growing in size \citep{Shultz:et:al:2009}. For instance, the number of core members of Al'Qaeda, arguably the most infamous terrorist group, is currently estimated to exceed 800, with their operations spreading over more than 20 countries.

Facing, on the one hand, the increased size of terrorist groups and, on the other hand, inevitable budget cuts, security agencies urgently require efficient techniques to identify who plays the most important role within a terrorist network and, therefore, where scarce resources should predominantly be focused. In this context, a number of authors have proposed using \emph{social network analysis} to investigate terrorist organisations (\emph{e.g.},~\cite{Carley:et:al:2003,Krebs:2002,Farley:2003,Lindelauf:et:al:2013,Ressler:2006}).

Now, the key strength of social network analysis is its bottom-up approach in that the structure and the functioning of a group is gradually revealed by considering pairwise relationships between the individuals who create the group (\emph{cfr}. \cite{Ressler:2006}). In this context the identification of the key individuals within the network is called a \emph{centrality analysis}. Analyst's Notebook 8 \citep{I2:2010}---a software package used worldwide by law enforcement and intelligence agencies---has recently included standard centrality metrics for networks (\textit{graphs}), such as degree, closeness and betweenness centralities \citep{Brandes:Erlebach:2005,Friedkin:1991}. But the usefulness of these metrics for terrorist networks is limited as they are often unable to capture the complex nature of these organisations \citep{Lindelauf:et:al:2013}.

Recently, to address these shortcomings, \cite{Lindelauf:et:al:2013}  developed a more advanced method specifically designed to measure centrality in terrorist networks. The new method belongs to a class of so-called \emph{Shapley value-based centrality} (see Definition~\ref{def:gt:cen}) and builds upon the notion of coalitional \emph{connectivity games} proposed by \cite{Amer:Gimenez:2004}. These games are defined on graphs, where a coalition of nodes gets a non-zero value if and only if the induced sub-graph is connected. Such a game has an appealing interpretation in our context: it reflects communication capabilities among various groups of terrorists within the network. Indeed, when applied to the terrorist networks responsible for the 9/11 WTC and 2002 Bali attacks (with 23 and 17 members, respectively), Lindelauf \textit{et al.}~showed that their method produces qualitatively better results than standard centrality metrics.

Unfortunately, the computational aspects of connectivity games by Amer and Gimenez have not been studied to date. This means that the current use of Lindelauf \textit{et al.}'s method is limited only to small terrorist networks (of 25 members or so) because general-purpose algorithms for coalitional games have to be applied, which exhaustively search the space of all possible coalitions. Thus, they are inapplicable to many real-world applications such as the terrorist networks responsible for the WTC 9/11 attack (from 35 to 63 nodes depending on the considered type of links between terrorist).% or train bombings in Madrid (70 nodes).

Against this background, we provide in this chapter the first computational analysis of connectivity games proposed by Amer and Gimenez:

\begin{itemize}
\item We prove that computing the Shapley value in connectivity games---including the centrality metrics of Lindelauf~\textit{et al.}---is \mbox{\#P-Hard}.
\item We propose a \textit{dedicated exact algorithm} for computing these centrality metrics. While the general-purpose Shapley value algorithm requires checking all subsets of vertices in the graph, our algorithm traverses through (most often) much smaller number of connected subgraphs. It also has minimal memory requirements.
\item We test our algorithm by analysing the aforementioned WTC 9/11 terrorist network with 36 members and 125 identified connections. In this setting, our algorithm returns the solution within 38 minutes, compared to weeks if a general-purpose approach was applied.
\item In order to study even bigger networks, we propose a \textit{dedicated approximate algorithm} based on Monte Carlo sampling. By comparing to our exact algorithm, we show that, after a reasonable number of iterations, the approximate algorithm yields a very accurate ranking of top nodes based on the centrality of Lindelauf \textit{et al}.
\end{itemize}

%Next, we present how connectivity games can be used for terrorist network analysis.

%The remainder of the paper is organized as follows. The Shapley value-based centrality for terrorist networks is presented in Section \ref{Section:ConnectivityGames}. In Section \ref{section:algorithm} we propose our algorithm and analyse its properties. The numerical results are presented in Section \ref{Section:simulations}. In Section \ref{section:related_work} we discuss related work. Conclusions follow.
%\vspace{-0.1cm}

\section{Related Work}\label{section:related_work}\label{chap:ter:related}

\noindent A rapidly growing body of work is directed to the analysis of terrorist organisations using the methods of social network analysis. A very good introduction to this line of research can be found in \citet{Ressler:2006}. Also worthy of note is the work of \cite{Farley:2003}, \cite{Carley:et:al:2003}, and \cite{Husslage:et:al:2012}, who conduct quantitative analysis of the terrorist networks.

Since the work of \cite{Grofman:Owen:1982}, a number of game-theoretic centrality measures have been developed either to enrich the existing well-known centralities or as completely new ones (see Chapter~\ref{chap:related}). In the terrorist network context, \cite{Lindelauf:et:al:2009a} and \cite{Lindelauf:et:al:2009b} employed the game-theoretic approach to analyze covert networks.

The hardness result presented in this chapter is consistent with other studies of the complexity of the Shapley value in various settings. For instance, computing the Shapley value was shown to be \#P-Complete for weighted majority games \citep{Deng:Papdimitriou:1994} and in minimum spanning tree games \citep{Nagamochi:et:al:1997}. \cite{Aziz:et:al:2009} obtained negative results for a related problem of computing the Shapley-Shubik power index for the spanning connectivity games that are based on undirected, unweighted multigraphs. Also, \cite{Bachrach:et:al:2008} showed that the computation of the Banzhaf index for connectivity games, in which agents own vertices and control adjacent edges and aim to become connected to the certain set of primary edges, is \#P-Complete.  A comprehensive review of these issues, including some positive results for certain settings, can be found in \citep{Chalkiadakis:et:al:2011}.

Finally, it should be mentioned that other types of connectivity games have been considered in the literature. These include vertex connectivity games proposed by \cite{Bachrach:et:al:2008} and the spanning connectivity games proposed by \cite{Aziz:et:al:2009}. The common denominator of these games is that they are interested in maintaining the connectivity between certain set of nodes and they further study the ability of each node to affect the outcome of these games. In particular, the concepts like Banzhaf power index and Shapley-Shubik power index are utilized to measure the influence of nodes \citep{Bachrach:et:al:2008,Aziz:et:al:2009}.

In the already classic work, \cite{Krebs:2002} applied standard centrality measures to determine the key players in the 9/11 terrorist network. \cite{Memon:et:al:2008} developed an algorithm based on the two well known centrality measures from social network analysis to automatically detect the hidden hierarchy in terrorist networks. Recently, \cite{Xu:et:al:2009} explored the use of social network analysis methods to analyze terrorist networks mostly focusing on assigning roles to actors in the network. Most of the above mentioned work in the literature focus on identifying key players in the terrorist networks. A complementary approach is proposed by \cite{Henke:2009} that combines aspects of traditional social network analysis with a novel multi-agent framework that describes how terrorist groups survive despite the aggressive counterterrorist operations.

\section{Connectivity Games for Terrorist Networks}\label{Section:ConnectivityGames}\label{chap:ter:basic}

\noindent Terrorist networks have been recently modeled using a weighted graph, $G$ (Definition~\ref{def:graph:weighted}), composed of individual terrorists and labeled edges \citep{Carley:et:al:2003,Krebs:2002,Farley:2003,Lindelauf:et:al:2013,Ressler:2006}. Based on available intelligence, an edge represents, for instance, a communication link between two terrorists, and the weight of the edge represents the frequency with which that link is used. Weights can be associated not only with edges but also with vertices. This, as argued by Lindelauf \textit{et al.}, allows for modeling additional information that intelligence agencies gather on individuals within the network, such as access to weapons, financial means, previous experience, and participation in a terrorist training camp. We will denote the  weight of vertex $v_i \in V(G)$ as $\gamma(i) \in \Gamma(G)$.

\begin{definition}[Vertex-weighted graph]\label{def:graph:vertex:weighted}
The vertex-weighted graph $ G =(V,E,\gamma) $ is a simple/directed/weighted graph, with a function  $ \gamma(v) : V \to \mathbb{R} $ that evaluates each node. The value $ \gamma(v) $ is called the weight of a node $v$.
\end{definition}

To address the limitations of standard centrality metrics, Lindelauf~\textit{et al.} proposed a new metric that builds upon \emph{connectivity games} by \citet{Amer:Gimenez:2004}. In these coalitional games on graphs, all subsets of set $V(G)$ are considered to be possible coalitions of terrorists. Apart from the empty set, every single coalition is classified as belonging to either the set of \emph{connected coalitions} (denoted $\mC(G)$) or \emph{disconnected coalitions} (denoted $\tilde{\mC}(G)$). We say that $C$ is connected if between any two nodes in $C$ there exists at least one path of which all nodes belong to $C$. Otherwise $C$ is disconnected. Importantly, any two terrorists in a connected coalition are able to communicate with each other (\textit{via} a path), whereas in a disconnected coalition this is not the case. In many connected coalitions, from this point of view, some nodes play more important role than others, as their removal makes a coalition disconnected. We will call them \emph{pivotal}.
\begin{definition}[Pivotal node]
Given connected coalition $C \in \mC(G)$, a node $v_i \in C$ is \textbf{pivotal} to $C$ iff $C \setminus \{v_i\} \not\in \mC(G)$.
\end{definition}

\noindent If such node is removed, the coalition becomes disconnected. In our context, communication between all terrorists within the subgroup becomes impossible.

\begin{definition}[Connectivity game]
The connectivity game is a coalitional game defined on graph $(V,\nu_{\mC})$, where the characteristic function is defined:
\[
\nu_{\mC}(C) =
\begin{cases}
1 & \textnormal{if $C \in \mC(G)$} \\
0 & \textnormal{ otherwise. }  \\
\end{cases}
%\vspace{-0.08cm}
\].
\end{definition}

In such a game defined by \cite{Amer:Gimenez:2004} every connected coalition is evaluated to $1$, and disconnected coalitions evaluated to $0$. Lindelauf~\textit{et al.} extend this definition by assuming that values of connected coalitions may depend on the network in a variety of ways, \emph{i.e.}, they can be a function of adjacent edges or nodes, their weights, \emph{etc.} More formally:

\begin{definition}[Lindelauf centrality]\label{def:lin:cen}
The Lindelauf centrality is a game-theoretic centrality $(\psi_f, \phi^{\SV})$, where the representative function $\psi_f$ is defined:
\[
\psi_f(G) = \nu_f(C) =
\begin{cases}
f(C,G) & \textnormal{if $C \in \mC(G)$} \\
0 & \textnormal{ otherwise. }  \\
\end{cases}
\].
\end{definition}

\noindent The exact definition of $f$ depends on the availability of information and analytical needs. For instance, to analyse the Jemaah Islamiyah network responsible for the 2002 Bali attack in Indonesia, Lindelauf~\emph{et al.} assume:

\begin{equation}
\label{equation:f1}
f(C)= \frac{|E(C)|}{\sum_{(v_i,v_j) \in E(C)\lambda(v_i,v_j)}},
\end{equation}
\noindent that is $f$ equals the number of edges in the connected coalition $C$ (denoted by $E(C)$) divided by their weight.

Lindelauf~\emph{et al.} show that the centrality ranking based on the Shapley value of the connectivity game is much more effective than degree, closeness and betweenness centralities in exposing the key players in the Bali attack. In particular, Azahari bin Husin---the network bomb expert who was considered the ``brain'' behind the entire operation---is ranked low by standard centralities but according to Lindelauf~\emph{et al.} metric is among top five Bali terrorists. Similarly, Feri (Isa)---again ranked low by standard centralities---was, in fact, the suicide bomber. Lindelauf~\emph{et al.} ranked him third.

For big networks networks, the only feasible approach to compute Shapley value currently outlined in the literature is Monte-Carlo sampling (see Section~\ref{chap:gtc:related:approxSV}). However, this method is not only inexact, but can be also very time-consuming. For instance, as shown in our simulations, for a weighted network of about 16,000 nodes and about 120,000 edges, the Monte Carlo approach has to iterate $300,000$ times through the entire network to produce the approximation of the Shapley value with a $40\%$ error margin.\footnote{See Section~\ref{chap:ter:simulations} for the exact definition of the error margin.} Exponentially more iterations are needed to further reduce this error margin.

\section{Computational Analysis \& Algorithms}\label{section:algorithm}\label{chap:ter:algos}

\noindent In this section, we first discuss the complexity of computing the centrality metrics of Lindelauf~\textit{et al.}. We then present our exact and approximate algorithms.
%\vspace{-0.1cm}

%%%%%%%%%%%%%%%%%%%%%%%%%%%%%%%%%%%%%%%%%%%%%%%%%%%%%%%%%%%%%%%%%%%%%%%%%%%%%%%%%%%%%%%%%%%%%%%%%%%%%%%%%%%%%%%%%%%%%%%
%\vspace*{1ex}\noindent\textbf{xxx:}

\subsection{Complexity}
%\vspace{-0.03cm}
\noindent First we show that, even for the simplest definition of the connectivity game, where $\forall_{C\in\mC(G)} f(C,G) = 1$, computing the Shapley value in an efficient way is impossible. The main problem of interest is as follows:
%Our problem is formally defined as follows:
% even in \emph{simple connectivity games}, where $\forall_{C \in\mC(G)} f(C) = 1$,
%we focus on the connectivity game originally defined by Amer and Gimenez, where $\forall_{C \in\mC(G)} f(C) = 1$. We prove that even for this simple case, computing the SV in an efficient way is impossible. We first formally define the problem.
%\vspace{-0.15cm}
\begin{definition}[\textsc{\#CG-Shapley}]
 Given a connectivity game on graph $G$, where $\forall_{C \in\mC(G)} f(C,G) = 1$, we are asked to compute the Shapley value for each node in this graph.
%$\mathit{\#CG\textrm{-}SHAPLEY}$: Given a connectivity game on a graph $G$ defined by Amer and Gimenez, we are asked to compute the SV for each node in this graph.
\end{definition}
%\vspace{-0.1cm}
%Before we start our proof we need to define some initial combinatorial problems.
%In what follows we will use the hardness result for the following problem:% from the literature regarding the problem of counting the number of connected spanning subgraphs. Formally:
In the first step, let us introduce the following problem:
%\vspace{-0.15cm}
\begin{definition}[\textsc{\#Connected-Spanning-Sub (\#CSS)}]
 Given a graph $G$, we are asked to compute the number of connected spanning subgraphs in~$G$.\footnote{\footnotesize A connected subgraph $F$ of the given graph, $G=(V,E)$, is called \emph{spanning} if it contains all nodes from $G$, \emph{i.e.}, $V(F)=V(G)$.}
\end{definition}
%\vspace{-0.1cm}
This problem is \#P-Complete even for bipartite and planar graphs \citep{Welsh:1997}. We will use this hardness result to prove \#P-Completeness of the following problem:
%\vspace{-0.05cm}
\begin{definition}[\textsc{\#Connected-Induced-Sub (\#CIS)}]Given a graph $G$, we are asked to compute the number of connected induced subgraphs in~$G$.
\end{definition}
%\vspace{-0.05cm}
\begin{theorem}
\label{first:theorem}
\textsc{\#Connected-Induced-Sub} is \#P-Complete.
\end{theorem}
\noindent \textit{Proof of Theorem~1} We note first that it is possible to check in polynomial time if a given subgraph is induced and connected. Thus, since a \emph{witness} can be verified in polynomial time, this problem is in \#P. Now, we will reduce a \textsc{\#CSS} instance to \textsc{\#CIS}. To this end, given a graph $G=(V,E)$, we will construct a transformed graph $G'$ and show that determining the number of connected induced subgraphs in $G'$ allows us to easily compute the number of connected spanning subgraphs in~$G$.

Our transformed graph G' is constructed by adding to each edge in G an extra node. Then, to each node from original graph G we attach a clique $K_A$ with $A$ nodes. This reduction is shown in Figure~\ref{fig:reduction}. More formally we define the following graph $G'$:
%\vspace{-0.2cm}
\begin{align}
V(G') = & V(G) \cup \{v_i: v\in V(G) \wedge i \in \{1,\ldots,A\}\} ~\cup \nonumber \\
 &  \{vu : (v,u) \in E(G)\} \nonumber \\
E(G') = & \{(v,uv) : (u,v) \in E(G)\} ~\cup  \nonumber \\
		&\{(v,v_i) : v \in V(G) \wedge i \in \{1,\ldots,A\}\} ~\cup  \nonumber \\
		&\{ (v_i,v_j): v \in v(G) \wedge i,j \in \{1,\ldots,A\}\} \wedge i<j \nonumber \}
\end{align}

Now, we arbitrarily choose some connected induced subgraph $F$ of $G'$. Either this subgraph intersects with the original set of vertices $V(G)$, or it does not. In the latter case, subgraph $F$ is contained within the single copy of $K_A$ and there are exactly $|V(G)|(2^A-1)$ such subgraphs. In the former case, we can define some \emph{pseudograph} $F'$, which consists of the following sets of vertices and edges:\footnote{\footnotesize We note that this tuple is not necessarily a properly defined graph, since it can contain some edge $(u,v)$ and does not contain node $v$.}

\begin{align}
V(F') = & V(F) \cap V(G) \nonumber \\
E(F')  = &  \{(u,v) : uv \in V(F) \wedge (u,v) \in E(G)\} \nonumber \
\end{align}

\noindent Note that since $F$ is connected, $F'$ also has to be connected. There are exactly $2^{|V(F')|A}$ choices of $F$ that can give us a particular pseudograph $F'$. Now, we can compute:
%\vspace{-0.05cm}
\begin{equation}\label{complexity_M}
M=|V(G)|(2^A-1) + \sum_{F'} 2^{|V(F')|A}
\end{equation}
\noindent which denotes the number of induced connected subgraphs in~$G'$.

The crucial observation here is that, if $V(F')=V(G)$, then $F'$ is a connected spanning subgraph of $G$. This holds because $F$ is an induced connected subgraph. Now, let $N$ denote the number of spanning connected subgraphs in $G$. Then, we can rewrite \eqref{complexity_M} as:
%\vspace{-0.12cm}
\begin{equation}\label{complexity_MN}
M=|V(G)|(2^A-1) + N(2^{|V(G)|A}) + \hspace{-0.5cm} \sum_{F': V(F') \neq V(G)}{\hspace{-0.5cm} 2^{|V(F')|A}}
%\vspace{-0.08cm}
\end{equation}
In order to compute $N$ we would like to bound the expression $X=M - N(2^{|V(G)|A})$. This expression is the number of induced connected subgraphs in $G'$ that do not correspond to any connected spanning subgraph in $G$. We have:
%\vspace{-0.12cm}
$$
0 \leq X \leq  |V(G)|(2^A-1) + 2^{|V(G)| + |E(G)|}2^{(|V(G)|-1)A}
$$
\noindent where the number $2^{|V(G)| + |E(G)|}$ is an upper bound on the number of all subgraphs in $G$. Now, we can use these bounds to transform equation \eqref{complexity_MN} and to determine the approximation of $N$:

\begin{align}
N  = &  \frac{M}{2^{|V(G)|A}} - \frac{X}{2^{|V(G)|A}}  \simeq \frac{M}{2^{|V(G)|A}} - \frac{|V(G)|+2^{|V(G)| + |E(G)|}}{2^{A}}  \nonumber \
\end{align}

%%%%%%%%%%%%%%%%%%%%%%%%%%%%%%%%%%%%%%%%%%%%%%%%%%%%%%%%%%%%%%%%%%%%%%%%%%%%%%%%%%%%%%%%%%%%%%%%%%%%%%%%%%%%%%%%%%%%%%%
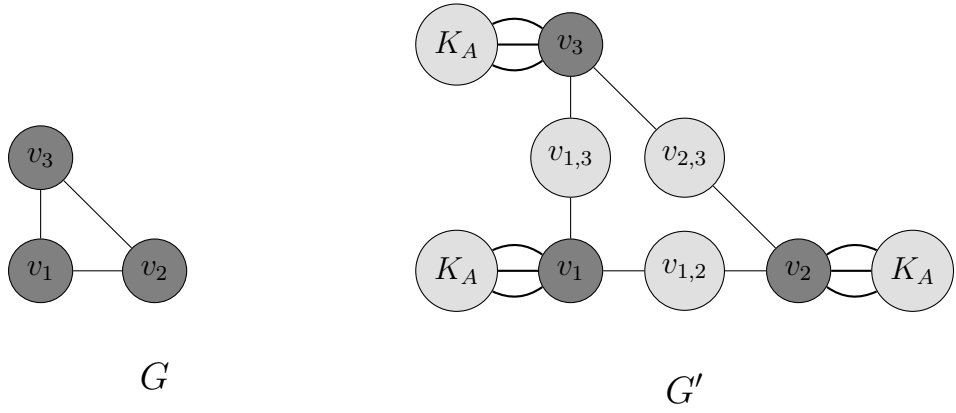
\begin{figure}[t]
\centering
\begin{tikzpicture}[node distance = 1.5cm]
\tikzstyle{every node}=[draw,shape=circle,minimum size=2em,color=black,font=\bfseries];

\definecolor{myGray}{RGB}{224,224,224}
\definecolor{myBest}{RGB}{128,128,128}

\node[fill=myBest] (1) {$v_1$};
\node[fill=myBest] (2) [right of= 1]  {$v_2$};
\node[fill=myBest] (3) [above of= 1] {$v_3$};

\draw (1) -- (2) (1) -- (3) (2) -- (3);

\node[fill=myGray] (K1) [right= 3cm  of 2] {$K_{A}$};
\node[fill=myBest] (1p) [right of=K1] {$v_{1}$};
\node[fill=myGray] (12) [right of=1p] {$v_{1,2}$};
\node[fill=myBest] (2p) [right of=12] {$v_{2}$};
\node[fill=myGray] (K2) [right of=2p] {$K_{A}$};

\node[fill=myGray] (13) [above of= 1p]  {$v_{1,3}$};
\node[fill=myBest] (3p) [above of= 13]  {$v_{3}$};
\node[fill=myGray] (K3) [left of=3p] {$K_{A}$};

\node[fill=myGray] (23) [above of= 12]  {$v_{2,3}$};

\draw (1p) -- (12) -- (2p) ;
\draw (1p) -- (13) -- (3p) ;
\draw (2p) -- (23) -- (3p) ;

\draw[thick] (K1) edge[bend left] (1p);
\draw[thick] (K1) edge (1p);
\draw[thick] (K1) edge[bend right] (1p);

\draw[thick] (K2) edge[bend left] (2p);
\draw[thick] (K2) edge (2p);
\draw[thick] (K2) edge[bend right] (2p);

\draw[thick] (K3) edge[bend left] (3p);
\draw[thick] (K3) edge (3p);
\draw[thick] (K3) edge[bend right] (3p);

\node[draw=none] (G) [below= 0.5cm  of 2] {$\mathlarger{\mathlarger{G}}$};
\node[draw=none] (Gp) [below= 0.5cm  of 12] {$\mathlarger{\mathlarger{G'}}$};

\end{tikzpicture}
\vspace{0.5cm}
\caption{The reduction in the proof of Theorem \ref{first:theorem}.}
\label{fig:reduction}
\end{figure}
%%%%%%%%%%%%%%%%%%%%%%%%%%%%%%%%%%%%%%%%%%%%%%%%%%%%%%%%%%%%%%%%%%%%%%%%%%%%%%%%%%%%%%%%%%%%%%%%%%%%%%%%%%%%%%%%%%%%%%%

In order to deal with the expression $\frac{|V(G)|+2^{|V(G)| + |E(G)|}}{2^{A}}$ we take $A > \log(|V(G)| + 2^{|V(G)|+|E(G)|}) $ so that this fraction becomes smaller than $1$. We note that $A$ is bounded by a polynomial in the order of the input $G$. Then, the number of spanning connected subgraphs of $G$ is the least integer $N$ such that $N \geq \frac{M}{2^{A|V(G)|}}$. This is easy to compute given the number of connected induced subgraphs $F$ of $G'$. This completes the proof. \hfill$\Box$
%%%%%%%%%%%%%%%%%%%%%%%%%PROOF%%%%%%%%%%%%%%%%%%%%%%%%%%%%%%%%%%%%%%%%%%%%%%%

From Theorem~\ref{first:theorem}, it trivially follows that the next problem is also \#P-Complete:
\begin{definition}[\textsc{\#Connected-Induced-Sub-\textit{k} (\#CIS-\textit{k})}]Given a graph $G$, we are asked to compute the number of connected induced subgraphs of size $k$ in $G$.
%(\emph{i.e.}, subgraphs with exactly $k$ nodes each)
\end{definition}
Clearly, if we can find in polynomial time an answer for the \textsc{\#CIS-\textit{k}} problem, we could efficiently compute \textsc{\#CIS}.
%\vspace{-0.1cm}
Now, the \#P-Hardness of \textsc{\#CG-Shapley} will be shown by the reduction from \textsc{\#CIS-\textit{k}}:
%Clearly, if we can find in polynomial time an answer for the $\mathit{\#CIS\textrm{-}k}$ problem, we could efficiently compute $\mathit{\#CIS}$.
%Now, we will prove the following:
%If we denote by $c^G_i$ the number of connected induced subgraphs of size $i$ in graph $G$, then we can easily compute the number of all connected induced subgraphs as $\sum_{i=0}^{|E(G)|}c^G_i$.
%Having proven that $\mathit{\#CIS}$ is \#P-Complete, we show that:
%\vspace{-0.1cm}
\begin{theorem}\label{second:theorem}
\textsc{\#CG-Shapley} is \#P-Hard.
%\vspace{-0.1cm}
\end{theorem}
%\vspace{-0.1cm}
\noindent \textit{Proof of Theorem~2}: We construct a proof by reduction. In particular, we demonstrate that if there exists an algorithm for solving \textsc{\#CG-Shapley} in polynomial time, then it is possible to solve \textsc{\#CIS-\textit{k}} in polynomial time. This contradicts the fact that \textsc{\#CIS-\textit{k}} is \#P-complete. Now, we will reduce \textsc{\#CIS-\textit{k}} to \textsc{\#CG-Shapley}.

Let $G=(V,E)$ be an arbitrary graph, where $|V|=n$ and $|E|=m$. We extend the set of nodes of $G$ by a single node $v$, while the set of edges remains unchanged. In other words, we obtain a new graph $G_0$ by adding a single node not connected to any node from $G$. Now, from the definition of connectivity games, we note that the marginal contribution of the new node $v$ to any coalition $C\subseteq G$ is either 0 or $-1$. More specifically, it is $-1$ if $C$ is connected, and $0$ otherwise. Based on this, the Shapley value of node $v$ can be computed as follows, where $c_s^G$ is the number of connected induced subgraphs in $G$ that contain exactly $s$ nodes:
%\vspace{-0.05cm}
\[
\phi^{\SV}_{v,G_0} = -\hspace{-0.15cm}\sum_{s=0}^{n}\hspace{-0.1cm}\frac{(s)!(n-s)!}{(n+1)!}c^G_{s}
\]
Now, let us consider a new graph $G_i$ constructed by adding to $G$ the set of $i$ nodes in addition to the node $v$, while keeping the set of edges just as in $G$. Analogously to $G_0$, the Shapley value of $v$ in $G_i$ is:
%\vspace{-0.15cm}
\begin{equation}\label{complexity_eq}
\phi^{\SV}_{v,G_i} =  -\hspace{-0.15cm}\sum_{s=0}^{n}\frac{(s)!(n+i-s)!}{(n+1+i)!}c^G_{s}
%\vspace{-0.15cm}
\end{equation}
\noindent This equation holds because each coalition containing more than $n$ nodes is disconnected, and so the contribution of $v$ to every such coalition is $0$. Now, we can build a system of linear equations using each equation (\ref{complexity_eq}) from graph $G_i$, where $i\in\{0,\ldots,n\}$. More precisely, we need to solve the following equation:
%\vspace{0.4cm}
\[
\arraycolsep=2pt
 \begin{bmatrix}
  0!n! & 1!(n-1)! & \cdots & n!0! \\
  0!(n+1)!  & 1!n!  & \cdots & n!1! \\
  \vdots  & \vdots  & \ddots & \vdots  \\
  0!(2n)! & 1!(2n-1)! & \cdots & n!n!
 \end{bmatrix}
 \begin{bmatrix}
  c^G_0 \\
  c^G_1  \\
  \vdots   \\
  c^G_n
 \end{bmatrix} =
  \begin{bmatrix}
  (n+1)! \phi^{\SV}_{v,G_0} \\
  (n+2)! \phi^{\SV}_{v,G_1}  \\
  \vdots   \\
  (2n+1)! \phi^{\SV}_{v,G_n}
 \end{bmatrix}
\]

\noindent that can also be written as $Ax=b$.

This equation has a unique solution if and only if the determinant of the matrix $A$ is non-zero. We can use Theorem~1.1 from \citep{Bacher:2002} to prove that it is non-zero. Thus, if we can compute in polynomial time the Shapley value for connectivity games, we would be able to solve this equation and determine all $c^G_i$ values. Specifically, we can use Gaussian elimination, which works in $O(n^3)$ time complexity.

We note that the largest possible number in our matrices is $n!n!$. According to the analysis in (Proposition 2)\citep{Aziz:et:al:2009} it is possible to store such a number in $km^2(\log m)^2$ bits. It is shown in (Theorem 4.10) \citep{Korte:Vygen:2005} that in Gaussian elimination each number occurring during the algorithm process can be stored in the number of bits quadric of the input size. This final observation ends our proof.
\hfill $\Box$

\subsection{Analysis of Marginal Contribitions}\label{subsection:mc}

\noindent In this section, we analyse how node $v_i \in V$ can marginally contribute to coalition $C \subseteq V \setminus \{v_i\}$. Four general cases, depicted in Figure~\ref{fig:example:contributions}, can be distinguished:% as an example:
%\vspace{-0.15cm}
\begin{itemize}
\item[(a)] Node $v_i$ can join a \emph{connected} coalition $C \in \mC$ and the resulting coalition is also \emph{connected}, \emph{i.e.}, $C\cup\{v_i\} \in \mC$. Here, the marginal contribution is equal to the difference in the value of $C \in \mC$ caused by the addition of $v_i$:

  $  \MC{C}{v_i} = \nu_f(C\cup \{v_i\})-\nu_f(C)=f(C \cup \{v_i\})-f(C)$
\item[(b)] Node $v_i$ can join a \emph{disconnected} coalition $C \in \tilde{\mC}$ and the resulting coalition becomes \emph{connected}, \emph{i.e.}, $C\cup\{v_i\} \in \mC$. Here, $v_i$'s contribution is the whole value of $C\cup\{v_i\}$:

$\MC{C}{v_i} = \nu_f(C \cup \{v_i\})-\nu_f(C)=f(C \cup \{v_i\})$
\item[(c)] Node $v_i$ can join a \emph{connected} coalition $C \in \tilde{\mC}$ and the resulting coalition becomes \emph{disconnected}, \emph{i.e.}, $C\cup\{v_i\} \in \tilde{\mC}$. This means that $v_i$ brings down the value of $C$ to 0:

$\MC{C}{v_i}= \nu_f(C \cup \{v_i\})-\nu_f(C)= - f(C)$
\item[(d)] Node $v_i$ can join a \emph{disconnected} coalition $C \in \tilde{\mC}$ and the resulting coalition remains \emph{disconnected}, \emph{i.e.}, $C\cup\{v_i\} \in  \tilde{\mC}$:

$\MC{C}{v_i} = \nu_f(C \cup \{v_i\})-\nu_f(C))=0$
\end{itemize}

%%%%%%%%%%%%%%%%%%%%%%%%%%%%%%%%%%%%%%%%%%%%%%%%%%%%%%%%%%%%%%%%%%%%%%%%%%%%%%%%%%%%%%%%%%%%%%%%%%%%%%%%%%%%%%%%%%%%%%%%

\begin{figure}[t]
\begin{minipage}[b]{0.5\linewidth}\centering
\begin{tikzpicture}[node distance = 1.5cm]
\tikzstyle{every node}=[draw,shape=circle,minimum size=2em,color=black,font=\bfseries];

\definecolor{myGray}{RGB}{224,224,224}
\definecolor{myBest}{RGB}{128,128,128}

\node (1) {$v_1$};
\node (2) [right of= 1]  {$v_2$};

\node (3) [above left= 0.75 cm and 0.25cm of 1] {$v_3$};
\node (4)[fill=myBest] [above right= 0.75 cm and 0.25cm of 1] {$v_4$};
\node (5) [above right= 0.75 cm and 0.25cm of 2] {$v_5$};

\node[fill=myGray] (6) [above of=3] {$v_6$};
\node[fill=myGray] (7) [above  of=4] {$v_7$};
\node (8) [above of=5] {$v_8$};

\draw (1) -- (2) (1) -- (3) (1) -- (4) (2) -- (4) (2) -- (5)
	  (3) -- (6) (4) -- (7) (5) -- (8) (6) -- (7) (7) -- (8);
\end{tikzpicture}
\caption*{\textbf{a)  ${\scriptstyle\text{\textnormal{MC}}(\set{v_6,v_7},v_4) = f(\set{v_6,v_7,v_4}) - f(\set{v_6,v_7})}$}}
\end{minipage}
\begin{minipage}[b]{0.5\linewidth}\centering
\begin{tikzpicture}[node distance = 1.5cm]
\tikzstyle{every node}=[draw,shape=circle,minimum size=2em,color=black,font=\bfseries];

\definecolor{myGray}{RGB}{224,224,224}
\definecolor{myBest}{RGB}{128,128,128}

\node (1)[fill=myGray] {$v_1$};
\node (2) [right of= 1]  {$v_2$};

\node (3)[fill=myGray] [above left= 0.75 cm and 0.25cm of 1] {$v_3$};
\node (4)[fill=myBest] [above right= 0.75 cm and 0.25cm of 1] {$v_4$};
\node (5) [above right= 0.75 cm and 0.25cm of 2] {$v_5$};

\node (6) [above of=3] {$v_6$};
\node (7)[fill=myGray] [above  of=4] {$v_7$};
\node (8)[fill=myGray] [above of=5] {$v_8$};

\draw (1) -- (2) (1) -- (3) (1) -- (4) (2) -- (4) (2) -- (5)
	  (3) -- (6) (4) -- (7) (5) -- (8) (6) -- (7) (7) -- (8);
\end{tikzpicture}
\caption*{\textbf{b) ${\scriptstyle\text{\textnormal{MC}}(\set{v_1,v_3,v_7,v_8},v_4) = f(\set{v_1,v_3,v_4,v_7,v_8})}$  }}
\end{minipage} \\
\begin{minipage}[b]{0.5\linewidth}\centering
\begin{tikzpicture}[node distance = 1.5cm]
\tikzstyle{every node}=[draw,shape=circle,minimum size=2em,color=black,font=\bfseries];

\definecolor{myGray}{RGB}{224,224,224}
\definecolor{myBest}{RGB}{128,128,128}

\node (1) {$v_1$};
\node (2) [right of= 1]  {$v_2$};

\node (3) [above left= 0.75 cm and 0.25cm of 1] {$v_3$};
\node (4)[fill=myBest] [above right= 0.75 cm and 0.25cm of 1] {$v_4$};
\node (5)[fill=myGray] [above right= 0.75 cm and 0.25cm of 2] {$v_5$};

\node (6) [above of=3] {$v_6$};
\node (7) [above  of=4] {$v_7$};
\node (8)[fill=myGray] [above of=5] {$v_8$};

\draw (1) -- (2) (1) -- (3) (1) -- (4) (2) -- (4) (2) -- (5)
	  (3) -- (6) (4) -- (7) (5) -- (8) (6) -- (7) (7) -- (8);
\end{tikzpicture}
\caption*{\textbf{c) ${\scriptstyle\text{\textnormal{MC}}(\set{v_5,v_8},v_4) = -f(\set{v_5,v_8})}$}}
\end{minipage}
\begin{minipage}[b]{0.5\linewidth}\centering
\begin{tikzpicture}[node distance = 1.5cm]
\tikzstyle{every node}=[draw,shape=circle,minimum size=2em,color=black,font=\bfseries];

\definecolor{myGray}{RGB}{224,224,224}
\definecolor{myBest}{RGB}{128,128,128}

\node (1) {$v_1$};
\node (2) [right of= 1]  {$v_2$};

\node (3)[fill=myGray] [above left= 0.75 cm and 0.25cm of 1] {$v_3$};
\node (4)[fill=myBest] [above right= 0.75 cm and 0.25cm of 1] {$v_4$};
\node (5)[fill=myGray] [above right= 0.75 cm and 0.25cm of 2] {$v_5$};

\node (6)[fill=myGray] [above of=3] {$v_6$};
\node (7) [above  of=4] {$v_7$};
\node (8)[fill=myGray] [above of=5] {$v_8$};

\draw (1) -- (2) (1) -- (3) (1) -- (4) (2) -- (4) (2) -- (5)
	  (3) -- (6) (4) -- (7) (5) -- (8) (6) -- (7) (7) -- (8);
\end{tikzpicture}
\caption*{\textbf{d) ${\scriptstyle\text{\textnormal{MC}}(\set{v_3,v_5,v_6,v_8},v_4) = 0}$}}
\end{minipage} \\
\caption{Four ways in which $v_4$ can contribute to a coalition.}
\label{fig:example:contributions}
\end{figure}

%%%%%%%%%%%%%%%%%%%%%%%%%%%%%%%%%%%%%%%%%%%%%%%%%%%%%%%%%%%%%%%%%%%%%%%%%%%%%%%%%%%%%%%%%%%%%%%%%%%%%%%%%%%%%%%%%%%%%%%%

The key conclusion to be drawn from the above analysis is that both connected and disconnected coalitions play a role when computing the Shapley value. This is because a node can contribute not only to a connected coalition, but also to a disconnected one (by making it connected). However, in the further sections, we will show that it is possible to develop an exact algorithm that only cycles through connected coalitions.

\subsection{The Basic Algorithm for Connectivity Games}
\noindent Based on the observation that both connected and disconnected coalitions play a role when computing the Shapley value, we develop the general-purpose Shapley value algorithm (\textbf{GeneralSV} presented in Algorithm~\ref{algo_brute}) for computing the Shapley value in connectivity games defined by Amer and Gimenez. In more details, for each coalition $C \in 2^V$ this algorithm considers four aforementioned cases (a), (b) and (c) as outlined in Figure~\ref{fig:example:contributions} in our chapter. Importantly, for (d), it can be disregarded since the marginal contribution in this case equals 0. In all of the rest three cases, a node $v_i$ can make a non-zero contribution by joining a coalition. We note that function $CheckConnectedness(C)$ runs in $O(|C|+|E(C)|)$.

%%%%%%%%%%%%%%%%%%%%%%%%%%%%%%%%%%%%%%%%%%%%%%%%%%%%%%%%%%%%%%%%%%%%%%%%%%%%%%%%%%%%%%%%%%%%%%%%%%%%
\RestyleAlgo{ruled} %The old version of this is: \RestyleAlgo

\begin{algorithm}[h]
\SetAlgoVlined %The old version of this is: \SetAlgoVlined
\LinesNumbered %The old version of this is: \LinesNumbered
\caption{\textbf{GeneralSV} Algorithm for the Shapley value }\label{algo_brute}
\KwIn{Graph $G = (V,E)$ and characteristic function $\nu_f$ }
\KwOut {Shapley Value, $\phi^{\SV}_i(\nu_f)$, for each node $v_i\in V$}

    \ForEach{$v_i \in V$}
     {
	      $\phi^{\SV}_i(\nu_f) \gets 0$;
	 }

    \ForEach{$C \in 2^{V}$}
    {
        $CheckConnectedness(C)$;\\
        \If{$C \in \mC$}
        {
        \ForEach{$v_i \in N(C)$}
            {
                $\phi^{\SV}_i(\nu_f) \gets \phi^{\SV}_i(\nu_f) + \frac{ |C|!(|V| - |C| - 1)!}{|V|!}(\nu_f(C\cup\{v_i\})-\nu_f(C))$
            }
        \ForEach{$v_i \not\in N(C)$}
            {
                $\phi^{\SV}_i(\nu_f) \gets \phi^{\SV}_i(\nu_f) -\frac{ |C|!(|V| - |C| - 1)!}{|V|!}\nu_f(C)$
            }
        }
        \Else
        {
            \ForEach{$v_i \in N(C)$}
            {
             $CheckConnectedness(C\cup\{v_i\})$ \\
             \If{$C\cup\{v_i\}\in \mC$}
                   {
                   $\phi^{\SV}_i(\nu_f) \gets \phi^{\SV}_i(\nu_f) +\frac{ |C|!(|V| - |C| - 1)!}{|V|!}\nu_f(C\cup\{v_i\})$
                   }
            }%end for each
        }%end else

    }%end for each

\end{algorithm}

Unfortunately, the use of this algorithm in practice is limited by the number of nodes in the network. It runs in $O((|V| + |E|)2^{|V|})$ and already for $|V|=50$, the Algorithm~\ref{algo_brute} has to cycle through more than $10^{15}$ coalitions. In this dissertation, we developed an algorithm that for many networks is able to compute the SV substantially faster.

\subsection{The Faster Algorithm for Connectivity Games}

\noindent Many real-world terrorist networks are sparse, \textit{i.e.}, $|\mC|\ll |\tilde{\mC}|$  \citep{Krebs:2002}; thus, if the Shapley value could be computed only by considering coalitions in $\mC$, it would be possible to analyse much larger terrorist networks. To this end, for each $v_i \in V$, let us define the following disjoint sets of coalitions:

\begin{align*}
\mC_i^{\#} =& \{C \subseteq V \setminus \{v_i\}:\ C\in \mC\ \ \wedge\ \ C \cup \{v_i\} \in \mC\}\\
\mC_i^{+} =& \{C \subseteq V \setminus \{v_i\}:\ C \in \tilde{\mC}\ \ \wedge\ \ C \cup \{v_i\} \in \mC\}\\
\mC_i^{-} =& \{C \subseteq V \setminus \{v_i\}:\ C \in \mC\ \ \wedge \ \ C \cup \{v_i\}  \in \tilde{\mC}\}\\
\end{align*}

which correspond to cases (a), (b) and (c) from the previous section. Based on this, the original formula for Shapley value (Definition~\ref{def:sv}) can be computed as:

\begin{align}\label{ter:SV:con}
\phi^{\SV}_i(V,\nu_f) & = \hspace{0.5cm} \sum_{C \subseteq V \setminus \{v_i\}} \frac{ |C|!(|V| - |C| - 1)!}{|V|!} (\nu_f(C \cup \{v_i\}) - \nu_f(C))\nonumber \\
& = \sum_{C \in \{\mC_i^{+}\cup \mC_i^{\#} \cup \mC_i^{-}\}} \frac{ |C|!(|V| - |C| - 1)!}{|V|!} (\nu_f(C \cup \{v_i\}) - \nu_f(C))
\end{align}

\noindent where case (d), when $mc_i(C)=0$, is simply omitted. The key idea behind our exact algorithm to compute the Shapley value in connectivity game (\textbf{FasterSVCG} presented in Algorithm~\ref{ter:algo:faster}) is to represent the sets $\mC_i^{+}$ and $\mC_i^{-}$ differently, such that $\tilde{\mC}$ does not appear in the new representation. As for $\mC_i^{\#}$, it does not depend on $\tilde{\mC}$, and so there is no need to represent it differently. In particular, we represent $\mC_i^{+}$ and $\mC_i^{-}$ as follows, where $\mathcal{P}(C)$ is the set of agents that are pivotal to $C$, and $N(C)$ is the set of neighbours of $C$:\\

\begin{align*}
\mC_i^{+} =& \{C \subseteq V \setminus \{v_i\}:\ C \cup \{v_i\} \in \mC\ \ \wedge\ \ v_i \in \mathcal{P}(C \cup \{v_i\})\}\\
\mC_i^{-} =& \{C \subseteq V \setminus \{v_i\}:\ C \in \mC\ \ \wedge \ \ v_i \notin N(C)\}\\
\end{align*}

Now since $\tilde{\mC}$ no longer appears in the definitions of $\mC_i^{\#}$, $\mC_i^{+}$ and $\mC_i^{-}$, it is possible to compute the Shapley value as in equation~\eqref{ter:SV:con} \textit{without enumerating any of the coalitions in $\tilde{\mC}$}. Based on this, our algorithm enumerates every connected coalition, $C\in\mC$, and determines for each agent $v_i\in C$ whether $C \setminus \{v_i\} \in \mC_i^{\#}$ or $C \setminus \{v_i\} \in \mC_i^{+}$ and for $v_i \not \in C$ if $C \in \mC_i^{-}$. Note, that we do not consider the impact of agent $v_i \not \in C$ that do not disconnect $C$ because the contribution of this agent will be calculated for connected coalition $C \cup \{v_i\}$ as $C \in \mC_i^{\#}$. The enumeration is carried out using \cite{Moerkotte:Neumann:2006}'s method---the fastest such enumeration method in the literature. Its basic idea is that, for each connected coalition $C\in\mC$, it expands $C$ by adding to it certain subsets of its neighbours. These subsets are chosen so as to ensure that no connected coalition is enumerated more than once (see Moerkotte an Neumann~[\citeyear{Moerkotte:Neumann:2006}] for more details).% An example is illustrated in Figure~

\RestyleAlgo{ruled} %The old version of this is: \restylealgo

\begin{algorithm}[!thbp]
\SetAlgoVlined %The old version of this is: \SetVline
\LinesNumbered %The old version of this is: \linesnumbered
\caption{\textbf{FasterSVCG} Faster algorithm for the Shapley value}\label{ter:algo:faster}
\KwIn{Graph $G\hspace{-0.1cm}=\hspace{-0.1cm}(V,E)$ and characteristic function~$\nu_f$}
\KwOut{Shapley value $\phi^{\SV}_i(\nu_f)$ of each node $v_i\in V$}

$X \gets V$; \Comment{initialize $X$, which is only used  for Moerkotte \& Neumann's enumeration} \\
    %initialize the Shapley values

	\textbf{foreach} $v_i \in V$ \textbf{do} $\phi^{\SV}_i((\nu_f) \gets 0$; \\
%    $SV_i(\nu_f) \gets 0$ for all $v_i \in V$; \\
%    \For{$i \gets |V|$ \KwTo $1$}{
%       $SV_i(\nu_f) \gets 0$;
%    }
%    \vspace{-0.05cm}
    %for every root, generate its tree
    \For{$i \gets |V|$ \KwTo $1$}{
     $computeSV(\{v_i\}, N(v_i),X, X \setminus (N(v_i)\cup \{v_i\}), \emptyset)$;
     $X \gets X \setminus \{v_i\}$
    }\Comment{----------- Next, we define $computeSV$ -----------}
\ \\

$computeSV(C, NC, X, V^-_{C}, V^+_{C})$ \Begin{
 $ X' \gets X \cup NC$; \Comment{where $NC$ consists of the neighbours of $C$}\\
 \ForEach{$ S \subseteq (NC \setminus X) \wedge S \neq \emptyset$} {

 $C' \gets C$; \Comment{ a new coalition that will be constructed from the old coalition $C$ } \nllabel{line:expand}
 $NC' \gets NC \setminus S$; \Comment{ the neighbours of $C'$ }

 $isCycle \gets false$;\Comment{ to indicate whether a new cycle has appeared while constructing $C'$ }

 \ForEach{$v \in S$} {
  $\mathit{TEMP} \gets \emptyset;$ \Comment{ a temporary set used to compute neighbours of: $C' \cup \{v\}$ }
  $\mathit{TEMP2} \gets \emptyset;$ \Comment{ a set used to compute the pivotal agents in: $C' \cup \{v\}$ }
  \ForEach{$u \in N(S)$} {
   \If{$u \notin (C \cup NC) $} {
     $\mathit{TEMP} \gets \mathit{TEMP} \cup \{u\} $; \\ \nllabel{out_update}
   }
         \ElseIf(\Comment{condition 1}){$(isCycle = false) \wedge (u \in C')$} {
           $\mathit{TEMP2} \gets \mathit{TEMP2} \cup \{u\}$; \nllabel{p_update}
           \vspace{-0.12cm}
          }
          \vspace{-0.12cm}
         }
         \vspace{-0.12cm}
  }

%\vspace{-0.1cm}
  \If(\Comment{condition 2}){$|\mathit{TEMP2}|>1$}{   \nllabel{cycle_detect2}
%  \vspace{-0.08cm}
   $isCycle \gets true$;
%   \vspace{-0.10cm}
  }
%  \vspace{-0.06cm}
  $C' \gets C' \cup \{v\}$; \nllabel{cycle_detect3}
%  \vspace{-0.06cm}
  $NC' \gets NC' \cup \mathit{TEMP}$; \\
%  \vspace{-0.02cm}
  $V^-_{C'} \gets V^-_{C'\setminus \{v\}} \setminus \mathit{TEMP};$ \\
%  \vspace{-0.1cm}
  \If(\Comment{condition 1}){$isCycle = false $}{
  $V^+_{C'} \gets V^+_{C'\setminus \{v\}} \cup \mathit{TEMP2};$
  }
%  \vspace{-0.18cm}
 }

 \If(\Comment{condition 3}){$|C'| = 2$}{
%  \vspace{-0.08cm}
  $V^+_{C'} \gets \emptyset$;
%  \vspace{-0.05cm}
 }
%  \vspace{-0.05cm}
 \ElseIf(\Comment{condition 2}){$isCycle = true$}{
%  \vspace{-0.08cm}
  $V^+_{C'} \gets FindPivotals(C')$; \nllabel{p_update2}
%  \vspace{-0.05cm}
 }
 \ForEach(\Comment{update Shapley value}){$v_i \in C'$} {
% \vspace{-0.08cm}
  \If{$v_i \in V^+_{C'}$} {
%   \vspace{-0.08cm}
   $\phi^{\SV}_i(\nu_f) \gets \phi^{\SV}_i(\nu_f) + \frac{ |C'\setminus\{v_i\}|!(|V| - |C'\setminus\{v_i\}| - 1)!}{|V|!}\nu_f(C')$   \nllabel{positive_con}
  }
  \Else(\Comment{deal with the set $V^{\#}_{C'}$}){
  $\phi^{\SV}_i(\nu_f)\hspace{-0.1cm}  \gets \phi^{\SV}_i(\nu_f) \hspace{-0.05cm}+\hspace{-0.05cm} \frac{ |C'\setminus\{v_i\}|!(|V| - |C'\setminus\{v_i\}| - 1)!}{|V|!}(\nu_f(C')-\nu_f(C'\setminus \{v_i\}))$ \nllabel{neutral_con}
  }
 }

 \ForEach(\Comment{update Shapley value}){$v_i \in V^-_{C'}$} {
  $\phi^{\SV}_i(\nu_f) \gets \phi^{\SV}_i(\nu_f) - \frac{ |C'|!(|V| - |C'| - 1)!}{|V|!}\nu_f(C')$ \nllabel{negative_con}
 }
%  \vspace{-0.08cm}
 $computeSV(C', NC', X', V^-_{C'}, V^+_{C'})$;
% \vspace{-0.03cm}
 }
%   \vspace{-0.15cm}
\end{algorithm}

%%%%%%%%%%%%%%%%%%%%%%%%%%%%%%%%%%%%%%%%%%%%%%%%%%%%%%%%%%%%%%%%%%%%%%%%%%%%%%%%%%%%%%%%%%%%%%%%%%%%%%%%%%%%%%%%%%%%%%%

 Next, we explain our algorithm. To enhance clarity, for every connected coalition $C\in\mC$, we will define three disjoint sets of agents: $V_C^{\#} = \{v_i\in C : C \setminus \{v_i\} \in \mC_i^{\#} \}$, $V_C^+ = \{v_i\in C : C \setminus \{v_i\} \in \mC_i^+ \}$, and $V_C^- = \{v_i\in V \setminus C : C \in \mC_i^- \}$.
If we compute the above sets for every $C\in\mC$, then we can compute the Shapley value. Let us take a closer look at the difficulty of computing those sets for a given $C$.
%\vspace{-0.1cm}
\begin{itemize}
\item Computing $V_C^-$ can be done in $O(|V|)$ time. This is because the agents in $V_C^-$ are basically all those that are not members, nor neighbours, of $C$.
%\vspace{-0.2cm}
\item Now, to compute $V_{C}^+$, we need to find the pivotal agents in $C$. This can be computed using a method \textit{``findPivotals''} that runs in $O(|V|+|E|)$~\citep{Alsuwaiyel:1999}.
% We will refer to this method as \textit{``findPivotals''}.
%\vspace{-0.2cm}
\item Having computed $V_{C}^+$, it becomes easy to compute $V_{C}^{\#}$. This is because $V_{C}^{\#} = C\setminus V_{C}^+$.
%\vspace{-0.1cm}
\end{itemize}

%\item Set $V_C^{\#}$ consists of all neighbours of $C$. This can be computed in $O(|V|)$.
%\vspace{-0.1cm}
%\item Having computed $V_{C}^{\#}$, it becomes easy to compute $V_{C}^-$, as any other player from $V \setminus C$ added to coalition $C$ will disconnect it. Formally: $V \setminus C = V_C^{\#} \cup V_C^-$.
%\vspace{-0.1cm}
%\item Now, to compute $V_{C}^+$, we need to find the pivotal agents in $C$. This can be computed using a method that runs in $O(|V|+|E|)$~\cite{Alsuwaiyel:1999}. We will refer to this method as \textit{``findPivotals''}.

From the above analysis, it is clear that the main difficulty lies in \textit{findPivotals}. Therefore, whenever possible, we would like to compute $V_C^+$ using some other, easier, technique. In particular, when we expand a connected coalition, $C$, into another connected coalition $C'=C\cup S$, we try to \textit{update} the set of pivotal agents, rather than compute it from scratch with \textit{findPivotals}. Here, we distinguish between three conditions:
%\vspace{-0.1cm}
\begin{itemize}
\item Condition~1: The cycles in $C'$ are exactly like those in $C$. In this case, the set $V_{C'}^+$ consists of elements of $V_{C}^+$, expanded by the nodes in $C$ that are connected to $S$.
%\vspace{-0.1cm}
\item Condition~2: $C'$ contains a cycle that is not in $C$. Here, we need to call \textit{findPivotals}.
%\vspace{-0.1cm}
\item Condition~3: $|C|=2$. In this case, since we assumed that a singleton is a connected coalition, none of the two agents in $C$ is pivotal.
%\vspace{-0.1cm}
\end{itemize}
The pseudocode of FasterSVCG is presented in Algorithm~\ref{ter:algo:faster}. It is easy to see that it runs in $O((|V|+|E|) |\mC|)$.

\subsection{Approximation algorithm for connectivity games}

\noindent Our FasterSVCG algorithm is much faster than the GeneralSV, but it is not fast enough to deal with large networks. In order to provide ranking for larger networks we should use an approximation algorithms. In this section we provide the new approximation algorithm for connectivity games (called \textbf{ApproximateSVCG}) that outperforms the standard \cite{Castro:et:al:2009} method. To see discussion about different methods of approximation the Shapley value see Section~\ref{chap:gtc:related:approxSV}.

Our Algorithm~\ref{ter:algo:approx} is the dedicated application of standard Monte Carlo sampling to connectivity games. The crucial difference between \cite{Castro:et:al:2009} method and ApproximateSVCG is that instead of sampling permutations our algorithm sample coalitions. The first observation is that during sampling a single set $C$ we can compute two contributions: one $\nu_f(C \cup v_i) - \nu_f(C)$, and the second $\nu_f(C) - \nu_f(C  \setminus \set{v_i})$. Generally speaking, in our algorithm, we will randomly select $C$ and and approximate the $\phi^{\SV}_i(\nu_f)$ using the resulting average. Looking at the equation~\eqref{ter:SV:con} we see that the  marginal contributions are calculated with different weights. Thus, to obtain an unbiased estimator we have to compute marginal contributions with appropriate probabilities. To this end, we propose the following process consisting of $3$ fazes. In \textbf{Faze 1}, we uniformly select the size of the coalition $k \in \{0, \ldots, |V|\}$. In \textbf{Faze 2}, we choose a random coalition $C$ of size $k$ and, in \textbf{Faze 3}, for every agent, compute the marginal contribution of this agent obtained by leaving/entering $C$. To better explain our motivation, let us transform the formula for the Shapley value (Definition~\ref{def:sv}) as follows:

\begin{equation}\label{Eq:ter:SV:approx}
\phi^{\SV}_i(\nu_f) = \frac{1}{|V|} \sum_{0 \le k < |V|} \cdot \frac{1}{\binom{|V|-1}{k}} \sum_{C \subseteq V \setminus \{i\}, |C|=k} (v(C\cup \{v_i\}) - v(C))
\end{equation}

\noindent The similar transformation was already made in equation~\ref{Eq:centrality:SV}, when we used it to compute efficiently Semivalues.

From the above formula we infer that to obtain an unbiased estimator the sampling method should satisfy two conditions:
\begin{itemize}
\item[(i)] the probability that a marginal contribution to randomly chosen $C$ is obtained from entering $v_i$ to a coalition of size $k$ is equal for every $k$:
$\frac{1}{n+1} \cdot \frac{n-k}{n} + \frac{1}{n+1} \cdot \frac{k+1}{n} = \frac{1}{n}$ ($n = |V|$);\footnote{Note that a given marginal contribution appears twice in our process---the first term represents the probability that we select a coalition of size $k$ without player $v_i$, while the second one---that we select the corresponding coalition of size $k+1$, with player $v_i$.}
%\vspace{-0.4cm}
\item[(ii)] marginal contributions to all coalitions of size $k$ are chosen with the same probability.
%\vspace{-0.1cm}
\end{itemize}

For instance consider a pivotal node $v_i$ in coalition $C$ of size $k+1$. The expected marginal contribution when $C \in \tilde{\mC}$ and $C \cup \set{v_i} \in \mC$ equals $\frac{1}{n+1} \cdot \binom{1}{\binom{n-1}{k-1}} \cdot \nu_f(C)$. On the other hand, the same marginal contribution is obtained from leaving connected coalition $C$ -- here expected value equals $\frac{1}{n+1} \cdot \binom{1}{\binom{n-1}{k}} \cdot \nu_f(C)$. As we won't consider transfers from the disconnected coalitions, we have to ensure that expected value from this marginal contribution does not change, so if we obtain $C$ with the probability $\frac{1}{n+1} \cdot \binom{1}{\binom{n-1}{k}}$ we add marginal contribution which equals $\nu_f(C) \cdot (1 + \frac{n-k+1}{k}) = \nu_f(C) \cdot (\frac{n+1}{k})$.

%%%%%%%%%%%%%%%%%%%%%%%%%%%%%%%%%%%%%%%%%%%%%%%%%%%%%%%%%%%%%%%%%%%%%%%%%%%%%%%%%%%%%%%%%%%%%%%%%%%%
\RestyleAlgo{ruled} %The old version of this is: \restylealgo
\begin{algorithm}[h]
\SetAlgoVlined %The old version of this is: \SetAlgoVlined
\LinesNumbered %The old version of this is: \LinesNumbered
\caption{\textbf{ApproximateSVCG} Approximation algorithm for the Shapley value}\label{ter:algo:approx}

\KwIn{Graph $G\hspace{-0.1cm}=\hspace{-0.1cm}(V,E)$ and characteristic function~$\nu_f$}
\KwOut{Shapley value $\phi^{\SV}_i(\nu_f)$ of each node $v_i \in  V$}
	\textbf{foreach} $v_i \in V$ \textbf{do} $\phi^{\SV}_i(\nu_f) \gets 0$; \\
%	\ForEach{$v_i \in V$} {
%		$SV_i(\nu_f) \gets 0$;
%	}
	
	\For{$it = 1$ {\normalfont\textbf{to}} $maxIter$} {
		$k \gets \textit{random number from } \{0, \ldots, |V|\}$; \\ \nllabel{sample_size}
		$C \gets \textit{random coalition of size }k$; \\
		\textbf{if} \textit{!CheckConnectedness(C)} \textbf{then} \textit{continue};\\ \nllabel{only_connected}
		$P \gets FindPivotals(C);$ \\
		\ForEach(\Comment{case (a)}){$v_i \in C \setminus P$} {  \nllabel{ter_contr_start}
%		\vspace{-0.05cm}
			$\phi^{\SV}_i(\nu_f)\gets \phi^{\SV}_i(\nu_f)
			+ \frac{|V|+1}{|C|} \cdot (\nu_f(C)-\nu_f(C \setminus \{v_i\}))$
%		\vspace{-0.05cm}			
		}
		\ForEach(\Comment{case (b)}){$v_i \in P$} {
%		\vspace{-0.05cm}
			$\phi^{\SV}_i(\nu_f)\gets \phi^{\SV}_i(\nu_f)
			 + \frac{|V|+1}{|C|} \cdot \nu_f(C)$
%		\vspace{-0.05cm}			
		}
		\ForEach(\Comment{case (c)}){$v_i \in (V \setminus C) \setminus N(C)$} {
%		\vspace{-0.05cm}
			$\phi^{\SV}_i(\nu_f)\gets \phi^{\SV}_i(\nu_f)
			- \frac{|V|+1}{|V|-|C|} \cdot \nu_f(C)$  \nllabel{ter_contr_stop}
%		\vspace{-0.1cm}			
		}
	}
	
	\textbf{foreach} $v_i \in V$ \textbf{do} $SV_i(\nu_f) \gets SV_i(\nu_f) / maxIter$;\\
\end{algorithm}
%%%%%%%%%%%%%%%%%%%%%%%%%%%%%%%%%%%%%%%%%%%%%%%%%%%%%%%%%%%%%%%%%%%%%%%%%%%%%%%%%%%%%%%%%%%%%%%%%%%%
%\vspace{-0.4cm}

This technique allows us to compute the marginal contributions of all agents for a randomly selected coalition $C$, which in the connectivity games can be performed much faster than estimating the Shapley value for each player separately \citep{Mann:Shapley:1962} or by sampling of a random permutation \citep{Castro:et:al:2009}, where we have to calculate the marginal contributions for a sequence of coalitions growing in size.

In the ApproximateSVCG algorithm we merge our technique with the analysis of marginal contribution presented in Figure~\ref{fig:example:contributions}. The pseudo code is presented in Algorithm~\ref{ter:algo:approx}. Line~\ref{sample_size} corresponds to \textbf{Faze 1}, where we sample the size of a coalition. Now, we modify \textbf{Faze 2} in order to select only \emph{connected coalitions} (line~\ref{only_connected}). See that generating a random \emph{connected coalition} uniformly will create a biased algorithm. We also modify \textbf{Faze 3}, when we consider cases (a), (b), and (c) from Figure~\ref{fig:example:contributions} (lines \ref{ter_contr_start}-\ref{ter_contr_stop}). The modification of \textbf{Faze 3} has to be done due to the following reason: since we no longer consider disconnected coalitions, any non-zero marginal contribution made to a disconnected coalition have to be transfered to a corresponding connected coalition (lines 8, 10, and 12). Furthermore, this should be done in a way that preserves appropriate probabilities (thus, in lines 8, 10, and 12 we multiply marginal contributions by adequate weights). Finally, we divide the sum of the contributions by the number of iterations (lines 13). For every sample, algorithm runs in time $O(|V|+|E|)$.

\section{Performance Evaluation}\label{Section:simulations}\label{chap:ter:simulations}

%An argument has been made[30] that terrorist networks may exhibit features of scale-free networks and can thus be treated as such in analysis and derivation of attack scenarios.

\noindent In this section, we empirically evaluate our FasterSVCG algorithm focusing on randomly generated networks as well as on real-life data from 9/11 terrorists network. We also evaluate how ApproximateSVCG approximates ranking, by benchmarking it against exact solution obtained by our algorithm.

The random graphs used in experiments are based on two topologies commonly found in social networks, and in terrorist organizations in particular \citep{Krebs:2002,Sageman:2004}: (i) \emph{scale-free graphs}, where the network is generated according to a power law; and (ii) \emph{random trees}, which model hierarchical organisations.  To construct scale-free graphs, we use the preferential attachment generation model \citep{Albert:et:al:2000}, with parameters $k = 2, 3$. In this model, while gradually constructing a graph, every new node $v_i$ is linked to $k$ already existing nodes such that the probability that $v_i$ is linked to the node $v_j$ is $\frac {degree(v_i)} {\sum_j degree(v_j)}$. To construct trees, every new node is attached to a randomly picked incumbent. Finally, we include in our analysis (iii) \emph{complete graphs}, where all coalitions are connected. Although complete graphs are unlikely to arise in a terrorist network context, they constitute a suitable benchmark for our simulations. For all the games on random graphs we assume that $f$ is defined as in Lindelauf's centrality (Definition~\ref{def:lin:cen}, equation~\ref{equation:f1}).

%%%%%%%%%%%%%%%%%%%%%%%%%%%%%%%%%%%%%%%%%%%%%%%%%%%%%%%%%%%%%%%%%%%%%%%%%%%%%%%%%%%%%%%%%%%%%%%%%%%%%%%%%%%%%%%%%%%%%%%
\begin{figure}[t]
\begin{center}
\includegraphics[]{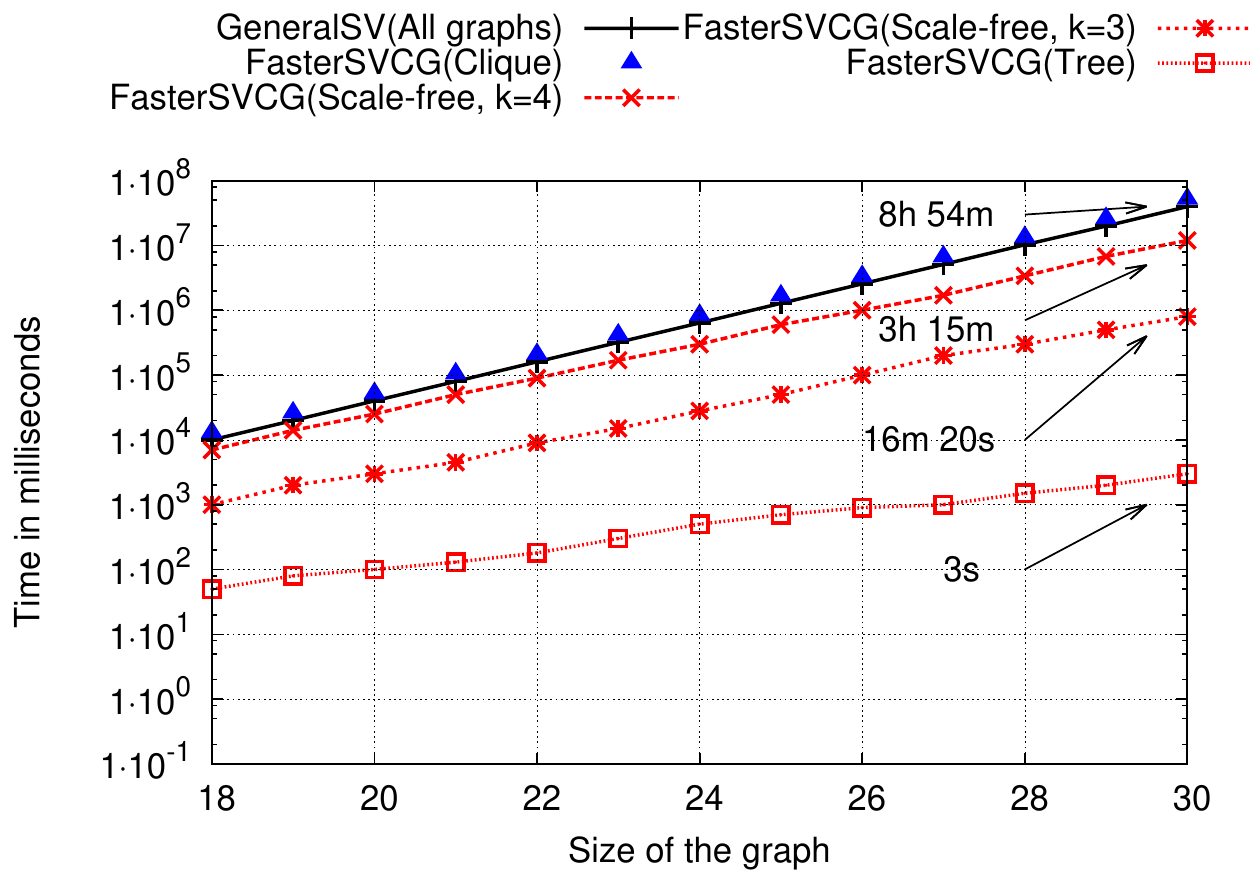}
%\vspace{-0.6cm}
\caption{Time performance of both algorithms.}
\label{figure:random_graphs}
\end{center}
%\vspace{-0.6cm}
\end{figure}

%\vspace*{1ex}
\noindent \textbf{FasterSVCG:} In Figure~\ref{figure:random_graphs}, we compare the performance of FasterSVCG to the general-purpose Shapley value algorithm (called GeneralSV and presented in Algorithm~\ref{algo_brute}). This latter algorithm does not take into account the structure of the network, so its runtime is the same for any type of graph. Unlike GeneralSV, FasterSVCG takes advantage of the sparsity of the network and, consequently, significantly outperforms the benchmark. For instance, for SFG($k=2$) and $|V|=30$, FasterSVCG needs only $0.39\%$ of GeneralSV runtime! Naturally, the best performance is obtained for tree graphs, where the game over a network of $30$ nodes can be computed in about $3$ seconds.

\noindent \textbf{ApproximateSVCG:} In Figure~\ref{figure:random_graphs2}, we evaluate the time performance of ApproximateSVCG. Importantly, for our purpose of identifying key terrorists, we are mostly interested in the approximation of the correct Shapley value-based ranking of top nodes, and less in the approximation of their actual Shapley values. To this end, we evaluate the time required by ApproximateSVCG to obtain the ranking of the top $\lceil \sqrt{n} \rceil$ nodes, with at most one error (i.e., one inversion compared to the exact ranking). We present an average time calculated over 500 iterations, each for a randomly selected SFG($k=4$). ApproximateSVCG returns top nodes much faster than FasterSVCG---for the graph of $24$ nodes the former algorithm runs in $0.25$ sec. and the exact one in $72$ sec.

Figure~\ref{figure:error_graphs} shows the comparison between the error convergence of the ApproximateSVCG and the error convergence of random permutation sampling studied by \cite{Castro:et:al:2009} (see Algorithm~\ref{algo:gtc:mc}). In this experiment we focus on the maximum absolute error of the Shapley value. It is computed as a percentage of the value of the grand coalition. The plot presents an average error from over $30$ iteration for Krebs' 9/11 WTC terrorist network with $36$ nodes. Since connected coalitions constitute only only $0.59\%$ of all coalitions, ApproximateSVCG outperforms Castro's method. Indeed, after $4$ seconds, the average error of ApproximateSVCG equals $0.029\%$, while the error for Castro et al. exceeds $0.2\%$. The absolute error in ApproximateSVCG converges to zero which indicates that this sampling method is not biased.
%%%%%%%%%%%%%%%%%%%%%%%%%%%%%%%%%%%%%%%%%%%%%%%%%%%%%%%%%%%%%%%%%%%%%%%%%%%%%%%%%%%%%%%%%%%%%%%%%%%%%%%%%%%%%%%%%%%%%%%
\begin{figure}[t]
\begin{center}
\includegraphics[]{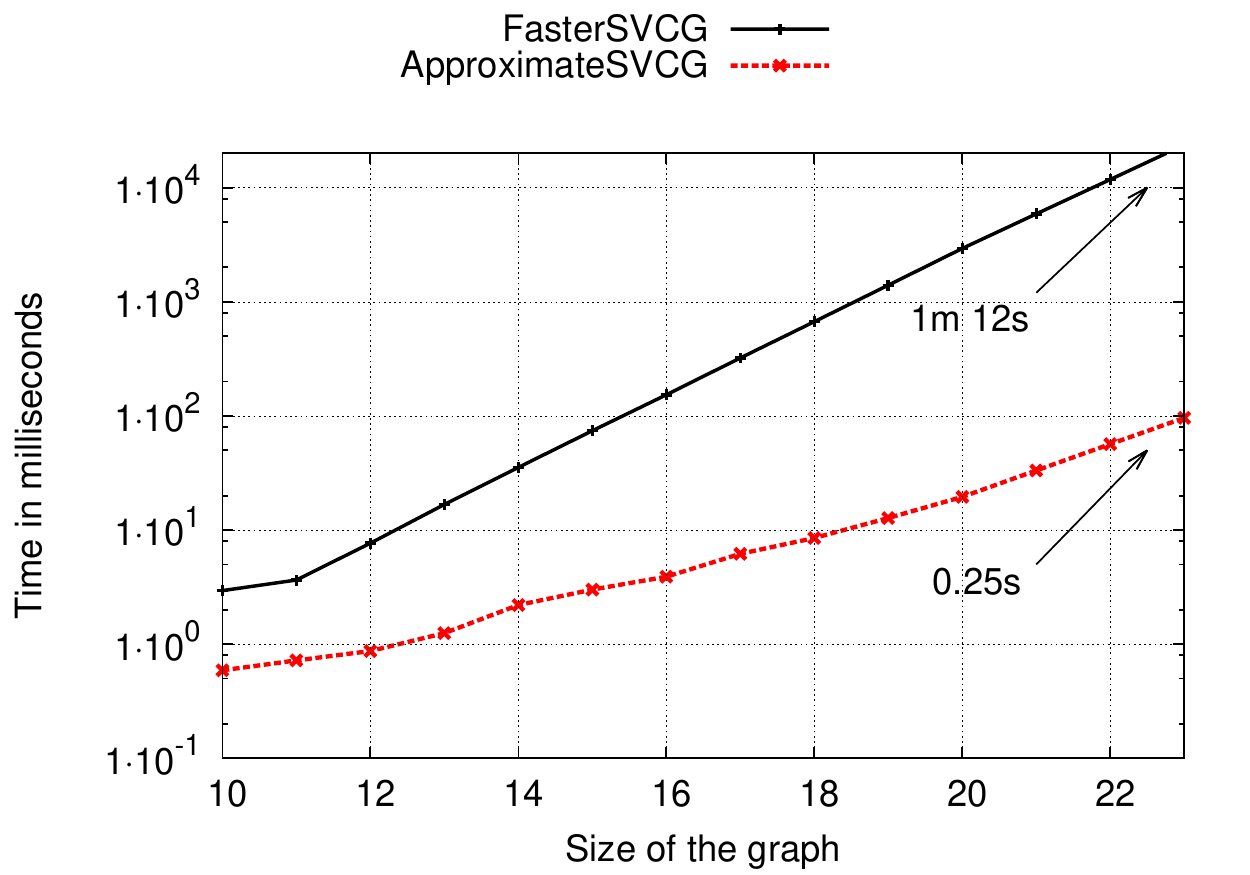}
%\vspace{-0.6cm}
\caption{Time performance of our exact and approximation algorithms.}
\label{figure:random_graphs2}
\end{center}
%\vspace{-0.6cm}
\end{figure}
\def\arraystretch{1.1}
\begin{table}[h]
\begin{center}
\begin{tabular}{|l|c|c|}
\hline
 & Lindelauf et al. & Our analysis \\
 & $|V| = 19$, $|E| = 32$ & $|V| = 26$, $|E| = 125$\\
\hline
1. & A. Aziz Al-Omari & Z. Moussaoui\\
2. & H. Alghamdi &  A. Aziz Al-Omari\\
3. & W. Alsheri & M. Atta\\
4. & H. Hanjour & W. Alshehri\\
5. & M. Al-Shehhi & N. Alhazmi\\
\hline
\end{tabular}
%\vspace{0.1cm}
\caption{FasterSVCG allows us to consider the bigger network of WTC 9/11 attack than Lindelauf \textit{et al.} which delivers new insights into the leadership structure of this network}
\label{table:comparison}
\end{center}
%\vspace{-0.3cm}
\end{table}

%%%%%%%%%%%%%%%%%%%%%%%%%%%%%%%%%%%%%%%%%%%%%%%%%%%%%%%%%%%%%%%%%%%%%%%%%%%%%%%%%%%%%%%%%%%%%%%%%%%%%%%%%%%%%%%%%%%%%%%
\begin{figure}[t]
\begin{center}
\includegraphics[]{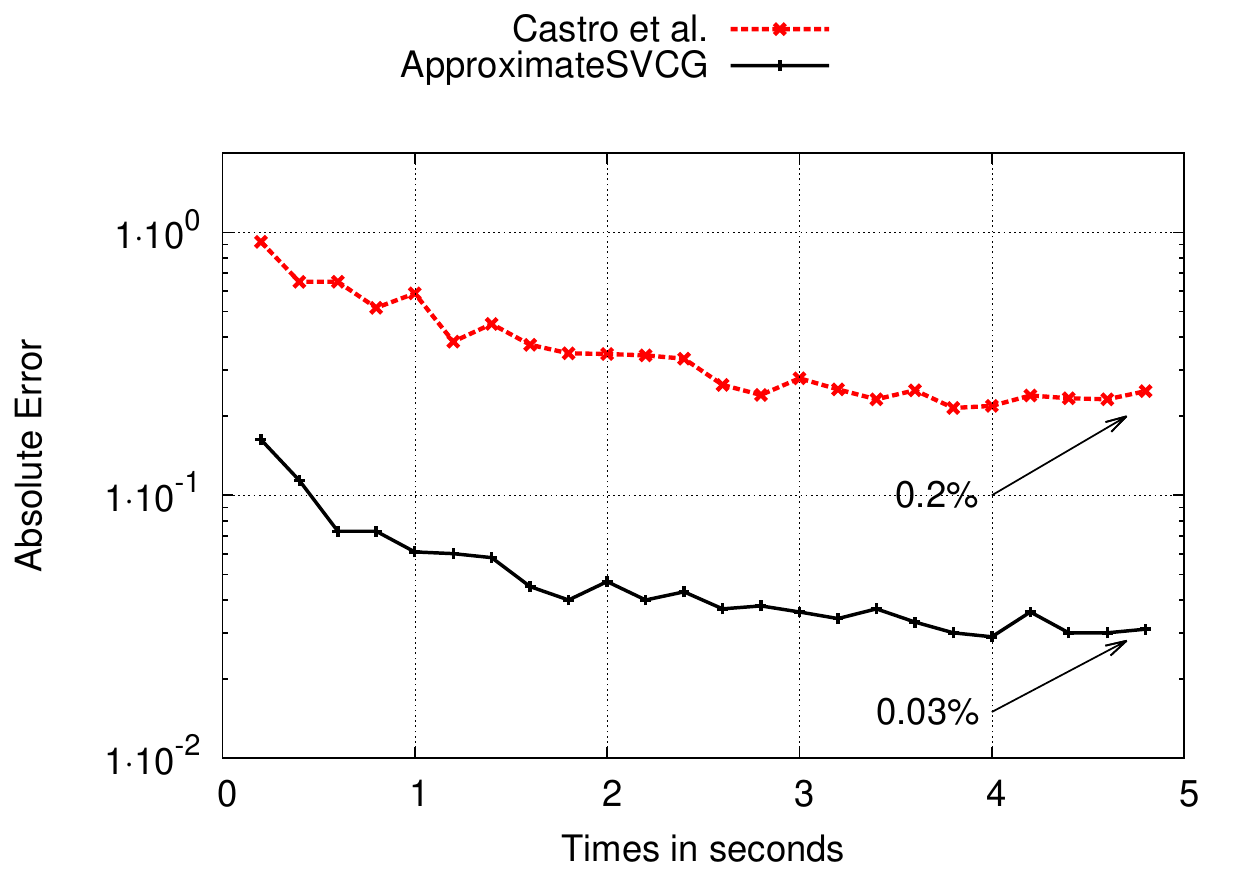}

\caption{Error performance of ApproximateSVCG.}
\label{figure:error_graphs}
\end{center}
%\vspace{-0.6cm}
\end{figure}
%%%%%%%%%%%%%%%%%%%%%%%%%%%%%%%%%%%%%%%%%%%%%%%%%%%%%%%%%%%%%%%%%%%%%%%%%%%%%%%%%%%%%%%%%%%%%%%%%%%%%%%%%%%%%%%%%%%%%%%
%%%%%%%%%%%%%%%%%%%%%%%%%%%%%%%%%%%%%%%%%%%%%%%%%%%%%%%%%%%%%%%%%%%%%%%%%%%%%%%%%%%%%%%%%%%%%%%%%%%%%%%%%%%%%%%%%%%%%%

%\vspace*{1ex}
\noindent \textbf{New insight on Krebs' 9/11 WTC network:} In the already classic work, \citet{Krebs:2002} constructed the 9/11 network from publicly available sources and computed standard centrality metrics to determine the key players in this network. Lindelauf~\textit{et al.} instead used their new centrality metric and, similarly to the case of the Jemaah Islamiyah's network, showed that the Shapley value-based approach delivers qualitatively better insights. However, Lindelauf~\textit{et al.}'s analysis focuses only on the network of 19 hijackers with 32 relationships, whereas Krebs reported also a bigger network of 36 nodes and 125 edges (mentioned above) that included accomplices. With FasterSVCG, we were able to analyse this bigger network as well.

Importantly, the five top terrorists (out of 36) identified by FasterSVCG differ from the top five terrorists reported by Lindelauf~\emph{et al.} when only the network of 19 terrorist was taken into account---see Table \ref{table:comparison}.
This difference is especially striking with respect to \textbf{M. Atta}, who was widely believed to be one of the ring-leaders of the conspiracy \citep{Krebs:2002} and was positioned the third place in our computations, but was classified in Lindelauf~\emph{et al.}'s work only on the 6th place. This clearly shows the importance of developing tools to analyse the biggest possible networks.

%Simulations_Figure01.png

\section{Conclusions}\label{section:summary}

\noindent In this chapter, the new algorithm was proposed to compute the Shapley value for the connectivity games of Amer and Gimenez defined over arbitrary graphs. These graphs are used by security services to represent terrorist networks. It was shown that the problem is \#P-Hard. Nevertheless, using the new exact algorithm we were able to analyse in moderate time the terrorist network responsible for the 9/11 WTC terrorist attacks. Additionally, in this chapter a new approximate algorithm was proposed that allows for an efficient study of bigger networks.

\chapter{The Marginal Contribution Networks for Generalized Coalitional Games}\label{chap:mcnets}

\noindent  \emph{Generalized characteristic function games} are a variation of characteristic function games, in which the value of a coalition depends not only on the identities of its members, but also on the order in which the coalition is formed. This class of games is a useful abstraction for a number of realistic settings and economic situations, such as modeling relationships in social networks. To date, two main extensions of the Shapley value have been proposed for generalized characteristic function games: the Nowak-Radzik \citeyearpar{Nowak:Radzik:1994} value and the  S\'{a}nchez-Berganti\~{n}os~\citeyearpar{Sanchez:Bergantinos:1999} value. In this context, the present chapter studies generalized characteristic function games from the point of view of representation and computation. Specifically, we make two key contributions. First, we propose  the Generalized Marginal-Contribution Nets that is a fully expressive mechanism that can in natural way represent generalized characteristic function games. Second, in order to facilitate an efficient computation, we propose the polynomial algorithms for computing Nowak-Radzik value and S\'{a}nchez-Berganti value. Our representation extends the results of \citet{Ieong:Shoham:2005} and \citet{Elkind:et:al:2009} for characteristic function games and retains their attractive computational properties. Our additional contribution is to use such a representation in order to compute game-theoretic centrality measures defined using either the Nowak-Radzik value or the S\'{a}nchez-Berganti\~{n}os values that extends path betweenness centrality introduced by \citet{Puzis:et:al:2007}.

All aforementioned contributions were made exclusively by the author of this dissertation.

\section{Preliminaries}

\noindent In this chapter we introduce the basic notation for \emph{generalized coalitional games} and also we present two extensions of the Shapley value (Definition~\ref{def:sv}) to these games.

\subsection{The Generalized Coalitional Games}

\noindent For each coalition $C \in 2^N\setminus \{\emptyset\}$, denote by $\Pi(C)$ the set of all possible permutations of the players in $C$. Any such permutation will be called an \emph{ordered coalition}. An arbitrary ordered coalition will often be denoted as $T$, while the set of all such coalitions will be denoted $\mathcal{T}$. That is, $\mathcal{T} = \bigcup_{C \in 2^N} \Pi(C)$.  A \textit{generalized characteristic  function} $\nu^g$ is a mapping $\nu^g : \mathcal{T} \rightarrow\R$, where it is assumed that $\nu^g(\emptyset) = 0$.

\begin{definition}[Generalized coalitional game]\label{def:gene:coalitional:game}
A generalized coalitional game consists of a set of \emph{players} (or \emph{agents}) $A = \{a_1, a_2, \ldots , a_{|A|}\}$, and a \emph{generalized characteristic function} $\nu^g: \mathcal{T} \rightarrow\R$, which assigns to each \emph{ordered coalition} of players~$T \in \mathcal{T}$ a real value (or payoff) indicating its performance, where $\nu^g(\emptyset) = 0$. Thus, the generalized coalitional game $g^g$ in characteristic function form is a pair $g^g=(A,\nu^g)$.
\end{definition}

\noindent A \textit{game in  generalized characteristic function form} is a tuple $(A,\nu^g)$, but will sometimes be denoted by $\nu^g$ alone. Sometimes, we will write $a_i$ instead of $(a_i)$ for brevity.

In some cases we will refer to the members of an ordered coalition $T$ using their names, e.g., write $T=(a_5,a_2,a_3)$, while other times we may refer to them using a lower case of the same letter: $T =(t_1,\dots, t_{|T|})$, meaning that $t_i$ is the $i^{\textnormal{th}}$ agent in $T$. Furthermore, given two disjoint ordered coalitions, $T=(t_1,\dots, t_{|T|}) \in \T$ and $S=(s_1,\dots, s_{|S|}) \in \T$, we write $(T,S)^k$ to denote the ordered coalition that results from inserting $S$ at the $k^{th}$ position in $T$. That is, $(T,S)^k = (t_1,\dots,t_{k-1}, s_1,\dots, s_{|S|}, t_k, \dots, t_{|T|})$. With a slight abuse of notation, we write $(T,a_i)^k$ to denote $(T,(a_i))^k$. Furthermore, we write $(a_i,T)$, and $(T,a_i)$, to denote the ordered coalition that results from inserting $a_i$ to $T$ as the first player, and the last player, respectively.

Next, we extend the notion of a subset to ordered sets.

\begin{definition}[The subset of ordered sets]\label{definition:specialSubset}
For any two ordered coalitions $S = (s_1,\ldots,s_{|S|}) \in \T$ and $T = (t_1, \ldots, t_{|T|}) \in \T$, we say that $T$ is a \textbf{subset} of $S$, and write $T\specialsubseteq S$, if and only if $T$ is a subsequence of $S$, i.e., the following two conditions hold:
\begin{itemize}
\item Every member of $T$ is a member of $S$. More formally:
\[\forall t_i\in T,\ \  \exists s_k\in S: s_k = t_i.\]
\item For any two players, $t_i,t_j\in T$, if $t_i$ appears before $t_j$ in $T$, then $t_i$ also appears before $t_j$ in $S$. More formally:
\[\forall t_i, t_j \in T : i<j, \exists s_k,s_w \in S: k<w \mbox{ and } s_k = t_i \mbox{ and } s_w = t_j.\]
Following convention, we say that $T$ is a \emph{strict} subset of $S$, and write $T \specialsubset S$ (instead of $T \specialsubseteq S$), if the above two conditions are met, and $T\neq S$.
\end{itemize}
\end{definition}

\subsection{The Nowak and Radzik Value}

\noindent Whereas four axioms defined in Section~\ref{sec:SV} uniquely determine the Shapley value for characteristic function games, the situation is more complex for generalized games, because a player's marginal contribution (and consequently the \emph{null-player} axiom) depends on where the new player in the coalition is placed. In this respect, \citet{Nowak:Radzik:1994} developed an extension of the Shapley value by making perhaps the most natural assumption that the marginal contribution of a player is computed when this player is placed \emph{last} in the coalition. Let us denote this marginal contribution of $a_i$ to $T \in \T(N\setminus \{a_i\})$ in game $\nu^g$ (according to Nowak and Radzik's definition) as $\MCNR{T}{a_i}$. Then:
\begin{equation}\label{formula:mc:NR}
\MCNR{T}{a_i} =\nu^g((T,a_i)) - \nu^g(T).
\end{equation}
In what follows, for any ordered coalition, $T$, let $T(a_i)$ denote the sequence of players in $T$ that appear before $a_i$ (if $a_i \notin T$ then $T(a_i)=T$). For example, given $T=(a_1,a_3,a_4,a_6)$, we have $T(a_4)=(a_1,a_3)$. Using this notation, we have to following definition.

\begin{definition}[The Nowak-Radzik value]
In a generalized coalitional game $(A,\nu^g)$ the Nowak-Radzik value (or \emph{NR value} for short) for a player $a_i$ is given by:
\begin{equation}\label{formula:NR:2nd}
\phi^{\NR}_i(A,\nu^g) = \frac {1}{|A|!} \sum_{T \in \Pi(A)}\MCNR{T(a_i)}{a_i} = \mathbb{E}[\MCNR{T(a_i)}{a_i}].
\end{equation}
\end{definition}

\noindent This can be written differently as follows:
\begin{equation}\label{formula:NR}
\phi^{\NR}_i(A,\nu^g) = \sum_{C \subseteq A\setminus \set{a_i}}\sum_{T \in \Pi(C)}   \frac {(|A|-|T|-1)!} {|A|!} [\nu^g((T,a_i)) - \nu^g(T)].
\end{equation}
The NR value is the unique value that satisfies the following ``fairness'' axioms:
\begin{center}
\begin{tabular}{ll}
\textbf{Efficiency:} & $\sum_{a_i \in A} \phi^{\NR}_i(\nu^g) = \frac{1}{|A|!}\sum_{T \in \Pi(A)} \nu^g(T).$\\[2ex]
\textbf{Null-Player:} & $\forall a_i\in A$, if $\nu^g(T) = \nu^g((T,a_i))$ $\forall T \in \mathcal{T} : a_i\notin T$, then $\phi^{\NR}_i(\nu^g) = 0.$ \\[2ex]
\textbf{Additivity:} & $\phi^{\NR}(\nu^g+\nu'^g) = \phi^{\NR}(\nu^g) + \phi^{\NR}(\nu'^g)$ for any two functions, $\nu^g$ and $\nu'^g$.
\end{tabular}
\end{center}

\subsection{The S\'{a}nchez and Berganti\~{n}os Value}

\noindent \citet{Sanchez:Bergantinos:1999} developed an alternative extension of the Shapley value based on the definition of the marginal contribution, where, instead of assuming that this player will be placed last, the authors take the average over all possible positions in which the player can be placed:
\begin{equation}\label{eq:marginalContribution:SB}
\MCSB{T}{a_i} = \frac{1}{(|T|+1)} \sum_{l=1}^{|T|+1} [\nu^g((T,a_i)^l) - \nu^g(T)].
\end{equation}
\noindent Using above, we can introduce the following definition.

\begin{definition}[The S\'{a}nchez-Berganti\~{n}os value]
In a generalized coalitional game $(A,\nu^g)$ the S\'{a}nchez-Berganti\~{n}os value (or \emph{SB value} for short) for a player $a_i$ is given by:

\begin{equation}\label{formula:SB:2nd}
\phi^{\SB}_i(A,\nu^g) = \frac {1}{|A|!} \sum_{T \in \Pi(A)} \MCSB{T(a_i)}{a_i} = \mathbb{E}[\MCSB{T(a_i)}{a_i}].
\end{equation}
\end{definition}

This also can be rewritten differently as follows:
\begin{equation}\label{formula:SB}
\phi^{\SB}_i(A,\nu^g) =  \sum_{C \subseteq A\setminus \set{a_i}}\sum_{T \in \Pi(C)}  \frac {(|A|-|T|-1)!} {|A|!(|T|+1)} \sum_{l=1}^{|T|+1}[\nu^g((T,a_i)^l) - \nu^g(T)].
\end{equation}

The SB value is the unique value that satisfies NR's efficiency and additivity axioms and the following axioms:

\begin{center}
\begin{tabular}{ll}
\textbf{Null-Player} & If $\forall T \in \mathcal{T}\ \forall l\in\{1,..,|T|+1\} \;:\; \nu^g((T,a_i)^l) = \nu^g(T)$, then $\phi^{\SB}_i(\nu^g) = 0$. \\[2ex]
\textbf{Symmetry} & If $\forall T \in \mathcal{T}_{-\{i,j\}}\forall l\in\{1,..,|T|+1\} \;:\; \nu^g((T,a_i)^l) = \nu^g((T,a_j)^l)$, then \\
& $\phi^{\SB}_i(\nu^g) = \phi^{\SB}_j(\nu^g)$.\\
\end{tabular}
\end{center}

The difference between the NR and SB values is illustrated in the following example:
\begin{example}\label{Example:Difference}
Consider a game with an ordered coalition $T^* \in \Pi(A)$ such that $\nu^g(T) = 1$ if $T = T^*$ and $\nu^g(T) = 0$ otherwise. Then, the average value of the grand coalition, taken over all possible orders, which is $\frac{1}{|A|!}$, needs to be distributed among the players. Using the NR value, we get $\phi^{\NR}_{t}(A) = \frac{1}{|A|!}$, where $a_t$ is the last player in the ordered coalition $T^*$, and we get $\phi^{\NR}_{i}(A) = 0$ for all $a_i \in A\setminus\{a_t\}$. In contrast, using the SB value, we get $\phi^{\SB}_i(A) = \frac{1}{|A|!\cdot|T^{\prime}|} = \frac{1}{|A|! \cdot |A|}$ for all $a_i \in A$. As can be seen, in this example, the NR value rewards the last player in the order, whereas the SB value rewards all players equally.
\end{example}

Having introduced the extensions of Shapley value to generalized characteristic function games, in the following section we consider the issue of representation.

\section{Representation}\label{Section:Representation}
\noindent In this section we consider the issue of efficient computation of the NR and SB values. Recall that a straightforward computation of the NR value, using either formulas \eqref{formula:NR:2nd} or \eqref{formula:NR}, requires iterating over \emph{all ordered coalition}. The same holds for the SB value and formulas \eqref{formula:SB:2nd} or \eqref{formula:SB}. Now since, for any game of $n$ players, there are $\sum_{i=1}^n\left( \binom {n} {i} \times i!\right)$ ordered coalitions, such an approach quickly becomes prohibitive with increasing $n$.\footnote{\footnotesize For instance, the size of the generalized characteristic function for $n = 5$ is $325$, whereas for $n=10$ is $9864100$.} To tackle this computational problem we introduce in this section a \emph{compact representation scheme} for games with ordered coalitions and evaluate its properties. We take as our starting point the \emph{marginal contribution nets} (\emph{MC-Nets}) representation by~\citet{Ieong:Shoham:2005} (see Section~\ref{sec:mcnets}) and its extension, \emph{read-once marginal contribution nets} (or \emph{read-once MC-Nets}) by \citet{Elkind:et:al:2009} (see Section~\ref{sec:read:once:mcnets}). This representation has a number of desirable properties: it is fully expressive, concise for many characteristic function games of interest, and facilitates a very efficient technique for computing the Shapley value. Importantly, our representation, which we call \emph{generalized read-once MC-Nets}, also has these desirable properties.

This section is divided into four subsections. In Section~\ref{subsection:generalizedMCnet}, we introduce the \emph{generalized read-once MC-Nets} representation to represent generalized characteristic function games compactly. After that, in Sections \ref{subsection:representation:NR} and~\ref{subsection:representation:SB}, we prove that this representation facilitates an efficient computation of the NR value and the SB value, respectively. Finally, in Section~\ref{subsection:representation:Centrality}, we present a sample application of generalized read-once MC-Nets to compute game-theoretic centrality based on the NR and SB values.

%%%%%%%%%%%%%%%%%%%%%%%%%%%%%%%%%%%%%%%%%%%%%%%%%%%%%%%%%%%%%%%%%%%%%%%%%%%%%%%
\subsection{Generalized Read-Once MC-Nets}\label{subsection:generalizedMCnet}
\noindent Having presented the MC-Net representation for characteristic function games in Section~\ref{sec:mcnets} and read-once MC-Net on Section~\ref{sec:read:once:mcnets}, we now present our extension of MC-Nets to generalized characteristic function games. More specifically, we extend the rules so that the formula $\mathcal{F}$ on the left hand side of a rule is applicable to an \textit{ordered} coalition. We introduce two \emph{atomic formulas}:

\[
\begin{array}{rl}
	\text{basic atomic formula (BAF)} & \{ a_m,\ldots,a_n \},\\
	\text{ordered atomic formula (OAF)} &  \langle a_i,\ldots, a_j \rangle.
\end{array}
\]

The ordered coalition $T$ meets a basic atomic formula, denoted $ \mathcal{F_{\mathit{BAF}}}$, if
all literals occurring in $ \mathcal{F_{\mathit{BAF}}} $ belong to
$T$. More formally
\[
T \models \mathcal{F_{\mathit{BAF}}} \iff \forall a_i \in \mathcal{F_{\mathit{BAF}}}, a_i \in T.
\]

\begin{example}
 Let $N=\{a_1,\ldots,a_{10}\}$ and consider $T=(a_5,a_1,a_4,a_3,a_2)$. We have:
 \[
 \begin{array}{l}
 T \vDash \{a_2,a_5\},  \\
 T \nvDash \{a_5,a_2,a_4,a_6\}.  \\
 \end{array}
 \]
\end{example}

Now, we introduce an ordered atomic formula, denoted $ \mathcal{F_{\mathit{OAF}}} $. This formula is our crucial extension to MC-Nets; it allows us to fully express the values of all ordered coalitions. Specifically, the ordered coalition $T$ meets $ \mathcal{F_{\mathit{OAF}}} $  if all literals occurring in $ \mathcal{F_{\mathit{OAF}}} $ appear in $T$ in the same order as in the $\mathcal{F_{\mathit{OAF}}}$  formula.

To state this more formally, let us first introduce the notion of an ordered coalition corresponding to $\mathcal{F_{\mathit{OAF}}}$. To this end, for every $\mathcal{F_{\mathit{OAF}}}=\langle a_i,\ldots, a_j\rangle$, the \textit{corresponding ordered coalition of} $\mathcal{F_{\mathit{OAF}}}$ is simply $S^{\mathit{OAF}}=(a_i,\ldots, a_j)$. The main difference between the two is that $a_i$ is a literal in $\mathcal{F_{\mathit{OAF}}}$, while it is an agent in $S^{\mathit{OAF}}$. Now, given a particular $\mathcal{F_{\mathit{OAF}}}$ and an ordered coalition $S^{\mathit{OAF}}$ corresponding to it, the following holds for every $T\in\T$:
$$
T \models \mathcal{F_{\mathit{OAF}}} \iff S^{\mathit{OAF}} \specialsubseteq T.
$$
\noindent In other words, $S^{\mathit{OAF}}$ is a subsequence of $T$. This has a natural interpretation: when a group of players joins any coalition in a certain order they contribute to it a certain value. Consider the following example:

\begin{example}
 Let $N=\{a_1,\ldots,a_{10}\}$ and consider $T=(a_5,a_1,a_4,a_3,a_2)$. We have:
 \[
 \begin{array}{l}
 T \vDash \langle a_5,a_3,a_2 \rangle,  \\
 T \nvDash \langle a_5,a_2,a_3 \rangle,\\
 T \nvDash \langle a_5,a_3,a_2,a_6 \rangle.\\
 \end{array}
 \]
\end{example}

Now, we are able to formalize our representation which we call the \emph{generalized read-once MC-Nets}. Specifically, in this representation, a rule, called \textit{generalized read-once rule}, is as follows:
\begin{equation}\label{generalizedrule}
\mathcal{F} \rightarrow V,
\end{equation}
\noindent where $V$ is a real value and the syntax of the \emph{generalized formula} $\mathcal{F}$ is:
\[
\begin{array}{rl}
	A := & \mathcal{F_{\mathit{BAF}}} ~|~ \mathcal{F_{\mathit{OAF}}}, \\
	\mathcal{F} := & (\mathcal{F}) ~|~ \neg \mathcal{F} ~|~ \mathcal{F} \wedge \mathcal{F} ~|~ \mathcal{F} \vee \mathcal{F} ~|~ \mathcal{F} \oplus \mathcal{F} ~|~ A,
\end{array}
\]
\noindent and each literal can only appear once in the formula. This
generalized formula is a binary rooted tree with atomic formulas on
leafs and internal nodes labeled with one of the following logical connectives: conjunction ($\wedge$), disjunction ($\vee$) or exclusive disjunction
($\oplus$). The logical connectives have the following standard interpretation:
\[
\begin{array}{lcl}
	T \vDash \neg \mathcal{F} &\iff& T \nvDash \mathcal{F}, \\
	T \vDash \mathcal{F}_1 \wedge \mathcal{F}_2 &\iff& T \vDash \mathcal{F}_1 ~and~ T \vDash \mathcal{F}_2, \\
	T \vDash \mathcal{F}_1 \vee \mathcal{F}_2 &\iff& T \vDash \mathcal{F}_1 ~or~ T \vDash \mathcal{F}_2, \\
	T \vDash \mathcal{F}_1 \oplus \mathcal{F}_2 &\iff& T \vDash (\mathcal{F}_1 \vee \mathcal{F}_2) \wedge \neg (\mathcal{F}_1 \wedge \mathcal{F}_2).
\end{array}
\]
\begin{example}
 Let $N=\{a_1,\ldots,a_{10}\}$ and consider $T=(a_5,a_1,a_4,a_3,a_2)$. We have:
 \[
 \begin{array}{l}
 T \vDash (\{a_2,a_5\} \wedge \langle a_1, a_3 \rangle ) \vee \{a_4,a_7,a_9\}, \\
 T \nvDash \langle a_4, a_5, a_3 \rangle \oplus (\{a_1\} \wedge \neg\{a_2\}). \\
 \end{array}
 \]
\end{example}

\begin{definition}[Generalized Read-Once MC-Nets]\label{def:generalized:mcnets}
The generalized coalitional game $(N,\nu^g)$ in the form of Generalized Read-once Marginal Contribution Network is a pair $(A,\mathcal{GR})$, where $ A $ is a set of players, and $\mathcal{GR}$ is a finite set of generalized rules. The value $\nu^g(T)$ of an ordered coalition $T$ is defined as the sum of all $\mathit{V}$ from the generalized read-once rules that are met
by $T$. More formally:
\begin{equation}\label{eqn:vg:MCnet}
\nu^g(T) = \sum_{\mathcal{GR} \ni \mathcal{F} \rightarrow \mathit{V}: ~T \models \mathcal{F}} \mathit{V}.
\end{equation}
\end{definition}

The following example presents a simple application of generalized read-once MC-Nets to the scheduling problem:

\begin{example}
Let $G=(V,E)$ be a directed acyclic graph representing a task-based domain, where each vertex in $V=\{v_1,\ldots,v_n\}$ represents a task. Each edge $(v_i,v_j) \in E$ illustrates that task $v_i$ has to be completed before task $v_j$. We say that the job is done if all tasks of a given graph $G$ are completed, in which case we get some revenue. Next we show how each such graph can be represented by a single generalized read-once rule, \textit{where agents represent tasks}. In particular, each formula in the rule represents an edge in $G$, as follows:
\[
\bigwedge_{(a_i,a_j) \in E}	\langle a_i, a_j \rangle  \to \text{revenue of the job defined by}\ G.
\]
An ordered coalition $T$ represents the workforce that can finish tasks in a given order. The total revenue of $T$ is then equal to the sum of revenues of all jobs feasible by $T$.
\end{example}

Having defined our representation formally, we now evaluate its properties. We start with \emph{expressiveness}:
\begin{proposition} [Expressiveness]\label{prop:expressive-power}
	Every ordered coalitional game can be expressed using generalized read-once MC-Nets.
\end{proposition}
\begin{proof}
	For any arbitrary generalized characteristic function $\nu^g$, we show how to construct a set of rules $\mathcal{GR}$ such that Equation~\eqref{eqn:vg:MCnet} holds for every $T \in \mathcal{T}$.

We will use only OAF to define a recursive set starting from singleton
sequences and empty set $\mathcal{GR}_0 = \emptyset$. Let $S^{\mathit{OAF}}$ be the ordered coalition corresponding to $\mathcal{F_{\mathit{OAF}}}$, and let $\left|\mathcal{F_{\mathit{OAF}}}\right|$ be the size of this formula.
\[
\begin{array}{rll}
	\mathcal{GR}_1 = & \big\{ \langle a_i \rangle \to \nu^g((a_i))\big\} & \forall a_i \in N, \\
	\mathcal{GR}_n = & \big\{\mathcal{F_{\mathit{OAF}}} \to \nu^g(S^{\mathit{OAF}}) - \sum_{k=1}^{n-1}\sum_{\mathcal{GR}_k \ni \mathcal{F} \rightarrow \mathit{V}:~ S^{\mathit{OAF}} \models \mathcal{F}} \mathit{V} \big\} &  \forall |\mathcal{F_{\mathit{OAF}}}| = n.
\end{array}
\]
Finally, we get $\mathcal{GR} = \bigcup_{i}\mathcal{GR}_i$.
\end{proof}

In the above proof of Proposition~\ref{prop:expressive-power}, it emerges that only OAFs are necessary in order to guarantee that the generalized MC-Nets are fully expressive. In the following basic example we present that in some games the BAF gives us desirable succinctness.

\begin{example}
Let $N=\{a_1,\ldots,a_n\}$ be the set of tasks. In order to get profit $V$, all tasks must be done in arbitrary chosen order. It is clear that representing this problem using only OAF requires $n!$ number of formulas:

\[
\begin{array}{rl}
	\langle a_1,\ldots,a_n \rangle & \to  V, \\
	\vdots &  \\
	\langle a_n,\ldots,a_1 \rangle & \to  V. \\
\end{array}
\]

By introducing BAF we can get exponentially more concise notation: $\{ a_1,\ldots,a_n \}  \to  V$.
\end{example}

It is worth making some observations about the conciseness of our
representation scheme. Our rule-based representation scheme is
ultimately based on propositional formulae, and the use of these
formulae is ultimately to define a Boolean function. The use of
propositional formulae to define Boolean functions is very standard: a
very fundamental result in the theory of propositional logic is that
any Boolean function of $n$ Boolean arguments can be represented by a
propositional formula in which the variables correspond to these
arguments. For many Boolean functions, it might be possible to obtain
\emph{compact} propositional formulae to represent them, but
crucially, this will not always be possible; some Boolean functions
will require formulae of size exponential in the number of variables
(see, e.g., \citet{Boppana:Sipser:1990}). It is in fact trivial to
produce examples of characteristic functions whose representation in
our framework is exponentially more succinct than explicitly listing
the value of every coalition, but from the above discussion, it also
follows that our representation cannot \emph{guarantee} compactness in
all cases. Finally, note that in the worst case, we can represent a
characteristic function with one rule for every possible input, which
essentially corresponds to the idea of listing the value of every
coalition. Thus, in the worst case, our representation scheme is no
worse than simply listing the value of every input. More formally:

\begin{corollary}[Conciseness] \label{corollary:at:least} Compared to the generalized function game representation, generalized read-once MC-nets are at least as concise, and for certain games exponentially more compact.
\end{corollary}

%%%%%%%%%%%%%%%%%%%%%%%%%%%%%%%%%%%%%%%%%%%%%%%%%%%%%%%%%%%%%%%%%%%%%%%%%%%%%%%
\subsection{Computing the Nowak-Radzik Value with Generalized Read-Once MC-Nets}\label{subsection:representation:NR}

\noindent In this subsection, we discuss the computational properties of the generalized read-once MC-Net representation
with respect to calculating the Nowak-Radzik value (i.e., the NR value).
We start our computational analysis with some observations about
atomic formulas. First, observe that all players occurring in some
$\mathcal{F_{\mathit{BAF}}}$ are indistinguishable. Furthermore, taking into
account that each player can appear in at most one formula, $\phi^{\NR}$ will assign to these
players the same values. The case of $\mathcal{F_{\mathit{OAF}}}$ is
different. Here, the NR value for all players except the last one
is zero (see Example \ref{Example:Difference}).

\begin{algorithm}[h]
\caption{\textbf{(COMP\_NR)} Computing the Nowak-Radzik value}\label{algo:mc:algo1}
\SetAlgoVlined %The old version of this is: \SetVline
\LinesNumbered %The old version of this is: \linesnumbered
\KwIn {$N = \{a_1,\ldots,a_n\},(\mathcal{F}_1 \to \mathit{V}_1),\ldots,(\mathcal{F}_r \to \mathit{V}_r)$}
\KwOut { $\phi^{\NR}_i$ for all $a_i \in N$}
\lFor{$a_i \in N$} {$\phi^{\NR}_i \gets 0;$}
\For{$j \gets 1$ \KwTo $r$} {
	$m \gets |\mathcal{F}_j|;$ \\
	$(\texttt{A},\texttt{B},\texttt{T},\texttt{F}) \gets Sh^{\NR};$  \Comment{the Pseudo code of $Sh^{\NR}$ is in Appendix~\ref{app:recursive}}\\
	\ForEach{$a_i \in X_{\mathcal{F}_j}$} {
		\If{$ \exists_k~ \texttt{A}_{k,i} \neq 0$}{
			$v_i \gets \frac{1}{m!}\sum_{k=0}^{m-1}(m-k-1)!\texttt{A}_{k,i};$  \label{nr1}
		}	
		\ElseIf{$ \exists_k~ \texttt{B}_{k,i} \neq 0$}{
			$v_i \gets -\frac{1}{m!}\sum_{k=0}^{m-1}(m-k-1)!\texttt{B}_{k,i};$ \label{nr2}
		} \Else {
			$v_i \gets 0;$
		}
		$\phi^{\NR}_i \gets \phi^{\NR}_i + \mathit{V}_j \cdot v_i;$
	}	
}
\end{algorithm}

\begin{theorem}\label{theorem:algorithmNR}
Given a generalized coalitional game $(N,\mathcal{GR})$, Algorithm \ref{algo:mc:algo1} computes the Nowak-Radzik value ($\phi^{\NR}_i$ for all $a_i \in N$) in polynomial time.
\end{theorem}

\begin{proof}
We start by introducing necessary notation, some of which is in the pseudo code of Algorithm~\ref{algo:mc:algo1}. For each formula $\mathcal{F}$, we denote by $X_{\mathcal{F}}$ the sets of players occurring in $\mathcal{F}$, and by $|\mathcal{F}|$ the size of $X_{\mathcal{F}}$, and by $\mathcal{T}_{\mathcal{F}}$ the set of all ordered coalitions within the set $X_{\mathcal{F}}$, i.e., $\mathcal{T}_{\mathcal{F}} = \bigcup_{C \subseteq X_{\mathcal{F}}} \Pi(C)$.  For all $a_i \in X_{\mathcal{F}}$ and $k \in \{0,\ldots,n\}$, we define the following quantities:
\[
\begin{array}{ll}
\texttt{A}_{k,i}(\mathcal{F})& = |T \in \mathcal{T}_{\mathcal{F}}: |T| = k, a_i \notin  T, T \nvDash \mathcal{F}, (T,a_i) \vDash \mathcal{F})|, \\
\texttt{B}_{k,i}(\mathcal{F})&  = |T \in \mathcal{T}_{\mathcal{F}}: |T| = k, a_i \notin  T, T \vDash \mathcal{F}, (T,a_i) \nvDash \mathcal{F})|, \\
\texttt{T}_{k}(\mathcal{F})&  = |T \in \mathcal{T}_{\mathcal{F}}: |T| = k, T \vDash \mathcal{F}|,  \\
\texttt{F}_{k}(\mathcal{F})&  = |T \in \mathcal{T}_{\mathcal{F}}: |T| = k, T \nvDash \mathcal{F}|.  \\
\end{array}
\]
The first quantity $\texttt{A}_{k,i}(\mathcal{F})$ is the number of ordered coalitions from $\mathcal{T}_{\mathcal{F}}$ of a given size $k$ (i.e., containing $k$ players) to which the addition of player $a_i$ at the end of the coalition causes the formula $\mathcal{F}$ to become satisfied. In contrast, the second quantity $\texttt{B}_{k,i}$ is the number of coalitions to which adding player $a_i$ causes the formula $\mathcal{F}$ to become unsatisfied. The last two quantities, $\texttt{T}_{k}(\mathcal{F})$ and $\texttt{F}_{k}(\mathcal{F})$, count the number of ordered coalition being satisfied, and not satisfied, respectively, by the formula $\mathcal{F}$. These quantities can be computed recursively as follows. Let us first focus on the case where $\mathcal{F}$ is an atomic formula. This case can be divided further into two cases:

\begin{itemize}
\item { \bf $ \mathcal{F} =  \{ a_1,\ldots, a_r \}$}: Here, $\texttt{A}_{k,i} = (r-1)! $ if $k = r-1 $ and $i \in X_{\mathcal{F}} $, otherwise $\texttt{A}_{k,i} = 0$. $\texttt{B}_{k,i}$ always equals 0. $\texttt{T}_{k} = r! $ if $k = r $, otherwise $\texttt{T}_{k} = 0$. Finally, $ \texttt{F}_{k} = k!  $ if $ k < r $, otherwise $\texttt{F}_k=0$.
\item { \bf $ \mathcal{F} =  \langle a_1,\ldots, a_r \rangle $}: Here $ \texttt{A}_{k,i} = 1 $ if $k=r-1$ and $i=r$, otherwise $ \texttt{A}_{k,i} = 0 $. $\texttt{B}_{k,i}$ always equals $0$. $\texttt{T}_{k}=1$ if $k=r$, otherwise $\texttt{T}_{k}=0$. Finally, $\texttt{F}_k = k! - \texttt{T}_k $ if $ k \leq r $, otherwise  $\texttt{F}_k = 0$.
\end{itemize}

Having discussed the case where $\mathcal{F}$ is an atomic formula, we now consider the case where $\mathcal{F}$ contains logical connectives. If such a connective involves
two subformulas $ \mathcal{F}_1$ and $ \mathcal{F}_2$, then we know
that $a_i$ can appear in $ \mathcal{F}_1$ or $ \mathcal{F}_2$, but not
in both. We also know that logical connectives used in the generalized read-once MC-Nets are symmetric, so we will consider only one case, $a_i \in \mathcal{F}_1$; the other is analogous.

In our analysis, for every two disjoint ordered coalitions, $T_1$ and $T_2$, we define a \emph{conflation} of $T_1$ and $T_2$ as an ordered coalition $T$ such that (1) every agent in $T$ appears in $T_1$ or $T_2$, and (2) the order of the agents in $T_1$ is retained in $T$, and the order of those in $T_2$ is also retained in $T$. The set of every conflation of $T_1$ and $T_2$ will be denoted $T_1 \times T_2$. More formally:
\[
T \in T_1 \times T_2 \iff X_{T} = X_{T_1} \cup X_{T_2} \wedge~  T' \specialsubset T \wedge T_2 \specialsubset T,
\]
\noindent where $X_{T}$ denotes a set of players in $T$.
We note that if $T_1 \vDash \mathcal{F}$ and $X_{\mathcal{F}} \cap T_2 =
\emptyset$ then $T \vDash \mathcal{F}  \forall T \in (T_1 \times T_2)$.

\begin{proposition}\label{proposition:permutation}
Two ordered coalitions
$T_1$ and $T_2$ can be conflated in ${{|T_1|+|T_2|}\choose {|T_2|}}$ ways.
\end{proposition}

\begin{proof}
Consider an sequence of $|T_1|+|T_2|$ empty slots, where each slot can take an agent. Out of all these slots, choose any $|T_2|$ slots, and place in them the agents from $T_2$ while retaining their order. As for the remaining slots, place in them the agents from $T_1$ while retaining their order. Clearly, every choice of the $|T_2|$ slots results in a unique conflation, and every conflation can be constructed by exactly one such choice of the $|T_2|$ slots. The number of conflations is then ${{|T_1|+1+|T_2|-1}\choose {|T_2|}}$.
\end{proof}

\noindent Next, we show how to compute $\texttt{A}_{k,i}$,$\texttt{B}_{k,i}$, $\texttt{T}_k$ and $\texttt{F}_k$. We will do this for each of the possible logical connective (i.e., $\neg$, $\land$, $\lor$ and $\oplus$) separately:

\begin{itemize}
\item { \bf $\mathcal{F} = \neg \mathcal{F}_1$}: In this case, the negation swaps quantities as follows. For all $i$ and $k$ we have $ \texttt{A}_{k,i}(\mathcal{F}) = \texttt{B}_{k,i}(\mathcal{F}_1) $, $ \texttt{B}_{k,i}(\mathcal{F}) = \texttt{A}_{k,i}(\mathcal{F}_1) $, $ \texttt{T}_{k}(\mathcal{F}) = \texttt{F}_k(\mathcal{F}_1) $, and $ \texttt{F}_{k}(\mathcal{F})= \texttt{T}_k(\mathcal{F}_1)$.

\item { \bf $\mathcal{F} = \mathcal{F}_1 \land \mathcal{F}_2$}. In this case, let $T\in T_1 \times T_2$, where $T_1 \in \Pi(X_{\mathcal{F}_1}), |T_1| = s$ and $T_2 \in \Pi(X_{\mathcal{F}_2}), |T_2| = k-s$. We have $(T,a_i) \vDash \mathcal{F}$ if and only if $(T_1,a_i) \vDash \mathcal{F}_1$ and  $T_2 \vDash \mathcal{F}_2$. In this case, $T \nvDash \mathcal{F}$ if and only if $T_1 \nvDash \mathcal{F}_1$. Consequently, using Proposition \ref{proposition:permutation} we get:
\[
\textstyle\texttt{A}_{k,i} = \sum_{s=0}^{k}{k \choose k-s}\texttt{A}_{s,i}(\mathcal{F}_1)\texttt{T}_{k-s}(\mathcal{F}_2).
\]
Furthermore, we have $T \vDash \mathcal{F}$ if and only if  $T_1
\vDash \mathcal{F}_1$ and  $T_2 \vDash \mathcal{F}_2$. Additionally in
this case we have:
$(T,a_i) \nvDash \mathcal{F} \Leftrightarrow (T_1,a_i) \nvDash
\mathcal{F}_1$. Consequently,
\[
\textstyle\texttt{B}_{k,i} = \sum_{s=0}^{k}{k \choose k-s}\texttt{B}_{s,i}(\mathcal{F}_1)\texttt{T}_{k-s}(\mathcal{F}_2).
\]
Eventually, for $\texttt{T}_k$ and $\texttt{F}_k$, we have:
\[
\begin{array}{ll}
\texttt{T}_{k} = &\sum_{s=0}^{k}{k \choose k-s}\texttt{T}_{s}(\mathcal{F}_1)\texttt{T}_{k-s}(\mathcal{F}_2),  \\
\texttt{F}_{k} = &\sum_{s=0}^{k}{k \choose k-s}(\texttt{F}_s\big(\mathcal{F}_1)\texttt{F}_{k-s}(\mathcal{F}_2) + \texttt{F}_s(\mathcal{F}_1)\texttt{T}_{k-s}(\mathcal{F}_2)
+ \texttt{T}_s(\mathcal{F}_1)\texttt{F}_{k-s}(\mathcal{F}_2)\big).
\end{array}
\]

\item { \bf $\mathcal{F} = \mathcal{F}_1 \lor \mathcal{F}_2$}. In this case, let $T\in T_1 \times T_2$, where $T_1 \in \Pi(X_{\mathcal{F}_1}), |T_1| = s$ and $T_2 \in \Pi(X_{\mathcal{F}_2}),|T_2| = k-s$. We know that $T \nvDash \mathcal{F}$ if and only if $T_1 \nvDash \mathcal{F}_1$ and $T_2 \nvDash \mathcal{F}_2$. Furthermore, in this case it holds that $(T,a_i) \vDash \mathcal{F} \Leftrightarrow (T_1,a_i) \vDash \mathcal{F}_1  $. Thus,
\[
\textstyle \texttt{A}_{k,i} = \sum_{s=0}^{k}{k \choose k-s}\texttt{A}_{s,i}(\mathcal{F}_1)\texttt{F}_{k-s}(\mathcal{F}_2).
\]
Analogously, we have $(T,a_i) \nvDash \mathcal{F}$ if and only if $(T_1,a_i) \nvDash \mathcal{F}_1$ and $T_2 \nvDash \mathcal{F}_2$, and in this case it holds that: $T \vDash \mathcal{F} \Leftrightarrow T_1 \vDash \mathcal{F}_1  $. Consequently,
\[
\textstyle \texttt{B}_{k,i} = \sum_{s=0}^{k}{k \choose k-s}\texttt{B}_{s,i}(\mathcal{F}_1)\texttt{F}_{k-s}(\mathcal{F}_2).
\]
Finally, for $\texttt{T}_k$ and $\texttt{F}_k$, we have:
\[
\begin{array}{ll}
\texttt{T}_{k} = &\sum_{s=0}^{k}{k \choose k-s}(\texttt{T}_s\big(\mathcal{F}_1)\texttt{T}_{k-s}(\mathcal{F}_2)   + \texttt{F}_s(\mathcal{F}_1)\texttt{T}_{k-s}(\mathcal{F}_2)
+ \texttt{T}_s(\mathcal{F}_1)\texttt{F}_{k-s}(\mathcal{F}_2)\big), \\
\texttt{F}_{k} = &\sum_{s=0}^{k}{k \choose k-s}\texttt{F}_{s}(\mathcal{F}_1)\texttt{F}_{k-s}(\mathcal{F}_2).
\end{array}
\]

\item { \bf $\mathcal{F} = \mathcal{F}_1 \oplus \mathcal{F}_2$}: Let $T=T_1 \times T_2$ such that $T_1 \in \Pi(X_{\mathcal{F}_1}), |T_1| = s$ and $T_2 \in \Pi(X_{\mathcal{F}_2}),|T_2| = k-s$. Now, the ordered coalition $T$ can contribute to $\texttt{A}_{k,i}$ in two ways. Firstly, when $T_1 \vDash \mathcal{F}_1$ and $T_2 \vDash \mathcal{F}_2$ and $(T_1,a_i) \nvDash \mathcal{F}_1$. Secondly, when $T_1 \nvDash \mathcal{F}_1$ and $T_2 \nvDash \mathcal{F}_2$ and $(T_1,a_i) \vDash \mathcal{F}_1$. Thus:
\[
%For the quantities $\texttt{T}_k$ and $\texttt{F}_k$ we have
\begin{array}{ll}
\textstyle \texttt{A}_{k,i} &= \sum_{s=0}^{k}{k \choose k-s}(\texttt{B}_{s,i}(\mathcal{F}_1)\texttt{T}_{k-s}(\mathcal{F}_2) + \texttt{A}_{s,i}(\mathcal{F}_1)\texttt{F}_{k-s}(\mathcal{F}_2)),\\
\textstyle \texttt{B}_{k,i} &= \sum_{s=0}^{k}{k \choose k-s}(\texttt{B}_{s,i}(\mathcal{F}_1)\texttt{F}_{k-s}(\mathcal{F}_2) + \texttt{A}_{s,i}(\mathcal{F}_1)\texttt{T}_{k-s}(\mathcal{F}_2)), \\
\texttt{T}_{k} &= \sum_{s=0}^{k}{k \choose k-s}(\texttt{T}_s(\mathcal{F}_1)\texttt{F}_{k-s}(\mathcal{F}_2)+ \texttt{F}_s(\mathcal{F}_1)\texttt{T}_{k-s}(\mathcal{F}_2)), \\
\texttt{F}_{k} &= \sum_{s=0}^{k}{k \choose k-s}(\texttt{T}_s(\mathcal{F}_1)\texttt{T}_{k-s}(\mathcal{F}_2) + \texttt{F}_s(\mathcal{F}_1)\texttt{F}_{k-s}(\mathcal{F}_2)).
\end{array}
\]

\end{itemize}

The above four quantities are computed recursively in the function $Sh^{\NR}$ which is invoked in line~4 of Algorithm~\ref{algo:mc:algo1} (the pseudo code of $Sh^{\NR}$ can be found in Appendix~A).

Note that $\texttt{T}_k$ and $\texttt{F}_k$ are only used to compute $\texttt{A}_{k,i}$ and $\texttt{B}_{k,i}$, which are then used in lines \ref{nr1} and \ref{nr2} of Algorithm~\ref{algo:mc:algo1} to compute $\phi^{\NR}$. Having proved the correctness of computing $\texttt{A}_{k,i}$ and $\texttt{B}_{k,i}$, it remains to prove that the way they are used in  lines \ref{nr1} and \ref{nr2} correctly computes $\phi^{\NR}$. It suffices to prove this for a single rule $\mathcal{F} \to 1$, as the extension to multiple rules is straight forward in the MC-Net representation. To this end, some player $a_i$ can contribute to any ordered coalition $T$ counted in $\texttt{A}_{k,i}$, or $\texttt{B}_{k,i}$ (one of these quantities is zero), if it appears in a random permutation right after the players from $T$. The remaining players not occurring in $T$ can be arranged in $(m-k-1)!$ ways (we have $|X_{\mathcal{F}}|=m$). In case $a_i$ contributes to $\texttt{A}_{k,i}$, we have
\[
\textstyle\phi_i^{\NR} = \frac{1}{m!}\sum_{k=0}^{m-1}(m-k-1)!\texttt{A}_{k,i}.
\]
In the second case $a_i$ contributes to $\texttt{B}_{k,i}$; we get
\[
\textstyle\phi_i^{\NR} = -\frac{1}{m!}\sum_{k=0}^{m-1}(m-k-1)!\texttt{B}_{k,i}.
\]
These are exactly lines \ref{nr1} and \ref{nr2} of Algorithm~\ref{algo:mc:algo1}. This concludes the proof of correctness of Algorithm~\ref{algo:mc:algo1}. In terms of runtime, it is clear from the pseudo codes (in Algorithm~\ref{algo:mc:algo1} and Appendix~\ref{app:recursive}) that the total number of operations is polynomial in the number of agents. This concludes the proof of Theorem~\ref{theorem:algorithmNR}.
\end{proof}

%%%%%%%%%%%%%%%%%%%%%%%%%%%%%%%%%%%%%%%%%%%%%%%%%%%%%%%%%%%%%%%%%%%%%%%%%%%%%%%
\subsection{Computing the S\'{a}nchez-Berganti\~{n}os Value with Generalized Read-Once MC-Nets}\label{subsection:representation:SB}
\noindent In this subsection, we discuss the computational properties of the generalized read-once MC-Net representation
with respect to calculating the S\'{a}nchez-Berganti\~{n}os value (i.e., the SB value).
We start our computational analysis with some observations about
atomic formulas. First, observe that all players occurring in some
$\mathcal{F_{\mathit{BAF}}}$ are indistinguishable. Furthermore, taking into
account that each player can appear in at most one formula, $\phi^{\SB} $ will assign to these
players the same values. Similarly, for the case of $\mathcal{F_{\mathit{OAF}}}$, the SB value is the same for all
players due to the symmetry axiom (see Example \ref{Example:Difference}).

\noindent Note that the SB value involves a different notion of marginal contribution. Whereas computing quantities $\texttt{T}_k$ and $\texttt{F}_k$ remains the same, we need to redefine $\texttt{A}_{k,i}$ and $\texttt{B}_{k,i}$.

\[
\begin{array}{ll}
\texttt{A}_{k,i,l}(\mathcal{F})& = |T \in \mathcal{T}_{\mathcal{F}}: |T| = k, a_i \notin  T, T \nvDash \mathcal{F}, (T,a_i)^l \vDash \mathcal{F})|, \\
\texttt{B}_{k,i,l}(\mathcal{F})&  = |T \in \mathcal{T}_{\mathcal{F}}: |T| = k, a_i \notin  T, T \vDash \mathcal{F}, (T,a_i)^l \nvDash \mathcal{F})|.  \\
\end{array}
\]

The first quantity $\texttt{A}_{k,i,l}(\mathcal{F})$ is the number of ordered coalitions of a given size $k$ for which adding player $a_i$ on position $l$ causes $\mathcal{F}$ to become satisfied. $\texttt{B}_{k,i,l}$ is the number of coalitions in which adding the player $a_i$ on position $l$ causes $\mathcal{F}$ to become unsatisfied.

\begin{algorithm}[h]
\caption{\textbf{(COMP\_SB)} Computing the S\'{a}nchez-Berganti\~{n}os value}\label{algo:mc:algo2}
\SetAlgoVlined %The old version of this is: \SetVline
\LinesNumbered %The old version of this is: \linesnumbered
\KwIn {$N = \{a_1,\ldots,a_n\},(\mathcal{F}_1 \to V_1),\ldots,(\mathcal{F}_r \to V_r)$}
\KwOut { $\phi^{\SB}_i$ for all $a_i \in N$}
\lFor{$a_i \in N$} {$\phi^{\SB}_i \gets 0;$}
\For{$j \gets 1$ \KwTo $r$} {
	$m \gets |\mathcal{F}_j|;$ \\
   $(\texttt{A},\texttt{B},\texttt{T},\texttt{F}) \gets Sh^{\SB}(\mathcal{F}_j);$    \Comment{the Pseudo code of $Sh^{\SB}$ is in Appendix~\ref{app:recursive}}\\
	\ForEach{$a_i \in X_{\mathcal{F}_j}$} {
		\If{$\exists_{k,l}~\texttt{A}_{k,i,l} \neq 0$}{
			$v_i \gets \frac{1}{m!}\sum_{k=0}^{m-1}\frac{1}{k+1}\sum_{l=1}^{k+1} (m-k-1)!\texttt{A}_{k,i,l};$ \nllabel{sb1}
		}	
		\ElseIf{$\exists_{k,l}~\texttt{B}_{k,i,l} \neq 0$}{
			$v_i \gets -\frac{1}{m!}\sum_{k=0}^{m-1}\frac{1}{k+1}\sum_{l=1}^{k+1} (m-k-1)!\texttt{B}_{k,i,l};$ \nllabel{sb2}
		} \Else {
			$v_i \gets 0;$
		}
		$\phi^{\SB}_i \gets \phi^{\SB}_i + V_j \cdot v_i;$
	}	
}
\end{algorithm}

\begin{theorem}\label{theorem:algorithmSB}
Given a generalized coalitional game $(N,\mathcal{GR})$, Algorithm \ref{algo:mc:algo2} computes the S\'{a}nchez-Berganti\~{n}os value ($\phi^{\SB}_i$ for all $a_i \in N$) in polynomial time.
\end{theorem}

\begin{proof}
As for $\texttt{T}_k$ and $\texttt{F}_k$, the computation is exactly the same as in the proof of Theorem~\ref{theorem:algorithmNR}. Therefore, we only need to show how to calculate  $\texttt{A}_{k,i,l}$ and $\texttt{B}_{k,i,l}$ in polynomial time. We will start from atomic formulas:

\begin{itemize}
\item { \bf $ \mathcal{F} =  \{ a_1,\ldots, a_r \}$}: In this case $ \texttt{A}_{k,i,l} = (r-1)! $ if $k = r-1 $, $i \in X_{\mathcal{F}} $ and $l \in \{1,\ldots,r\}$, otherwise, $\texttt{A}_{k,i,l} = 0$. On the other hand, $\texttt{B}_{k,i,l}=0$ for all $k$, $i$ and $l$.
\item { \bf $ \mathcal{F} =  \langle a_1,\ldots, a_r \rangle $}: Here $ \texttt{A}_{k,i,l} = 1 $ if $k=r-1$, $i\in X_{\mathcal{F}}$ and $l=i$, otherwise $ \texttt{A}_{k,i,l} = 0 $. On the other hand, $\texttt{B}_{k,i,l}$ is always equal to $0$.
\end{itemize}

In the next part of proof we will use following proposition.
\begin{proposition}\label{proposition:permutation2}
 The number of conflations of two ordered coalitions $T_1$ and $T_2$
 such that $s$ players from $T_2$ go before the $l$-th member of
 $~T_1$ equals
 \[
 \textstyle{l+s-1 \choose s}{ {(|T_1|-l+1+1) + (|T_2|-s) -1} \choose {|T_2|-s}}.
   \]
\end{proposition}

\begin{proof}
Simply use combination with repetition twice.
\end{proof}

\noindent Next, we show how to compute $\texttt{A}_{k,i}$,$\texttt{B}_{k,i}$ for each of the possible logical connective (i.e., $\neg$, $\land$, $\lor$ and $\oplus$) separately. We will assume that $a_i \in X_{\mathcal{F}_1}$.

\begin{itemize}
\item { \bf $\mathcal{F} = \neg \mathcal{F}_1$}: Here, the negation causes a swap between sets. That is, we have 	$ \texttt{A}_{k,i,l}(\mathcal{F}) = \texttt{B}_{k,i,l}(\mathcal{F}_1) $ and $ \texttt{B}_{k,i,l}(\mathcal{F}) = \texttt{A}_{k,i,l}(\mathcal{F}_1)$ for all $k,i,l$.
\item { \bf $\mathcal{F} = \mathcal{F}_1 \land \mathcal{F}_2$}: Let $T=T_1 \times T_2$, where $T_1 \in \Pi(X_{\mathcal{F}_1}), |T_1| = s$ and $T_2 \in \Pi(X_{\mathcal{F}_2}),|T_2| = k-s$.	We have $(T,a_i)^l \vDash \mathcal{F}$ if and only if  $(T_1,a_i)^m \vDash \mathcal{F}_1$ for each $m \leq l$ when $l-m$ members of coalition $T_2$ go before the $m$-th member of coalition $T_1$ during the conflation operation and $T_2 \vDash \mathcal{F}_2$. In this case, $T \nvDash \mathcal{F}$ if and only if $T_1 \nvDash \mathcal{F}_1$. Consequently, using Proposition~\ref{proposition:permutation2} we get:

\[
	\textstyle \texttt{A}_{k,i,l} = \sum_{s=0}^{k}\sum_{m=1}^{s+1}{l-1 \choose l-m}{k-l+1 \choose k-l+m-s}  \texttt{A}_{s,i,m}(\mathcal{F}_1)\texttt{T}_{k-s}(\mathcal{F}_2).
\]

Furthermore, we have $T \vDash \mathcal{F}$ if and only if  $T_1 \vDash \mathcal{F}_1$ and  $T_2 \vDash \mathcal{F}_2$. Additionally in this case it holds that: $(T,a_i)^l \nvDash \mathcal{F} \Leftrightarrow (T_1,a_i)^m \nvDash \mathcal{F}_1$ for $m \leq l$, when $l-m$ members of coalition $T_2$ go before the $m$-th member of coalition $T_1$ during the conflation operation and $T_2   \nvDash \mathcal{F}_1 $. Again, using Proposition~\ref{proposition:permutation2} we get:,

\[
\textstyle \texttt{B}_{k,i,l} = \sum_{s=0}^{k}\sum_{m=1}^{s+1}{l-1 \choose l-m}{k-l+1 \choose k-l +m-s}  \texttt{B}_{s,i,m}(\mathcal{F}_1)\texttt{T}_{k-s}(\mathcal{F}_2).
\]

\item { \bf $\mathcal{F} = \mathcal{F}_1 \lor \mathcal{F}_2$}: Once again let $T=T_1 \times T_2$, where $T_1 \in \Pi(X_{\mathcal{F}_1}), |T_1| = s$ and $T_2 \in \Pi(X_{\mathcal{F}_2}),|T_2| = k-s$. We know that $T \nvDash \mathcal{F}$ if and only if $T_1 \nvDash \mathcal{F}_1$ and $T_2 \nvDash \mathcal{F}_2$. Furthermore, in this case it holds that $(T,a_i)^l \vDash \mathcal{F} \Leftrightarrow (T_1,a_i)^m \vDash \mathcal{F}_1 $ for $m \leq l$. Thus,

\[
		\textstyle \texttt{A}_{k,i,l} = \sum_{s=0}^{k}\sum_{m=1}^{s+1}{l-1 \choose l-m}{k-l+1 \choose k-l+m-s}  \texttt{A}_{s,i,m}(\mathcal{F}_1)\texttt{F}_{k-s}(\mathcal{F}_2).
\]

Analogously, we have $(T,a_i)^l \nvDash \mathcal{F}$ if and only if $(T_1,a_i)^m \nvDash \mathcal{F}_1$ for $m \leq l$ when $l-m$ members of coalition $T_2$ go before the $m$-th member of coalition $T_1$ during the conflation operation, and $T_2 \nvDash \mathcal{F}_2$. Furthermore, in this case it holds that: $T \vDash \mathcal{F} \Leftrightarrow T_1 \vDash \mathcal{F}_1  $. Consequently,
\[
	\textstyle \texttt{B}_{k,i,l} = \sum_{s=0}^{k}\sum_{m=1}^{s+1}{l-1 \choose l-m}{k-l+1 \choose k-l+m-s}  \texttt{B}_{s,i,m}(\mathcal{F}_1)\texttt{F}_{k-s}(\mathcal{F}_2).
\]

\item { \bf $\mathcal{F} = \mathcal{F}_1 \oplus \mathcal{F}_2$}: Let $T=T_1 \times T_2$ such that $T_1 \in \Pi(X_{\mathcal{F}_1}), |T_1| = s$ and $T_2 \in \Pi(X_{\mathcal{F}_2}),|T_2| = k-s$. Now, the ordered coalition $T$ can contribute to $\texttt{A}_{k,i,l}$ in two ways. Firstly, when $T_1 \vDash \mathcal{F}_1$ and $T_2 \vDash \mathcal{F}_2$ and $(T_1,a_i)^m \nvDash \mathcal{F}_1$. Secondly, when $T_1 \nvDash \mathcal{F}_1$ and $T_2 \nvDash \mathcal{F}_2$ and $(T_1,a_i)^m \vDash \mathcal{F}_1$, so
\[
		\textstyle \texttt{A}_{k,i,l} =  \sum_{s=0}^{k}\sum_{m=1}^{s+1}{l-1 \choose l-m}{k-l+1 \choose k-l+m-s}  \big( \texttt{B}_{s,i,m}(\mathcal{F}_1)\texttt{T}_{k-s}(\mathcal{F}_2) +
		\texttt{A}_{s,i,m}(\mathcal{F}_1)\texttt{F}_{k-s}(\mathcal{F}_2)\big).	
\]
\[			
		\textstyle \texttt{B}_{k,i,l} =\sum_{s=0}^{k}\sum_{m=1}^{s+1}{l-1 \choose l-m}{k-l+1 \choose k-l+m-s}  \big( \texttt{B}_{s,i,m}(\mathcal{F}_1)\texttt{F}_{k-s}(\mathcal{F}_2) + \\
		\texttt{A}_{s,i,m}(\mathcal{F}_1)\texttt{T}_{k-s}(\mathcal{F}_2)\big).	
\]

\end{itemize}

The above two quantities (together with $\texttt{T}_k$ and $\texttt{F}_k$ ) are computed recursively in the function $Sh^{\SB}$  which is invoked in line~4 of Algorithm~\ref{algo:mc:algo2} (the pseudo code of $Sh^{\SB}$ can be found in Appendix~\ref{app:recursive}).

Note that $\texttt{T}_k$ and $\texttt{F}_k$ are only used to compute $\texttt{A}_{k,i,l}$ and $\texttt{B}_{k,i,l}$, which are then used in lines \ref{sb1} and \ref{sb2} of Algorithm~\ref{algo:mc:algo2} to compute $Sh^{\SB}$. Having proved the correctness of computing $\texttt{A}_{k,i,l}$ and $\texttt{B}_{k,i,l}$, it remains to prove that the way they are used in  lines \ref{sb1} and \ref{sb2} correctly computes $Sh^{\SB}$. It suffices to prove this for a single rule $\mathcal{F} \to 1$, as the extension to multiple rules is straight forward in the MC-Net representation. To this end, some player $a_i$ can contribute to any ordered coalition $T$ counted in $\texttt{A}_{k,i,l}$, or $\texttt{B}_{k,i,l}$,  if it appears in a permutation right after the $(l-1)^{th}$ player from $T$. The remaining players (i.e., those not belonging to $T$) can be arranged in $(m-k-1)!$ ways (we have $|X_{\mathcal{F}}|=m$). In case $a_i$ contributes to $\texttt{A}_{k,i,l}$, we have
\[
\textstyle\phi_i^{\SB} = \frac{1}{m!}\sum_{k=0}^{m-1}\frac{1}{k+1}\sum_{l=1}^{k+1} (m-k-1)!\texttt{A}_{k,i,l}.
\]
In the second case, when $a_i$ contributes to $\texttt{B}_{k,i,l}$, we get
\[
\textstyle\phi_i^{\SB} = -\frac{1}{m!}\sum_{k=0}^{m-1}\frac{1}{k+1}\sum_{l=1}^{k+1} (m-k-1)!\texttt{B}_{k,i,l}.
\]
These are exactly lines \ref{sb1} and \ref{sb2} of Algorithm~\ref{algo:mc:algo2}. This concludes the proof of correctness of Algorithm~\ref{algo:mc:algo2}. In terms of runtime, it is clear from the pseudo codes (in Algorithm~\ref{algo:mc:algo2} and Appendix~\ref{app:recursive}) that the total number of operations is polynomial in the number of agents. This concludes the proof of Theorem~\ref{theorem:algorithmSB}.
\end{proof}

Finally, it is worth making a few remarks about the problem of \emph{generating} a representation of a game in our framework. We know that given a representation of a game in our framework, we can for example compute the NR and SB values in polynomial time. But how hard is it to come up with a representation of a game in our framework?

Suppose we are given a black box representation of a characteristic function, i.e., an oracle for the function that can be queried in unit time to determine the value of any ordered coalition. Then, it is not hard to see that it will require time exponential in the number of players to produce a representation of it in our framework. This is because we cannot avoid having to query the oracle for the value of every possible ordered coalition, and there are clearly exponentially many of them.

Another obvious question is whether other representations can easily be ``translated'' to our representation.
Here, we have better news. For example, a representation of a game in Ieong and Shoham's MC-Nets representation can be transformed into our representation in polynomial time: MC-Net rules correspond directly to rules in our framework requiring only basic atomic formulae (BAFs). In the same way, the subgraph representation introduced by \citet{Deng:Papdimitriou:1994} (which is, in fact, a special case of the MC-Net representation) can also be translated into our framework in polynomial time. Each edge of the graph translates to a rule with a BAF containing only two players.

It is worth noting that we are not aware of any other attempts in the literature to develop representations for generalized coalitional games. If such were developed, then it would be interesting to consider the problem of translating between representations.

%%%%%%%%%%%%%%%%%%%%%%%%%%%%%%%%%%%%%%%%%%%%%%%%%%%%%%%%%%%%%%%%%%%%%%%%%%%%

\section{Computing Generalized Game-Theoretic Betweenness Centrality}\label{subsection:representation:Centrality}

\noindent As mentioned in the introduction, it is possible to use the game-theoretic solution concepts, including the Nowak-Radzik and S\'{a}nchez-Berganti\~{n}os values, to measure centrality of nodes in networks~\citep{delPozo:et:al:2011}. In this section, we show how to extend the classic notion of betweenness centrality to generalized game-theoretic betweenness centrality which builds upon generalized characteristic function games and the NR and SB values. After that, we show that it is possible to use generalized read-once MC-Nets to compute the new centrality in time polynomial in the number of shortest paths in the network.

While betweenness centrality (Definition~\ref{def:cen:bet}) is focused on the role in the network played by individual nodes, and group betweenness centrality (Definition~\ref{def:group:bet}) is designed to measure the power of the set of nodes, the measure proposed by \citet{Puzis:et:al:2007} captures the notion of the sequences of nodes. Their measure, called \emph{path betweenness centrality}, evaluates sequences of nodes and counts the proportion of shortest paths that traverse along a given sequence of nodes in an appropriate order. Formally:

\begin{definition}[Path betweenness centrality]\label{def:path:bet}
Given a graph $G=(V,E)$, we define the set of all ordered coalitions of nodes in the graph, i.e., $\mathcal{T} = \bigcup_{C \in 2^V} \Pi(C)$. The path betweenness centrality of the ordered group of nodes $T$ is defined as a function $ c_{pb}: \mathcal{T} \rightarrow \mathbb{R}$ such that:
\begin{equation}
c_{pb}(T) =  \sum_{ \substack{s \notin T \\ t \notin T}} \frac{\sigma_{st}(T)}{\sigma_{st}},
\end{equation}
\noindent where $\sigma_{st}$ denotes the number of all shortest paths from $s$ to $t$ and $\sigma_{st}(T) $ is the number of shortest paths from $s$ to $t$ passing through all nodes from $T$ (if $s \in T$ or $t \in T$ then $ \sigma_{st}(T) = 0$) such that the order of visited nodes is maintained (for path $p$ we have $T \specialsubseteq{p}$).
\end{definition}

Similarly to other game-theoretic extension of classic centrality measures recently considered in the literature (see Section~\ref{sec:rw:gtc}), the above function, that assigns values to all ordered groups (coalitions) of nodes, can be considered to be a value function, i.e., in our case, the generalized characteristic function. Now, the NR and SB values computed for this function become the \textit{generalized game-theoretic betweenness centrality}.\footnote{\footnotesize We note that the game-theoretic betweenness centrality was introduced in Chapter~\ref{chap:bet} but this measure does not account for the order of players.}

\begin{definition}[Generalized representation function]\label{def:general:rep:fun}
A \emph{generalized representation function} is a function~$\psi^g : \mathscr{G} \to \mathscr{V}^{g}_{G} $ that maps every graph~$G=(V,E)$ onto a generalized cooperative game~$(NA,\nu^{g}_G) \in \mathscr{V}^{g}_{G}$ with~$A=V$.
\end{definition}

Now, solution concepts like Nowak-Radzik and S\'{a}nchez-Berganti\~{n}os values can be used in the network setting by applying them to the characteristic function that a network represents.

\begin{definition}[Generalized Game-theoretic centrality measure]\label{def:general:gt:cen}
Formally, we define a \textit{generalized game-theoretic centrality measure} as a pair~$(\psi^g, \phi)$ consisting of a representation function $\psi^g$ and a solution concept $\phi$.
\end{definition}

\begin{definition}[Generalized Game-theoretic Betweenness Centrality]\label{def:general:bet:gt:cen}
The  generalized game-theoretic betweenness centrality is a pair~$(\psi^g_B, \phi)$, where $\psi^g_B(G)=c_{pb}$ is the path betweenness centrality, the value of each node $v \in V$ is given by $\phi_v(c_{pb})$, and $\phi$ is the  Nowak-Radzik  value ($\phi^{\NR}$) or S\'{a}nchez-Berganti\~{n}os value ($\phi^{\SB}$).
\end{definition}

In the previous chapters we have shown that some game-theoretic centrality measures can be computed in polynomial time in the size of the network, i.e., in the number of nodes and edges. While in the literature there is no such a positive result for the generalized game-theoretic betweenness centrality, we will now show that it is possible to compute it in time polynomial in the number of shortest paths in the network and with the help of our generalized read-once MC-Nets.

\begin{algorithm}[t]
\caption{Computing the generalized game-theoretic betweenness centrality}\label{algo:path:betweenness}
\SetAlgoVlined %The old version of this is: \SetVline
\LinesNumbered %The old version of this is: \linesnumbered
\KwIn {Graph $G=(V,E)$, path betweenness centrality $c_{pb}$, game solution $\phi$ }
\KwOut { $\phi_v$ for all $v \in V$}
\lFor{$v \in V$} {$\phi_v \gets 0;$}
$\mathcal{GR} \gets \emptyset$; \Comment{initialize the set of generalized read-once formulas} \\
\For{$v \in V$} {
	\For{$u \in V$} {
		$\text{SP} \gets shortestPaths(v,u)$; \Comment{compute using a BFS algorithm all shortest paths from v to u} \\
	    \For{$(v,\ldots, u) \in \text{SP} $} {
			$\mathcal{F} \gets \langle v,\ldots, u \rangle $;\\
			$\mathcal{F} \gets \mathcal{F}  \wedge \neg \{V \setminus X_{\mathcal{F}} \}$;\\
			$V \gets c_{pb}((v,\ldots, u))$;\\
			$\mathcal{GR} \gets \mathcal{GR} \cup (\mathcal{F} \to V)$;
		}
	}
}
\If{$\phi == \phi^{\NR}$} {
$\text{COMP\_NR}(\mathcal{GR})$; \Comment{use Algorithm~\ref{algo:mc:algo1} }
} \ElseIf{$\phi == \phi^{\SB}$} {
$\text{COMP\_SB}(\mathcal{GR})$; \Comment{use Algorithm~\ref{algo:mc:algo2} }
}

\end{algorithm}
In particular, given a network  $G=(V,E)$, Algorithm~\ref{algo:path:betweenness} builds, for each shortest path in the network, a corresponding generalized read-once rule. More specifically, after initialization, we iterate over all shortest paths in the network. For each shortest path, in line 7, we initialize the formula $\mathcal{F}$ with ordered atomic formula corresponding to the sequence of all the nodes on the path. Next, in line 8, we add a conjunction that ensures that no other nodes can appear in an ordered coalition that satisfies our formula.

The size of the set $\mathcal{GR}$ constructed by this algorithm equals to the number of shortest paths in the graph and can be exponential in the number of vertices. However, in many  graphs we have $|\mathcal{GR}| \ll \mathcal{T}$, which makes this algorithm much more efficient than computing this measure directly using equation  \eqref{formula:NR} or equation \eqref{formula:SB} for $\phi^{\NR}$ and $\phi^{\SB}$, respectively.

\section{Conclusions}\label{Conclusions}

\noindent Generalized characteristic function games are attracting increasing interest in the literature due to their manifold potential applications. In this chapter we considered both the representation and computational aspects of these games. In particular, building upon the state-of-the-art read-once MC-Nets for the Shapley value by \citet{Elkind:et:al:2009}, we proposed the first computationally-convenient formalism to model games with ordered coalitions. Moreover, our Generalized Marginal-Contribution Nets representation allows to compute effectively the Nowak-Radzik and S\'{a}nchez-Berganti\~{n}os values. Finally, we proved the usefulness of such representation by constructing efficient algorithm to compute generalized game-theoretic betweenness centrality.

\chapter{Summary and Potential Extensions}

\noindent In this dissertation, we developed $10$ and analyzed additional $5$ efficient algorithms computing different coalitional game solutions on graphs (the most important of them were presented in Table~\ref{tab:final:summary}). These positive results in the domain of networks involved three fundamental centrality metrics: degree, closeness and betweenness. The game-theoretic solution concepts, such as the Shapley value, or Nowak-Radzik value, computed for these three measures gave us the new class of centrality metrics---the game-theoretic centrality measures.

In Chapter~\ref{chap:gtcm}, two algorithms were introduced: Algorithm~\ref{algo:gtc:game1} for Shapley value-based degree centrality and  Algorithm~\ref{algo:gtc:game4} for Shapley value-based closness centrality. Both algorithms can be easily applied to weighted and unweighted graphs. In this chapter we also presented three generalizations of these two measures. Additionally, Algorithm~\ref{algo:gtc:game1} was an inspiration for \cite{Szczepanski:et:al:2015:c}, who computed efficiently on networks Interaction Index based on degree centrality.

In the next Chapter~\ref{chap:bet} four algorithms were introduced to compute different version of game-theoretic betweenness centrality. Algorithm~\ref{algorithms:SV} and Algorithm~\ref{algorithms:SV:weighted} compute Shapley value-based betweenness centrality for weighted and unweighted networks, respectively. The another two Algorithms~\ref{algorithms:PSV} for unweighted and Algorithms~\ref{algorithms:PSV:weighted} for weighted graphs are an implementation of Semivalue-based betweenness centrality.

Owen value-based degree centrality was introduced in Chapter~\ref{chap:OV} together with two polynomial time algorithms. Firstly, in Algorithm~\ref{algo:semi} the general method for computing Coalitional Semivalue-based centrality measure was introduced. Secondly, in Algorithm~\ref{algo:OV} we introduced the faster method to compute Owen value-based degree centrality, which was an optimization of Algorithm~\ref{algo:semi} to calculate this particular Coalitional Semivalue.

The contribution of this dissertation not only consists of positive results. In particular, in Chapter~\ref{chap:connectivity} it was proven that computing Shapley value for connectivity games is \#P-hard. However, against this negative results the Algorithm~\ref{ter:algo:faster} was developed to compute Shapley value much faster than the brute force approach for connectivity games. This algorithm allows us to study much bigger real-life networks, and also was an inspiration for developing even faster methods to analyze connectivity games and graph-restricted games \citep{Skibski:et:al:2014,Michalak:et:al:2015b}.

In the last Chapter~\ref{chap:mcnets} the new representation of generalized coalitional games was introduced. Importantly, we developed two algorithms that in time polynomial in the size of the representation computes the Nowak-Radzik value (Algorithms~\ref{algo:mc:algo1}) and the S\'{a}nchez-Berganti\~{n}os value (Algorithms~\ref{algo:mc:algo2}). Based on this representation the Algorithm~\ref{algo:path:betweenness} was introduced to compute Nowak-Radzik value-based and S\'{a}nchez-Berganti\~{n}os value-based betweenness centrality.

Apart from computational results, it was experimentally shown that game-theoretic centrality measures were useful to analyze scale-free and real-life networks. In Chapter~\ref{chap:bet}, the extensive simulations showed that game-theoretic betweenness centrality was a better measure to protect network against simultaneous nodes failures than the standard betweenness centrality.

The experiments performed on scientific citation network in Chapter~\ref{chap:OV} suggested that Owen value-based weighted degree centrality could be successfully used to evaluate the significance of publications, while taking into account the importance of journals and conference venues. Speaking more generally, the importance of a single node stems not only from its individual power, but also from the power of community it belongs to. Owen value-based centrality measure is the first measure in literature that takes this fact into account.

In Chapter~\ref{chap:connectivity} a terrorist network was analyzed with a centrality build upon connectivity games. In particular, our algorithm proposed in this chapter allowed to analyze real-life WTC 9/11 terrorist network with 36 members and 125 identified connections. 

There are a number of directions for future work. Table~\ref{tab:final:summary} presents our main computational results vs. the gaps in the literature. For instance it is still an open question if it is possible to compute Owen value-based betweenness centrality in polynomial time? Such measure would find interesting applications to analyze financial institution networks, in which many institutions form strictly collaborating conglomerates. As can be seen from Table~\ref{tab:final:summary} there are more open questions regarding computation aspects of the game-theoretic centrality measures.

In this dissertation we analyzed three types of coalitional games: the standard games, games with coalitional structure, and generalized coalitional games. Nevertheless, there are still other classes of coalitional games, such as games with either positive or negative externalities \citep{Yi:1997,Michalak:et:al:2009}, that have been extensively studied in game theory and that may yield interesting results when applied to network centrality.

It is also interesting to analyse the properties of game-theoretic network centralities constructed on solution concepts from cooperative game theory other than the Shapley value, or Semivalues. In particular, more general solution concepts such as the core \citep{Osborne:Rubinstein:1994}  or the nucleolus \citep{Schmeidler:1969} could be also used as a centrality metric.

\renewcommand{\arraystretch}{1.3}
\setlength{\tabcolsep}{.4em}
\begin{table}[t!]
\begin{center}
\footnotesize
\begin{tabular} {c  c  c  c}
\hline
Standard centrality &  degree  centrality&  closeness centrality &  betweenness centrality \\
\hline\hline
Shapley value-based centrality      &   Algorithm~\ref{algo:gtc:game1}  &  Algorithm~\ref{algo:gtc:game4}  & Algorithm~\ref{algorithms:SV}   \\ \hline
Semivalue-based centrality   &    \multicolumn{2}{c}{\citep{Szczepanski:et:al:2015b} }  & Algorithm~\ref{algorithms:PSV} \\ \hline 
Owen value-based centrality & Algorithm~\ref{algo:OV} &  \citep{Tarkowski:et:al:2016} &   \textbf{unknown} \\ \hline
Coalitional Semivalue-based centrality & Algorithm~\ref{algo:semi} &  \citep{Tarkowski:et:al:2016} &  \textbf{unknown} \\ \hline
Nowak-Radzik-based centrality                    &  \textbf{unknown}   & \textbf{unknown} & Algorithm~\ref{algo:path:betweenness}$^{(*)}$ \\ \hline
S\'{a}nchez-Berganti\~{n}os-based centrality  & \textbf{unknown}    &  \textbf{unknown} & Algorithm~\ref{algo:path:betweenness}$^{(*)}$ \\\hline\hline
\end{tabular}
\begin{flushleft}
\scriptsize{$^{(*)}$This algorithm does not run in polynomial time.}
\end{flushleft}
\caption{Summary of the $7$ algorithms for game-theoretic centrality measures presented in this dissertation. The above table shows also gaps in the literature regarding other computational results.}
\label{tab:final:summary}
\end{center}
\end{table}
\renewcommand{\arraystretch}{1}

There are several other potential directions for future work.  Firstly, it would be interesting to extend the algorithms proposed in this thesis to streaming graphs \citep{Green:et:al:2012}. The challenge here is to efficiently re-compute the centralities and update the ranking of nodes after each change in the structure of the network, i.e., after adding/removing some nodes and/or edges. Secondly, we are keen to develop polynomial-time algorithms for the most specific variation of centralities analyzed in this dissertation, such as routing betweenness centrality \citep{Dolev:et:al:2010} built upon betweenness centrality. Finally, we would like to identify various axiomatic systems that each uniquely characterize our centrality measure. Such an axiomatic approach is common in the literature on cooperative game theory, and we believe it would be interesting to follow a similar approach to analyse centrality measures in general, and ours in particular.

Furthermore, our contribution on representation can be extended in various directions, two of them seem particularly interesting. On the one hand, it would be interesting to consider representation formalism other than MC-Nets, such as Algebraic Decision Diagrams \citep{Aadithya:et:al:2011,Sakurai:et:al:2011}, to model generalized characteristic function games. On the other hand, it is worth analyzing if our representation can produce algorithms that compute generalized game-theoretic centrality measures build upon degree and closness centralities.

Finally, it would be worth developing a more formal and general approach that would allow us to construct coalitional games defined over networks that correspond to other known centrality metrics or even entire families of them. Such an approach would involve developing a group centrality first and then building a characteristic function of a coalitional game upon it. Of course, while developing new centrality metrics based on coalitional games, one should keep in mind computational properties of the proposed solutions. Although we were able to obtain satisfactory computational results for a number of games, the computation of the game-theoretic network centrality may become much more challenging for more complex definitions of the characteristic function.

\appendix
\renewcommand\chaptername{Appendix}

\begin{appendices}
\makeatletter
\renewcommand*\FormatBlockHeading[1]{%
  \leftskip\@titleindent
  #1{\noindent
  \ifHeadingNumbered
    \@chapapp\ \mw@seccntformat\HeadingNumber
  \fi
  \ignorespaces\HeadingText\@@par}
  }
\makeatother

\chapter{Main Notation Used in the Dissertation}

\fontsize{9}{9}\selectfont{
 \renewcommand*{\arraystretch}{0.9}
\begin{longtable}[h!]{c c c}

%\hline
& & \\
\smallskip\hspace*{-0.2in} & \multicolumn{2}{c}{\textbf{Graphs/Networks}} \\
\hline \\
\smallskip\hspace*{-0.2in} &  $ V = \{v_1,v_2,\ldots\,v_{|V|}\} $ & \parbox{4.5 in}{The set of nodes.}\\
\smallskip\hspace*{-0.2in} &  $ E = \{e_1,e_2,\ldots\,e_{|E|}\} $ & \parbox{4.5 in}{The set of undirected (two-elements sets) or directed (two-elements ordered pairs) edges $e=(v_i,v_j)$.}\\
\smallskip\hspace*{-0.2in} &  $G=(V,E) $ & \parbox{4.5 in}{The simple graph, where $E$ is set of undirected edges.}\\
\smallskip\hspace*{-0.2in} &  $G=(V,E) $ & \parbox{4.5 in}{The directed graph, where $E$ is set of directed edges.}\\
\smallskip\hspace*{-0.2in} &  $ G =(V,E,\lambda) $ & \parbox{4.5 in}{The weighted graph with weight function: $ \lambda(v,u) : E \to \mathbb{R} $.}\\
\smallskip\hspace*{-0.2in} &  $ G =(V,E,\gamma) $ & \parbox{4.5 in}{The weighted graph with weights on nodes: $ \gamma : V \to \mathbb{R} $.}\\
\smallskip\hspace*{-0.2in} &  $N(v)$ & \parbox{4.5 in}{The neighbors of the node $v$, that is: $\set{u\midd (v,u)\in E}$.}\\
\smallskip\hspace*{-0.2in} &  $\mathit{deg}(v)$ & \parbox{4.5 in}{Degree of the node $v$, $ \mathit{deg}(v)= |N(v)|$.}\\
\smallskip\hspace*{-0.2in} &  $N(C)$ & \parbox{4.5 in}{The neighbors of the set of nodes $C \subseteq V$, i.e. $\bigcup_{v_i \in C}N(v_i)\setminus C$.}\\
\smallskip\hspace*{-0.2in} &  $E(C)$ & \parbox{4.5 in}{The set of edges within set of nodes $C$: $E(C) = \set{(u,v) \in E : u,v \in C}$.}\\
\smallskip\hspace*{-0.2in} &  $\mathit{deg}(C)$ & \parbox{4.5 in}{The degree of the set $C$, i.e. $\mathit{deg}(C) = |N(C)|$.}\\
\smallskip\hspace*{-0.2in} &  $\mathit{dist}(v, u)$ & \parbox{4.5 in}{The distance between two nodes: $v$ and $u$.}\\
\smallskip\hspace*{-0.2in} &  $\mathit{dist}(C, v)$ & \parbox{4.5 in}{The distance between two nodes: $v$ and set of nodes $C$: $\mathit{dist}(C, v) = \min_{u \in C} \textit{dist}(u, v)$.}\\
\smallskip\hspace*{-0.2in} &  $c_d, c_b, c_c: V \to \mathbb{R}$ & \parbox{4.5 in}{The degree, betweeness and closeness centrality, respectively.}\\
\smallskip\hspace*{-0.2in} &  $c_{gd}, c_{gb}, c_{gc}: 2^V \to \mathbb{R}$ & \parbox{4.5 in}{The group degree, group betweeness and group closeness centrality, respectively.}\\
\smallskip\hspace*{-0.2in} &  $C \subseteq V$ & \parbox{4.5 in}{A community $C$.}\\
\smallskip\hspace*{-0.2in} &  $\mathit{CS} = \{C_1,C_2,\ldots\,C_{|\mathit{CS}|}\} $ & \parbox{4.5 in}{The community structure.} \\
\smallskip\hspace*{-0.2in} &  $ \sigma_{st} $ & \parbox{4.5 in}{ The number of shortest paths between $ s $ and $ t $.} \\
\smallskip\hspace*{-0.2in} &  $ \sigma_{st}(v) $ & \parbox{4.5 in}{ The number of shortest paths between $ s $ and $ t $ passing through the node $v$.} \\
\smallskip\hspace*{-0.2in} &  $ \sigma_{st}(C) $ & \parbox{4.5 in}{ The number of shortest paths between $ s $ and $ t $ passing through at least one node from $C$.} \\
\smallskip\hspace*{-0.2in} & \multicolumn{2}{c}{\textbf{Coalitional Games}} \\
\hline \\
\smallskip\hspace*{-0.2in} &  $A = \{a_1,a_2,\ldots\,a_{|A|}\}$ & \parbox{4.5 in}{The set of players/agents.}\\
\smallskip\hspace*{-0.2in} &  $\nu: 2^A \rightarrow \mathbb{R}$ & \parbox{4.5 in}{The characteristic function.}\\
\smallskip\hspace*{-0.2in} &  $g=(A,\nu)$ & \parbox{4.5 in}{A coalitional game $g$ in a characteristic function form.}\\
\smallskip\hspace*{-0.2in} &  $C \subseteq A$ & \parbox{4.5 in}{A coalition $C$.}\\
\smallskip\hspace*{-0.2in} &  $\Pi(C)$ & \parbox{4.5 in}{The set of all orders (permutations) of players in $C$.}\\
\smallskip\hspace*{-0.2in} &  $\MC{C}{a_i}$ & \parbox{4.5 in}{The marginal contribution of player $a_i$ to coalition $C$.}\\
\smallskip\hspace*{-0.2in} &  $\phi_i(A,\nu)$ & \parbox{4.5 in}{A solution concept for a player $a_i$ in coalitional game $(A,\nu)$.}\\
\smallskip\hspace*{-0.2in} &  $\phi_i(A)$ & \parbox{4.5 in}{Shorter notation for $\phi_i(A,\nu)$.}\\
\smallskip\hspace*{-0.2in} &  $\phi^{\SV}_i(A,\nu)$ & \parbox{4.5 in}{The Shapley value for a player $a_i$ in coalitional game $(A,\nu)$.}\\
\smallskip\hspace*{-0.2in} &  $\phi^{\SEMI}_i(A,\nu)$ & \parbox{4.5 in}{The Semivalue for a player $a_i$ in coalitional game $(A,\nu)$.}\\
\smallskip\hspace*{-0.2in} &  $\psi : \mathscr{G} \to \mathscr{V} $ & \parbox{4.5 in}{The representation function ($\mathscr{G}$ is the set of all graphs, and $\mathscr{V}$ the set of all coalitional games).}\\
\smallskip\hspace*{-0.2in} &  $ (A,\mathcal{R}) $ & \parbox{4.5 in}{The coalitional game in the form of Marginal Contribution Network.}\\
\smallskip\hspace*{-0.2in} &  $\delta^i_j$ & \parbox{4.5 in}{The difference between $\phi_i(N) - \phi_{i}(N_{-j})$ and $\phi_j(N) - \phi_j(N_{-i})$.}\\
\smallskip\hspace*{-0.2in} & \multicolumn{2}{c}{\textbf{Coalitional Games with Coalitional Strcuture}} \\
\hline \\
\smallskip\hspace*{-0.2in} &  $\mathit{CS} = \{C_1,C_2,\ldots\,C_{|\mathit{CS}|}\} $ & \parbox{4.5 in}{The coalitional structure.} \\
\smallskip\hspace*{-0.2in} &  $\phi^{\OV}_i(A,\nu, \mathit{CS})$ & \parbox{4.5 in}{The Owen value for a player $a_i$ in coalitional game $(A,\nu, \mathit{CS})$ with coalitional structure $\mathit{CS}$.} \\
\smallskip\hspace*{-0.2in} &  $\phi^{\CSEMI}_i(A,\nu, \mathit{CS})$ & \parbox{4.5 in}{The Coalitional Semivalue for a player $a_i$ in coalitional game $(A,\nu, \mathit{CS})$ with coalitional structure $\mathit{CS}$.} \\
\smallskip\hspace*{-0.2in} & \multicolumn{2}{c}{\textbf{Generalized Coalitional Games}} \\
\hline \\
\smallskip\hspace*{-0.2in} &  $T$ & \parbox{4.5 in}{An ordered coalition.}\\
\smallskip\hspace*{-0.2in} & $\mathcal{T} = \bigcup_{C \in 2^N} \Pi(C)$ & \parbox{4.5 in}{The set of all ordered coalitions.}\\
\smallskip\hspace*{-0.2in} & $\nu^g : \mathcal{T} \rightarrow\R$ & \parbox{4.5 in}{ A generalized characteristic function.}\\
\smallskip\hspace*{-0.2in} & $g^g=(A,\nu^g)$ & \parbox{4.5 in}{ A generalized coalitional game in characteristic function form.}\\
\smallskip\hspace*{-0.2in} &  $(T,T')^k$ & \parbox{4.5 in}{The ordered coalition that results from inserting $T'$ at the $k^{th}$ position in $T$.}\\
\smallskip\hspace*{-0.2in} &  $(T,a_i)$ & \parbox{4.5 in}{the ordered coalition that results from inserting $a_i$ at the end of $T$.}\\
\smallskip\hspace*{-0.2in} &  $(A,\nu^g)$ & \parbox{4.5 in}{A coalitional game in a generalized characteristic function form.}\\
\smallskip\hspace*{-0.2in} &  $T(a_i)$ & \parbox{4.5 in}{An ordered coalition made of all predecessors of player $a_i$ in $T$.}\\
\smallskip\hspace*{-0.2in} &  $\phi^{\SB}_i(N,\nu^g)$ & \parbox{4.5 in}{The S\'{a}nchez-Berganti\~{n}os value of player $a_i$ in game $(N,\nu^g)$.}\\
\smallskip\hspace*{-0.2in} &  $\phi^{\NR}_i(N,\nu^g)$ & \parbox{4.5 in}{The Nowak-Radzik value of player $a_i$ in game $(N,\nu^g)$.}\\
\smallskip\hspace*{-0.2in} &  $v^{\NR}_g(C)$ & \parbox{4.5 in}{The Nowak-Radzik value of coalition $C$ (i.e., $\frac{1}{|C|!}\cdot \sum_{T \in \Pi(C).} \nu^g(T)$).}\\
\smallskip\hspace*{-0.2in} &  $\MCNR{T}{i}$ & \parbox{4.5 in}{$\nu^g((T,a_i)) - \nu^g(T)$.}\\
\smallskip\hspace*{-0.2in} &  $\MCSB{T}{i}$ & \parbox{4.5 in}{$\frac{1}{(|T|+1)} \sum_{l=1}^{|T|+1} [\nu((T,a_i)^l) - \nu(T)] $.}\\
\smallskip\hspace*{-0.2in} &   $(N,\mathcal{GR})$& \parbox{4.5 in}{generalized coalitional game represented by generalized MC-net}\\
\smallskip\hspace*{-0.2in} & \multicolumn{2}{c}{\textbf{Game-theoretic Network Centralities}} \\
\hline \\
\smallskip\hspace*{-0.2in} &  $\nu_{G}: 2^V \rightarrow \mathbb{R}$ & \parbox{4.5 in}{The characteristic function defined on graph.}\\
\smallskip\hspace*{-0.2in} &  $g=(N,\nu_{G})$ & \parbox{4.5 in}{A coalitional game $g$ in a characteristic function form defined on graph.}\\
\smallskip\hspace*{-0.2in} & \multicolumn{2}{c}{\textbf{Other}} \\
\hline \\
\smallskip\hspace*{-0.2in} &  $\mathbb{E}[\cdot]$ & \parbox{4.5 in}{The expectation operator.}\\
\smallskip\hspace*{-0.2in} &  $(z \wedge \lnot z') \rightarrow V$ & \parbox{4.5 in}{Basic MC-Nets $z$ positive literals, $z'$ positive literals,.}\\
\smallskip\hspace*{-0.2in} &  $\mathcal{F}$ & \parbox{4.5 in}{Formula}\\
\smallskip\hspace*{-0.2in} &  $\mathcal{F}_{BAF}=\{ a_m,\ldots,a_n \}$ & \parbox{4.5 in}{basic atomic formula (BAF)}\\
\smallskip\hspace*{-0.2in} &  $\mathcal{F}_{OAF}=\langle a_i,\ldots, a_j \rangle$ & \parbox{4.5 in}{ordered atomic formula (OAF)}\\
\smallskip\hspace*{-0.2in} &  $\langle a^1_i,\ldots, a^f_j\rangle$ & \parbox{4.5 in}{upper indices indicate positions in formula/ordered coalition}\\
\smallskip\hspace*{-0.2in} &  $\mathcal{F}_1 \oplus \mathcal{F}_2$ & \parbox{4.5 in}{exclusive disjunction}\\
\smallskip\hspace*{-0.2in} &  $T \vDash \mathcal{F}$ & \parbox{4.5 in}{ordered coalition $T$ meets or satisfies formula $\mathcal{F}$}\\
\smallskip\hspace*{-0.2in} &  $T \nvDash \mathcal{F}$ & \parbox{4.5 in}{ordered coalition $T$ does not meet or satisfy formula $\mathcal{F}$}\\
\smallskip\hspace*{-0.2in} &   $X_{\mathcal{F}}$ & \parbox{4.5 in}{set of players occurring in $\mathcal{F}$}\\
\smallskip\hspace*{-0.2in} &   $\mathcal{T}_{\mathcal{F}}$ & \parbox{4.5 in}{$\bigcup\{\Pi(C):C \subseteq X_{\mathcal{F}}\}$}\\
\smallskip\hspace*{-0.2in} &   $T_1 \times T_2$ & \parbox{4.5 in}{conflation of two ordered coalitions }\\
\smallskip\hspace*{-0.2in} &   $\texttt{A}_{k,i}(\mathcal{F})$& \parbox{4.5 in}{$|\{T \in \mathcal{T}_{\mathcal{F}}: |T| = k, a_i \notin  T, T \nvDash \mathcal{F}, (T,a_i) \vDash \mathcal{F})\}| $}\\
\smallskip\hspace*{-0.2in} &   $\texttt{B}_{k,i}(\mathcal{F})$ & \parbox{4.5 in}{$  |\{T \in \mathcal{T}_{\mathcal{F}}: |T| = k, a_i \notin  T, T \vDash \mathcal{F}, (T,a_i) \nvDash \mathcal{F})\}| $}\\
\smallskip\hspace*{-0.2in} &   $\texttt{T}_{k}(\mathcal{F})$ & \parbox{4.5 in}{$|\{T \in \mathcal{T}_{\mathcal{F}}: |T| = k, T \vDash \mathcal{F}\}| $}\\
\smallskip\hspace*{-0.2in} &   $\texttt{F}_{k}(\mathcal{F})$ & \parbox{4.5 in}{$|\{T \in \mathcal{T}_{\mathcal{F}}: |T| = k, T \nvDash \mathcal{F}\}| $}\\
\smallskip\hspace*{-0.2in} &   $\texttt{T}_{k}(\mathcal{F})$ & \parbox{4.5 in}{$ |\{T \in \mathcal{T}_{\mathcal{F}}: |T| = k, a_i \notin  T, T \nvDash \mathcal{F}, (T,a_i)^l \vDash \mathcal{F})\}|  $}\\
\smallskip\hspace*{-0.2in} &   $\texttt{B}_{k,i,l}(\mathcal{F})$ & \parbox{4.5 in}{$|\{T \in \mathcal{T}_{\mathcal{F}}: |T| = k, a_i \notin  T, T \vDash \mathcal{F}, (T,a_i)^l \nvDash \mathcal{F})\}|$}\\
%\\&&\\&&\\&&\\&&\\&&\\&&\\&&\\&&\\&&\\&&\\&&\\&&\\&&\\&&\\&&\\&&
\end{longtable}
}

\chapter{The Recursive Functions}\label{app:recursive}

\begin{algorithm}[h]
\label{algo1_recursive}
{\fontsize{7}{7}\selectfont
\caption{$Sh^{\NR}(\mathcal{F})$}
\KwIn {Formula $\mathcal{F}$}
\KwOut {Quantities  $(\texttt{A},\texttt{B},\texttt{T},\texttt{F})$}

\Begin{
\If{$ \mathcal{F} =  \{ a_1,\ldots, a_r \}$} {
	$(\texttt{A},\texttt{B},\texttt{T},\texttt{F}) \gets (0,0,0,0);$ \\
	\For{$k \in \{0,\ldots,r\} $} {
		\For{$i \in X_{\mathcal{F}} $} {
			\lIf{$k = r-1 $} {
				$ \texttt{A}_{k,i} \gets (r-1)! $
			}
		}
		\lIf{$k = r $}{
			$ \texttt{T}_{k} \gets r! $
		}
		$ \texttt{F}_{k} \gets  k! - \texttt{T}_k $;
	}
	\vspace*{-0.1cm}
}
\vspace*{-0.1cm}
\ElseIf{$ \mathcal{F} =  \langle a_1,\ldots, a_r \rangle $} {
	$(\texttt{A},\texttt{B},\texttt{T},\texttt{F}) \gets (0,0,0,0);$ \\
	\For{$k \in \{0,\ldots,r\} $} {
			\lIf{$k=r-1$} {
				$ \texttt{A}_{k,r} \gets 1 $
			}
		\lIf{$k = r $}{
			$ \texttt{T}_{k} \gets 1 $
		}
		$ \texttt{F}_{k} \gets  k! - \texttt{T}_k $;
	}
	\vspace*{-0.1cm}
}
\vspace*{-0.1cm}
\ElseIf{$ \mathcal{F} =   \neg \mathcal{F}_1 $} {
	$(\texttt{A}^1,\texttt{B}^1,\texttt{T}^1,\texttt{F}^1) \gets Sh^{\NR}(\mathcal{F}_1);$ \\
	$(\texttt{A},\texttt{B},\texttt{T},\texttt{F}) \gets (\texttt{B}^1,\texttt{A}^1,\texttt{F}^1,\texttt{T}^1) ;$ \\
}
\vspace*{-0.1cm}
\ElseIf (\Comment{We assume that $\forall_{i \in X_{\mathcal{F}}} i \in \mathcal{F}_1$, the second situation is analogous.}){$ \mathcal{F} = \mathcal{F}_1 \land \mathcal{F}_2$} {
	$(\texttt{A}^1,\texttt{B}^1,\texttt{T}^1,\texttt{F}^1) \gets Sh^{\NR}(\mathcal{F}_1);$ \\
	$(\texttt{A}^2,\texttt{B}^2,\texttt{T}^2,\texttt{F}^2) \gets Sh^{\NR}(\mathcal{F}_2);$ \\
	\For{$k \in \{0,\ldots,r\} $} {		
		\For{$i \in X_{\mathcal{F}} $} {
			$\texttt{A}_{k,i} \gets \sum_{s=0}^{k}{k \choose k-s}\texttt{A}^1_{s,i}\texttt{T}^2_{k-s}$ \\
			$\texttt{B}_{k,i} \gets \sum_{s=0}^{k}{k \choose k-s}\texttt{B}_{s,i}^1\texttt{T}^2_{k-s} $\\
		}
		$\texttt{T}_{k} \gets \sum_{s=0}^{k}{k \choose k-s}\texttt{T}^1_{s}\texttt{T}^2_{k-s} $\\
		$\texttt{F}_{k} \gets \sum_{s=0}^{k}{k \choose k-s}(\texttt{F}^1_s\texttt{F}^2_{k-s} + \texttt{F}^1_s\texttt{T}^2_{k-s} + \texttt{T}^1_s\texttt{F}^2_{k-s} ) $ \\
	}
	\vspace*{-0.1cm}
}
\vspace*{-0.1cm}
\ElseIf{$ \mathcal{F} = \mathcal{F}_1 \lor \mathcal{F}_2$} {
	$(\texttt{A}^1,\texttt{B}^1,\texttt{T}^1,\texttt{F}^1) \gets Sh^{\NR}(\mathcal{F}_1);$ \\
	$(\texttt{A}^2,\texttt{B}^2,\texttt{T}^2,\texttt{F}^2) \gets Sh^{\NR}(\mathcal{F}_2);$ \\
	\For{$k \in \{0,\ldots,r\} $} {
		\For{$i \in X_{\mathcal{F}} $} {
			$\texttt{A}_{k,i} \gets \sum_{s=0}^{k}{k \choose k-s}\texttt{A}^1_{s,i}\texttt{F}^2_{k-s}$ \\
			$\texttt{B}_{k,i} \gets \sum_{s=0}^{k}{k \choose k-s}\texttt{B}_{s,i}^1\texttt{F}^2_{k-s} $\\
		}
		$\texttt{T}_{k} \gets \sum_{s=0}^{k}{k \choose k-s}(\texttt{T}^1_s\texttt{T}^2_{k-s} + \texttt{F}^1_s\texttt{T}^2_{k-s} + \texttt{T}^1_s\texttt{F}^2_{k-s} ) $ \\
		$\texttt{F}_{k} \gets \sum_{s=0}^{k}{k \choose k-s}\texttt{F}^1_{s}\texttt{F}^2_{k-s} $\\
	}
	\vspace*{-0.1cm}
}
\vspace*{-0.1cm}
 \ElseIf { $\mathcal{F} = \mathcal{F}_1 \oplus \mathcal{F}_2$} {
	$(\texttt{A}^1,\texttt{B}^1,\texttt{T}^1,\texttt{F}^1) \gets Sh^{\NR}(\mathcal{F}_1);$ \\
	$(\texttt{A}^2,\texttt{B}^2,\texttt{T}^2,\texttt{F}^2) \gets Sh^{\NR}(\mathcal{F}_2);$ \\
	\For{$k \in \{0,\ldots,r\} $} {	
		\For{$i \in X_{\mathcal{F}} $} {
			$\texttt{A}_{k,i} \gets \sum_{s=0}^{k}{k \choose k-s}(\texttt{B}^1_{s,i}\texttt{T}^2_{k-s} + \texttt{A}^1_{s,i}\texttt{F}^2_{k-s});$\\
            $\texttt{B}_{k,i} \gets \sum_{s=0}^{k}{k \choose k-s}(\texttt{B}^1_{s,i}\texttt{F}^2_{k-s} + \texttt{A}^1_{s,i}\texttt{T}^2_{k-s});$ \\
        }
		$\texttt{T}_{k} \gets \sum_{s=0}^{k}{k \choose k-s}(\texttt{T}^1_s\texttt{F}^2_{k-s}+ \texttt{F}^1_s\texttt{T}^2_{k-s});$ \\
		$\texttt{F}_{k} \gets \sum_{s=0}^{k}{k \choose k-s}(\texttt{T}^1_s\texttt{T}^2_{k-s} + \texttt{F}^1_s\texttt{F}^2_{k-s});$ \\
	}
	\vspace*{-0.1cm}
}
\vspace*{-0.1cm}
}
}
\end{algorithm}

\begin{algorithm}[h]
\label{algo2_recursive}
{\fontsize{8}{8}\selectfont
\caption{$Sh^{\SB}(\mathcal{F})$}
\KwIn {Formula $\mathcal{F}$}
\KwOut {Quantities  $(\texttt{A},\texttt{B},\texttt{T},\texttt{F})$}

\Begin{

\If{$ \mathcal{F} =  \{ a_1,\ldots, a_r \}$} {
	$(\texttt{A},\texttt{B},\texttt{T},\texttt{F}) \gets (0,0,0,0);$ \\
	\For{$k \in \{0,\ldots,r\} $} {
		\For{$i \in X_{\mathcal{F}},~l \in \{1,\ldots,r\}  $} {
			\lIf{$k = r-1 $} {
				$ \texttt{A}_{k,i,l} \gets (r-1)! $
			}
		}
		\lIf{$k = r $}{
			$ \texttt{T}_{k} \gets r! $
		}
		$ \texttt{F}_{k} \gets  k! - \texttt{T}_k $;
	}
	\vspace*{-0.1cm}
}
\vspace*{-0.1cm}
\ElseIf{$ \mathcal{F} =  \langle a_1,\ldots, a_r \rangle $} {
	$(\texttt{A},\texttt{B},\texttt{T},\texttt{F}) \gets (0,0,0,0);$ \\
	\For{$k \in \{0,\ldots,r\} $} {
		\For{$i \in X_{\mathcal{F}} $} {
			\lIf{$k=r-1$} {
				$ \texttt{A}_{k,i,i} \gets 1 $
			}
		}
		\lIf{$k = r $}{
			$ \texttt{T}_{k} \gets 1 $
		}
		$ \texttt{F}_{k} \gets  k! - \texttt{T}_k $;
	}
	\vspace*{-0.1cm}
}
\vspace*{-0.1cm}
\ElseIf{$ \mathcal{F} =   \neg \mathcal{F}_1 $} {
	$(\texttt{A}^1,\texttt{B}^1,\texttt{T}^1,\texttt{F}^1) \gets Sh^{\SB}(\mathcal{F}_1);$ \\
	$(\texttt{A},\texttt{B},\texttt{T},\texttt{F}) \gets (\texttt{B}^1,\texttt{A}^1,\texttt{F}^1,\texttt{T}^1) ;$ \\
}
\vspace*{-0.1cm}
\ElseIf(\Comment{We assume that $\forall_{i \in X_{\mathcal{F}}} i \in \mathcal{F}_1$, the second situation is analogous.}){$ \mathcal{F} = \mathcal{F}_1 \land \mathcal{F}_2$} {
	$(\texttt{A}^1,\texttt{B}^1,\texttt{T}^1,\texttt{F}^1) \gets Sh^{\SB}(\mathcal{F}_1);$ \\
	$(\texttt{A}^2,\texttt{B}^2,\texttt{T}^2,\texttt{F}^2) \gets Sh^{\SB}(\mathcal{F}_2);$ \\
	\For{$k \in \{0,\ldots,r\} $} {		
		\For{$i \in X_{\mathcal{F}},~l \in \{1,\ldots,r\}  $} {
			$ \texttt{A}_{k,i,l} \gets \sum_{s=0}^{k}\sum_{m=1}^{s+1}{l-1 \choose l-m}{k-l+1 \choose k-l+m-s}  \texttt{A}^1_{s,i,m}\texttt{T}^2_{k-s};$ \\
			$\texttt{B}_{k,i,l} \gets \sum_{s=0}^{k}\sum_{m=1}^{s+1}{l-1 \choose l-m}{k-l+1 \choose k-l +m-s}  \texttt{B}^1_{s,i,m}\texttt{T}^2_{k-s};$\\
		}
		$\texttt{T}_{k} \gets \sum_{s=0}^{k}{k \choose k-s}\texttt{T}^1_{s}\texttt{T}^2_{k-s} ;$\\
		$\texttt{F}_{k} \gets \sum_{s=0}^{k}{k \choose k-s}(\texttt{F}^1_s\texttt{F}^2_{k-s} + \texttt{F}^1_s\texttt{T}^2_{k-s} + \texttt{T}^1_s\texttt{F}^2_{k-s} ) ;$ \\
	}
	\vspace*{-0.1cm}
}
\vspace*{-0.1cm}
\ElseIf{$ \mathcal{F} = \mathcal{F}_1 \lor \mathcal{F}_2$} {
	$(\texttt{A}^1,\texttt{B}^1,\texttt{T}^1,\texttt{F}^1) \gets Sh^{\SB}(\mathcal{F}_1);$ \\
	$(\texttt{A}^2,\texttt{B}^2,\texttt{T}^2,\texttt{F}^2) \gets Sh^{\SB}(\mathcal{F}_2);$ \\
	\For{$k \in \{0,\ldots,r\} $} {
		\For{$i \in X_{\mathcal{F}},~l \in \{1,\ldots,r\}  $} {
			$\texttt{A}_{k,i,l} \gets \sum_{s=0}^{k}\sum_{m=1}^{s+1}{l-1 \choose l-m}{k-l+1 \choose k-l+m-s}  \texttt{A}^1_{s,i,m}\texttt{F}^2_{k-s};$ \\
			$\texttt{B}_{k,i,l} \gets \sum_{s=0}^{k}\sum_{m=1}^{s+1}{l-1 \choose l-m}{k-l+1 \choose k-l+m-s}  \texttt{B}^1_{s,i,m}\texttt{F}^2_{k-s};$\\
		}
		$\texttt{T}_{k} \gets \sum_{s=0}^{k}{k \choose k-s}(\texttt{T}^1_s\texttt{T}^2_{k-s} + \texttt{F}^1_s\texttt{T}^2_{k-s} + \texttt{T}^1_s\texttt{F}^2_{k-s} ) $ \\
		$\texttt{F}_{k} \gets \sum_{s=0}^{k}{k \choose k-s}\texttt{F}^1_{s}\texttt{F}^2_{k-s} $\\
	}
	\vspace*{-0.1cm}
}
\vspace*{-0.1cm}
\ElseIf{ $\mathcal{F} = \mathcal{F}_1 \oplus \mathcal{F}_2$} {
	$(\texttt{A}^1,\texttt{B}^1,\texttt{T}^1,\texttt{F}^1) \gets Sh^{\SB}(\mathcal{F}_1);$ \\
	$(\texttt{A}^2,\texttt{B}^2,\texttt{T}^2,\texttt{F}^2) \gets Sh^{\SB}(\mathcal{F}_2);$ \\
	\For{$k \in \{0,\ldots,r\} $} {	
		\For{$i \in X_{\mathcal{F}},~l \in \{1,\ldots,r\}  $} {
			$\texttt{A}_{k,i,l} \gets  \sum_{s=0}^{k}\sum_{m=1}^{s+1}{l-1 \choose l-m}{k-l+1 \choose k-l+m-s} ( \texttt{B}^1_{s,i,m}\texttt{T}^2_{k-s} + \texttt{A}^1_{s,i,m}\texttt{F}^2_{k-s});$\\
            $\texttt{B}_{k,i,l} \gets \sum_{s=0}^{k}\sum_{m=1}^{s+1}{l-1 \choose l-m}{k-l+1 \choose k-l+m-s}  \big( \texttt{B}^1_{s,i,m}\texttt{F}^2_{k-s} + \texttt{A}^1_{s,i,m}\texttt{T}^2_{k-s} \big);$ \\
        }
		$\texttt{T}_{k} \gets \sum_{s=0}^{k}{k \choose k-s}(\texttt{T}^1_s\texttt{F}^2_{k-s}+ \texttt{F}^1_s\texttt{T}^2_{k-s});$ \\
		$\texttt{F}_{k} \gets \sum_{s=0}^{k}{k \choose k-s}(\texttt{T}^1_s\texttt{T}^2_{k-s} + \texttt{F}^1_s\texttt{F}^2_{k-s});$ \\
	}
	\vspace*{-0.1cm}
}
\vspace*{-0.1cm}
}
}
\end{algorithm}

{\color{white}.}\\
{\color{white}.}\\
{\color{white}.}\\

\end{appendices}

\bibliography{references}
%\nocite{*}
\bibliographystyle{plainnat}

\end{document}